\documentclass[10pt]{article}
\usepackage{amsmath,amssymb,amsfonts,amsthm}
\usepackage{mathrsfs}
\usepackage{indentfirst}
\usepackage[pdftex]{color,graphicx}
\usepackage[pdftex,bookmarks,unicode,colorlinks]{hyperref}
\usepackage{enumerate} 
\usepackage[numbers]{natbib} 
\usepackage{yfonts} 
\usepackage{stmaryrd} 
\usepackage[title]{appendix} 
\usepackage{supertabular} 
\usepackage{caption} 

\usepackage{longtable}

\usepackage{tikz-cd}
\usetikzlibrary{matrix,arrows,decorations.pathmorphing}

\makeatletter
\let\@fnsymbol\@arabic 
\makeatother

\allowdisplaybreaks[4] 

\theoremstyle{plain}
\newtheorem{theorem}{Theorem}[section]
\newtheorem{proposition}[theorem]{Proposition}
\newtheorem{corollary}[theorem]{Corollary}
\newtheorem{lemma}[theorem]{Lemma}

\theoremstyle{remark}
\newtheorem{remark}[theorem]{\underline{Remark}}
\newtheorem{example}[theorem]{\underline{Example}}

\theoremstyle{definition}
\newtheorem{definition}[theorem]{Definition}

\newcommand{\E}{\mathbf E}
\renewcommand{\P}{\mathbf P}

\newcommand{\tr}{\mathrm{tr}\,}
\newcommand{\R}{\mathbb R}
\newcommand{\N}{\mathbb N}

\newcommand{\C}{\mathcal C}
\newcommand{\D}{\mathcal D}
\newcommand{\A}{\mathcal A}
\renewcommand{\L}{\mathcal L}
\newcommand{\e}{\epsilon}
\newcommand{\F}{\mathcal F}
\newcommand{\B}{\mathcal B}

\newcommand{\Pred}{\mathcal P}

\newcommand{\id}{\mathbf{Id}}

\newcommand{\I}{\mathbf I}

\newcommand{\vf}[1]{\frac{\partial}{\partial{#1}}}
\newcommand{\pt}{\partial}
\newcommand{\Sym}{\mathrm{Sym}}

\renewcommand{\D}{\mathbf{D}}
\newcommand{\Q}{\mathbf{Q}}
\newcommand{\GL}{\mathrm{GL}}
\newcommand{\Ric}{\mathrm{Ric}}
\newcommand{\sign}{\mathrm{sign}}
\newcommand{\Vol}{\mathrm{Vol}}
\newcommand{\divg}{\mathrm{div}\,}
\newcommand{\ts}[1]{\textstyle{{#1}}}
\newcommand{\arctanh}{\mathrm{arctanh}\,}
\renewcommand{\d}{\emph{\textbf{d}}}

\numberwithin{equation}{section} 

\begin{document}

\title{\bf\Large From Second-Order Differential Geometry to \\
Stochastic Geometric Mechanics 
}
\author{\bf\normalsize{
Qiao Huang\footnote{Group of Mathematical Physics (GFMUL), Department of Mathematics, Faculty of Sciences, University of Lisbon, Campo Grande, Edif\'{\i}cio C6, PT-1749-016 Lisboa, Portugal. Email: \texttt{qhuang@fc.ul.pt}},
Jean-Claude Zambrini\footnote{Group of Mathematical Physics (GFMUL), Department of Mathematics, Faculty of Sciences, University of Lisbon, Campo Grande, Edif\'{\i}cio C6, PT-1749-016 Lisboa, Portugal. Email: \texttt{jczambrini@fc.ul.pt}}
}}

\date{}
\maketitle
\vspace{-0.3in}

\begin{abstract}
  Classical geometric mechanics, including the study of symmetries, Lagrangian and Hamiltonian mechanics, and the Hamilton-Jacobi theory, are founded on geometric structures such as jets, symplectic and contact ones. In this paper, we shall use a partly forgotten framework of second-order (or stochastic) differential geometry, developed originally by L.~Schwartz and P.-A.~Meyer, to construct second-order counterparts of those classical structures. These will allow us to study symmetries of stochastic differential equations (SDEs), to establish stochastic Lagrangian and Hamiltonian mechanics and their key relations with second-order Hamilton-Jacobi-Bellman (HJB) equations. Indeed, stochastic prolongation formulae will be derived to study symmetries of SDEs and mixed-order Cartan symmetries. Stochastic Hamilton's equations will follow from a second-order symplectic structure and canonical transformations will lead to the HJB equation. A stochastic variational problem on Riemannian manifolds will provide a stochastic Euler-Lagrange equation compatible with HJB one and equivalent to the Riemannian version of stochastic Hamilton's equations. 
  A stochastic Noether's theorem will also follow. The inspirational example, along the paper, will be the rich dynamical structure of Schr\"odinger's problem in optimal transport, where the latter is also regarded as a Euclidean version of hydrodynamical interpretation of quantum mechanics.
  \bigskip\\
  \textbf{AMS 2020 Mathematics Subject Classification:} 70L10, 35F21, 70H20, 49Q22. \\ 
  \textbf{Keywords and Phrases:} Stochastic Hamiltonian mechanics, stochastic Lagrangian mechanics, Hamilton-Jacobi-Bellman equations, stochastic Hamilton's equations, stochastic Euler-Lagrange equation, stochastic Noether's theorem, Schr\"odinger's problem, second-order differential geometry. 
\end{abstract}


\section{Introduction}

Hamilton-Jacobi (HJ) partial differential equations and the associated theory lie at the center of classical mechanics \cite{AM78,Arn89,MR99,GPS02}. Motivated by Hamilton's approach to geometrical optic where the action represents the time needed by a particle to move between two points and a variational principle due to Fermat, Jacobi extended this approach to Lagrangian and Hamiltonian mechanics. Jacobi designed a concept of ``complete'' solution of HJ equations allowing him to recover all solutions simply by substitutions and {differentiations}. Although, in general, it is  more complicated to solve than a system of ODEs like Hamilton's ones, HJ equations proved to be powerful tools of integration of classical equations of motion. In addition, Jacobi's approach suggested him to ask what diffeomorphisms of the cotangent bundle, the geometric arena of  canonical equations, preserve the structure of these first order equations. Those are called today symplectic or canonical transformations and Jacobi's method of integration is precisely one of them.

It is not always recognized as it should be that HJ equations were also fundamental in the construction of quantum mechanics. The reading of Schr\"odinger \cite{Sch26}, Fock \cite{Foc78}, Dirac \cite{Dir33} and others until Feynman \cite{Fey48} makes  abundantly clear that most of new ideas  in the field made use of  HJ equations for  the classical system to be ``quantized'', or some quantum deformation
of them. There are at least two ways to express this deformation.
On the one hand, one can exponentiate the $L^2$ wave function, call $S$ its complex exponent and look for the equation solved by $S$ (see \cite{GPS02}). When the system is a single particle in a scalar potential, one obtain the classical HJ equation with an additional Laplacian term and factor $i\hbar$, representing the regularization expected from the quantization of the system. This complex factor is symptomatic of the basic quantum probability problem, at least for pure states. In a nutshell, it is the reason why Feynman's diffusions, in his path integral approach, do not exist.
On the other hand, there is an hydrodynamical interpretation of quantum mechanics, founded on Madelung transform, a polar representation of the wave function whose real part is the square root of a probability density. The argument solves another deformation of HJ equation. The geometry of this transform has been thoroughly investigated recently, highlighting its relations with optimal transport theory \cite{KMM21,Von12}.

However, the probabilistic content of quantum mechanics, especially for pure states, remained a vexing mathematical mystery right from its beginning, despite several interesting (but unsuccessful) attempts \cite{Nel01}. The current consensus is that regular probability theory and stochastic analysis have little or nothing to teach us about it. And, in particular, that all that can be saved from Feynman path integral theory is Wiener's measure and perturbations of it by  potential terms. This is the ``Euclidean approach'', one of the starting points of mathematical quantum field theory.

In 1931, however, Schr\"odinger suggested in a paper almost forgotten until the eighties \cite{Sch32} (but insightfully commented by the probabilist  S. Bernstein \cite{Ber32}) the existence of a completely different Euclidean approach to quantum dynamics. In short, a stochastic variational boundary value problem for probability densities characterizes  optimal diffusions on a given time interval as having a density product of two positive solutions of time adjoint heat equations. This idea, revived and elaborated from 1986 \cite{Zam86}, is known today as ``Schr\"odinger's  problem'' in the community of optimal transport, where it has proved to provide, among other results, very efficient regularization of fundamental problems of this field \cite{Leo14}. In fact, Schr\"odinger's problem hinted toward the existence of a stochastic dynamical theory of processes, considerably more general than its initial quantum motivation. In it, various regularizations associated with the tools of stochastic calculus should play the role of those involved in quantum mechanics in Hilbert space, where the looked-for measures do not exist.

The variational side of the stochastic theory has been developed in the last decades, inspired by number of results in stochastic optimal control \cite{Hau86,FS06} and stochastic optimal transport \cite{Mik21}. In this context, the crucial role of (second-order) Hamilton-Jacobi-Bellman (HJB)
equation has been known for a long time. It provides the proper regularization of the (first-order) HJ equation needed to construct well defined stochastic dynamical theories. In contrast, for instance, with the notion of viscosity solution, whose initial target was the study of the classical PDE, HJB equation becomes central, there, as natural stochastic deformation of this one, compatible with It\^o’s calculus. It is worth mentioning that in any fields like AI or reinforcement learning, where HJB equations play a fundamental role \cite{PCV19}, it is natural to expect that such a stochastic dynamical framework, built on them, should present some interest.

The geometric, and especially, Hamiltonian side of the dynamical theory had resisted until now and constitutes the main contribution of this paper. It is our hope that it will be useful far beyond its initial motivation referred to, afterward, as its ``inspirational examples''. In this sense, it can clearly be interpreted as a general contribution to stochastic geometric mechanics. More precisely, we are trying to answer the following questions:

\begin{itemize}
  \item Do we have any geometric interpretation of the Hamilton-Jacobi-Bellman equation? That is, can we derive the HJB equation from some sort of canonical transformations?
  \item Can we formulate some variational problem that leads to a Euler-Lagrange equation which is equivalent to the HJB equation?
  \item More systematically, can we develop some counterpart of Lagrangian and Hamiltonian mechanics that are associated with the HJB equation?
\end{itemize}

The first question indicates that canonical transformations should be somehow  second-order, so that the corresponding symplectic and contact structures are also second-order. Meanwhile, the stochastic generalization of optimal control and optimal transport suggests that the variational problem of the second question should be formulated in stochastic sense. Combining these hints, the third question amounts to seeking a new theory of geometric mechanics that integrates stochastics and second-order together.

The cornerstone of stochastic analysis, the well-known It\^o's formula, tells us that the generator of a diffusion process is a second-order differential operator. This provides a very natural way to connect the stochastics with the second-order. That is, in order to build a stochastic or second-order counterpart of geometric mechanics, we need to encode the rule of It\^o's formula into the geometric structures.

There is a theory named second-order differential geometry (``stochastic differential geometry'' is also used by some authors but we would like to keep the original terminology), which was devised by L. Schwartz and P.-A. Meyer around 1980 \cite{Sch80,Sch82,Sch84,Mey79,Mey81}, and later on developed by Belopolskaya and Dalecky \cite{BD90}, Gliklikh \cite{Gli11}, Emery \cite{Eme89}, etc. See \cite{Eme07} for a survey of this aspect. Compared with the theory of stochastic analysis on manifolds (or geometric stochastic analysis) developed by It\^o \cite{Ito62,Ito75}, Malliavin \cite{Mal97}, Bismut \cite{Bis81} and Elworthy \cite{Elw82} etc., which focus on Stratonovich stochastic differential equations on classical geometric structures, like Riemannian manifolds, frame bundles and Lie groups, so that the Leibniz's rule is preserved, Schwartz' second-order differential calculus alter the underlying geometric structures to include second-order It\^o correction terms, and provide a broader picture even though it loses Leibniz's rule and is much less known.

In this paper, we will adopt the viewpoint of Schwartz--Meyer and enlarge their picture to develop a theory of stochastic geometric mechanics. We first give an equivalent and more intuitive description for the second-order tangent bundle by equivalent classes of diffusions, via Nelson's mean derivatives. And then we generalize this idea to construct stochastic jets, from which stochastic prolongation formulae are proved and the stochastic counterpart of Cartan symmetries is studied. The second-order cotangent bundle is also studied, which helps us to establish stochastic Hamiltonian mechanics. We formulate the stochastic Hamilton's equations, a system of stochastic equations on the second-order cotangent bundle in terms of mean derivatives. By introducing the second-order symplectic structure and the mixed-order contact structure, we derive the second-order HJB equations via canonical transformations. Finally, we set up a stochastic variational problem on the space of diffusion processes, also in terms of mean derivatives. Two kinds of stochastic principle of least action are built: stochastic Hamilton's principle and stochastic Maupertuis's principle. Both of them yield a stochastic Euler-Lagrange equation. The equivalence between the stochastic Euler-Lagrange equation and the HJB equation is proved, which exactly leads to the equivalence between our stochastic variational problem and Schr\"odinger's problem in optimal transport. Last but not least (actually vital), a stochastic Noether's theorem is proved. It says that every symmetry of HJB equation corresponds to a martingale that is exactly a conservation law in the stochastic sense. It should be observed, however, that the Schwartz-Meyer approach, together with the one of Bismut \cite{Bis81}, has also inspired a distinct, Stratonovich-type stochastic Hamiltonian framework \cite{LO08} leading to a stochastic HJ equation \cite{LO09}, without relations with Schr\"odinger's problem or optimal transport.

The key results of the present paper and the dependence among them are briefly expressed in the following diagram:

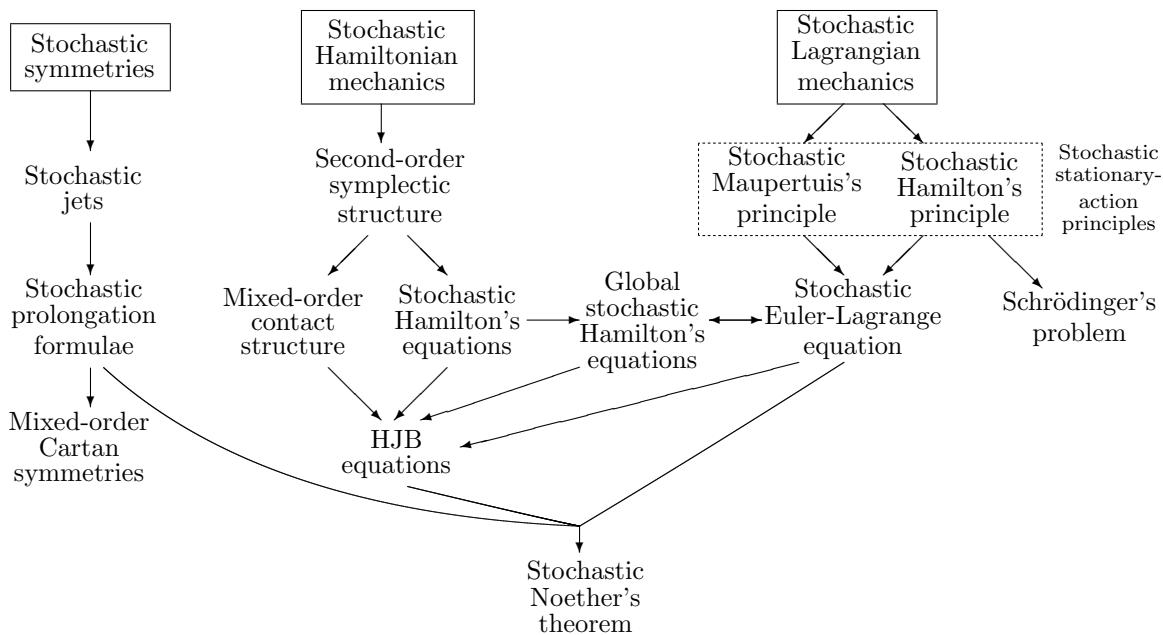
\begin{center}
\begin{picture}(450,260)(-5,-240)
  \put(0,-30){\shortstack{Stochastic \\ symmetries}}
  \put(-5,-35){\framebox(60,25){}}
  \put(25,-35){\vector(0,-1){23}}
  \put(0,-80){\shortstack{Stochastic \\ jets}}
  \put(25,-85){\vector(0,-1){20}}
  \qbezier[800](30,-140)(90,-195)(210,-200)
  \put(-5,-135){\shortstack{Stochastic \\ prolongation \\ formulae}}
  \put(25,-140){\vector(0,-1){15}}
  \put(-6,-183){\shortstack{Mixed-order \\ Cartan \\ symmetries}}

  \put(110,-35){\shortstack{Stochastic \\ Hamiltonian \\ mechanics}}
  \put(105,-40){\framebox(65,35){}}
  \put(135,-40){\vector(0,-1){15}}
  \put(110,-85){\shortstack{Second-order \\ symplectic \\ structure}}
  \put(130,-90){\vector(-1,-1){15}}
  \put(145,-90){\vector(1,-1){15}}
  \put(75,-135){\shortstack{Mixed-order \\ contact \\ structure}}
  \put(115,-140){\vector(1,-1){20}}
  \put(140,-135){\shortstack{Stochastic \\ Hamilton's \\ equations}}
  \put(160,-140){\vector(-1,-1){20}}
  \put(190,-122){\vector(1,0){20}}
  \put(210,-140){\shortstack{Global \\ stochastic \\ Hamilton's \\ equations}}
  \put(259,-122){\vector(1,0){20}}
  \put(279,-122){\vector(-1,0){20}}
  \put(210,-140){\vector(-3,-1){60}}
  \put(120,-180){\shortstack{HJB \\ equations}}
  \put(293,-138){\vector(-4,-1){128}}
  \qbezier(145,-185)(195,-197)(210,-200)

  \put(290,-35){\shortstack{Stochastic \\ Lagrangian \\ mechanics}}
  \put(285,-40){\framebox(60,35){}}
  \put(310,-40){\vector(-1,-1){15}}
  \put(325,-40){\vector(1,-1){15}}
  \put(390,-88){\shortstack{\footnotesize Stochastic \\ \footnotesize stationary- \\ \footnotesize action \\ \footnotesize principles}}
  \put(255,-90){\dashbox(130,35){}}
  \put(260,-85){\shortstack{Stochastic \\ Maupertuis's \\ principle}}
  \put(295,-90){\vector(1,-1){15}}
  \put(330,-85){\shortstack{Stochastic \\ Hamilton's \\ principle}}
  \put(340,-90){\vector(-1,-1){15}}
  \put(365,-90){\vector(1,-1){20}}
  \put(370,-130){\shortstack{Schr\"odinger's \\ problem}}
  \put(280,-135){\shortstack{Stochastic \\ Euler-Lagrange \\ equation}}
  \qbezier(310,-138)(260,-170)(210,-200)
  \put(210,-200){\vector(0,-1){10}}
  \put(190,-240){\shortstack{Stochastic \\ Noether's \\ theorem}}
\end{picture}
\end{center}
\vspace{5mm}

The organization of this paper is the following.

Chapter \ref{sec-2} is a summary on the theory of stochastic differential equations on manifolds, in the perspective appropriate to our goal. In particular, diffusions will be characterized by their mean and quadratic mean derivatives as in Nelson's stochastic mechanics \cite{Nel01} although the resulting dynamical content of our theory will have very little to do with his. In this way, we are able to rewrite It\^o SDEs on manifolds as ODE-like equations that have better geometric nature. The notion of second-order tangent bundle answers to the question: the drift parts of It\^o SDEs are sections of what? 

Chapter \ref{sec-3} is devoted to the notion of Stochastic jets. In the same way as tangent vector on $M$ are defined as equivalence classes of smooth curves through a given point and then generalized to higher-order cases to produce the notion of jets, the stochastic tangent vector is defined as equivalence classes of diffusions so that the stochastic tangent bundle is isomorphic to the elliptic subbundle of the second-order tangent bundle. Stochastic jets are also constructed. This provides an intrinsic definition of SDEs under consideration.

Chapter \ref{sec-4} illustrates the use of the above geometric formulation of SDEs for the study of their symmetries. Prolongations of $M$-valued diffusions are defined as new processes with values on the stochastic tangent bundle. Among all deterministic space-time transformations, bundle homomorphisms will be the only subclass to transform diffusions to diffusions. Total mean and quadratic derivative are defined in conformity with the rules of It\^o's calculus. The prolongation of diffusions allows to define symmetries of SDEs and their infinitesimal versions. Stochastic prolongation formulae are derived for infinitesimal symmetries, which yield determining equations for It\^o SDEs. 

In Chapter \ref{sec-5}, the second-order cotangent bundle, as dual bundle of second-order tangent bundle, is defined and analyzed. 
The properties of second-order differential operator, pushforwards and pullbacks are described. When time is involved, i.e., the base manifold is the product manifold $\R\times M$, the corresponding bundles are mixed-order tangent and cotangent bundles, where ``mixed-order'' means they are second-order in space but first-order in time. More about this topic, like mixed-order pushforwards and pullbacks, pushforwards and pullbacks by diffusions, and Lie derivatives, can be found in Appendix \ref{app-1}. 
An generalized notion to stochastic Cartan distribution and its symmetries are discussed in Appendix \ref{app-2} based on the mixed-order contact structure.

The point of Chapter \ref{sec-6} is to use the tools developed before in the construction of the stochastic Hamiltonian mechanics which is one of the main goals of the paper. One of our inspirational example will be the one underlying the dynamical content of Schr\"odinger's problem.
By analogy with  Poincar\'e $1$-form in the cotangent bundle of classical mechanics and its associated symplectic form, one can construct  counterparts in the second-order cotangent bundle. 
Using the canonical second-order symplectic form on second-order cotangent bundles, one defines second-order symplectomorphisms. The generalization of classical Hamiltonian vector fields becomes second-order operators, for a given real-valued Hamiltonian function on the second-order cotangent bundle. The resulting stochastic Hamiltonian system involves pairs of extra equations compared with their classical versions. Bernstein's reciprocal processes inspired by Schr\"odinger's problem are described in this framework, corresponding to a large class of second-order Hamiltonians on Riemannian manifolds. 
A mixed-order contact structure describes time-dependent stochastic Hamiltonian systems. The last section of this chapter is devoted to canonical transformations preserving the form of stochastic Hamilton's equations. The corresponding generating function satisfies the Hamilton-Jacobi-Bellman equation.

Chapter \ref{sec-7} treats the stochastic version of classical Lagrangian mechanics on Riemannian manifolds. It\^o's stochastic deformation of the classical notion of parallel displacements are recalled. Another one, called damped parallel displacement in the mathematical literature, involving the Ricci tensor, is also indicated. Each of these displacements corresponds to a mean covariant derivative along diffusions. 
The action functional is defined as expectation of Lagrangian and the stochastic Euler-Lagrange equation involves the damped mean covariant derivative.
The dynamics of Schr\"odinger's problem is, again, used as illustration. 
The equivalence between stochastic Hamilton's equations on Riemannian manifolds and the stochastic Euler-Lagrange one as well as the HJB equation are derived via the Legendre transform. Relations with stochastic control are also mentioned. The chapter ends with the stochastic Noether's theorem. The stochastic version of Maupertuis principle, as the twin of stochastic Hamilton's principle, is left into Appendix \ref{app-3}.

We end the introduction with a list of notations and abbreviations frequently used in the paper, for reader's convenience.

\subsection{List of main notations and abbreviations}

\begin{longtable}{ll}
\hline
HJB equation & Hamilton-Jacobi-Bellman equation \\
MDE & Mean differential equations \\
2nd-order & Second-order \\
SDE & Stochastic differential equations \\
S-EL equation & Stochastic Euler-Lagrange equation \\
S-H equations & Stochastic Hamilton's equations \\
\hline
$A$ & A general second-order differential operator or second-order vector field \\
$A^X$ & Generator of the diffusion $X$ \\
$\circ\,d$ & Stratonovich stochastic differential \\
$d$ & Exterior differential on the manifold $M$, or It\^o stochastic differential \\
$\d$ & Linear operator extended from the exterior differential on the tangent bundle $TM$ \\
$d^2$ & Second-order differential on $M$ \\
$d^\circ$ & Mixed-order differential on $\R\times M$ \\
$d_x$ & Horizontal differential on the tangent bundle $TM$ or cotangent bundle $T^*M$ \\
$d_{\dot x}$ & Vertical differential on $TM$ \\
$(DX, QX), D_\nabla X, Q(X,Y)$ & Mean derivatives \\
$\D_t,\Q_t$ & Total mean derivatives \\
$\frac{\D}{dt}, \frac{\overline\D}{dt}$ & Mean covariant derivative and damped mean covariant derivative \\
$\Delta$, $\Delta_{\text{LD}}$ & Connection Laplacian and Laplace-de Rham operator \\
$F^S_*, F^{S*}$ & Second-order pushforward and pullback of a smooth map $F:M\to N$ \\
$F^R_*, F^{R*}$ & Mixed-order pushforward and pullback of $F$ \\
$\Gamma$ &  Christoffel symbols or stochastic parallel displacement \\
$\overline \Gamma$ & Damped parallel displacement \\
$I_t(M), I_{(t,q)}(M), I_{(t,q)}^{T,\mu}(M)$ & Various sets of $M$-valued diffusions starting at time $t$ \\
$j_q X, j_q^\nabla X, j_{(t,q)}X, j_t X$ & Stochastic
tangent vectors and stochastic jets \\
$\L$ & Lie derivatives \\
$\nabla$ & Linear connection, Levi-Civita connection, covariant derivative, or gradient on $M$ \\
$\nabla^2$ & Hessian operator \\
$\nabla_p$ & Vertical gradient on $T^*M$ \\
$(\Omega, \F, \P)$ & Probability space $\Omega$ with $\sigma$-field $\F$ and probability measure $\P$ \\
$\{\Pred_t\}_{t \in \R}$, $\{\F_t\}_{t \in \R}$ & Past (nondecreasing) filtration and future (nonincreasing) filtration \\
$\vf{t}, \pt_{t}$ & Differential operator with respect to coordinate $t$ \\
$\vf{x^i}, \pt_i$ & Differential operator with respect to coordinate $x^i$ \\
$\frac{\pt^2}{\pt x^j\pt x^k}, \pt_{jk}$ & Second-order differential operator with respect to coordinates $x^j$ and $x^k$ \\
$\vf{p_i}, \pt_{p_i}$ & Differential operator with respect to coordinate $p_i$ \\
$R$, $\Ric$ & Riemann curvature tensor and Ricci $(1,1)$-tensor \\
$\mathcal T^O M, \mathcal T^E M$ & Second-order tangent bundle and second-order elliptic tangent bundle \\
$\mathcal T^S M$ & Stochastic tangent bundle \\
$\mathcal T^{S*} M$ & Second-order cotangent bundle \\
$V$ & A general vector field \\
$(x^i,D^i x, Q^{jk} x)$ & Canonical coordinates on $\mathcal T^S M$ \\
$(x^i,p_i, o_{jk})$ & Canonical coordinates on $\mathcal T^{S*} M$ \\
$X_*, X^*$ & Pushforward and pullback of the diffusion $X$ \\
$\mathbf X$ & A horizontal diffusion valued on a general bundle $E$ or on $\mathcal T^{S*} M$ \\
$\mathbb X$ & A horizontal diffusion valued on $T^* M$ \\
\hline
\end{longtable}

\section{Stochastic differential equations on manifolds}\label{sec-2}

In this chapter, we will study several types of stochastic differential equations on manifolds which are weakly equivalent to It\^o SDEs. We start with a $d$-dimensional smooth manifold $M$ and a probability space $(\Omega, \F, \P)$, and equip the latter with a filtration $\{\Pred_t\}_{t \in \R}$, i.e., a family of nondecreasing sub-$\sigma$-fields of $\F$. We call $\{\Pred_t\}_{t \in \R}$ a \emph{past filtration}. Unless otherwise specified, the manifold $M$ will not be endowed with any structures other than the smooth structure. In some cases, it will be endowed with a linear connection, a Riemannian metric, or a Levi-Civita connection.

Recall from \cite[Definition 1.2.1]{Hsu02} that by an $M$-valued (forward) $\{\Pred_t\}$-\emph{semimartingale}, we mean a $\{\Pred_t\}$-adapted continuous $M$-valued process $X= \{X(t)\}_{t\in[t_0,\tau)}$, where $t_0\in\R$ and $\tau$ is a $\{\Pred_t\}$-stopping time satisfying $t_0<\tau\le+\infty$, such that $f(X)$ is a real-valued $\{\Pred_t\}$-semimartingale on $[t_0,\tau)$ for all $f\in C^\infty(M)$. The stopping time $\tau$ is called the lifetime of $X$. If we adopt the convention to introduce the one-point compactification of $M$ by $M^*:= M \cup \{\pt_M\}$, then the process $X$ can be extended to the whole time line $[t_0,+\infty)$ by setting $X(t) = \pt_M$ for all $t\ge\tau$. The point $\pt_M$ is often called the cemetery point in the context of Markovian theory.

\subsection{It\^o SDEs on manifolds}

Given $N+1$ time-dependent vector field $b,\sigma_r, r=1,\cdots,N$ on $M$, one can introduce a Stratonovich SDE in local coordinates, which has the same form as in Euclidean space \cite[Section 1.2]{Hsu02}. The form of Stratonovich SDEs on $M$ is invariant under changes of coordinates, as Stratonovich stochastic differentials obey the Leibniz's rule.

However, for It\^o stochastic differentials this is not the case because of It\^o's formula. Hence, we cannot directly write an Euclidean form of It\^o SDE on $M$ in local coordinates, since it is no longer invariant under changes of coordinates. Indeed, a change of coordinates will always produce an additional term. To balance this term, a common way is to add a correction term to the drift part of the Euclidean form of It\^o SDE, by taking advantage of a linear connection. More precisely, under local coordinates $(x^i)$, we consider the following It\^o SDE \cite[Section 7.1, 7.2]{Gli11}:
\begin{equation}\label{SDE-mfld}
  dX^i(t) = \left[ b^i(t,X(t)) - \frac{1}{2} \sum_{r=1}^N \Gamma^i_{jk}(X(t)) (\sigma^j_r \sigma^k_r)(t,X(t)) \right]dt + \sigma^i_r(t,X(t)) dW^r(t),
\end{equation}
where $(\Gamma^i_{jk})$ is the family of Christoffel symbols for a given linear connection $\nabla$ on $TM$. When conditioning on $\{X(t)=q\}$ and taking $(x^i)$ as normal coordinates at $q\in M$, \eqref{SDE-mfld} turns to the Euclidean form, since at $q$,
\begin{equation}\label{normal-vanishing}
  \sum_{r=1}^N \Gamma^i_{jk} \sigma^j_r \sigma^k_r = \frac{1}{2} \sum_{r=1}^N \left(\Gamma^i_{jk} + \Gamma^i_{kj}\right) \sigma^j_r \sigma^k_r = 0.
\end{equation}

If we denote
\begin{equation*}
  \sigma\circ\sigma^* := \sum_{r=1}^N \sigma_r \otimes \sigma_r = \sum_{r=1}^N \sigma_r^j \sigma_r^k \vf{x^j}\otimes \vf{x^k}.
\end{equation*}
Then, clearly $\sigma\circ\sigma^*$ is a symmetric and positive semi-definite $(2,0)$-tensor field. We also introduce formally a modified drift $\mathfrak b$ which has the following coordinate expression
\begin{equation}\label{modified-drift}
  \mathfrak b^i = b^i - \frac{1}{2} \sum_{r=1}^N \Gamma^i_{jk} \sigma^j_r \sigma^k_r.
\end{equation}
We change the coordinate chart from $(U,(x^i))$ to $(V,(\tilde x^j))$ with $U\cap V\ne \emptyset$. Since each $\sigma_r$ transforms as a vector, we apply the change-of-coordinate formula for Christoffel symbols (e.g., \cite[Proposition III.7.2]{KN63}) to derive that
\begin{equation*}
  \begin{split}
    \Gamma^i_{jk} \sigma^j_r \sigma^k_r &= \left( \tilde\Gamma^l_{mn} \frac{\pt\tilde x^m}{\pt x^j} \frac{\pt\tilde x^n}{\pt x^k} \frac{\pt x^i}{\pt\tilde x^l} + \frac{\pt^2\tilde x^l}{\pt x^j \pt x^k} \frac{\pt x^i}{\pt\tilde x^l} \right) \sigma^j_r \sigma^k_r = \left( \tilde\Gamma^l_{mn} \tilde \sigma^m_r \tilde \sigma^n_r + \frac{\pt^2\tilde x^l}{\pt x^j \pt x^k} \sigma^j_r \sigma^k_r \right) \frac{\pt x^i}{\pt\tilde x^l}.
  \end{split}
\end{equation*}
It follows that the coefficients of the modified drift $\mathfrak{b}$ in \eqref{modified-drift} transform as
\begin{equation}\label{modified-drift-tsfm}
  \begin{split}
    \tilde{\mathfrak{b}}^l &= \tilde b^l - \frac{1}{2} \sum_{r=1}^N \tilde\Gamma^l_{mn} \tilde \sigma^m_r \tilde \sigma^n_r = b^i\frac{\pt \tilde x^l}{\pt x^i} - \frac{1}{2} \sum_{r=1}^N \left( \Gamma^i_{jk} \frac{\pt \tilde x^l}{\pt x^i} - \frac{\pt^2 \tilde x^l}{\pt x^j \pt x^k} \right) \sigma^j_r \sigma^k_r = \mathfrak{b}^i \frac{\pt \tilde x^l}{\pt x^i} + \frac{1}{2} \frac{\pt^2 \tilde x^l}{\pt x^j \pt x^k} \sum_{r=1}^N \sigma^j_r \sigma^k_r.
  \end{split}
\end{equation}
Therefore, $\mathfrak{b}$ is \emph{not} a vector field as it does not pointwisely transform as a vector.

Finally, using It\^o's formula, we derive the transformation of \eqref{SDE-mfld} as follows:
\begin{equation*}
  \begin{split}
    d\tilde x^l &= \frac{\pt \tilde x^l}{\pt x^i} dx^i + \frac{1}{2} \frac{\pt^2 \tilde x^l}{\pt x^j \pt x^k} d[x^j,x^k] \\
    &= \left[ \frac{\pt \tilde x^l}{\pt x^i} \left( b^i - \frac{1}{2} \sum_{r=1}^N \Gamma^i_{jk} \sigma^j_r \sigma^k_r \right) + \frac{1}{2} \sum_{r=1}^N \frac{\pt^2 \tilde x^l}{\pt x^j \pt x^k} \sigma^j_r \sigma^k_r \right]dt + \frac{\pt \tilde x^l}{\pt x^i} \sigma^i_r dW^r \\
    &= \left( \tilde b^l - \frac{1}{2} \sum_{r=1}^N \tilde\Gamma^l_{mn} \tilde \sigma^m_r \tilde \sigma^n_r \right)dt + \tilde \sigma^l_r dW^r,
  \end{split}
\end{equation*}
where the bracket $[\cdot,\cdot]$ on the right hand side (RHS) of the first equality denotes the quadratic variation. This shows that equation \eqref{SDE-mfld} is indeed invariant under changes of coordinates.

\begin{remark}
  One can regard $\sigma = (\sigma_r)_{r=1}^N \in (\R^N)^*\otimes\mathfrak{X}(M)$ as an $(\R^N)^*$-valued vector field on $M$. In this way, the pair $(b,\sigma)$ is called an It\^o vector field in \cite[Chapter 7]{Gli11}, while the pair $(\mathfrak{b},\sigma)$ is called an It\^o equation therein.
\end{remark}

%
%

Now we present the definition of weak solutions to \eqref{SDE-mfld}.

\begin{definition}[Weak solutions to It\^o SDEs]
  Given a linear connection on $M$, a weak solution of the It\^o SDE \eqref{SDE-mfld} is a triple $(X, W)$, $(\Omega,\F,\P)$, $\{\Pred_t\}_{t\in\R}$, where
  \begin{itemize}
    \item[(i)] $(\Omega,\F,\P)$ is a probability space, and $\{\Pred_t\}_{t\in\R}$ is a past (i.e., nondecreasing) filtration of $\F$ satisfying the usual conditions,
    \item[(ii)] $X = \{X(t)\}_{t\in[t_0,\tau)}$ is a continuous, $\{\Pred_t\}$-adapted $M$-valued process with $\{\Pred_t\}$-stopping time $\tau >t_0$, $W$ is an $N$-dimensional $\{\Pred_t\}$-Brownian motion, and
    \item[(iii)] for every $q\in M$, $t\ge t_0$ and any coordinate chart $(U,(x^i))$ of $q$, it holds under the conditional probability $\P(\cdot | X(t_0) = q)$ that almost surely in the event $\{X(t)\in U\}$,
        \begin{equation*}
          X^i(t) = X^i(t_0) + \int_{t_0}^t \left( b^i(s,X(s)) - \frac{1}{2} \sum_{r=1}^N \Gamma^i_{jk}(X(s)) (\sigma^j_r \sigma^k_r)(s,X(s)) \right)ds + \int_{t_0}^t \sigma^i_r(s,X(s)) dW^r(s).
        \end{equation*}
  \end{itemize}
\end{definition}

\begin{definition}[Uniqueness in law]\label{unique-Nelson}
  We say that uniqueness in the sense of probability law holds for the It\^o SDE \eqref{SDE-mfld} 
  if, for any two weak solutions $(X, W)$, $(\Omega,\F,\P)$, $\{\Pred_t\}_{t\in\R}$, and $(\hat X, \hat W)$, $(\hat\Omega,\hat\F,\hat\P)$, $\{\hat\Pred_t\}_{t\in\R}$ 
  with the same initial data, i.e., $\P(X(0)=x_0) = \hat\P(\hat X(0)=x_0) = 1$, the two processes $X$ and $\tilde X$ have the same law.
\end{definition}

Note that it is possible to change $\sigma$ and $W$ in the It\^o SDE \eqref{SDE-mfld} but keep the same weak solution in law. In other words, the form of \eqref{SDE-mfld} does not univocally correspond to its weak solution in law. For this reason, we will reformulate SDEs in a fashion that makes them look more like ODEs and have better geometric nature. Moreover, we will see that it is the pair $(\mathfrak{b}, \sigma\circ\sigma^*)$ that univocally corresponds to the weak solution of \eqref{SDE-mfld}.

\subsection{Mean derivatives and mean differential equations on manifolds}\label{sec-2-2}

In this part, we will recall the definitions of Nelson's mean derivatives and extend them to $M$-valued processes. In Nelson's stochastic mechanics \cite{Nel01}, the probability space $(\Omega, \F, \P)$ is equipped with two different filtrations. The first one is just an usual nondecreasing filtration $\{\Pred_t\}_{t \in \R}$, a past filtration. The second is a family of nonincreasing sub-$\sigma$-fields of $\F$, which is denoted by $\{\F_t\}_{t \in \R}$ and called a \emph{future} filtration. 
For an $\R^d$-valued process $\{X(t)\}_{t \in I}$, its forward mean derivative $DX$ and forward quadratic mean derivative $Q X$ are defined by conditional expectations as follows:
  \begin{equation*}
    DX(t) = \lim_{\e\to0^+} \E\left[ \frac{X(t+\e)-X(t)}{\e} \bigg| \Pred_t \right], \quad Q X(t) = \lim_{\e\to0^+} \E\left[ \frac{(X(t+\e)-X(t))\otimes (X(t+\e)-X(t))}{\e} \bigg| \Pred_t \right],
  \end{equation*}
Their backward versions, i.e., the backward mean derivative and backward quadratic mean derivative, are defined as follows:
  \begin{equation*}
    \overleftarrow{D}X(t) = \lim_{\e\to0^+} \E\left[ \frac{X(t)-X(t-\e)}{\e} \bigg| \F_t \right], \quad \overleftarrow{Q} X(t) = \lim_{\e\to0^+} \E\left[ \frac{(X(t)-X(t-\e))\otimes(X(t)-X(t-\e))}{\e} \bigg| \F_t \right].
  \end{equation*}

In our present paper, we will only focus on the ``forward'' case, so that only the past filtration $\{\Pred_t\}_{t \in \R}$ will be invoked. The ``backward'' case is analogous and every part of this paper can have its ``backward'' counterpart (cf. \cite{Zam15}). 

Denote by $\Sym^2(TM)$ (and $\Sym^2_+(TM)$) the fiber bundle of symmetric (and respectively, symmetric positive semi-definite) $(2,0)$-tensors on $M$. Now we define quadratic mean derivatives for $M$-valued semimartingales, cf. \cite[Chapter 9]{Gli11}.

\begin{definition}[Quadratic mean derivatives]
  The (forward) quadratic mean derivative of the $M$-valued semimartingale $\{X(t)\}_{t \in [t_0,\tau)}$ is a $\Sym^2_+(TM)$-valued process $QX$ on $[t_0,\tau)$, whose value at time $t\in[t_0,\tau)$ in any coordinate chart $(U,(x^i))$ and in the event $\{X(t) \in U\}$ is given by
  \begin{equation}\label{quad-deriv}
    (Q X)^{ij}(t) = \lim_{\e\to0^+} \E\left[ \frac{(X^i(t+\e)-X^i(t)) (X^j(t+\e)-X^j(t))}{\e} \bigg| \Pred_t \right],
  \end{equation}
  where the limits are assumed to exist in $L^1(\Omega, \F, \P)$.
\end{definition}

More generally, we can define the (forward) quadratic mean derivative for two $M$-valued semimartingales $X$ and $Y$ in local coordinates by
\begin{equation*}
    (Q (X,Y))^{ij}(t) = \lim_{\e\to0^+} \E\left[ \frac{(X^i(t+\e)-X^i(t)) (Y^j(t+\e)-Y^j(t))}{\e} \bigg| \Pred_t \right].
  \end{equation*}

Due to It\^o's formula for semimartingales, $Q X(t)$ does transform as a $(2,0)$-tensor and is obviously symmetric, so that the definition is independent of the choice of $U$. However, the formal limit $\E[ \frac{1}{\e} (X^i(t+\e)-X^i(t)) | \Pred_t ]$ under any coordinates $(x^i)$, no longer transforms as a vector, as can be guessed from \eqref{modified-drift-tsfm}. In order to turn it into a vector we need to specify a coordinate system. A natural choice is the normal coordinate system. For this purpose, we endow $M$ with a linear connection $\nabla$, which determines a normal coordinate system near each point on $M$.

\begin{definition}[$\nabla$-mean derivatives]\label{def-vector-deriv}
  Given a linear connection $\nabla$ on $M$, the (forward) $\nabla$-mean derivative of the $M$-valued semimartingale $\{X(t)\}_{t \in [t_0,\tau)}$ is a $TM$-valued process $D_\nabla X$ on $[t_0,\tau)$, whose value at time $t\in[t_0,\tau)$ is defined in \emph{normal} coordinates $(x^i)$ on the normal neighborhood $U$ of $q\in M$ and under the conditional probability $\P(\cdot | X(t) = q)$ as follows:
  \begin{equation*}
    (D_\nabla X)^i(t) = \lim_{\e\to0^+} \E\left[ \frac{X^i(t+\e)-X^i(t)}{\e} \bigg| \Pred_t \right],
  \end{equation*}
  where the limits are assumed to exist in $L^1(\Omega, \F, \P)$.
\end{definition}

As we force $D_\nabla X(t)$ to be vector-valued by definition, its coordinate expression under any other coordinate system can be calculated via Leibniz's rule. Let us stress that the notation $D_\nabla$ should not be confused with the one of covariant derivatives in geometry. 

Now we formally take forward mean derivatives in It\^o SDE \eqref{SDE-mfld}, and note that the correction term in the modified drift involving Christoffel symbols vanishes by \eqref{normal-vanishing}. Then, we get an ODE-like system:
\begin{equation}\label{Nelson-SDE-mfld}\left\{
  \begin{aligned}
    &D_\nabla X(t) = b(t,X(t)), \\
    &Q X(t) = (\sigma\circ\sigma^*)(t,X(t)).
  \end{aligned} \right.
\end{equation}
We call equations \eqref{Nelson-SDE-mfld} a system of mean differential equations (MDEs). Note that both MDEs \eqref{Nelson-SDE-mfld} and It\^o SDE \eqref{SDE-mfld} rely on linear connections on $M$.


\begin{definition}[Solutions to MDEs]\label{weak-sol-Nelson}
  Given a linear connection on $M$, a solution of MDEs \eqref{Nelson-SDE-mfld} is a triple $X$, $(\Omega,\F,\P)$, $\{\Pred_t\}_{t\in\R}$, where
  \begin{itemize}
    \item[(i)] $(\Omega,\F,\P)$ is a probability space, and $\{\Pred_t\}_{t\in\R}$ is a past filtration of $\F$ satisfying the usual conditions,
    \item[(ii)] $X = \{X(t)\}_{t\in[t_0,\tau)}$ is a continuous, $\{\Pred_t\}$-adapted $M$-valued semimartingale with lifetime a $\{\Pred_t\}$-stopping time $\tau >t_0$, and
    \item[(iii)] the $\nabla$-mean derivative and quadratic mean derivative of $X$ exist and satisfy \eqref{Nelson-SDE-mfld}.
    \end{itemize}
\end{definition}

\subsection{Second-order operators and martingale problems}

\begin{definition}[Second-order operators]\label{2-operator}
  A second-order operator on $M$ is a linear operator $A : C^\infty(M) \to C^\infty(M)$, which has the following expression in a coordinate chart $(U,(x^i))$,
  \begin{equation}\label{2-tangent-rep}
    Af = A^i \frac{\partial f}{\partial x^i} + A^{ij} \frac{\partial^2 f}{\partial x^i\partial x^j}, \quad f\in C^\infty(M),
  \end{equation}
  where $(A^{ij})$ is a symmetric $(2,0)$-tensor field, and the expression is required to be invariant under changes of coordinates. If $(A^{ij})$ is positive semi-definite, then we say the second-order operator $A$ is elliptic; if $(A^{ij})$ is positive definite, we say $A$ is nondegenerate elliptic.
\end{definition}

There is a coordinate-free definition of second-order operators. A linear map $A_q: C^\infty(M) \to \R$ is called a second-order derivation at $q\in M$, if there is a symmetric $(2,0)$-tensor $\Gamma_{A_q}$ at $q$ such that
$A_q(fg) = f(q) A_q g + g(q) A_qf + (df\otimes dg) (\Gamma_{A_q})$ for all $f,g \in C^\infty(M)$. Then, a second-order operator is nothing but a smooth field of second-order derivations.
From this, we see that for $A$ in \eqref{2-tangent-rep}, $A^i = A(x^i)$, $A^{ij} = A(x^i x^j) - x^iA(x^j) - x^jA(x^i)$, and
\begin{equation}\label{squared-field}
  \Gamma_A = A^{ij} \vf{x^i}\otimes\vf{x^j}.
\end{equation}
We call $\Gamma_A$ the \emph{squared field operator} (originally ``op\'erateur carr\'e du champ'') associated with $A$. We also denote $\Gamma_A(f,g) := (df\otimes dg) (\Gamma_A)$. Clearly, for a classical vector field $V$, $\Gamma_V \equiv 0$ by Leibniz's rule.

It is easy to verify from the coordinate-change invariance that the coefficients $A^i$'s and $A^{ij}$'s transform under the change of coordinates from $(x^i)$ to $(\tilde x^j)$ by the following rule (e.g., \cite[Section V.4]{IW89}),
\begin{equation}\label{2-tangent-tsfm}
  \tilde A^i = \frac{\partial \tilde x^i}{\partial x^j} A^j + \frac{\partial^2 \tilde x^i}{\partial x^j \partial x^k} A^{jk}, \quad \tilde A^{ij} = \frac{\partial \tilde x^i}{\partial x^k} \frac{\partial \tilde x^j}{\partial x^l} A^{kl}.
\end{equation}
The formal generator of It\^o SDE \eqref{SDE-mfld} is given by,
\begin{equation}\label{generator}
  A^X_t = \mathfrak{b}^i(t) \frac{\partial}{\partial x^i} + \frac{1}{2} \sum_{r=1}^N \sigma^i_r(t) \sigma^j_r(t) \frac{\partial^2}{\partial x^i\partial x^j},
\end{equation}
which is a time-dependent second-order elliptic operator due to the change-of-coordinate formula \eqref{modified-drift-tsfm}.

Denote by $\C_{t_0}$ the subspace of $C([t_0,\infty),M^*)$ consisting of all paths always staying in $M$ or eventually stopped at $\pt_M$. That is, $\omega\in\C_{t_0}$ if and only if there exists $\tau(\omega)\in (t_0,\infty]$ such that $\omega(t)\in M$ for $t\in[t_0,\tau(\omega))$ and $\omega(t) = \pt_M$ for $t\in [\tau(\omega),\infty)$. Let $\B(\C_{t_0})$ be the $\sigma$-field generated by Borel cylinder sets. Let $X(t) : \C_{t_0}\to M^*, X(t,\omega) = \omega(t), t\ge t_0$ be the coordinate mapping. For each $t\in\R$, define a sub-$\sigma$-field by $\B_t = \sigma\{X(s): t_0 \le s\le t_0\vee t\}$. Then $\{\B_t\}_{t\in\R}$ is a past filtration of $\B(\C_{t_0})$ and $\tau$ is a $\{\B_t\}$-stopping time.

\begin{definition}[Martingale problems on manifolds, {\cite[Definition 1.3.1]{Hsu02}}]\label{mtgl-prob}
  Given a time-dependent second-order elliptic operator $A=(A_t)_{t\ge t_0}$, 
  a solution to the martingale problem associated with $A$ is a triple $X$, $(\Omega,\F,\P)$, $\{\Pred_t\}_{t\in\R}$, where
  \begin{itemize}
    \item[(i)] $(\Omega,\F,\P)$ is a probability space, and $\{\Pred_t\}_{t\in\R}$ is a past filtration of $\F$ satisfying the usual conditions,
    \item[(ii)] $X:\Omega\to \C_{t_0}$ is an $M^*$-valued $\{\Pred_t\}$-semimartingale, and
    \item[(iii)] for every $f\in C^\infty(\R\times M)$, the process
    $M^{f,X}(t) := f(t,X(t)) - f(t_0,X(t_0)) - \int_{t_0}^t (\frac{\pt}{\pt t}+A_s) f(s,X(s)) ds$, $t\in[t_0,\tau(X))$,
    is a real-valued continuous $\{\Pred_t\}$-martingale.
    \end{itemize}
    The process $\{X(t)\}_{t\in[t_0,\tau(X))}$ is called an $M$-valued $\{\Pred_t\}$-diffusion process with generator $A$ (or simply an $A$-diffusion).
\end{definition}

The uniqueness in the sense of probability law for both MDEs and martingale problems can be defined in a similar fashion to Definition \ref{unique-Nelson}. Note that unlike It\^o SDEs or MDEs, the definition for martingale problems does not rely on linear connections.

When provided with a linear connection on $M$, one can see, in the same way as in Stroock and Varadhan's theory (e.g., \cite[Section 5.4]{KS91}), that the existence of a solution to the martingale problem associated with $A^X=(A^X_t)_{t\ge t_0}$ in \eqref{generator} is equivalent to the existence of a weak solution to the It\^o SDE \eqref{SDE-mfld}, and also equivalent to the existence of a solution to MDEs \eqref{Nelson-SDE-mfld}; their uniqueness in law of are also equivalent.

%
%


\subsection{The second-order tangent bundle}\label{sec-2-4}

As we have seen, the modified drift $\mathfrak{b}$ in \eqref{modified-drift} is not a vector field. Is $\mathfrak{b}$ a section (and, in the affirmative, of what)? In fact, it is not a section of any bundle, as its changes-of-coordinate formula \eqref{modified-drift-tsfm} involves $\sigma$. But if we look at the formal generator $A^X$ in \eqref{generator}, or the pair $(\mathfrak{b}, \sigma\circ\sigma^*)$ of its coefficients, then we can construct a bundle whose structure group is governed by the changes-of-coordinate formulae \eqref{2-tangent-tsfm}, so that the sections are just second-order operators.

We denote by $\Sym^2(\R^d)$ the space of all symmetric $(2,0)$-tensors on $\R^d$, and by $\Sym^2_+(\R^d)$ the subspace of it consisting of all positive semi-definite $(2,0)$-tensors. Also denote by $\L(\R^n,\R^d)$ the space of all linear maps from $\R^n$ to $\R^d$.

\begin{definition}[The second-order tangent bundle]
  (i). \cite[Definition 7.14]{Gli11} The It\^o group $G_I^d$ is the Cartesian product (but not direct product of groups) $\GL(d,\R) \times \L(\R^d\otimes\R^d,\R^d)$ equipped with the following binary operation:
  \begin{equation*}
    (g_2, \kappa_2) \circ (g_1, \kappa_1) = (g_2\circ g_1, g_2\circ \kappa_1 + \kappa_2\circ (g_1\otimes g_1)),
  \end{equation*}
  for all $g_1, g_2 \in \GL(d,\R)$, $\kappa_1, \kappa_2\in \L(\R^d\otimes\R^d,\R^d)$.

  (ii). The left group action of $G_I^d$ on $\R^d \times \Sym^2(\R^d)$ is defined by
  \begin{equation}\label{left-action}
    (g, \kappa)\cdot(\mathfrak b, a) = (g\mathfrak b + \kappa a, (g\otimes g) a),
  \end{equation}
  for all $(g, \kappa) \in G_I^d$, $\mathfrak b\in \R^d$, $a\in\Sym^2(\R^d)$.

  (iii). The second-order tangent bundle $(\mathcal T^O M, \tau^O_M, M)$ is the fiber bundle with base space $M$, typical fiber $\R^d \times \Sym^2(\R^d)$, and structure group $G_I^d$.

  (iv). The fiber $\mathcal T^O_q M$ at $q\in M$ is called second-order tangent space to $M$ at $q$. An element $(\mathfrak b, a)_q\in \mathcal T^O_q M$ is called a second-order tangent vector at $q$. A (global or local) section of $\tau^O_M$ is called a second-order vector field.

  (v). Denote by $\mathcal T^E M$ the subbundle of $\mathcal T^O M$ consisting of all elements $(\mathfrak b, a)_q\in \mathcal T^O_q M$, $q\in M$, with $a_q$ a positive semi-definite $(2,0)$-tensors. Let $\tau^E_M = \tau^O_M|_{\mathcal T^E M}$. We call $(\mathcal T^E M, \tau^E_M, M)$ the second-order elliptic tangent bundle.
\end{definition}


\begin{remark}
  (i). We indulge in some abuse of notions. For example, the second-order vector fields should not be confused with the semisprays which are sections of the double tangent bundle $T^2M$ (e.g., \cite[Section 1.4]{Sau89}, \cite[Section IV.3]{Lan99}).

  (ii). Some authors just defined second-order vector fields as second-order operators as in Definition \ref{2-operator} (\cite[Definition 6.3]{Eme89} or \cite[Definition 2.74]{Gli11}). As soon as we choose a frame for $\mathcal T^O M$, it will be clear that second-order vector fields are identified with second-order operators.

  (iii). The authors in \cite{BD90,Gli11} define a bundle which has the It\^o group as its structure group and has the pair $(\mathfrak{b}, \sigma)$ of coefficients in It\^o SDE \eqref{SDE-mfld} as its section. They name it It\^o's bundle and denote it as $\mathcal I M$. 
  The difference is that, in our formulation, the pair $(\mathfrak{b}, \sigma\circ\sigma^*)$ of coefficients of the generator of It\^o SDE \eqref{SDE-mfld} is a section of second-order elliptic tangent bundle $\tau^E_M$. The advantage of the bundle $\tau^E_M$ is that it is a natural generalization of tangent bundle to second-order and has a good geometric interpretation, as we will see in Proposition \ref{TS-TO}.

  (iv). Note that the typical fiber $\R^d \times \Sym^2(\R^d)$ of $\tau^O_M$ is a vector space of dimension $d+\frac{d(d+1)}{2}$. But $\tau^E_M$ is not a vector bundle, since its structure group $G_I^d$ is not a linear group (subgroup of general linear group). The typical fiber of $\tau^E_M$ is $\R^d \times \Sym^2_+(\R^d)$, which is not even a vector space, so that $\tau^E_M$ is not a vector bundle either. Indeed, we may call them quadratic bundles, just as the way they call It\^o's bundle in \cite[Chapter 4]{BD90}.

  (v). The It\^o's bundle $\mathcal I M$ defined in \cite[Definition 7.17]{Gli11} is the fiber bundle over manifold $M$, with fiber $\R^d \times \L(\R^N,\R^d)$ and structure group $G_I^d$ which acts on the fiber from the left by
  \begin{equation*}
    (g, \kappa)(\mathfrak b, \sigma) = \left( g\mathfrak b + \ts{\frac{1}{2}} \tr (\kappa\circ (\sigma\otimes\sigma)), g \circ \sigma \right),
  \end{equation*}
  for all $(g, \kappa) \in G_I^d$, $\mathfrak b\in \R^d$, $\sigma\in\L(\R^N,\R^d)$. For the same reason as $\mathcal T^O M$ or $\mathcal T^E M$, It\^o's bundle $\mathcal I M$ is not a vector bundle. There is a bundle homomorphism over $M$ from $\mathcal I M$ to $\mathcal T^E M$, which maps in fibers from $\mathcal I_q M$ to $\mathcal T^E_q M$, $q\in M$, by $(\mathfrak b, \sigma) \to (\mathfrak b, \sigma\circ\sigma^*)$. It is easy to see that this bundle homomorphism is also a subjective submersion. If we identify $g\in \GL(d,\R)$ with $(g,0) \in G_I^d$, then $\GL(d,\R)$ is a subgroup of $G_I^d$. We define the \emph{Stratonovich's bundle} $\mathcal S M$ to be the reduction of $\mathcal I M$ to the structure group $\GL(d,\R)$, that is, the fiber bundle over $M$, with fiber $\R^d \times \L(\R^N,\R^d)$ and structure group $\GL(d,\R)$ which acts on the fiber from the left by
  \begin{equation*}
    g(\mathfrak b, \sigma) = (g\mathfrak b, g \circ \sigma).
  \end{equation*}
  Unlike $\mathcal T^O M$ or $\mathcal I M$, Stratonovich's bundle $\mathcal S M$ is indeed a vector bundle, and the tangent bundle $T M$ is a vector subbundle of $\mathcal S M$. It can be expected that Stratonovich's bundle is a natural bundle to formulate Stratonovich SDEs. But, in this paper, we mainly focus on It\^o SDEs and their generators.
\end{remark}

It is natural to regard the differential operators
\begin{equation}\label{frame}
  \left\{ \vf{x^i}, \frac{\pt^2}{\pt x^j \pt x^k}: 1\le i\le d, 1\le j\le k \le d \right\}
\end{equation}
as a local frame of $\mathcal T^O M$ over the local chart $(U,(x^i))$ on $M$. 
In the sequel, we will usually shorten them by
\begin{equation*}
  \left\{ \pt_i,\ \pt_j\pt_k: 1\le i\le d, 1\le j\le k \le d \right\}.
\end{equation*}
We make the convention that $\pt_k\pt_j = \pt_j\pt_k$ for all $1\le j\le k \le d$. A second-order vector field $(\mathfrak b,a)$ is expressed in terms of this local frame by
$$(\mathfrak b,a) = \mathfrak{b}^i \pt_i + \ts{\frac{1}{2}a^{jk}} \pt_j\pt_k.$$
In this way, every second-order vector field can be regarded as a second-order operator and vice versa. In particular, the generator $A^X$ of an $M$-valued diffusion process $X$, for example the generator \eqref{generator} of the It\^o SDE, is a time-dependent second-order vector field, so that we can rewrite $A^X$ as $A^X_t = (\mathfrak b(t),(\sigma\circ\sigma^*)(t))$.

The tangent bundle $T M$ is a subbundle (but not a vector subbunddle) and also an embedded submanifold of $\mathcal T^O M$, as the bundle monomorphism
\begin{equation}\label{iota}
  \iota: (TM, \tau_M, M)\to\left(\mathcal T^O M, \tau^O_M, M\right), \quad v_q\mapsto (v,0)_q
\end{equation}
is also an embedding. However, there is no canonical bundle epimorphism from $\mathcal T^O M$ to $T M$ which is a left inverse of $\iota$ and linear in fiber. We call such a bundle epimorphism a \emph{fiber-linear bundle projection} from $\mathcal T^O M$ to $T M$. 
The choice of such a bundle epimorphism is exactly the choice of a linear connection on $M$. More precisely, we have the following connection correspondence properties, the first of which can also be found in \cite[Section 2.9]{Gli11}.

\begin{proposition}[Connection correspondence]\label{induced-conn}
  Any linear connection on $M$ induces a fiber-linear bundle projection from $\mathcal T^O M$ to $T M$. Conversely, any fiber-linear bundle projection from $\mathcal T^O M$ to $T M$ induces a torsion-free linear connection on $M$.
\end{proposition}

\begin{remark}
  The connection correspondence is similar to the correspondence between horizontal subbundles of the tangent bundle of a vector bundle and connections on this vector bundle, cf. \cite[Section 3.1]{Sau89}.
\end{remark}

\begin{proof}
  Let $(\Gamma_{ij}^k)$ be the Christoffel symbols of a linear connection $\nabla$ on $M$. Define a projection by
  \begin{equation}\label{varrho}
    \varrho_\nabla : \mathcal T^O M \to T M, \quad (\mathfrak b, a)_q \mapsto \left( \mathfrak b^i + \ts{\frac{1}{2}} a^{jk} \Gamma^i_{jk}(q) \right) \pt_i\big|_q.
  \end{equation}
  Clearly, $\varrho_\nabla$ is linear in fiber and $\varrho_\nabla\circ \iota = \id_{TM}$. Conversely, let $\varrho: \mathcal T^O M\to T M$ be a fiber-linear bundle projection. Then, on each coordinate chart $(U,(x^i))$ around $q\in M$, there exists a diffeomorphism $B_U: U \to \L(\Sym^2(\R^d), \R^d)$, such that
  \begin{equation*}
    \varrho( \mathfrak b, a ) = \left( \mathfrak b^i + B_U(q)(a)^i \right) \pt_i\big|_q, \quad ( \mathfrak b, a ) \in \mathcal T^O_q M, q\in U.
  \end{equation*}
  The family of diffeomorphisms $(B_U)$ determines a spray and then a torsion-free linear connection on $M$ (see, e.g., \cite[Section IV.3]{Lan99}). The torsion-freeness follows from the symmetry of $B_U$'s.
\end{proof}

Observe that a group action of $\GL(d,\R)$ on $\Sym^2(\R^d)$ can be separated  from \eqref{left-action}, which is given by $g\cdot a = (g\otimes g) a$. Thus the second component $a$ of each element $( \mathfrak b, a ) \in \mathcal T^O_q M$ can be regarded as a $(2,0)$-tensor. Recall that we denote by $\Sym^2(TM)$ the bundle of $(2,0)$-tensors on $M$, then there is a canonical bundle epimorphism
\begin{equation}\label{varrho-hat}
  \hat\varrho: \mathcal T^O M \to \Sym^2(TM), \quad (\mathfrak b, a)_q \mapsto a_q,
\end{equation}
whose kernel is the image of $\iota$. Conversely, we also have a similar connection correspondence property for $\Sym^2(TM)$, as in Proposition \ref{induced-conn}. That is, a linear connection $\nabla$ on $M$ induces a fiber-linear bundle monomorphism from $\Sym^2(TM)$ to $\mathcal T^O M$, which is a right inverse of $\hat\varrho$ and given by
\begin{equation}\label{iota-hat}
  \hat\iota_{\nabla}: \Sym^2(TM) \to \mathcal T^O M, \quad a_q \mapsto a^{ij} \left( \pt_i\pt_j \big|_q - \Gamma^k_{ij}(q) \pt_k \big|_q \right) = a^{ij} \nabla^2_{\pt_i,\pt_j}\big|_q
\end{equation}
where $\nabla^2$ is the second covariant derivative \cite[Subsection 2.2.2.3]{Pet16} (which is also called the Hessian operator when acting on smooth functions \cite{Jos17}). In other words, $\nabla^2_{\pt_i,\pt_j}|_q = \hat\iota_{\nabla}(dx^i \odot dx^j |_q)$, where $\odot$ is the symmetrization operator on $T^2 M$.

Combining \eqref{iota} and \eqref{varrho} together, we have the following short exact sequence:
\begin{equation}\label{exact-seq}
  0 \longrightarrow TM \stackrel{\iota}{\longrightarrow} \mathcal T^O M \stackrel{\hat\varrho}{\longrightarrow} \Sym^2(TM) \longrightarrow 0.
\end{equation}
Proposition \ref{induced-conn} and \eqref{varrho-hat}, \eqref{iota-hat} imply that when a linear connection $\nabla$ is given, the sequence is also split, in the fiber-wise sense. The induced decomposition
\begin{equation}\label{dcpst}
  \mathcal T^O M = \iota(TM) \oplus \hat\iota_{\nabla} \left( \Sym^2(TM) \right) \cong TM \oplus \Sym^2(TM),
\end{equation}
where both the first direct sum $\oplus$ and the isomorphism $\cong$ are in the fiber-wise sense (but not bundle isomorphism and Whitney sum), while the second direct sum is the Whitney sum, and is given by
\begin{equation}\label{dcpst-tang}
  (\mathfrak b, a)_q = b^i \pt_i\big|_q + \ts{\frac{1}{2}} a^{ij} \nabla^2_{\pt_i,\pt_j}\big|_q \mapsto (b_q, a_q),
\end{equation}
for $b_q = ( \mathfrak b^i + \ts{\frac{1}{2}} a^{jk} \Gamma^i_{jk}(q) ) \pt_i |_q \in T_qM$. A similar short exact sequence as \eqref{exact-seq} holds with $\mathcal T^E M$ and $\Sym^2_+(TM)$ in place of $\mathcal T^O M$ and $\Sym^2(TM)$, respectively.

Now we introduce a subclass of semimartingales on manifolds which contains diffusions. We call the $M$-valued process $X= \{X(t)\}_{t\in[t_0,\tau)}$ an \emph{It\^o process}, if there exists a $\{\Pred_t\}$-adapted continuous $\mathcal T^E M$-valued process $\{(\mathfrak b, a)(t)\}_{t\in[t_0,\tau)}$ satisfying $(\mathfrak b, a)(t) \in \mathcal T^E_{X(t)} M$ for each $t\in[t_0,\tau)$, such that for every $f\in C^\infty(\R\times M)$,
$M^{f,X}(t) := f(t,X(t)) - f(t_0,X(t_0)) - \int_{t_0}^t (\vf t+ \A^X )f(s,X(s)) ds$, $t\in [t_0,\tau)$ is a real-valued $\{\Pred_t\}$-martingale, where
$\A^X_t = (\mathfrak b, a)(t) = \mathfrak{b}^i(t) \pt_i + \ts{\frac{1}{2}} a^{ij}(t) \pt_i\pt_j$. We call the process $\{(\mathfrak b, a)(t)\}_{t\in[t_0,\tau)}= \{\A^X_t\}_{t\in[t_0,\tau)}$ the random generator of $X$. A similar notion ``Brownian semimartingale'' is also used in the literature (e.g., \cite{Dri92}). If $X$ is a diffusion with generator $A^X_t = (\mathfrak b(t), a(t))$, then it is an It\^o process with random generator $\A^X_t = A^X_{(t,X(t))} = (\mathfrak b(t,X(t)), a(t,X(t)))$. The difference between It\^o processes and diffusions is that the randomness of the random generator of the former can not only appear on the base manifold $M$, but also on the fibers.

Then, we can define forward mean derivatives in a coordinate-free way, without relying on linear connections.

\begin{definition}[Mean derivatives]
  For an $M$-valued It\^o process $X= \{X(t)\}_{t\in[t_0,\tau)}$, we define its (forward) mean derivatives $(DX(t), QX(t))$ at time $t\in [t_0,\tau)$ by
  \begin{equation*}
    (DX(t), QX(t)) = (\mathfrak{b}, a)(t) \in \mathcal T^E_{X(t)} M,
  \end{equation*}
  where $(\mathfrak{b}, a)$ is the random generator of $X$.
\end{definition}

Comparing with forward mean derivatives defined in local coordinates before, we have the following relations. The proof follows the lines of \cite[Lemma 9.4]{Gli11}.

\begin{lemma}
  Given an $M$-valued It\^o process $X= \{X(t)\}_{t\in[t_0,\tau)}$ and a coordinate chart $(U,(x^i))$ centered at $q\in M$.

  (i). In the event $\{X(t) \in U\}$, $QX(t)$ has the coordinate expression \eqref{quad-deriv} and
  \begin{equation*}
    (D X)^i(t) = \lim_{\e\to0^+} \E\left[ \frac{X^i(t+\e)-X^i(t) }{\e} \bigg| \Pred_t \right].
  \end{equation*}

  (ii). Given a linear connection $\nabla$ on $M$, we have, under the conditional probability $\P(\cdot | X(t) = q)$, that
  \begin{equation}\label{relation-two-drvtv}
    (D_\nabla X)^i(t) = (D X)^i(t) + \frac{1}{2} \Gamma^i_{jk}(X(t)) (QX)^{jk}(t).
  \end{equation}
\end{lemma}

It follows from \eqref{relation-two-drvtv} that the map $\varrho_\nabla$ in \eqref{varrho} acts on the generator $A^X$ of a diffusion $X$ by
\begin{equation}\label{varrho-nabla-gen}
  \varrho_\nabla (A^X_{(t,X(t))}) = \varrho_\nabla (DX(t), QX(t)) = D_\nabla X(t)
\end{equation}

For a time-dependent second-order vector field $A_t = (\mathfrak{b}(t), a(t))$, we can take MDEs \eqref{Nelson-SDE-mfld} to set up a new type of MDEs by using the mean derivatives as follows:
\begin{equation}\label{Nelson-SDE-mfld-2}\left\{
  \begin{aligned}
    D X(t) &= \mathfrak b(t,X(t)), \\
    Q X(t) &= a(t,X(t)).
  \end{aligned} \right.
\end{equation}
Then, similarly to Definitions \ref{weak-sol-Nelson} and \ref{unique-Nelson}, we may also define solutions and uniqueness in law for MDEs \eqref{Nelson-SDE-mfld-2}. We call a solution of \eqref{Nelson-SDE-mfld-2} an \emph{integral process} of $A = (A_t)$. Note that the system \eqref{Nelson-SDE-mfld-2} does not rely on linear connections. The equivalence of the well-posedness of \eqref{Nelson-SDE-mfld-2} and the martingale problem in Definition \ref{mtgl-prob} is easy to verify. When a linear connection is specified, the system \eqref{Nelson-SDE-mfld-2} and martingale problem associated with $A^X$ in \eqref{generator} are both equivalent to the It\^o SDE \eqref{SDE-mfld} and MDEs \eqref{Nelson-SDE-mfld}.

\section{Stochastic jets}\label{sec-3}

In classical differential geometry, a tangent vector to a manifold may be defined as an equivalence class of curves passing through a given point, where two curves are equivalent if they have the same derivative at that point \cite[Chapter 3]{Lee13}. 
This idea can be generalized to higher-order cases, which leads to the notion of jets. The jet structures allow us to translate a system of differential equations to a system of algebraic equations, and make it more intuitive to study the symmetries of systems of differential equations.

In this chapter we shall generalize these ideas to the stochastic case. We will first give an equivalent description to the second-order elliptic tangent bundle $\tau^E_M$ by constructing an equivalence relation on diffusions. Then we will define the stochastic jets and figure out the ``jet-like'' bundle structure involved in the space of stochastic jets. Finally, we shall see that the bundle structure is the appropriate platform to formulate SDEs intrinsically. In the next chapter, we will apply stochastic jets to study stochastic symmetries.


\subsection{The stochastic tangent bundle}

Recall that a tangent vector can be represented as a equivalence classes of smooth curves that have the same velocity at the base point. This leads to the following equivalent definition of tangent bundle $TM$:
\begin{equation}\label{equiv-tang-bd}
  TM \cong \left\{[\gamma]_q: \gamma\in C^\infty_{(0,q)}(M), q\in M \right\},
\end{equation}
where $C^\infty_{(0,q)}(M)$ is the set of all smooth curves on $M$ that pass through $q$ at time $t=0$, and the equivalence relation is defined as $\gamma,\tilde\gamma\in C^\infty_{(0,q)}(M)$ are equivalent if and only if $(f\circ\gamma)'(0)=(f\circ\tilde\gamma)'(0)$ for every real-valued smooth function $f$ defined in neighborhood $q$. If we replace smooth curves by diffusion processes, and time derivatives by mean derivatives, then we get the following definition.

\begin{definition}[The stochastic tangent bundle]\label{st-tg-bd}
  Two $M$-valued diffusion processes $X=\{X(t)\}_{t\in[0,\tau)}$, $Y=\{Y(t)\}_{t\in[0,\sigma)}$ are said to be \emph{stochastically equivalent at $(t,q)\in\R\times M$}, if, almost surely, $X(t)=Y(t)=q$ and $D(f\circ X)(t) = D(f\circ Y)(t)$ for all $f\in C^\infty(M)$. The equivalence class containing $X$ is called the \emph{stochastic tangent vector} of $X$ at $q$ and is denoted by $j_{(t,q)} X$. When $t=0$, we denote $j_q X:= j_{(0,q)}X$ in short. Let $I_{(t,q)}(M)$ be the set of all $M$-valued diffusion processes starting from $q$ at time $t$. The \emph{stochastic tangent bundle} of $M$ is the set
  \begin{equation*}
    \mathcal T^S M = \{ j_q X: X\in I_{(0,q)}(M), q\in M \}.
  \end{equation*}
\end{definition}

Note that since $X, Y$ are $M$-valued diffusion processes, $f(X)$ and $f(Y)$ are real-valued It\^o processes, and hence their mean derivatives exists.

At this stage, we have not yet touched the jet-like formulation even though we used the jet-like notation $j_q X$. Indeed, if one follows strictly the definition of jet bundles over the trivial bundle $(\R\times M, \pi, \R)$, it is more rational to use the time line $\R$ as ``source'' and the manifold $M$ as ``target'' (cf. \cite[Example 4.1.16]{Sau89}). But here we just assign the ``target'' to the manifold $M$, because, roughly speaking, one can talk about the velocity of a smooth curve at a moment $t$, but not about the generator of a diffusion at a moment $t$. Instead, we can talk about the generator of a diffusion at a position $q\in M$. Later on, we will define the ``bona fide'' stochastic jet space which possess the time line $\R$ as ``source'' and the manifold $M$ as ``target''.


Similarly to the one-to-one correspondence between tangent space and space of equivalence classes of smooth curves, we have the following:

\begin{proposition}\label{TS-TO}
  There is a one-to-one correspondence between the stochastic tangent bundle $\mathcal T^S M$ and the second-order elliptic tangent bundle $\mathcal T^E M$.
\end{proposition}
\begin{proof}
  For an $M$-valued diffusion process $X\in I_{(0,q)}(M)$, $q\in M$, we denote by $A^X$ its generator. Then the map $j_q X \mapsto A^X_{(0,q)} = (DX(0), QX(0))$ defines a one-to-one correspondence between $\mathcal T^S M$ and $\mathcal T^E M$. The inverse map is $A_q = (\mathfrak b,a)_q \mapsto j_q X^A$, where $A$ is a section of $\mathcal T^E M$ (i.e., an elliptic second-order operator) smoothly extending the element $A_q\in \mathcal T^E_q M$, and $X^A\in I_{(0,q)}(M)$ is a diffusion processes having $A$ as its generator.
\end{proof}

Therefore, the stochastic tangent bundle $\mathcal T^S M$ admit a smooth structure which makes it to be a smooth manifold diffeomorphic to $\mathcal T^E M$, and hence it is a bona fide fiber bundle over $M$. In the sequel, we will identify $\mathcal T^S M$ with $\mathcal T^E M$ without ambiguity. And the projection map from $\mathcal T^S M$ to $M$ will be denoted by $\tau^S_M$, that is, $\tau^S_M(j_q X) = q$ for any $j_q X\in\mathcal T^S M$.

\begin{definition}[Canonical coordinate system on $\mathcal T^S M$]
  Let $(U, (x^i))$ be an coordinate system on $M$. The induced canonical coordinate chart $(U^{(1)}, x^{(1)})$ on $\mathcal T^S M$ is defined by
  \begin{equation*}
    U^{(1)} := \left\{ j_q X: q \in U, X\in I_{(0,q)}(M) \right\}, \quad x^{(1)} := (x^i, D^i x, Q^{jk} x),
  \end{equation*}
  where $x^i(j_q X) = x^i(q)$, $D^i x(j_q X) = (DX)^i(0)$ and $Q^{jk} x(j_q X) = (Q X)^{jk}(0)$.
\end{definition}

Our slightly ambiguous notations $D^i x$ and $Q^{jk} x$ are chosen so as to avoid the worse one $Qx^{jk}$.

When a linear connection $\nabla$ is provided, we can also define the coordinates via the $\nabla$-mean derivative $D_\nabla$ instead of $D$, as follows:
\begin{equation*}
  D^i_\nabla x(j_q X) := (D_\nabla X)^i(0).
\end{equation*}
Then, $x^{(1)}_\nabla := (x^i, D^i_\nabla x, Q^{jk} x)$ also forms a coordinate system on $\mathcal T^S M$, which we call the $\nabla$-canonical coordinate system. It follows from relation \eqref{relation-two-drvtv} that
\begin{equation}\label{partial-coord}
  D_\nabla^i x = D^ix + \ts{\frac{1}{2}} (\Gamma^i_{jk}\circ x) Q^{jk} x.
\end{equation}

Using the identification of elements $j_q X \in \mathcal T^S_q M$ and $(\mathfrak b,a)_q \in \mathcal T^E_q M$ via Proposition \ref{TS-TO}, as well as their relations with the element $(b_q, a_q)\in TM \oplus \Sym^2(TM)$, via \eqref{dcpst-tang}, we have $D^i x(j_q X) = \mathfrak b^i$, $D^i_\nabla x(j_q X) = b^i = \mathfrak b^i + \ts{\frac{1}{2}} a^{jk} \Gamma^i_{jk}(q)$ and $Q^{jk} x(j_q X) = a^{jk}$. In this way the fiber-linear bundle projection $\varrho_\nabla$ of \eqref{varrho} maps, under the canonical coordinates $(x,\dot x)$ on $TM$, as follows:
\begin{equation}\label{varrho-2}
  \dot x^i \circ \varrho_\nabla (j_q X) = \left( D^ix + \ts{\frac{1}{2}} (\Gamma^i_{jk}\circ x) Q^{jk} x \right)(j_q X) = D^i_\nabla x(j_q X),
\end{equation}
so that $D_\nabla^ix = \dot x^i \circ \varrho_\nabla$. Therefore, $(x^i, D_\nabla^ix)$ is a partial coordinate system on $\mathcal T^S M$ that coincides with $(x^i,\dot x^i)$ when restricted on $TM$. Moreover, the decomposition in \eqref{dcpst-tang} yields the following expressions for second-order vector fields:
\begin{equation}\label{2-tangent-equiv}
  (Dx, Qx) = D^i x \pt_i + \ts{\frac{1}{2}} Q^{jk} x \pt_j\pt_k = D_\nabla^i x \pt_i + \ts{\frac{1}{2}} Q^{jk} x \nabla^2_{\pt_j,\pt_k}.
\end{equation}

Similarly to Definition \ref{st-tg-bd}, we define a $\nabla$-dependent equivalence relation as follows:
\begin{definition}
  Two $M$-valued diffusion processes $X=\{X(t)\}_{t\in[0,\tau)}$, $Y=\{Y(t)\}_{t\in[0,\sigma)}$ are said to be $\nabla$-\emph{stochastically equivalent at $(t,q)\in\R\times M$}, if, almost surely, $X(t)=Y(t)=q$ and $D_\nabla X(t) = D_\nabla X(t)$. The equivalence class containing $X$ is called the $\nabla$-\emph{tangent vector} of $X$ at $q$ and is denoted by $j^\nabla_{(t,q)} X$. When $t=0$, we denote $j^\nabla_q X:= j^\nabla_{(0,q)}X$ for short.
\end{definition}

Then, similarly to Proposition \ref{TS-TO}, one can show that the tangent bundle $TM$ can be identified with the following set of equivalent classes of diffusions:
\begin{equation}\label{equiv-tang-bd-2}
  \left\{ j_q^\nabla X: X\in I_{(0,q)}(M), q\in M \right\},
\end{equation}
via $j_q^\nabla X\mapsto D_\nabla X(0)$. Under this identification, it follows from \eqref{varrho-nabla-gen} that $j_q^\nabla X = \varrho_\nabla(j_q X)$. Clearly, if we regard all smooth curves as special diffusions, then the partition determined by \eqref{equiv-tang-bd} is the restriction of the one determined by \eqref{equiv-tang-bd-2} to the set of all smooth curves.

\begin{remark}\label{symm-bundle-tang}
In presence of a linear connection $\nabla$ on $M$, one can easily follow Definition \ref{st-tg-bd} and Proposition \ref{TS-TO} with $D_\nabla$ in place of $D$, to verify the one-to-one correspondence between the set $\mathcal T^S M$ of equivalent classes and the Whitney sum $TM \oplus \Sym^2_+(TM)$, which brings back to the fiber-wise isomorphism \eqref{dcpst}. But since such kind of correspondence need to specify beforehand a linear connection, we still endow $\mathcal T^S M$ with the structure of $\mathcal T^E M$ instead of that of $TM \oplus \Sym^2(TM)$ in this paper, although the latter is also feasible and may provide easier calculations.
\end{remark}

\subsection{The stochastic jet space}

In classical jet theory, for the trivial bundle $(\R\times M, \pi, \R)$, there is a one-to-one correspondence between 1-jets and tangent vectors, and there is a canonical diffeomorphism between the first-order jet bundle $J^1 \pi$ and $\R\times TM$ \cite[Example 4.1.16]{Sau89}.

Now using similar ideas, we will introduce the ``bona fide'' stochastic jet space. The key is to modify the definition of stochastic tangent vectors, to involve the time line $\R$ as the ``source'' as well as to randomize the initial datum of the diffusion processes. Intuitively, an $M$-valued diffusion process $X$ can be regarded as a random ``section'' of the trivial ``bundle'' $(\R\times M, \pi, \R)$ which is merely continuous in time and depends on the sample point $\omega$.

For a metric space $(F,d)$, we denote by $L^0(\Omega, F)$ the quotient space of all $F$-valued random elements, by the following equivalence relation: two random elements are equivalent if and only if they are identical almost surely. We endow $L^0(\Omega, F)$ with the topology of the following $\P$-essential metric (cf. \cite[Section 43]{Mun75}):
\begin{equation*}
  \rho(\xi,\zeta) = \inf\{c>0: \P(d(\xi,\zeta)>c) =0 \} \wedge 1.
\end{equation*}

\begin{definition}
  Two $M$-valued diffusion processes $X=\{X(s)\}_{s\in[t,\tau)}$, $Y=\{Y(s)\}_{s\in[t,\sigma)}$ starting at time $t$, are said to be \emph{stochastically equivalent at $t\in \R$}, if, almost surely, $X(t)= Y(t)$ and $(DX(t), QX(t)) = (DY(t), QY(t))$. The equivalence class containing $X$ is called the \emph{stochastic jet} of $X$ at $t$, denoted by $j_t X$. Let $I_t(M)$ be the set of all $M$-valued diffusion processes starting at time $t$. Then the \emph{stochastic jet space} of $M$ is the set
  \begin{equation*}
    \mathcal J^S M = \{ j_t X: X\in I_t(M), t\in\R \}.
  \end{equation*}
  The functions $\pi^S_1$ and $\pi^S_{1,0}$, called stochastic source and target projections, are defined by
  \begin{equation*}
    \pi^S_1 : \mathcal J^S M \to \R, \quad j_t X \mapsto t,
  \end{equation*}
  and
  \begin{equation*}
    \pi^S_{1,0} : \mathcal J^S M \to \R\times L^0(\Omega, M), \quad j_t X \mapsto (t,X(t)).
  \end{equation*}
\end{definition}

In the above definition, since $\pi_M\circ\phi = \id_M$, we have $\pi(Y) = \pi_M\circ\phi(X) = X$ a.s., that is, $X$ is the projection of $Y$.

To characterize the relation between $\mathcal J^S M$ and $\mathcal T^S M$ (or $\mathcal T^E M$), we need the following definitions.

\begin{definition}[Horizontal subspace]\label{horizontal-subspace}
  Let $(E,\pi_M, M)$ be a fiber bundle. The horizontal subspace of $L^0(\Omega,E)$ is defined by
  \begin{equation*}
    L^h(\Omega; \pi_M) := \{ \phi \circ \xi \in L^0(\Omega, E): \phi \text{ is a section of } \pi_M, \xi\in L^0(\Omega, M) \}.
  \end{equation*}
\end{definition}


An element of the horizontal subspace $L^h(\Omega; \tau^E_M)$ of $L^0(\Omega, \mathcal T^E M)$ is then of the form $A \circ \xi$, where $A$ is a section of $\tau^E_M$ and $\xi\in L^0(\Omega, M)$. Such an element $A \circ \xi$ will be denoted by $A_\xi$. By the correspondence of $\mathcal T^S M$ and $\mathcal T^E M$, one can easily get the following equivalent definition for $L^h(\Omega; \tau^E_M)$,
\begin{equation*}
  L^h(\Omega; \tau^E_M) = L^h(\Omega; \tau^S_M) := \{j_{X(0)} X: X\in I_0(M) \} \subset L^0(\Omega, \mathcal T^S M).
\end{equation*}
The correspondence is given explicitly by
\begin{equation*}
  j_{X(0)} X = A^X_{X(0)} = (DX(0), QX(0)), \quad \text{or} \quad
  A_\xi = j_\xi X^{A_\xi}.
\end{equation*}
where $X^{A_\xi}$ is an $M$-valued diffusion with generator $A$ and with $X^{A_\xi}(0) = \xi$ a.s..

\begin{proposition}\label{diff}
  The stochastic jet space $\mathcal J^S M$ is trivial. More precisely, we have the homeomorphism
  $$\mathcal J^S M \cong \R\times L^h(\Omega; \tau^S_M),$$
  given by $j_t X \mapsto (t, j_{X(t)} (\theta_t X))$, for any $X\in I_t(M)$, where $\theta_t$ is the shift operator on $\C$, that is, $\theta_t \omega(\cdot) = \omega(\cdot +t)$.
\end{proposition}
\begin{proof}
  The homeomorphism $\mathcal J^S M \cong \R\times \mathcal J^S_0 M$ is given by $j_t X \mapsto (t, j_0 (\theta_t X))$. The homeomorphism $\mathcal J^S_0 M \cong L^h(\Omega; \tau^S_M)$ is given by $j_0 X \mapsto j_{X(0)} X$, whose inverse map is $A_\xi \mapsto j_0 X^{A_\xi}$.
\end{proof}


\begin{definition}[Stochastic fibered space]
  (i) Given a fiber bundle $(E,\pi_M, M)$ with total space $E$, base space $M$ and typical fiber manifold $F$, the stochastic fibered space associated with it is the triplet $(E^S,\pi^S_M, M)$ where
  \begin{equation*}
    E^S := \{ (q, \xi): q\in M, \xi\in \hat L(\Omega, E_q) \},
  \end{equation*}
  $\pi^S_M: E^S\to M$ is the natural projection given by $\pi^S_M(q, \xi) = q$, and $\hat L(\Omega,F)$ is a subspace of $L^0(\Omega,F)$, with $E_q$ denoting the fiber of $\pi_M$ over $q$. The fiber bundle $E$ is called model bundle of $E^S$. There is a family of projections $\{\pi_\omega\}_{\omega\in\Omega}$ from the stochastic fiber manifold $E^S$ to its model bundle $E$, defined by
  \begin{equation*}
    \pi_\omega: E^S\to E, \quad (q, \xi) \mapsto (q, \xi(\omega)).
  \end{equation*}

  (ii) A global section of $(E^S,\pi^S_M, M)$ is called a random global section. A random local section is a map $\sigma: U \to E$ defined on some measurable subset $U\subset \Omega\times M$ and such that, for almost all $\omega\in \Omega$, $\sigma(\omega): U_\omega \to E$ is a local section of $(E,\pi_M, M)$, where $U_\omega = U\cap (\{\omega\}\times M)$.
\end{definition}

Note that a random global section is a random local section defined on all $\Omega\times M$.

It follows from Proposition \ref{diff} that the stochastic jet space $(\mathcal J^S M, \pi_1^S, \R)$ is a stochastic fibered space, whose associated model bundle is $(\R\times \mathcal T^S M, \pi_1, \R)$. Just like the first-order jet bundle $J^1 \pi$ which is diffeomorphic to $\R\times TM$, the model bundle $\R\times \mathcal T^S M$ is itself a jet bundle and also has two bundle structures, with base space $\R$ and $\R\times M$, respectively. The corresponding source and target projections are defined, respectively by
  \begin{equation*}
    \pi_1 : \R\times \mathcal T^S M \to \R, \quad (t, j_q X) \mapsto t,
  \end{equation*}
  and
  \begin{equation*}
    \pi_{1,0} : \R\times \mathcal T^S M \to \R \times M, \quad (t, j_q X)\mapsto (t,q).
  \end{equation*}
Moreover, we will denote the natural projection from $\R\times \mathcal T^S M$ to $\mathcal T^S M$ by $\pi_{0,1}$. This projection map is indeed a bundle homomorphism from $(\R\times \mathcal T^S M, \pi_{1,0}, \R\times M)$ to $(\mathcal T^S M, \tau^S_M, M)$, whose projection is the natural projection from $\R\times M$ to $M$, denoted by $\hat\pi$.

Similarly to Proposition \ref{diff}, we have the following diffeomorphisms for the model bundle $\R\times \mathcal T^S M$:
\begin{equation*}
  \{ j_{(t,q)} X: X\in I_{(t,q)}(M), t\in\R, q\in M \} \cong \R\times \mathcal T^S M \cong \R\times \mathcal T^E M,
\end{equation*}
which is given by
\begin{equation}\label{diff-1}
  j_{(t,q)} X \mapsto (t, j_q (\theta_t X)) \mapsto A^X_{(t,q)} = (t, DX(t), QX(t)),
\end{equation}
for any $X\in I_{(t,q)}(M)$, where $A^X$ is the generator of $X$ as a section of $\R\times \mathcal T^E M$ (i.e., a time-dependent elliptic second-order differential operator). Furthermore, the proof of Proposition \ref{TS-TO} allows us to find simply the inverse maps, especially for the second diffeomorphism. That is, for any $(t,A_q) = (t,\mathfrak b,a) \in \pi_{1,0}^{-1}(t,q)$,
\begin{equation}\label{diff-2}
  (t,A_q) = (t,\mathfrak b,a) \mapsto \left(t, j_q (\theta_t X^A)\right) \mapsto j_{(t,q)} X^A,
\end{equation}
where $A$ is a section of $\R\times \mathcal T^E M$ such that $A_{(t,q)} = A_q$, and $X^A\in I_{(t,q)}(M)$ is a diffusion process having $A$ as its generator.

The ``stochastic target'' of $\mathcal J^S M$, i.e., the trivial bundle $(\R\times L^0(\Omega, M), \pi^S, M)$, is another example of stochastic fibered spaces. Its model bundle is the trivial bundle $(\R\times M, \pi, \R)$. The graph of an $M$-valued stochastic process defined on a random time interval $[0,\tau)$ is a random (local) section of $(\R\times L^0(\Omega, M), \pi^S, \R)$.
The projection of $\pi_\omega$ on the targets from $\R\times L^0(\Omega, M)$ to $\R\times M$ is denoted by $\hat \pi_\omega$.

We may summarize how all these maps fit together by the following diagram:

\vspace{1mm}
\begin{center}
\begin{tikzcd}
\mathcal J^S M \cong \R\times L^h(\Omega; \tau^S_M) \arrow{rr}{\pi_\omega} \arrow{d}[swap]{\pi^S_{1,0}} \arrow{rdd}[crossing over]{\pi^S_1} & &
\R\times \mathcal T^S M \arrow{d}{\pi_{1,0}} \arrow{ldd}[crossing over, left]{\pi_1} \arrow{r}{\pi_{0,1}} & \mathcal T^S M\cong\mathcal T^E M \arrow{d}{\tau^S_M} \arrow[hookleftarrow]{r}{\iota} & T M \arrow{ld}{\tau_M} \\
\R\times L^0(\Omega, M) \arrow{rr}{\hat\pi_\omega} \arrow{dr}[swap]{\pi^S} & &
\R\times M \arrow{ld}{\pi} \arrow{r}{\hat\pi} & M & \\
& \R & & &
\end{tikzcd}
\end{center}
\vspace{1mm}

When a linear connection is specified on $M$, one can easily obtain, similarly to \eqref{diff-1}, the following homeomorphism:
\begin{equation*}
  \left\{ j^\nabla_t X: X\in I_t(M), t\in\R \right\}\cong \R\times L^h(\Omega; \tau_M), \quad j^\nabla_t X \mapsto \left(t, j^\nabla_{X(t)} (\theta_t X) \right),
\end{equation*}
and the following diffeomorphisms:
\begin{equation*}
  \left\{ j^\nabla_{(t,q)} X: X\in I_{(t,q)}(M), t\in\R, q\in M \right\} \cong \R\times \left\{ j_q^\nabla X: X\in I_{(0,q)}(M), q\in M \right\} \cong \R\times T M \cong J^1 \pi,
\end{equation*}
where the first two diffeomorphisms are given by
\begin{equation*}
  j^\nabla_{(t,q)} X \mapsto \left(t, j^\nabla_q (\theta_t X) \right) \mapsto (t, D_\nabla X(t)),
\end{equation*}
and the last one is due to the classical theory.

\subsection{Intrinsic formulation of SDEs}

With the classical machinery of jet structures, it is possible to translate differential equations into algebraic equations on jet bundle \cite{Sau89}. In this section, we follow this way to formulate intrinsic SDEs.

For a subset $S$ of the model bundle $\R\times \mathcal T^S M$ and $t\in\R$, we denote by $S_t$ the intersection of $S$ with the fiber $\{t\} \times \mathcal T^S M$.

\begin{definition}
  A stochastic differential equation on $M$ is a closed embedded submanifold $S$ of the model jet bundle $\R\times \mathcal T^S M$ with $S_0 \ne\emptyset$. A (local) solution of the stochastic differential equation $S$ is a triple $X$, $(\Omega,\F,\P)$, $\{\Pred_t\}_{t\ge0}$, where
  \begin{itemize}
    \item[(i)] $(\Omega,\F,\P)$ is a probability space, and $\{\Pred_t\}_{t\ge0}$ is a past filtration of $\F$ satisfying the usual conditions,
    \item[(ii)] $X = \{X(t)\}_{t\in[0,\tau)}$ is a $\{\Pred_t\}$-adapted $M$-valued diffusion process over $[0,\tau)$, where $\tau$ is a $\{\Pred_t\}$-stopping time, and
    \item[(iii)] almost surely $j_t X = (t, j_{X(t)} (\theta_t X)) \in S$ for every $t\in [0,\tau)$.
  \end{itemize}
\end{definition}

\begin{remark}
  (i). The condition that $S_0 \ne\emptyset$ is just for convenience, in order to set the initial time at $t=0$.

  (ii). There is an equivalent way to formulate the solution of a stochastic differential equation $S$. That is, a (local) solution is a pair $(P,\tau)$, where $P$ is a probability measure on $(\C,\B(\C),\{\B_t\})$ and $\tau$ is a $\{\B_t\}$-stopping time, such that for $P$-almost surely $\omega$, $j_t \omega = (t, j_{\omega(t)} (\theta_t \omega)) \in S$ for every $t\in [0,\tau(\omega))$.

\end{remark}

This definition does not look like the traditional definition of a stochastic differential equation, but we can see the relationship between the two by using coordinates. Since $S$ is a embedded submanifold of $\R\times \mathcal T^S M$, it admits a \emph{local defining function} in a neighborhood of each of its points \cite[Proposition 5.16]{Lee13}. That is, for a coordinate chart $(\R\times U^{(1)}, (t, x^{(1)}))$ of the point $(0,j_q X) \in S_0$, there is a function $\Theta: \R\times U^{(1)} \to\R^K$ where $K = \dim \mathcal T^S M - \dim S$, such that $S\cap (\R\times U^{(1)}) = \Theta^{-1}(0)$ and 0 is a regular value of $\Theta$. Then, the condition $j_t X = (t, j_{X(t)} (\theta_t X)) \in S$ before $X(t)$ leaves the neighborhood $U=\tau^S_M(U^{(1)})$ reads in local coordinates as
\begin{equation}\label{SDE-new}
  \Theta(t,x, Dx, Qx)(j_t X) = \Theta(t, X(t), DX(t), QX(t)) = 0,
\end{equation}
which defines a general MDE (in terms of mean derivatives). The use of a submanifold $S$ is therefore a way to distinguish the definition of the equation from a definition of its solutions.

As an example, the system of MDEs \eqref{Nelson-SDE-mfld-2} can be rewritten to the form \eqref{SDE-new} by setting the defining function
\begin{equation}\label{SDE-F}
  \Theta(t, x, Dx, Q x) = \left( D x - \mathfrak{b}(t,x), Q x- (\sigma\circ\sigma^*)(t,x) \right).
\end{equation}

So far we have not done anything but reformulate the basic problem of finding solutions of systems of stochastic differential equations in a more geometrical form, ideally suited to our investigation into symmetry groups thereof.

\section{Stochastic symmetries}\label{sec-4}

The symmetry group of a system of differential equations is the largest local group of transformations acting on the independent and dependent variables of the system with the property that it transform solutions of the system into other solutions \cite{Olv98}. In the stochastic case, we can proceed analogously.

All methods of this chapter work in the local case, that is, the vector fields are not necessarily complete and the bundle homomorphisms could be only locally defined.

\subsection{Prolongations of diffusions and bundle homomorphisms}\label{sec-4-1}

\begin{definition}[Prolongations of diffusions]\label{prog-diff}
  Let $X$ be an $M$-valued diffusion process defined on a stopping time interval $[t_0,\tau)$. The prolongation of $X$ is a $\mathcal T^S M$-valued process $jX$ defined by, for $\theta_t$ the shift operator,
  \begin{equation*}
    j X(t) = j_{X(t)} (\theta_t X), \quad t\in [t_0,\tau).
  \end{equation*}
\end{definition}

Note that $j_t X = (t, j_{X(t)} (\theta_t X)) = (t, j X(t))$. Thus the graph of the prolongation process $jX$ is nothing but the random section $j X$ of the stochastic jet space $\mathcal J^S M$. It is easy to see that if $X$ is an $M$-valued diffusion process, then $jX$ is a $\mathcal T^S M$-valued diffusion process.

Given two smooth manifolds $M$ and $N$, a bundle homomorphism $F$ from $(\R\times M, \pi, \R)$ to $(\R\times N, \rho, \R)$ is a projectable (or fiber-preserving) smooth map, which means it maps fibers of $\pi$ to fibers of $\rho$. Hence, there exist two smooth maps $F^0:\R\to\R$ and $\bar F:\R\times M \to N$ such that $F(t,q) = (F^0(t), \bar F(t,q))$. This leads to $\rho\circ F = F^0\circ \pi$ which is the original definition of bundle homomorphisms. We denote $F = (F^0, \bar F)$ and say that $F$ projects to $F^0$.

The following lemma shows that a bundle homomorphisms has the property that it always transforms diffusions into diffusions. One can find a proof of it in Lemma \ref{push-mixed-coord} or Corollary \ref{push-ext-generator}.

\begin{lemma}\label{push-diff-bh}
  Given a bundle homomorphism $F = (F^0, \bar F)$ from $(\R\times M, \pi, \R)$ to $(\R\times N, \rho, \R)$, where $F^0$ is a diffeomorphism, for every $M$-valued diffusion process $X = \{X(t)\}_{t\in[t_0,\tau)}$, the image of its graph (or its corresponding random local section) $\{(t,X(t)): t\in[t_0,\tau) \}$ by $F$, i.e.,
  \begin{equation*}
    \{ F(t, X(t)): t \in [t_0,\tau)\}
  \end{equation*}
  is almost surely the graph of a well-defined $N$-valued diffusion process $\tilde X$ given by
  \begin{equation}\label{F-X}
    \tilde X(s) = \bar F\left( (F^0)^{-1}(s), X((F^0)^{-1}(s)) \right),
    \quad s \in [F^0(t_0),F^0(\tau)).
  \end{equation}
\end{lemma}

As observed in Remark \ref{remark-1}, among all (deterministic) smooth maps from $\R\times M$ to $\R\times N$, the class of bundle homomorphisms is the only subclass that maps diffusions to diffusions.

\begin{definition}[Pushforwards of diffusions by bundle homomorphisms]
  We call the diffusion $\tilde X$ of Lemma \ref{push-diff-bh} the pushforward of $X$ by $F$, and write $\tilde X = F\cdot X$. When $M=N$ and $F$ is a bundle endomorphism on $(\R\times M, \pi, \R)$, we also call $F\cdot X$ the transform of $X$ by $F$.
\end{definition}

We now introduce the idea of stochastic prolongation whereby a bundle homomorphism may be extended to act upon the model jet bundle.

\begin{definition}[Stochastic prolongations of bundle homomorphisms]\label{prog-bundle-homo}
  Let $F$ be a bundle homomorphism from $(\R\times M, \pi, \R)$ to $(\R\times N, \rho, \R)$ projecting to a diffeomorphism $F^0:\R\to\R$. The stochastic prolongation of $F$ is the map $j F: \R\times \mathcal T^S M \to \R\times \mathcal T^S N$ defined by
  \begin{equation}\label{prog-morph}
    j F (j_{(t,q)} X) = j_{F(t,q)} (F\cdot X).
  \end{equation}
\end{definition}

It is easy to see from \eqref{F-X} that if $j_{(t,q)} X = j_{(t,q)} Y$, then $j_{F(t,q)} (F\cdot X) = j_{F(t,q)} (F\cdot Y)$. Therefore, the map $jF$ is well defined. By letting $F = (F^0, \bar F)$, definition \eqref{prog-morph} can be rewritten in a more evident way:
\begin{equation}\label{prog-morph-1}
  j F (t, j_q(\theta_t X)) = \big( F^0(t), j_{\bar F(t,q)} \theta_{F^0(t)}(F\cdot X) \big).
\end{equation}

The following properties are easy to check.

\begin{corollary}\label{jf-bd-morph}
  (i) The map $jF: \pi_1 \to \rho_1$ is a bundle homomorphism projecting to $F^0$. \\
  (ii) The map $jF: \pi_{1,0} \to \rho_{1,0}$ is a bundle homomorphism projecting to $F$. \\
  (iii) $j(\id_{\R\times M}) = \id_{\R\times\mathcal T^S M}$. Let $F$ and $G$ be two bundle endomorphisms on $(\R\times M, \pi, \R)$ projecting to diffeomorphisms. Then $j(F\circ G) = jF \circ jG$.
\end{corollary}

By virtue of \eqref{prog-morph-1} and Corollary \ref{jf-bd-morph}.(i), we may write $jF = (F^0, \overline{jF})$, where $\overline{jF}: \R\times \mathcal T^S M \to \mathcal T^S N$ is the smooth map given by
\begin{equation}\label{prog-morph-2}
  \overline{jF} (t, j_q(\theta_t X)) = j_{\bar F(t,q)} \theta_{F^0(t)}(F\cdot X).
\end{equation}
We can also consider the pushforward of the $\mathcal T^S M$-valued process $jX$ by the bundle homomorphism $jF$.

\begin{corollary}\label{push-diff-bd}
  Given a bundle homomorphism $F: (\R\times M, \pi, \R)\to (\R\times N, \rho, \R)$ projecting to a diffeomorphism on $\R$, and an $M$-valued diffusion process $X$, we have
  \begin{equation*}
    jF \cdot jX = j(F\cdot X).
  \end{equation*}
\end{corollary}

\begin{proof}
  It follows from \eqref{F-X}, \eqref{prog-morph-2} and Definition \ref{prog-diff} that
  \begin{equation*}
    \begin{split}
      jF \cdot jX(s) &= \overline{jF}\left( (F^0)^{-1}(s),j X((F^0)^{-1}(s)) \right) = \overline{jF}\left( (F^0)^{-1}(s),j_{X((F^0)^{-1}(s))} (\theta_{(F^0)^{-1}(s)} X) \right) \\
      &= j_{\tilde X(s)} (\theta_s \tilde X) = j\tilde X(s).
    \end{split}
  \end{equation*}
  The result follows.
\end{proof}

Now we need to investigate the coordinate representation of $jF$, in stochastic analysis terms. Before that, we introduce the stochastic version of the notion of total derivatives.

\begin{definition}[Total mean derivatives]\label{total-mean-d}
  Let $f$ be a smooth real-valued function on $\R\times M$. The total mean derivative and total quadratic mean derivative of $f$ are the unique smooth functions $\D_t f$ and $\Q_t f$ defined on $\R\times \mathcal T^S M$, with the property that if $X\in I_{(t_0,q)}(M)$ is a representative diffusion process of $j_{(t_0,q)} X$, then
  \begin{align*}
    (\D_t f)(j_{(t_0,q)} X) &= D[f(t_0, X(t_0))], \\
    (\Q_t f)(j_{(t_0,q)} X) &= Q[f(t_0, X(t_0))].
  \end{align*}
\end{definition}

There is an abuse of notations in the above definition. Indeed, the left-hand sides (LHSs) of the above two equations both involve subscripts $t$, but their RHS's do not depend on $t$. Those two equations need to be understood as that functions $\D_t f,\Q_t f$ taking their values on the point $j_{(t_0,q)} X\in\R\times \mathcal T^S M$ equal to the RHS's.

It is easy to check that the definitions of total mean derivatives are independent of the choice of representative diffusions. By It\^o's formula, we have the following coordinate representation for total mean derivatives in the local chart $(\R\times U^{(1)}, (t, x^{(1)}))$ on $\R\times \mathcal T^S M$,
\begin{align}
  \D_t f &= \frac{\pt f}{\pt t} + \frac{\pt f}{\pt x^i} D^i x + \frac{1}{2} \frac{\pt^2 f}{\pt x^j \pt x^k} Q^{jk} x, \label{local-rep-total-mean} \\
  \Q_t f &= \frac{\pt f}{\pt x^j} \frac{\pt f}{\pt x^k} Q^{jk} x. \notag
\end{align}
If a linear connection $\nabla$ is specified, we can use \eqref{2-tangent-equiv} to rewrite $\D_t$ as follows:
\begin{equation}\label{local-rep-total-mean-2}
  \D_t = \pt_t + D_\nabla^i x \pt_i + \ts{\frac{1}{2}} Q^{jk} x \nabla^2_{\pt_j,\pt_k}.
\end{equation}

\begin{lemma}\label{push-mixed-coord}
  Let us be given a bundle homomorphism $F = (F^0, \bar F)$ from $(\R\times M, \pi, \R)$ to $(\R\times N, \rho, \R)$ projecting to a diffeomorphism $F^0$ and an $M$-valued diffusion process $X = \{X(t)\}_{t\in[t_0,\tau)}$. If $\tilde X = F\cdot X$, then in local coordinates $(t,x^i)$ around $(t_0,q)$ and $(s,y^j)$ around $F(t_0,q)$,
  \begin{align*}
    (D\tilde X)^{j}(F^0(t)) &= (\D_t \bar F^j) \left( j_{(t, X(t))} X \right) \frac{d(F^0)^{-1}}{ds}(F^0(t)), \notag\\
    (Q\tilde X)^{kl}(F^0(t)) &= \left( \frac{\pt\bar F^k}{\pt x^i}\frac{\pt\bar F^l}{\pt x^j} \right) \left( t, X(t) \right) (QX)^{ij} \left( t \right) \frac{d(F^0)^{-1}}{ds}(F^0(t)).
  \end{align*}
\end{lemma}

\begin{proof}
  Assume that the diffusion $X$ can be represented in local coordinates by
  \begin{equation*}
    dX^i(t) = \mathfrak{b}^i(t,X(t)) dt + \sigma^i_r(t,X(t)) dW^r(t), \quad X^i(t_0) = x^i(q).
  \end{equation*}
  where $W$ is an $N$-dimensional Brownian motion, so that
  $$j_t X = (DX(t), QX(t)) = (\mathfrak{b}, \sigma\circ\sigma^*)(t,X(t)).$$
  Let $(s_0,\tilde q)=F(t_0,q) = (F^0(t_0), \bar F(t_0,q))$. Then
  \begin{equation*}
    X^i((F^0)^{-1}(s)) = x^i(q) + \int_{(F^0)^{-1}(s_0)}^{(F^0)^{-1}(s)} \mathfrak{b}^i(u,X(u)) du + \int_{(F^0)^{-1}(s_0)}^{(F^0)^{-1}(s)} \sigma^i_r(u,X(u)) dW^r(u).
  \end{equation*}
  Define
  \begin{equation*}
    B(s) = \int_0^{(F^0)^{-1}(s)} \sqrt{(F^0)'(u)} dW(u).
  \end{equation*}
  Then \cite[Theorem 8.5.7]{Oks10} says that $B$ is an $N$-dimensional $\{\F_{(F^0)^{-1}(s)}\}$-Brownian motion, as by a change of variable $u=(F^0)^{-1}(v)$, we have
  \begin{equation*}
    \int_{(F^0)^{-1}(s_0)}^{(F^0)^{-1}(s)} \sigma^i_r(u,X(u)) dW^r(u) = \int_{s_0}^{s} \sigma^i_r((F^0)^{-1}(v),X((F^0)^{-1}(v))) \left( \frac{d(F^0)^{-1}}{ds}(v) \right)^{\frac{1}{2}} dB^r(v).
  \end{equation*}
  Therefore,
  \begin{equation*}
    \begin{split}
      X^i((F^0)^{-1}(s)) =&\ x^i(q) + \int_{s_0}^{s} \mathfrak{b}^i((F^0)^{-1}(v),X((F^0)^{-1}(v))) d (F^0)^{-1}(v) \\
      &\ + \int_{s_0}^{s} \sigma^i_r((F^0)^{-1}(v),X((F^0)^{-1}(v))) \left( \frac{d(F^0)^{-1}}{ds}(v) \right)^{\frac{1}{2}} dB^r(v).
    \end{split}
  \end{equation*}
  Recall that $\tilde X(s) = \bar F\left( (F^0)^{-1}(s), X((F^0)^{-1}(s)) \right)$. Using It\^o's formula, we have
  \begin{equation*}
    \begin{split}
      \tilde X^j(s) =&\ y^j(\tilde q) + \int_{s_0}^{s} \frac{\pt\bar F^j}{\pt t} \left( (F^0)^{-1}(v), X((F^0)^{-1}(v)) \right) d(F^0)^{-1}(v) \\
      &\ + \int_{s_0}^{s} \frac{\pt\bar F^j}{\pt x^i} \left( (F^0)^{-1}(v), X((F^0)^{-1}(v)) \right) dX^i((F^0)^{-1}(v)) \\
      &\ + \frac{1}{2} \int_{s_0}^{s} \frac{\pt^2\bar F^j}{\pt x^k \pt x^l} \left( (F^0)^{-1}(v), X((F^0)^{-1}(v)) \right) d\langle X^k\circ (F^0)^{-1}, X^l\circ (F^0)^{-1} \rangle(v) \\
      =&\ y^j(q) + \int_{s_0}^{s} \left[ \frac{\pt\bar F^j}{\pt t} + \frac{\pt\bar F^j}{\pt x^i} \mathfrak{b}^i + \frac{1}{2} \frac{\pt^2\bar F^j}{\pt x^k \pt x^l} \sigma^k_r \sigma^l_r \right] \left( (F^0)^{-1}(v), X((F^0)^{-1}(v)) \right) \frac{d(F^0)^{-1}}{ds}(v) dv \\
      &\ + \int_{s_0}^{s} \left( \frac{\pt\bar F^j}{\pt x^i} \sigma^i_r \right) \left( (F^0)^{-1}(v), X((F^0)^{-1}(v)) \right) \left( \frac{d(F^0)^{-1}}{ds}(v) \right)^{\frac{1}{2}} dB^r(v).
    \end{split}
  \end{equation*}
  It follows that
  \begin{align*}
    (D\tilde X)^j(s) &= \left[ \frac{\pt\bar F^j}{\pt t} + \frac{\pt\bar F^j}{\pt x^i} \mathfrak{b}^i + \frac{1}{2} \frac{\pt^2\bar F^j}{\pt x^k \pt x^l} \sigma^k_r \sigma^l_r \right] \left( (F^0)^{-1}(v), X((F^0)^{-1}(v)) \right) \frac{d(F^0)^{-1}}{ds}(v) \\
    &= (\D_t\bar F^j) \left( j_{((F^0)^{-1}(s), X((F^0)^{-1}(s)))} X \right) \frac{d(F^0)^{-1}}{ds}(s), \\
    (Q\tilde X)^{kl}(s) &= \left( \frac{\pt\bar F^k}{\pt x^i} \sigma^i_r \frac{\pt\bar F^l}{\pt x^j} \sigma^j_r \right) \left( (F^0)^{-1}(s), X((F^0)^{-1}(s)) \right) \frac{d(F^0)^{-1}}{ds}(s) \\
    & = \left( \frac{\pt\bar F^k}{\pt x^i}\frac{\pt\bar F^l}{\pt x^j} \right) \left( (F^0)^{-1}(s), X((F^0)^{-1}(s)) \right) (QX)^{ij} \left( (F^0)^{-1}(s) \right) \frac{d(F^0)^{-1}}{ds}(s).
  \end{align*}
  This completes the proof.
\end{proof}

We denote the induced local coordinates on $\mathcal T^S N$ by $(y^j, D^j y, Q^{kl} y)$. Then clearly, $y^j \circ jF =  y^j \circ \overline{jF} = y^j \circ F = \bar F^j$. Now take $j_{(t,q)} X \in \R\times \mathcal T^S M$. Then
\begin{gather}
  D^j y \circ jF(j_{(t,q)} X) = D^j y (j_{F(t,q)} \tilde X) = (D\tilde X)^j (F^0(t)) = (\D_t\bar F^j) (j_{(t, q)} X) \left( \frac{dF^0}{dt}(t) \right)^{-1}, \label{local-rep-jF-1}\\
  Q^{kl} y \circ jF(j_{(t,q)} X) = Q^{kl} y (j_{F(t,q)} \tilde X) = (Q\tilde X)^{kl}(F^0(t)) = \left( \frac{\pt\bar F^k}{\pt x^i}\frac{\pt\bar F^l}{\pt x^j} \right) (t, X(t))(QX)^{ij}(t) \left( \frac{dF^0}{dt}(t) \right)^{-1}. \label{local-rep-jF-2}
\end{gather}

\subsection{Symmetries of SDEs}

As an important application of the prolongations of diffusions and bundle homomorphisms, we now study the symmetries of stochastic differential equations. As in classical Lie's theory of symmetries of ODEs, a symmetry of a stochastic differential equation is a space-time transformation that maps solutions to solutions. But this is not sufficient for the stochastic case. As we have mentioned in Section \ref{sec-4-1}, the only smooth transformation on $\R\times M$ mapping diffusions to diffusions are bundle endomorphisms. Moreover, a solution of stochastic differential equation is always accompanied by a filtration, which will also be altered under space-time transformations. Thus, we have the following definition:

\begin{definition}[Symmetries]
  Given a stochastic differential equation $S\subset \R\times \mathcal T^S M$, a symmetry of $S$ is a bundle automorphism $F$ on $(\R\times M, \pi, \R)$ projecting to $F^0$ such that if $(X, \{\Pred_t\})$ is a solution of $S$, then so is $(F\cdot X, \{\Pred_{(F^0)^{-1}(s)}\})$.
\end{definition}

Using the definitions of stochastic differential equations and pushforwards, we have the following equivalent characterization of symmetries.

\begin{lemma}
  Let $S$ be a stochastic differential equation on $M$. A bundle automorphism $F$ on $(\R\times M, \pi, \R)$ is a symmetry of $S$, if and only if, whenever $j_{(t,q)} X \in S$ we have $j F (j_{(t,q)} X) \in S$, or equivalently, $j F(S)\subset S$.
\end{lemma}


Recall that the infinitesimal version of bundle homomorphisms are the so called projectable or fiber-preserving vector fields. More precisely, a vector field $V$ on $\R\times M$ is called $\pi$-projectable, if the (local) flow (or one-parameter group action) generated by $V$ consists of (local) bundle endomorphisms on $(\R\times M, \pi, \R)$ (cf. \cite[Example 2.22]{Olv98} or \cite[Proposition 3.2.15]{Sau89}). For such a vector field, we define its prolongation to be the infinitesimal generator of the prolongated flow.

\begin{definition}[Stochastic prolongations of projectable vector fields]\label{prog-vf}
  Let $V$ be a $\pi$-projectable vector field on $\R\times M$, with corresponding (local) flow $\psi = \{\psi_\e\}_{\e\in(-\varepsilon,\varepsilon)}$. Then, the stochastic prolongation of $V$, denoted by $j V$, will be a vector field on the model jet bundle $\R\times \mathcal T^S M$, defined as the infinitesimal generator of the corresponding prolonged flow $\{j\psi_\e\}_{\e\in (-\varepsilon,\varepsilon)}$. In other words, $j V$ is a vector field on $\R\times \mathcal T^S M$ defined by
  \begin{equation*}
    j V \big|_{j_{(t,q)} X} = \frac{d}{d\e} \bigg|_{\e=0} (j\psi_\e) (j_{(t,q)} X),
  \end{equation*}
  for any $j_{(t,q)} X\in \R\times \mathcal T^S M$.
\end{definition}

Now we can define infinitesimal versions of symmetries.

\begin{definition}[Infinitesimal symmetries]
  Let $S$ be a stochastic differential equation on $M$. An infinitesimal symmetry of $S$ is a $\pi$-projectable vector field $V$ on $\R\times M$ whose stochastic prolongation $jV$ is tangent to $S$.
\end{definition}

The following properties follow straightforwardly from definitions.
\begin{lemma}\label{inf-symm}
  Given a stochastic differential equation $S$ on $M$, let $V$ be a complete $\pi$-projectable vector field on $\R\times M$ and $\psi = \{\psi_\e\}_{\e\in\R}$ be its flow. Then \\
  (i) $V$ is an infinitesimal symmetry of $S$ if and only if $jV(\Theta) = 0$ for every local defining function $\Theta$ of $S$; \\
  (ii) $V$ is an infinitesimal symmetry of $S$ if and only if for each $\e\in\R$, $\psi_\e$ is a symmetry of $S$.
\end{lemma}


\subsection{Stochastic prolongation formulae}

We consider a coordinate chart $(\R\times U^{(1)}, (t, x^{(1)}))$ on the model jet bundle $\R\times \mathcal T^S M$, which is induced by the coordinate chart $(U, (x^i))$ on $M$. A $\pi$-projectable vector field $V$ on $\R\times M$ has the following local coordinate representation
\begin{equation}\label{general-vf}
  V_{(t,q)} = V^0(t) \vf t\bigg|_t + V^i(t,q) \vf{x^i}\bigg|_q.
\end{equation}
Its prolongation $jV$ is a vector field $\R\times \mathcal T^S M$ of the form
\begin{equation*}
  j V\big|_{j_{(t,q)} X} = V^0(t) \vf t\bigg|_t + V^i(t,q) \vf{x^i}\bigg|_{j_{(t,q)} X} + V^i_1(j_{(t,q)} X) \vf{D^i x}\bigg|_{j_{(t,q)} X} + V^{jk}_2(j_{(t,q)} X) \vf{Q^{jk} x}\bigg|_{j_{(t,q)} X}.
\end{equation*}


Now we use Lemma \ref{push-mixed-coord} to compute the coefficients $V^i_1$'s and $V^{jk}_2$'s.

\begin{theorem}\label{prog-proj-vf}
  Suppose $V$ is complete and $\pi$-projectable and has the local representation \eqref{general-vf}. Then in the canonical coordinates $(t, x^{(1)})$, the coefficient functions of its prolongation $jV$ are given by the following formulae:
  \begin{align}
    V^i_1(t, x^{(1)}) &= (\D_t V^i)(t,x^{(1)}) - \dot V^0(t) D^i x, \label{prog-1} \\
    V^{jk}_2(t, x^{(1)}) &= \frac{\pt V^j}{\pt x^i}(t,x) Q^{ik} x + \frac{\pt V^k}{\pt x^i}(t,x) Q^{ij} x - \dot V^0(t) Q^{jk} x. \label{prog-2}
  \end{align}
\end{theorem}

\begin{proof}
  Let $\psi = \{\psi_\e\}_{\e\in\R}$ be the flow generated by $V$. Since $V$ is complete and $\pi$-projectable, each $\psi_\e$ is a bundle endomorphism on $\R\times M$ projecting to a diffeomorphism on $\R$. Let $\psi_\e (t, q) = (\psi^0_\e(t), \bar \psi_\e(t,q))$. Note that $\psi^0_0(t)=t$, $\bar \psi_0(t,q) = q$ and
  \begin{equation*}
    V^0(t) = \frac{d}{d\e} \bigg|_{\e=0} \psi^0_\e(t), \quad V^i(t,q) = \frac{d}{d\e} \bigg|_{\e=0} \bar\psi^i_\e(t,q).
  \end{equation*}
  Let $X=\{X(t)\}_{t\in[t_0,\tau)}$ be a representative diffusion of $j_{(t_0,q)} X \in U^{(1)}$. Then by Lemma \ref{push-diff-bh} and Definition \ref{prog-bundle-homo}, a representative diffusion of $j\psi_\e (j_{(t,q)} X)$ is
  \begin{equation*}
    \tilde X_\e(s) = \psi_\e \cdot X(s) = \bar \psi_\e\left( (\psi^0_\e)^{-1}(s), X((\psi^0_\e)^{-1}(s)) \right),
    \quad s \in [\psi^0_\e(t_0), \psi^0_\e(\tau)).
  \end{equation*}
  Now we apply Lemma \ref{push-mixed-coord} and take derivatives with respect to $\e$. Since $\frac{d}{d\e}$ commutes with the total mean derivative $\D_t$ as is clear from the coordinate representation, we have
  \begin{equation*}
    \begin{split}
      V_1^i (j_{(t,q)} X) &= \frac{d}{d\e} \bigg|_{\e=0} (D\tilde X_\e)^i(\psi^0_\e(t)) = \frac{d}{d\e} \bigg|_{\e=0} \left[ (\D_t \bar \psi_\e^i) \left( j_{(t, X(t))} X \right) \frac{d(\psi_\e^0)^{-1}}{ds}(\psi^0_\e(t)) \right] \\
      &= \D_t V^i (j_{(t,q)} X) - (DX)^i(t) \dot V^0(t).
    \end{split}
  \end{equation*}
  Also,
  \begin{equation*}
    \begin{split}
      V_2^{kl}(j_{(t,q)} X) &= \frac{d}{d\e} \bigg|_{\e=0} (Q\tilde X_\e)^{kl}(\psi^0_\e(t)) \\
      &= \frac{d}{d\e} \bigg|_{\e=0} \left[ \left( \frac{\pt\bar \psi_\e^k}{\pt x^i}\frac{\pt\bar \psi_\e^l}{\pt x^j} \right) \left( t, X(t) \right) (QX)^{ij} \left( t \right) \frac{d(\psi_\e^0)^{-1}}{ds}(\psi^0_\e(t)) \right] \\
      &= \left( \frac{\pt V^k}{\pt x^i} \delta_j^l + \delta_i^k \frac{\pt V^l}{\pt x^j} \right) (t,X(t)) (QX)^{ij}(t) - \delta_i^k \delta_j^l (QX)^{ij}(t) \dot V^0(t) \\
      &= \frac{\pt V^k}{\pt x^i}(t,q) (QX)^{il}(t) + \frac{\pt V^l}{\pt x^j} (t,q) (QX)^{jk}(t) - (QX)^{kl}(t) \dot V^0(t).
    \end{split}
  \end{equation*}
  In the induced coordinate system $(t, x^{(1)})= (t, x^i, D^i x, D^{jk}_2 x)$, the last two formulae read as \eqref{prog-1} and \eqref{prog-2}, respectively.
\end{proof}

Stochastic analogs of contact structure on $\R\times \mathcal T^S M$ and Cartan symmetries will be discussed in Appendix \ref{app-2}. It turns out that the infinitesimal symmetry of the mixed-order Cartan distribution is equivalent to stochastic prolongation formulae of Theorem \ref{prog-proj-vf}.

Applying Theorem \ref{prog-proj-vf} to the system of mean differential equations \eqref{Nelson-SDE-mfld-2}, we have
\begin{corollary}\label{symm-Ito}
  The complete and $\pi$-projectable vector field $V$ in \eqref{general-vf} is an infinitesimal symmetry of MDEs \eqref{Nelson-SDE-mfld-2} if and only if the coefficients $V^0$ and $V^i$'s satisfy the following ``determining equations'':
  \begin{align}
    V^0 \frac{\pt \mathfrak b^i}{\pt t} + V^j \frac{\pt \mathfrak b^i}{\pt x^j} = \frac{\pt V^i}{\pt t} + \frac{\pt V^i}{\pt x^j} \mathfrak{b}^j + \frac{1}{2} \frac{\pt^2 V^i}{\pt x^j \pt x^k} \sigma^j_r \sigma^k_r - \dot V^0 \mathfrak b^i, \notag \\
    V^0 \frac{\pt (\sigma_r^j \sigma_r^k)}{\pt t} + V^i \frac{\pt (\sigma_r^j \sigma_r^k)}{\pt x^i} = \frac{\pt V^j}{\pt x^i} \sigma^i_r \sigma^k_r + \frac{\pt V^k}{\pt x^i} \sigma^i_r \sigma^j_r - \dot V^0 \sigma^j_r \sigma^k_r. \label{symm-Ito-2}
  \end{align}
\end{corollary}

\begin{proof}
  We apply Lemma \ref{inf-symm}.(i) to \eqref{SDE-F} and then use Theorem \ref{prog-proj-vf}, to get
  \begin{align*}
    V^0 \frac{\pt \mathfrak b^i}{\pt t} + V^j \frac{\pt \mathfrak b^i}{\pt x^j} &= \D_t V^i - \dot V^0 D^i x, \\
    V^0 \frac{\pt (\sigma_r^j \sigma_r^k)}{\pt t} + V^i \frac{\pt (\sigma_r^j \sigma_r^k)}{\pt x^i} &= \frac{\pt V^j}{\pt x^i} Q^{ik} x + \frac{\pt V^k}{\pt x^i} Q^{ij} x - \dot V^0 Q^{jk} x.
  \end{align*}
  Then we use the coordinate representation \eqref{local-rep-total-mean} for the total mean derivative $\D_t$ and plug equation \eqref{SDE-F} in; the results follow.
\end{proof}

\begin{remark}
In \cite{GQ99}, the author proved a result similar to Corollary \ref{symm-Ito}, with the following equation instead of equation \eqref{symm-Ito-2}:
\begin{equation}\label{symm-Ito-3}
  V^0 \frac{\pt \sigma_r^j}{\pt t} + V^i \frac{\pt \sigma_r^j }{\pt x^i} = \frac{\pt V^j}{\pt x^i} \sigma^i_r - \frac{1}{2} \dot V^0 \sigma^j_r.
\end{equation}
By multiplying both sides of \eqref{symm-Ito-3} with $\sigma_r^k$, and using the symmetry for index $j,k$, one gets easily \eqref{symm-Ito-2}. So our determining equations for infinitesimal symmetries are more general than those of \cite{GQ99}. Basically, the paper \cite{GQ99} concerns symmetries of the It\^o equation $(\mathfrak b, \sigma)$, while we consider symmetries of the diffusion with generator $(\mathfrak b, \sigma\circ\sigma^*)$, or equivalently, a weak formulation of SDE. The former symmetries belong to the latter obviously, but not vice versa.
\end{remark}

Now given a linear connection $\nabla$ on $M$, we define the $\nabla$-dependent versions of Definitions \ref{prog-diff}, \ref{prog-bundle-homo} and \ref{prog-vf}. More precisely, for a diffusion $X$ on $M$, we define its $\nabla$-prolongation to be a $TM$-valued diffusion $j^\nabla X$ given by $j^\nabla X(t) = j^\nabla_{X(t)} (\theta_t X)$. For a bundle homomorphism from $F:(\R\times M, \pi, \R)\to(\R\times N, \rho, \R)$ projecting to a diffeomorphism $F^0:\R\to\R$, the $\nabla$-prolongation of $F$ is the map $j^\nabla F: \R\times T M \to \R\times T N$ defined by $j^\nabla F (j^\nabla_{(t,q)} X) = j^\nabla_{F(t,q)} (F\cdot X)$. The $\nabla$-prolongation of $V$, denoted by $j^\nabla V$, is defined to be the infinitesimal generator of the corresponding prolonged flow $\{j^\nabla\psi_\e\}_{\e\in (-\varepsilon,\varepsilon)}$, so that $j^\nabla V$ is a vector field on $\R\times T M$ and has the form
\begin{equation*}
  j^\nabla V\big|_{j^\nabla_{(t,q)} X} = V^0(t) \vf t\bigg|_t + V^i(t,q) \vf{x^i}\bigg|_{j^\nabla_{(t,q)} X} + V^i_\nabla(j^\nabla_{(t,q)} X) \vf{\dot x^i}\bigg|_{j^\nabla_{(t,q)} X},
\end{equation*}
for $V$ of the form \eqref{general-vf}. If we denote $\bar V = V^i \vf{x^i}$ so that $V = V^0 + \bar V$, we have

\begin{corollary}\label{prog-tang}
  Under the canonical coordinates $(t, x, \dot x)$, the coefficient $V^i_\nabla$ of the $\nabla$-prolongation $j^\nabla V$ are given by:
  \begin{equation*}
    V^i_\nabla(t, x,\dot x) = \left(\pt_t + \dot x^j \pt_j \right) V^i(t,x) + \ts{\frac{1}{2}} Q^{jk}x \left[ \nabla^2_{\pt_j,\pt_k} \bar V + R(\bar V,\pt_j)\pt_k \right]^i (t,x) - \dot V^0(t) \dot x^i,
  \end{equation*}
  where $R$ is the curvature tensor.
\end{corollary}

\begin{proof}
  By \eqref{prog-1} and \eqref{prog-2}, we have
  \begin{equation*}
    \begin{split}
      V_\nabla^i (j_{(t,q)} X) &= \frac{d}{d\e} \bigg|_{\e=0} (D_\nabla\tilde X_\e)^i(\psi^0_\e(t)) = \frac{d}{d\e} \bigg|_{\e=0} \left[ (D\tilde X_\e)^i(\psi^0_\e(t)) + \frac{1}{2} \Gamma^i_{jk}(\tilde X_\e(t)) (Q\tilde X_\e)^{jk}(\psi^0_\e(t)) \right] \\
      &= V_1^i (j_{(t,q)} X) + \frac{1}{2} \Gamma^i_{jk}(X(t)) V_2^{jk} (j_{(t,q)} X) + \frac{1}{2} \frac{\pt\Gamma^i_{jk}}{\pt x^l}(X(t)) (QX)^{jk}(t) V^l(X(t)) \\
      &= \left[ \frac{\pt}{\pt t} + \left( (D_\nabla X)^l(t) - \frac{1}{2}\Gamma_{jk}^l(X(t)) (QX)^{jk}(t) \right) \frac{\pt}{\pt x^l} + \frac{1}{2} (QX)^{jk}(t) \frac{\pt^2}{\pt x^j \pt x^k} \right] V^i (t,X(t)) - (DX)^i(t) \dot V^0(t) \\
      &\quad + \frac{1}{2} \Gamma^i_{jk}(X(t)) \left[ \frac{\pt V^j}{\pt x^l}(t,(X(t)) (QX)^{kl}(t) + \frac{\pt V^k}{\pt x^m} (t,(X(t)) (QX)^{jm}(t) - (QX)^{jk}(t) \dot V^0(t) \right] \\
      &\quad + \frac{1}{2} \frac{\pt\Gamma^i_{jk}}{\pt x^l}(X(t)) (QX)^{jk}(t) V^l(t,X(t)) \\
      &= \left[ \frac{\pt}{\pt t} + (D_\nabla X)^l(t) \frac{\pt}{\pt x^l} \right] V^i (t,X(t)) + \frac{1}{2} (Q_\nabla X)^{jk}(t) \left[ \nabla^2_{\pt_j,\pt_k} \bar V + R(\bar V,\pt_j)\pt_k \right]^i (t,X(t)) \\
      &\quad - (D_\nabla X)^i(t) \dot V^0(t).
    \end{split}
  \end{equation*}
  The proof is complete.
\end{proof}

\section{The second-order cotangent bundle}\label{sec-5}

\subsection{Second-order covectors}

\begin{definition}[Second-order cotangent space]
  The second-order cotangent space at $q\in M$ is the dual vector space of $\mathcal T^O_q M$, denoted by $\mathcal T^{S*}_q M$. The pairing of $\alpha\in \mathcal T^{S*}_q M$ and $A\in \mathcal T^O_q M$ is denoted by $\langle \alpha, A \rangle$ or $\alpha(A)$.  Elements of $\mathcal T^{S*}_q M$ are called second-order covectors at $q$. The disjoint union $\mathcal T^{S*} M:= \amalg_{q\in M} \mathcal T^{S*}_q M$ is called the stochastic cotangent bundle of $M$. The natural projection map from $\mathcal T^{S*} M$ to $M$ is denoted by $\tau^{S*}_M$. A (local or global) smooth section of $\mathcal T^{S*} M$ is called a second-order covector field or a second-order form. 
\end{definition}

Dual to the left action \eqref{left-action} of $G_I^d$ on fibers of $\mathcal T^S M$, $G_I^d$ will act on those of $\mathcal T^{S*} M$ from the right.

\begin{lemma}\label{struct-grp-cot}
  The stochastic cotangent bundle $(\mathcal T^{S*} M, \tau^{S*}_M, M)$ is the fiber bundle dual to $(\mathcal T^S M, \tau^S_M, M)$, with structure group $G_I^d$ acting on the typical fiber $(\R^d \times \Sym^2(\R^d))^*$ from the right by
  \begin{equation*}
    (p, o)\cdot(g, \kappa) = (g^* p, \kappa^* p + (g^*\otimes g^*) o),
  \end{equation*}
  for all $(g, \kappa) \in G_I^d$, $p\in (\R^d)^*$, $o\in(\Sym^2(\R^d))^*$.
\end{lemma}

The notion of second-order forms should not be confused with the classical one of 2-forms. There are two basic examples of second-order forms, say, $d^2 f$ and $df\cdot dg$, where $f$ and $g$ are given smooth functions on $M$. They are defined as follows: for $A\in\mathcal T^S M$,
\begin{equation}\label{forms}
  \langle d^2 f, A \rangle := Af, \qquad \langle df\cdot dg, A \rangle := \Gamma_A(f,g) = A(fg) - fAg - gAf,
\end{equation}
where $\Gamma_A$ is the squared field operator defined in \eqref{squared-field}.
These notations go back to L.~Schwartz \cite{Sch84} and P.A.~Meyer \cite{Mey81} (see also \cite[Chapters VI and VII]{Eme89}), where the term $d^2 f$ is called the \emph{second differential} of $f$, and the term $df\cdot dg$ is called the \emph{symmetric product} of $df$ and $dg$. Note that in these original references, there is a factor $\frac{1}{2}$ at the RHS of the definition of $df\cdot dg$. Here we drop this factor. Obviously, when restricted to $T M$, the second differential $d^2f$ is just the differential $df$ but the symmetric product $df\cdot dg$ vanishes.


The definition of the symmetric product $df\cdot dg$ yields two properties: $df\cdot dg$ is symmetric in $f$ and $g$; and $(df\cdot dg)_q = 0$ if one of $df_q$ and $dg_q$ vanishes. These lead to a more general definition for symmetric products of two 1-forms. More precisely, let $\omega, \eta \in \mathcal T^*_q M$, then there exist smooth functions $f$ and $g$ on $M$ such that $\omega = df_q$ and $\eta = dg_q$. By the preceding property, the second-order covector $(df\cdot dg)_q$ does not depend on the choice of $f$ and $g$, and we will denote it by $\omega\cdot\eta$. Now if $\omega, \eta$ are second-order forms, then their symmetric product is defined pointwisely through $(\omega\cdot\eta)_q = \omega_q \cdot\eta_q$. More formally, we have
\begin{definition}[Symmetric product, {\cite[Chapter VI]{Eme89}}]
  There exists a unique fiber-linear bundle homomorphism $\bullet$ from $T^* M \otimes T^* M$ to $\mathcal T^{S*} M$, which is called the symmetric product, such that for all $\omega, \eta \in T^* M$, $\bullet(\omega\otimes\eta) = \omega\cdot\eta$.
\end{definition}

It is easy to verify from \eqref{forms} that the local frame, dual to \eqref{frame}, for $(\mathcal T^{S*} M, \tau^{S*}_M, M)$ over the local chart $(U,(x^i))$ is given by (see also \cite[Chapter VI]{Eme89})
\begin{equation*}
  \left\{ d^2 x^i, \textstyle{\frac{1}{2}} dx^i\cdot dx^i, dx^j\cdot dx^k: 1\le i\le d, 1\le j< k \le d \right\}.
\end{equation*}
We adopt the convention that $dx^k\cdot dx^j = dx^j\cdot dx^k$ for all $1\le j< k \le d$. Under this frame, a second-order covector $\alpha\in \mathcal T^{S*}_q M$ has a local expression
\begin{equation}\label{eqn-8}
  \alpha = \alpha_i d^2 x^i|_q + \textstyle{\frac{1}{2}} \alpha_{jk} dx^j \cdot dx^k|_q,
\end{equation}
where $\alpha^{jk}$ is symmetric in $j,k$. The coordinates $(x^i)$ induce a canonical coordinate system on $\mathcal T^{S*} M$, denoted by $(x^i, p_i, o_{jk})$ and defined by
\begin{equation}\label{induced-sys-2-cot}
  x^i(\alpha) = x^i(q), \quad p_i(\alpha) = \alpha_i, \quad o_{jk}(\alpha) = \alpha_{jk}.
\end{equation}
for $\alpha$ in \eqref{eqn-8}. Since the coefficients $(\alpha_i)$ do transform like a covector, as indicated in Lemma \ref{struct-grp-cot}, it will cause no ambiguity to retain $(x^i,p_i)$ as canonical coordinates on $T^* M$. As in classical geometric mechanics \cite{AM78,HSS09}, we still call the coordinates $(p_i)$ the conjugate momenta. And we shall call the second-order coordinates $(o_{jk})$ the \emph{conjugate diffusivities}. 

The pairing of $\alpha$ and the second-order vector field $A$ in \eqref{2-tangent-rep} is then
\begin{equation*}
  \langle \alpha, A \rangle = \alpha_i A^i + \alpha_{jk} A^{jk}.
\end{equation*}
It follows from \eqref{forms} and \eqref{squared-field} that for smooths functions $f$ and $g$ on $M$,
\begin{equation*}
  d^2 f = \frac{\pt f}{\pt x^i} d^2 x^i + \frac{1}{2} \frac{\pt^2 f}{\pt x^j \pt x^k} dx^j\cdot dx^k, \qquad df\cdot dg = \frac{\pt f}{\pt x^i} \frac{\pt g}{\pt x^j} dx^i\cdot dx^j.
\end{equation*}
More generally, for 1-forms $\omega$ and $\eta$ with local expressions $\omega = \omega_i dx^i$ and $\eta = \eta_i dx^i$, the symmetric product $\omega\cdot\eta$ has local expression
\begin{equation}\label{product-general}
  \omega\cdot\eta = \omega_i \eta_j dx^i\cdot dx^j.
\end{equation}

Dual to the tangent case, there is indeed a canonical bundle epimorphism $\hat\varrho^*: (\mathcal T^{S*} M, \tau^{S*}_M, M) \to (T^* M, \tau^*_M, M)$, given by
\begin{equation*}
  \hat\varrho^*(\alpha) = \alpha|_{TM}.
\end{equation*}
In particular $\hat\varrho^*(d^2 f) = df$. In local coordinates, $\hat\varrho^*$ reads as
\begin{equation*}
  \hat\varrho^*\left( \alpha_i d^2 x^i|_q + \textstyle{\frac{1}{2}} \alpha_{jk} dx^j \cdot dx^k|_q \right) = \alpha_i d x^i|_q,
\end{equation*}
The map $\hat\varrho^*$ is well defined since $\alpha|_{TM}$ is a covector. 
Clearly, $\hat\varrho^*$ is also a surjective submersion, so that $\mathcal T^{S*} M$ is a fiber bundle over $T^* M$. Occasionally, we will use the notation $\hat\varrho^*_M$ to indicate the base manifold $M$.

However, there is no canonical bundle monomorphism from $T^* M$ to $\mathcal T^{S*} M$ which is a left inverse of $\hat\varrho^*$ and linear in fiber. We call such a bundle epimorphism a \emph{fiber-linear bundle injection} from $T^* M$ to $\mathcal T^{S*} M$. Similarly to Proposition \ref{induced-conn}, we also have a connection correspondence property. Namely, if we are given a linear connection $\nabla$ on $M$, then it induces a fiber-linear bundle injection from $T^* M$ to $\mathcal T^{S*} M$ by
\begin{equation}\label{iota-conn}
  \hat\iota^*_\nabla : T^* M \to \mathcal T^{S*} M, \quad d x^i|_q \mapsto d^2 x^i|_q + \textstyle{\frac{1}{2}} \Gamma_{jk}^i(q) dx^j \cdot dx^k|_q =: d^\nabla x^i|_q,
\end{equation}
or in local coordinates $\hat\iota^*_\nabla(x, p) = (x, p, (\Gamma_{jk}^i(x) p_i))$.
Any fiber-linear bundle injection from $T^* M$ to $\mathcal T^{S*} M$ induces a torsion-free linear connection on $M$.

Denote by $\Sym^2(T^* M)$ the subbundle of $T^* M \otimes T^* M$ consisting of all $(0,2)$-tensors on $M$. Then the symmetric product $\bullet$, when restricting to $\Sym^2(T^* M)$, is a bundle monomorphism whose image is the kernel of $\hat\varrho^*$. Conversely, still by the connection correspondence, a linear connection $\nabla$ induces a fiber-linear bundle epimorphism from $\mathcal T^{S*} M$ to $\Sym^2(T^* M)$ which is a right inverse of $\bullet$ and is given by
\begin{equation*}
  \varrho^*_\nabla: \mathcal T^{S*} M \to \Sym^2(T^* M), \quad \alpha_i d^2 x^i|_q + \textstyle{\frac{1}{2}} \alpha_{jk} dx^j \cdot dx^k|_q \mapsto \left(\alpha_{jk} - \alpha_i\Gamma_{jk}^i(q)\right) dx^j \otimes dx^k|_q.
\end{equation*}
We introduce the $\nabla$-dependent coordinates $(o_{jk}^\nabla)$ by $o_{jk}^\nabla(\alpha) = \alpha_{jk} - \alpha_i \Gamma_{jk}^i(q)$ for $\alpha$ in \eqref{eqn-8}, i.e.,
\begin{equation}\label{tensorial-diffusivities}
  o_{jk}^\nabla = o_{jk} - p_i(\Gamma_{jk}^i\circ x).
\end{equation}
Then $\varrho^*_\nabla(\alpha) = o_{jk}^\nabla(\alpha) dx^j \otimes dx^k|_q$ and in particular
\begin{equation*}
  \varrho^*_\nabla(d^2 f) = \left( \frac{\pt^2 f}{\pt x^j \pt x^k} - \Gamma_{jk}^i \frac{\pt f}{\pt x^i} \right) dx^j \otimes dx^k = \nabla^2 f.
\end{equation*}
The coordinates $(x^i, p_i, o_{jk}^\nabla)$ form a coordinate system on $\mathcal T^{S*} M$, which we call the $\nabla$-canonical coordinate system. The coordinates $(x^i, o_{jk}^\nabla)$ also form a coordinate system on $\Sym^2(T^* M)$ when restricted to it. We will call the coordinates $(o^\nabla_{jk})$ the \emph{tensorial conjugate diffusivities}.

To sum up, we have the following short exact sequence which is split when a linear connection is provided:
\begin{equation}\label{dual-exact-seq}
  0 \longrightarrow \Sym^2(T^* M) \stackrel{\bullet}{\longrightarrow} \mathcal T^{S*} M \stackrel{\hat\varrho^*}{\longrightarrow} T^* M \longrightarrow 0.
\end{equation}
It is easy to check that the bundle homomorphisms $\hat\varrho^*$, $\hat\iota^*_\nabla$, $\bullet$ and $\varrho^*_\nabla$ are dual to $\iota$, $\varrho_\nabla$, $\hat\varrho$ and $\hat\iota_\nabla$ in \eqref{iota}, \eqref{varrho}, \eqref{varrho-hat} and \eqref{iota-hat}, respectively, so that the short exact sequence \eqref{dual-exact-seq} is dual to \eqref{exact-seq}. Similarly to \eqref{dcpst}, we have the following decomposition if a linear connection $\nabla$ is given,
\begin{equation*}
  \mathcal T^{S*} M = \hat\iota^*_\nabla(T^*M) \oplus \bullet \left( \Sym^2(T^* M) \right) \cong T^*M \oplus \Sym^2(T^*M),
\end{equation*}
with fiber-wise isomorphism $\cong$ and first direct sum $\oplus$, which is given by
\begin{equation*}
  \alpha = \alpha_i d^\nabla x^i|_q + \textstyle{\frac{1}{2}} \left(\alpha_{jk} - \alpha_i\Gamma_{jk}^i(q)\right) dx^j \cdot dx^k|_q \mapsto \left( \alpha_i d x^i|_q, \left(\alpha_{jk} - \alpha_i\Gamma_{jk}^i(q)\right) dx^j \otimes dx^k|_q \right).
\end{equation*}
In particular,
\begin{equation}\label{eqn-35}
  d^2 f = \pt_i f d^\nabla x^i + \ts{\frac{1}{2}} \nabla^2_{\pt_j, \pt_k} f dx^j \cdot dx^k \mapsto (df, \nabla^2 f).
\end{equation}

Similarly to the classical cotangent space, the second-order cotangent space may be defined via germs. To be precise, we denote by $C_q^\infty(M)$ the set of all germs of smooth functions at $q\in M$, and define a equivalence relation between germs: $[f]_q, [g]_q \in C_q^\infty(M)$ are equivalent if and only if they have the same Taylor expansion at $q$ higher than order zero and up to order two. Then, one can easily check that there is a one-to-one correspondence between $\mathcal T^{S*}_q M$ and the quotient space of $C_q^\infty(M)$ by this equivalence relation. Along this way, we can also observe the following diffeomorphism:
  \begin{equation}\label{germ}
    \mathcal T^{S*} M\times\R \cong J^2\hat\pi,
  \end{equation}
  by mapping $(d^2 f_q, f(q))$ to $j^2_q f$, where $J^2\hat\pi$ is the classical second-order jet bundle of $(M\times \R, \hat\pi, M)$. This is similar to $T^*M\times\R$ is diffeomorphic to the first-order jet bundle $J^1\hat\pi$ (e.g., \cite[Example 2.5.11 ]{Gei08} or \cite[Example 4.1.15 ]{Sau89}). We denote the natural projection maps from $\mathcal T^{S*} M\times\R$ to $\R$ and from $T^* M\times\R$ to $\R$ by $\hat\pi^2_{0,1}$ and $\hat\pi^1_{0,1}$, respectively.

  The relations and projection maps are integrated into the following commutative diagram:
  \vspace{1mm}
\begin{center}
\begin{tikzcd}
J^1\hat\pi \cong T^* M\times\R \arrow{ddrr}[left, crossing over]{\hat\pi_1} \arrow{ddr}[left, crossing over]{\hat\pi_{1,0}} \arrow{ddd}{\hat\pi^1_{0,1}} & J^2\hat\pi \cong \mathcal T^{S*} M\times\R \arrow{dd}{\hat\pi_{2,0}} \arrow{lddd}{\hat\pi^2_{0,1}} \arrow{l}{\hat\pi_{2,1}} \arrow{ddr}{\hat\pi_2} \arrow{r}{\hat\pi_{1,1}} & \mathcal T^{S*} M \arrow{dd}{\tau^{S*}_M} \arrow{r}{\hat\varrho^*} & T^* M \arrow{ldd}{\tau^*_M} \\
& & & \\
 & M\times \R \arrow{ld}{\pi} \arrow{r}{\hat\pi} & M & \\
\R & & &
\end{tikzcd}
\end{center}

\begin{remark}
  (i). As in Remark \ref{symm-bundle-tang}, given a linear connection $\nabla$, we can obtain a one-to-one correspondence between $(T^*M \oplus \Sym^2(T^*M))\times\R$ and $J^2\hat\pi$ by mapping $(df_q, \nabla^2 f_q, f(q))$ to $j^2_q f$.  One can find in \cite{DDD19} an application of the jet-like structure on $T^*M \oplus \Sym^2(T^*M)$ and higher-order bundles to Martin Hairer's theory of regularity structures \cite{Hai14}.

  (ii). As we have seen, the product $\R\times \mathcal T^S M$ is the model bundle of the stochastic jet space $\mathcal J^S M$, while the product $\mathcal T^{S*} M\times\R$ is diffeomorphic to the second-order jet bundle $J^2\hat\pi$. So, in a way, we can say that the ``stochastic'' and the ``second-order'' are dual to each other. This stochastic--second-order duality is somehow analogous to the particle--wave duality in quantum mechanics.
\end{remark}

\subsection{Second-order tangent and cotangent maps}

\begin{definition}[Second-order tangent and cotangent maps, {\cite[Chapter VI]{Eme89}}]\label{push-pull-map-point}
  Let $M$ and $N$ be two smooth manifolds, $F: M\to N$ be a smooth map. The second-order tangent map of $F$ at $q\in M$ is a linear map $d^2 F_q: \mathcal T^S_q M \to \mathcal T^S_{F(q)} N$ defined by
  \begin{equation*}
    d^2 F_q (A) f = A(f\circ F), \quad\text{for } A\in \mathcal T^S_q M, f\in C^\infty(N).
  \end{equation*}
  The second-order cotangent map of $F$ at $q\in M$ is a linear map $d^2 F^*_q: \mathcal T^{S*}_{F(q)}N \to \mathcal T^{S*}_q M$ dual to $d^2 F_q$, that is,
  \begin{equation*}
    d^2 F^*_q (\alpha) (A) = \alpha (d^2 F_q( A)), \quad\text{for } A\in \mathcal T^S_q M, \alpha\in \mathcal T^{S*}_{F(q)} N.
  \end{equation*}
\end{definition}

The restrictions of $d^2 F_q$ to $T_q M$ coincide with the classical tangent map $d F_q$. But this is not the case for $d^2 F^*_q$ when restricting to $T^{*}_{F(q)} N$, since for $\alpha\in T^*_{F(q)} N$, $d^2 F^*_q (\alpha)$ is still a linear map on $\mathcal T^S_q M$. A manifestation of these phenomena may be seen through local coordinates in the following lemma.

\begin{lemma}\label{push-pull-map-prop}
  Let $(U,(x^i))$ and $(V,(y^j))$ be local coordinate charts around $q$ and $F(q)$, respectively. If
  \begin{equation*}
    A = A^i \frac{\partial}{\partial x^i}\bigg|_q + A^{ij} \frac{\partial^2}{\partial x^i\partial x^j}\bigg|_q \quad \text{and} \quad \alpha = \alpha_i d^2y^i|_{F(q)} + \alpha_{ij} dy^i\cdot dy^j|_{F(q)}.
  \end{equation*}
  Then
  \begin{gather*}
    d^2 F_q (A) = (A F^i) \frac{\partial}{\partial y^i}\bigg|_{F(q)} + \Gamma_A(F^i, F^j) \frac{\partial^2}{\partial y^i\partial y^j}\bigg|_{F(q)}, \\
    d^2 F^*_q(\alpha) = \alpha_i d^2F^i|_q + \alpha_{ij} dF^i\cdot dF^j|_q.
  \end{gather*}
\end{lemma}

Now if $A\in\mathcal T_q M$, then all $(A^{ij})$ vanish and thereby so do $\Gamma_A(F^i, F^j)$'s. Thus, $d^2 F_q(A) = (A F^i) \frac{\partial}{\partial y^i}|_{F(q)} = d F_q(A)$. This makes clear that $d^2 F_q|_{\mathcal T_{q} M} = d F_q$. But if $\alpha\in \mathcal T^*_{F(q)} N$, then $\alpha^{ij}$'s vanish and
\begin{equation*}
  d^2 F^*_q (\alpha) = \alpha_i d^2F^i|_q = \alpha_i \frac{\pt F^i}{\pt x^j}(q) d^2 x^j|_q + \alpha_i \frac{\pt^2 F^i}{\pt x^j \pt x^k}(q) dx^j\cdot dx^k|_q,
\end{equation*}
while $d F^*_q (\alpha) = \alpha_i d F^i|_q = \alpha_i \frac{\pt F^i}{\pt x^j}(q) d^2 x^j|_q$. Hence, $d^2 F^*_q|_{\mathcal T^*_{F(q)} N} \ne d F^*_q$.

\begin{definition}[Second-order pushforwards and pullbacks]\label{push-pull-map}
  Let $F: M\to N$ be smooth map. The second-order pushforward by $F$ is a bundle homomorphism $F^S_*: (\mathcal T^S M, \tau^S_M, M) \to (\mathcal T^S N, \tau^S_N, N)$ defined by
  $$F^S_*|_{\mathcal T^S_q M} = d^2 F_q.$$
  Given a second-order form $\alpha$ on $N$, the second-order pullback of $\alpha$ by $F$ is a second-order form $F^{S*}\alpha$ on $M$ defined by
  \begin{equation*}
    (F^{S*}\alpha)_q = d^2 F^*_q \left( \alpha_{F(q)} \right), \quad q\in M.
  \end{equation*}
  Let $F$ be a diffeomorphism. The second-order pullback by $F$ is a bundle isomorphism $F^{S*}: (\mathcal T^{S*} N, \tau^{S*}_N, N) \to (\mathcal T^{S*} M, \tau^{S*}_M, M)$ defined by
  $$F^{S*}|_{\mathcal T^{S*}_{q'} N} = d^2 F^*_{F^{-1}(q')}.$$
  Given a second-order vector field $A$ on $M$, the second-order pushforward of $A$ by $F$ is a second-order vector field $F^S_*A$ on $N$ defined by
  \begin{equation*}
    (F^S_*A)_{q'} = d^2 F_{F^{-1}(q')} \left( A_{F^{-1}(q')} \right), \quad q'\in N.
  \end{equation*}
\end{definition}

Clearly, $F^S_*|_{T M} = F_*$ is the usual pushforward, but $F^{S*}|_{T^* N} \ne F^*$. The following properties are straightforward.

\begin{lemma}\label{prop-push-pull}
  Let $F: M\to N$, $G:N\to K$ be two smooth maps. Let $A$ be a second-order vector field on $M$ and $f, g$ be two smooth functions on $N$. \\
  (i) $G^S_*\circ F^S_* = (G\circ F)^S_*$. \\
  (ii) If $F$ is a diffeomorphism, then $((F^S_*A)f)\circ F = A(f\circ F)$. \\
  (iii) $F^{S*}(d^2 f) = d^2 (f\circ F)$, $F^{S*} (df\cdot dg) = d(f\circ F)\cdot d(g\circ F)$.
\end{lemma}

\subsection{Mixed-order tangent and cotangent bundles}

In this section, we will extend the notions of the previous two sections to the product manifold $\R\times M$.

\begin{definition}
  The mixed-order tangent bundle of $\R\times M$ is the product bundle (\cite[Definition 1.4.1]{Sau89}) $(T \R \times \mathcal T^S M, \tau_\R\times \tau^S_M, \R\times M)$. The mixed-order cotangent bundle of $\R\times M$ is the product bundle $(T^* \R \times \mathcal T^{S*} M, \tau^*_\R\times \tau^{S*}_M, \R\times M)$. A section of the mixed-order tangent or cotangent bundle is called a mixed-order vector field or mixed-order form, respectively.
\end{definition}

The mixed-order tangent and cotangent bundles are dual to each other. The mixed-order tangent (or cotangent) bundle is the bundle that mixes the first-order tangent (or cotangent) bundle in time and the second-order one in space (this is why we use the terminology ``mixed-order''). It also matches the fundamental principle of stochastic analysis, whose It\^o's logo is $(dX(t))^2 \sim dt$.


For an $M$-valued diffusion $X$ with (time-dependent) generator $A^X$, we call the operator $\vf t + A^X$ its extended generator. This extended generator is a mixed-order vector field on $\R\times M$. Also note that the extended generator $\vf t + A^X$ of $X\in I_{t_0}(M)$ can be characterized by the property that for every $f\in C^\infty(\R\times M)$, the process
\begin{equation*}
  f(t,X(t)) - f(t_0,X(t_0)) - \int_{t_0}^t \left(\vf t + A^X \right) f(s,X(s)) ds, \quad t\ge t_0,
\end{equation*}
is a real-valued continuous $\{\Pred_t\}$-martingale. In general, a mixed-order vector field $A$ has the following local expression:
  \begin{equation*}
    A = A^0 \vf t + A^i \vf{x^i} + A^{jk} \frac{\pt^2}{\pt x^j \pt x^k}.
  \end{equation*}
To give an example of mixed-order forms, we consider a smooth function $f$ on $\R\times M$, and define in local coordinates
\begin{equation*}
  d^\circ f := \frac{\pt f}{\pt t} dt + \frac{\pt f}{\pt x^i} d^2 x^i + \frac{1}{2} \frac{\pt^2 f}{\pt x^j \pt x^k} dx^j\cdot dx^k.
\end{equation*}
Then $d^\circ f$ is a mixed-order form, and we call it the \emph{mixed differential} of $f$. Clearly, the pairing of the mixed differential $d^\circ f$ and a mixed-order vector field $A$ is $\langle d^\circ f, A \rangle = Af$.

Given a bundle homomorphism from $F:(\R\times M, \pi, \R)\to (\R\times N, \rho, \R)$, we define its mixed-order tangent map at $(t,q)\in\R\times M$ by
\begin{equation*}
    d^\circ F_{(t,q)} = d^2 F_{(t,q)}|_{T_t \R\times \mathcal T^S_q M}: T \R\times \mathcal T^S M|_{(t,q)} \to T \R\times\mathcal T^S N|_{F(t,q)}.
\end{equation*}
Its mixed-order cotangent map at $(t,q)\in\R\times M$ is defined as the linear map $d^\circ F^*_{(t,q)}: T^* \R\times\mathcal T^{S*} N|_{F(t,q)} \to T^* \R\times\mathcal T^{S*} M|_{(t,q)}$ dual to $d^\circ F_{(t,q)}$. If, moreover, $F$ is a bundle isomorphism, its mixed-order pushforward and pullback, denoted by $F^R_*$ and $F^{R*}$, respectively, can be defined in a similar manner to Definition \ref{push-pull-map}. We leave their detailed but cumbersome definitions and properties to Appendix \ref{sec-A-1}.

\section{Stochastic Hamiltonian mechanics}\label{sec-6}

\subsection{Horizontal diffusions}

In this section, we consider a general fiber bundle $(E,\pi_M, M)$ over a manifold $M$, with fiber dimension $n$. We first introduce a special class of diffusions on this fiber bundle, which we call horizontal diffusions. They are defined in a similar fashion as the horizontal subspaces in Definition \ref{horizontal-subspace}. Roughly speaking, a horizontal diffusion process on $E$ is a diffusion that is random only ``horizontally'', but not on fibers.

\begin{definition}[Horizontal diffusions on fiber bundles]
  Let $(E,\pi_M, M)$ be a fiber bundle. A $E$-valued diffusion process $\mathbf X$ is said to be horizontal, if there exists an $M$-valued diffusion process $X$ and a smoothly time-dependent section $\phi=(\phi_t)$ of $\pi_M$, such that a.s. $\mathbf X(t) = \phi(t, X(t))$ for all $t$. 
\end{definition}

The process $X$ in the above definition is just the projection of $\mathbf X$, for $\pi_M(\mathbf X(t)) = \pi_M(\phi(t,X(t))) = X(t)$ a.s.. Since the projection map $\pi_M$ is smooth, $X$ is still a diffusion process.

Now we are going to 
define a subclass of ``integral processes'' for second-order vector fields on $E$ by making use of horizontal diffusions. We use $(x^i,u^\mu)$ for an adapted coordinate system on $E$ (see \cite[Definition 1.1.5]{Sau89}), where we use Greek alphabet to label the coordinates of fibers.


Given a 
second-order vector field with local expression
\begin{equation}\label{hrz-2-vf}
  A = A^i \vf{x^i} + A^\mu \vf{u^\mu} + A^{jk} \frac{\pt^2}{\pt x^j \pt x^k} + A^{j\mu} \frac{\pt^2}{\pt x^j \pt u^\mu} + A^{\mu\nu} \frac{\pt^2}{\pt u^\mu \pt u^\nu},
\end{equation}
where $A^i, A^\mu, A^{jk}, A^{j\mu}, A^{\mu\nu}$ are smooth functions in the local chart of $E$,
by a \emph{horizontal integral process} of $A$ in \eqref{hrz-2-vf} we mean an $E$-valued horizontal diffusion process $\mathbf X$ such that $\mathbf X$ is an integral process of $A$ in the sense of \eqref{Nelson-SDE-mfld-2},
that is, it is determined by the system
\begin{equation}\label{hrz-int-prc}\left\{
  \begin{aligned}
    (D(x\circ \mathbf X))^i(t) &= A^i(\mathbf X(t)), \\
    (Q(x\circ \mathbf X))^{jk}(t) &= 2A^{jk}(\mathbf X(t)), \\
    (D(u\circ \mathbf X))^\mu(t) &= A^\mu(\mathbf X(t)), \\
    (Q(x\circ \mathbf X, u\circ \mathbf X))^{j\nu}(t) &= A^{j\mu}(\mathbf X(t)), \\
    (Q(u\circ \mathbf X))^{\mu\nu}(t) &= 2A^{\mu\nu}(\mathbf X(t)),
  \end{aligned}\right.
\end{equation}
where the expression $x\circ \mathbf X$ means that the family of coordinate functions $(x^i)$ acts on $\mathbf X$, and so on.
Set $\mathbf X(t) = \phi(t, X(t))$ for some time-dependent section $\phi$ of $\pi_M$ and $M$-valued diffusion $X$. Denote $\phi^\mu = u^\mu\circ \phi$. By It\^o's formula, the system \eqref{hrz-int-prc} can be written as
\begin{equation}\label{hrz-int-prc-2}\left\{
  \begin{aligned}
    (DX)^i(t) &= A^i(\phi(t, X(t))), \\
    (QX)^{jk}(t) &= 2A^{jk}(\phi(t, X(t))), \\
    \bigg( \frac{\pt}{\pt t} + A^i&(\phi(t, X(t))) \vf{x^i} + A^{jk}(\phi(t, X(t))) \frac{\pt^2}{\pt x^j\pt x^k}\bigg)\phi^\mu(t, X(t)) = A^\mu(\phi(t, X(t))) \\
    2 A^{jk}(\phi(t&, X(t))) \frac{\pt \phi^\mu}{\pt x^k}(t, X(t)) = A^{j\mu}(\phi(t, X(t))) \\
    A^{jk}(\phi(t, &X(t))) \frac{\pt \phi^\mu}{\pt x^j}\frac{\pt \phi^\nu}{\pt x^k}(t, X(t)) = A^{\mu\nu}(\phi(t, X(t))).
  \end{aligned}\right.
\end{equation}
If $X(t)$ has full support for all $t$, then the last three equations in \eqref{hrz-int-prc-2} translate into a system of (possibly degenerate) parabolic equations on $E$,
\begin{equation}\label{eqn-section}\left\{
  \begin{aligned}
    &\bigg( \frac{\pt}{\pt t} + A^i(\phi(t, q)) \vf{x^i} + A^{jk}(\phi(t, q)) \frac{\pt^2}{\pt x^j\pt x^k} \bigg)\phi^\mu(t, q) = A^\mu(\phi(t, q)), \\
    &2 A^{jk}(\phi(t, q)) \frac{\pt \phi^\mu}{\pt x^k}(t, q) = A^{j\mu}(\phi(t, q)) \\
    &A^{jk}(\phi(t, q)) \frac{\pt \phi^\mu}{\pt x^j}\frac{\pt \phi^\nu}{\pt x^k}(t, q) = A^{\mu\nu}(\phi(t, q)).
  \end{aligned}\right.
\end{equation}
Therefore, under suitable assumptions for the coefficients $A^i, A^\mu, A^{jk}, A^{j\mu}, A^{\mu\nu}$, equation \eqref{eqn-section} is solvable, at least locally, by some time-dependent local section $\phi = (\phi_t)$ over a time interval $[0,T]$. Then, plugging $\phi(t)$ into the first two equations of \eqref{hrz-int-prc-2}, we can find $X$ and hence $\mathbf X$. We call $X$ an \emph{projective integral process} of $A$.


\subsection{The second-order symplectic structure on $\mathcal T^{S*} M$ and stochastic Hamilton's equations}\label{sec-6-2}

It is well known that the classical cotangent bundle $T^* M$ has a natural symplectic structure, given by the canonical symplectic form $\omega_0 = dx^i \wedge dp_i$, where $(x^i,p_i)$ are the canonical local coordinates on $T^* M$ induced by local coordinates $(x^i)$ on $M$. Clearly $\omega_0$ is closed, because it is exact as $\omega_0 = -d \theta_0$, where $\theta_0 = p_i dx^i$ is called the  Poincar\'e (or tautological) 1-form.

Now we need to define a similar structure on the second-order cotangent bundle $\mathcal T^{S*} M$, which is a second-order counterpart of the symplectic structure. Firstly, we adapt the coordinate-free definition of the tautological 1-form to the second-order case.

\begin{definition}
  The second-order tautological form $\theta$ is a second-order form on $\mathcal T^{S*} M$ defined by
  \begin{equation*}
    \theta_{\alpha} = d^2(\tau_M^{S*})_{\alpha}^* (\alpha), \quad \alpha\in \mathcal T^{S*}_q M.
  \end{equation*}
\end{definition}

Under the induced coordinate system on $\mathcal T^{S*} M$ defined in \eqref{induced-sys-2-cot}, the second-order tautological form $\theta$ has the following coordinate representation
\begin{equation}\label{2-taut}
  \theta = p_i d^2 x^i + \textstyle{\frac{1}{2}} o_{jk} dx^j\cdot dx^k.
\end{equation}
We introduce the canonical second-order symplectic form $\omega$ on $\mathcal T^{S*} M$ by writing $\omega = -d^2 \theta$. Although we do not define the exterior differential for second-order forms, we can still take $d^2$ formally on both sides of \eqref{2-taut}, using Leibniz's rule and the composition rule $d\circ d=d^2$ (cf. \cite[Section 6.(e)]{Mey81a}), and forcing $d^3 = 0$ and $(d^2-)\cdot (d-) = (d-)\cdot (d^2-) = 0$. Then, we get
\begin{equation}\label{2-symp-form}
  \begin{split}
    \omega =&\ d\left( d^2 x^i \wedge d p_i + \textstyle{\frac{1}{2}} dx^j\cdot dx^k \wedge d o_{jk} - p_i d^3 x^i + o_{jk} d^2x^j\wedge dx^k \right) \\
    =&\ d^2 x^i \wedge d^2 p_i + \textstyle{\frac{1}{2}} dx^j\cdot dx^k \wedge d^2 o_{jk}.
  \end{split}
\end{equation}
We call the pair $(\mathcal T^{S*} M, \omega)$ a second-order symplectic manifold. The complete axiom system for a second-order differential system $(d,d^2,\wedge,\cdot)$ is beyond the scope of this paper.

\begin{remark}\label{remark-5}
  In the formal expression $(d\circ d) f=d^2 f$, $f\in C^\infty(M)$, the two differential operators $d$ at LHS are \emph{different}. The second $d$ is still de Rham's exterior differential on $M$, while the first needs to be understood as the exterior differential on $TM$ by regarding the first differential $df$ as a function on $TM$. Thus the complete expression should be $d_{TM}\circ d_M = d^2$. Along this way, the differential operator $d_{TM}$ can be extended to a linear transform that maps 1-forms to 2nd-order forms and satisfies Leibniz's rule, see \cite[Theorem 7.1]{Eme89}. We shall denote the linear operator extended from $d_{TM}$ by $\d$ in order to distinguish. In local coordinates, it acts on a 1-form $\eta = \eta_i dx^i$ by $\d \eta = \eta_i d^2x^i + \frac{1}{2} \frac{\pt \eta_i}{\pt x^j} dx^i\cdot dx^j$, so that $\hat\varrho^*(\d\eta) = \eta$ and $d^2 = \d\circ d$. When a linear connection $\nabla$ is specified, $\d \eta = \eta_i d^\nabla x^i + \frac{1}{2} \nabla\eta(\pt_i,\pt_j) dx^i\cdot dx^j$ which covers \eqref{eqn-35}.
\end{remark}

As in the classical case, we have the following property for the second-order tautological form.
\begin{lemma}\label{prop-taut}
  The second-order tautological form $\theta$ is the unique second-order form on $\mathcal T^{S*} M$ with the property that, for every second-order form $\alpha$ on $M$, $\alpha^{S*} \theta = \alpha$.
\end{lemma}
\begin{proof}
  From Lemma \ref{prop-push-pull}, we have, for any second-order vector $A\in \mathcal T^S_q M$,
  \begin{equation*}
    \langle (\alpha^{S*} \theta)_q, A \rangle = \langle \theta_{\alpha_q}, d^2 \alpha_q(A) \rangle = \langle d^2(\tau_M^{S*})_{\alpha_q}^* (\alpha_q), d^2 \alpha_q(A) \rangle = \langle \alpha_q, d^2(\tau_M^{S*})_{\alpha_q} \circ d^2 \alpha_q(A) \rangle = \langle \alpha_q, A \rangle,
  \end{equation*}
  since $\tau_M^{S*} \circ \alpha = \id_M$.
\end{proof}

Recall that, in Definition \ref{push-pull-map}, we have defined the second-order pullbacks of second-order forms. Now, given a smooth map $\mathbf F: \mathcal T^{S*} M \to \mathcal T^{S*} N$ and a second-order 2-form $\eta$ on $\mathcal T^{S*} N$, we may also define the second-order pullback $\mathbf F^{S*} \eta$ of $\eta$ by $\mathbf F$ by allowing $\mathbf F^{S*}$ to be exchangeable with the symmetric product $\cdot$ as well as the wedge product $\wedge$. Then, as a corollary of Lemma \ref{prop-taut}, we have
$$\alpha^{S*} \omega = -d^2\alpha.$$

\begin{definition}
  Let $\omega$ and $\eta$ be the canonical second-order symplectic forms on $\mathcal T^{S*} M$ and $\mathcal T^{S*} N$, respectively. A bundle homomorphism $\mathbf F: (\mathcal T^{S*} M, \hat\varrho^*_M, T^* M) \to (\mathcal T^{S*} N, \hat\varrho^*_N, T^*N)$ is called second-order symplectic or a second-order symplectomorphism if $\mathbf F^{S*} \eta = \omega$.
\end{definition}

\begin{theorem}
  Let $F: N\to M$ be a diffeomorphism. The second-order pullback $F^{S*}: \mathcal T^{S*} M \to \mathcal T^{S*} N$ by $F$ is second-order symplectic; in fact $(F^{S*})^{S*} \vartheta = \theta$, where $\vartheta$ is the second-order tautological form on $\mathcal T^{S*} N$.
\end{theorem}

\begin{proof}
  For $q\in M$, $\alpha_q\in \mathcal T^{S*}_q M$ and $A \in \mathcal T^S_{\alpha_q} T^{S*} M$,
  \begin{equation*}
    \begin{split}
      \langle (F^{S*})^{S*} \vartheta, A \rangle &= \langle \vartheta, d^2(F^{S*})_{\alpha_q} A \rangle = \langle d^2(\tau_N^{S*})_{F^{S*}(\alpha_q)}^* (F^{S*}(\alpha_q)), d^2(F^{S*})_{\alpha_q} A \rangle \\
      &= \langle F^{S*}(\alpha_q), d^2(\tau_N^{S*})_{F^{S*}(\alpha_q)} \circ d^2(F^{S*})_{\alpha_q} A \rangle \\
      &= \langle \alpha_q, d^2F_{F^{-1}(q)} \circ d^2(\tau_N^{S*})_{F^{S*}(\alpha_q)} \circ d^2(F^{S*})_{\alpha_q} A \rangle \\
      &= \langle \alpha_q, d^2(\tau_M^{S*})_{\alpha_q} A \rangle \\
      &= \langle d^2(\tau_M^{S*})^*_{\alpha_q} (\alpha_q), A \rangle \\
      &= \langle \theta_{\alpha_q}, A \rangle,
    \end{split}
  \end{equation*}
  where we used the fact that $F\circ \tau_N^{S*}\circ F^{S*} = \tau_M^{S*}$ in the fourth line.
\end{proof}

Clearly, the counterparts of Hamiltonian vector fields on $T^* M$ are now second-order vector fields on $\mathcal T^{S*} M$. Remark that for a second-order vector field $A$ on $\mathcal T^{S*} M$, the form $A \lrcorner\, \omega$ take values in the 
cotangent bundle $\mathcal T^{S*} \mathcal T^{S*} M$. 

\begin{definition}
  Let $H: \mathcal T^{S*} M \to \R$ be a given smooth function. A 
  second-order vector field $A_H$ on $\mathcal T^{S*} M$ satisfying
  \begin{equation}\label{2-Hamiltonian-vf}
    A_H \lrcorner\, \omega = d^2 H
  \end{equation}
  is called a second-order Hamiltonian vector field of $H$. We call the triple $(\mathcal T^{S*} M, \omega, H)$ a second-order Hamiltonian system. The function $H$ is called the second-order Hamiltonian of the system.
\end{definition}


According to \eqref{2-Hamiltonian-vf}, the 2nd-order vector field $A_H$ satisfies
\begin{equation}\label{2-Hamiltonian-vf-prop}
  A_H H = d^2H (A_H) = \omega(A_H,A_H) = 0.
\end{equation}
The condition \eqref{2-Hamiltonian-vf} cannot uniquely determine $A_H$.
It is easy to verify that $A_H$ is of the general form
\begin{equation}\label{Hamiltonian-vf}
  \begin{split}
    A_H =&\ \frac{\pt H}{\pt p_i} \vf{x^i} - \frac{\pt H}{\pt x^i} \vf{p_i} + \frac{\pt H}{\pt o_{jk}} \frac{\pt^2}{\pt x^j \pt x^k} - \left( \frac{\pt^2 H}{\pt x^j \pt x^k} + C_{jk} \right) \vf{o_{jk}} \\
    &\ + A_{jk} \frac{\pt^2}{\pt p_j \pt p_k} + A_{ijkl} \frac{\pt^2}{\pt o_{ij} \pt o_{kl}} + A^j_k \frac{\pt^2}{\pt x^j \pt p_k} + A^j_{kl} \frac{\pt^2}{\pt x^j \pt o_{kl}} + A_{jkl} \frac{\pt^2}{\pt p_j \pt o_{kl}},
  \end{split}
\end{equation}
where the coefficients $C_{jk}, A_{jk}, A_{ijkl}, A^j_k, A^j_{kl}, A_{jkl}$ are smooth functions on local chart satisfying
\begin{equation*}
  C_{jk} \frac{\pt H}{\pt o_{jk}} = A_{jk} \frac{\pt^2 H}{\pt p_j \pt p_k} + A_{ijkl} \frac{\pt^2 H}{\pt o_{ij} \pt o_{kl}} + A^j_k \frac{\pt^2 H}{\pt x^j \pt p_k} + A^j_{kl} \frac{\pt^2 H}{\pt x^j \pt o_{kl}} + A_{jkl} \frac{\pt^2 H}{\pt p_j \pt o_{kl}},
\end{equation*}
such that the local expression \eqref{Hamiltonian-vf} is invariant under the canonical change of coordinates on $\mathcal T^{S*} M$ induced by a change of coordinates on $M$, governed by the structure group in Lemma \ref{struct-grp-cot}.

Given such a second-order Hamiltonian vector field of $H$, its horizontal integral process is a $\mathcal T^{S*} M$-valued horizontal diffusion $\textbf{X}$ determined by the following MDEs on $\mathcal T^{S*} M$,
\begin{equation}\label{stoch-Hamilton-eqns}\left\{
  \begin{aligned}
    (D (x\circ\textbf{X}))^i(t) &= \frac{\pt H}{\pt p_i}(\textbf{X}(t)), \\
    (Q (x\circ\textbf{X}))^{jk}(t) &= 2\frac{\pt H}{\pt o_{jk}}(\textbf{X}(t)), \\
    (D (p\circ\textbf{X}))_i(t) &= - \frac{\pt H}{\pt x^i}(\textbf{X}(t)), \\
    (D (o\circ\textbf{X}))_{jk}(t) &= - \left( \frac{\pt^2 H}{\pt x^j \pt x^k} + C_{jk} \right) (\textbf{X}(t)), \\
    \left(C_{ij} \frac{\pt H}{\pt o_{ij}}\right) (\textbf{X}(t)) &= \frac{1}{2} (Q (p\circ\textbf{X}))_{jk}(t) \frac{\pt^2 H}{\pt p_j \pt p_k} (\textbf{X}(t)) + \frac{1}{2} (Q (o\circ\textbf{X}))_{ijkl}(t) \frac{\pt^2 H}{\pt o_{ij} \pt o_{kl}} (\textbf{X}(t)) \\
    &\quad + (Q (x\circ\textbf{X}, p\circ\textbf{X}))^j_k \frac{\pt^2 H}{\pt x^j \pt p_k} (\textbf{X}(t)) + (Q (x\circ\textbf{X}, o\circ\textbf{X}))^j_{kl} \frac{\pt^2 H}{\pt x^j \pt o_{kl}} (\textbf{X}(t)) \\
    &\quad + (Q (p\circ\textbf{X}, o\circ\textbf{X}))_{jkl} \frac{\pt^2 H}{\pt p_j \pt o_{kl}} (\textbf{X}(t)),
  \end{aligned}\right.
\end{equation}
or, in coordinates,
\begin{equation*}
  \left\{
  \begin{aligned}
    D^i x &= \frac{\pt H}{\pt p_i}, \\
    Q^{jk} x &= 2\frac{\pt H}{\pt o_{jk}}, \\
    D_i p &= - \frac{\pt H}{\pt x^i}, \\
    D_{jk} o &= - \left( \frac{\pt^2 H}{\pt x^j \pt x^k} + C_{jk} \right), \\
    C_{ij} \frac{\pt H}{\pt o_{ij}} &= \frac{1}{2} Q_{jk} p \frac{\pt^2 H}{\pt p_j \pt p_k} + \frac{1}{2} Q_{ijkl} o \frac{\pt^2 H}{\pt o_{ij} \pt o_{kl}} + Q^j_k(x,p) \frac{\pt^2 H}{\pt x^j \pt p_k} \\
    &\quad + Q^j_{kl}(x,o) \frac{\pt^2 H}{\pt x^j \pt o_{kl}} + Q_{jkl}(p,o) \frac{\pt^2 H}{\pt p_j \pt o_{kl}},
  \end{aligned}\right.
\end{equation*}
where $\big( x^i, p_i, o_{jk}, D^i x, D_i p, D_{jk} o, Q^{jk} x, Q_{jk} p, Q_{ijkl} o, Q^j_k(x,p), Q^j_{kl}(x,o), Q_{jkl}(p,o) \big)$ are canonical coordinates on $\mathcal T^S \mathcal T^{S*} M$. The first and third equations has been conjectured in \cite{Zam15} as stochastic Hamilton's equations in the Euclidean space, since they have the same form as classical Hamilton's equations (e.g., \cite[Proposition 3.3.2]{AM78}) except that mean derivative $D$ replaces classical time derivative.

At first glance, one may think that the system \eqref{stoch-Hamilton-eqns} is underdetermined, as there are fewer equations than unknowns (the number of unknowns is equal to the fiber dimension of $\mathcal T^S \mathcal T^{S*} M$). Besides, we haven not yet given \eqref{stoch-Hamilton-eqns} initial or terminal data. These will become clear after we make the following observations. Firstly, the first two equations of \eqref{stoch-Hamilton-eqns} constitute MDEs that are equivalent to an It\^o SDE for $x(\mathbf X)$ in weak sense, as we have seen in Section \ref{sec-2-4}. So $x(\mathbf X)$ should be assigned an initial value, say,
\begin{equation}\label{init-value}
  \text{Law}((x\circ \mathbf X)(0)) = \mu_0,
\end{equation}
where $\mu_0$ is a given probability measure on $M$. Secondly, in the third and fourth equations of \eqref{stoch-Hamilton-eqns}, only the ``drift'' information of $p(\mathbf X)$ and $o(\mathbf X)$ is clear. To overcome the lack of information, we need to assign $p(\mathbf X)$ and $o(\mathbf X)$ terminal values, say,
\begin{equation}\label{term-value}
  \left\{
  \begin{aligned}
    (p\circ \mathbf X)(T) &= p^*(x\circ \mathbf X(T)), \\
    (o\circ \mathbf X)(T) &= o^*(x\circ \mathbf X(T)),
  \end{aligned}\right.
\end{equation}
where $(p^*, o^*)$ is a given second-order form. Therefore, the third and fourth equations are understood as backward SDEs, whose drifts rely on diffusion coefficients via the last equation. The system \eqref{stoch-Hamilton-eqns} together with boundary values \eqref{init-value} and \eqref{term-value} could be understood as a (coupled) forward-backward system of SDEs \cite{Yon99} (where ``backward'' is taken in a different sense from ours in Chapter \ref{sec-2}).

Notice that those forward-backward SDEs are not necessarily solvable (see \cite[Proposition 7.5.2]{Yon99} for an example). In order to solve \eqref{stoch-Hamilton-eqns}--\eqref{term-value}, we have to take the horizontal condition into consideration, and make some compatibility assumption. More precisely, we set $X=\tau_M^{S*}(\mathbf X)$ and
\begin{equation}\label{eqn-31}
  \mathbf X(t) = \alpha(t, X(t)),
\end{equation}
for some time-dependent second-order form $\alpha$ on $M$, and denote $p_i(t,x) = p_i(\alpha(t,x))$ and $o_{jk}(t,x) = o_{jk}(\alpha(t,x))$, so that $\alpha(t,x) = (p(t,x), o(t,x))$. Assume that for each $t\in(0,T)$, $X(t)$ has full support. Then, by applying It\^o's formula, in the same way as in \eqref{eqn-section}, the system \eqref{stoch-Hamilton-eqns} reduces to
\begin{equation}\label{eqn-32}\left\{
  \begin{aligned}
    \bigg( \frac{\pt}{\pt t}& + \frac{\pt H}{\pt p_j} \vf{x^j} + \frac{\pt H}{\pt o_{jk}} \frac{\pt^2}{\pt x^j\pt x^k}\bigg)p_i = - \frac{\pt H}{\pt x^i}, \\
    \bigg( \frac{\pt}{\pt t}& + \frac{\pt H}{\pt p_k} \vf{x^k} + \frac{\pt H}{\pt o_{kl}} \frac{\pt^2}{\pt x^k\pt x^l}\bigg)o_{ij} = - \left( \frac{\pt^2 H}{\pt x^i \pt x^j} + C_{ij} \right), \\
    C_{ij}& \frac{\pt H}{\pt o_{ij}} = \frac{\pt H}{\pt o_{ij}} \bigg( \frac{\pt p_k}{\pt x^i} \frac{\pt p_l}{\pt x^j} \frac{\pt^2 H}{\pt p_k \pt p_l} + \frac{\pt o_{kl}}{\pt x^i} \frac{\pt o_{mn}}{\pt x^j} \frac{\pt^2 H}{\pt o_{kl} \pt o_{mn}} + 2 \frac{\pt p_k}{\pt x^i} \frac{\pt^2 H}{\pt x^j \pt p_k} \\
    &\qquad\qquad\qquad + 2 \frac{\pt o_{kl}}{\pt x^i} \frac{\pt^2 H}{\pt x^j \pt o_{kl}} + 2 \frac{\pt p_k}{\pt x^i} \frac{\pt o_{lm}}{\pt x^j} \frac{\pt^2 H}{\pt p_k \pt o_{lm}} \bigg).
  \end{aligned}\right.
\end{equation}
Next, 
by taking partial derivative $\vf{x^j}$ on both sides of the first equation of \eqref{eqn-32} and comparing with the next two, we find the following sufficient condition for the last two equations of \eqref{eqn-32}:
\begin{equation}\label{non-degenerate}
  o_{ij}(t, x) = \frac{\pt p_i}{\pt x^j}(t, x) = \frac{\pt p_j}{\pt x^i}(t, x),
\end{equation}
or equivalent, for the terminal value $(p^*,o^*)$,
\begin{equation}\label{non-degenerate-2}
  o^*_{ij}(x) = \frac{\pt p^*_i}{\pt x^j}(x) = \frac{\pt p^*_j}{\pt x^i}(x).
\end{equation}
Equation \eqref{non-degenerate} implies that $\alpha$ in \eqref{eqn-31} is ``exact'', in the sense that $\alpha = \d\eta$ for the time-dependent 1-form $\eta = p_i dx^i$, where $\d$ is the extended differential operator defined in Remark \ref{remark-5}. Similarly, equation \eqref{non-degenerate-2} implies that $(p^*,o^*) = \d\eta^*$ for 1-form $\eta^* = p^*_i dx^i$. The second equality of \eqref{non-degenerate} (or \eqref{non-degenerate-2}), called \emph{Onsager reciprocity} or \emph{Maxwell relations} \cite[Section 5.3]{AM78}, implies that the 1-form $\eta$ (or $\eta^*$) is closed. We will refer to equation \eqref{non-degenerate} or \eqref{non-degenerate-2} as second-order Maxwell relations.

Under the 2nd-order Maxwell relations, the original stochastic Hamilton's system \eqref{stoch-Hamilton-eqns} turns to the following MDE-PDE coupled system.
\begin{equation}\label{stoch-Hamilton-eqns-2}\left\{
  \begin{aligned}
    (DX)^i(t) &= \frac{\pt H}{\pt p_i}(X(t), p(t, X(t)), o(t, X(t)) ), \\
    (QX)^{jk}(t) &= 2\frac{\pt H}{\pt o_{jk}}(X(t), p(t, X(t)), o(t, X(t)) ), \\
    \bigg( \frac{\pt}{\pt t} + \frac{\pt H}{\pt p_j}&(x, p(t, x), o(t, x)) \vf{x^j} + \frac{\pt H}{\pt o_{jk}}(x, p(t, x), o(t, x)) \frac{\pt^2}{\pt x^j\pt x^k}\bigg)p_i(t, x) = - \frac{\pt H}{\pt x^i}(x, p(t, x), o(t, x)), \\
    o_{ij}(t, x) &= \frac{\pt p_i}{\pt x^j}(t, x),
  \end{aligned}\right.
\end{equation}
The boundary values in \eqref{init-value} and \eqref{term-value} now read
\begin{equation}\label{bdr-value}
  \text{Law}(X(0)) = \mu_0, \quad (p,o)(T) = \d\eta^*.
\end{equation}
We first use the terminal value in \eqref{bdr-value}, which satisfies \eqref{non-degenerate-2}, to solve the last two PDEs in \eqref{stoch-Hamilton-eqns-2}. This gives $(p,o)$ and hence the 2nd-order form $\alpha$. Then we plug $p$ and $o$ into the first two MDEs and solve them with initial distribution in \eqref{bdr-value}. This yields in law the $M$-valued diffusion $X = \tau_M^{S*}(\textbf{X})$ as a projective integral process of $A_H$.

We call system \eqref{stoch-Hamilton-eqns} or \eqref{stoch-Hamilton-eqns-2} the \emph{stochastic Hamilton's equations} (S-H equations in short). The second-order Maxwell relations are sufficient for the component $o$ of $\alpha$ in \eqref{eqn-31} to solve the last two equations of \eqref{stoch-Hamilton-eqns}, so we refer to it as an \emph{integrability condition} of \eqref{stoch-Hamilton-eqns}. When restricting settings to Riemannian manifolds, the S-H equations \eqref{stoch-Hamilton-eqns} can be simplified to a global Hamiltonian-type system on $T^*M$, as we will see in Subsection \ref{sec-7-4-2}.

\begin{lemma}
  Let $H: \mathcal T^{S*} M\times\R \to \R$ be a time-dependent 2nd-order Hamiltonian, and $\mathbf X$ be a horizontal integral process of $A_H$. Then, the total mean derivative of $H$ along $\mathbf X$ is
  \begin{equation*}
    \mathbf D_t H = \frac{\pt H}{\pt t}.
  \end{equation*}
\end{lemma}

\begin{proof}
  We use \eqref{stoch-Hamilton-eqns} and local coordinates to derive
  \begin{equation*}
    \begin{split}
      \mathbf D_t H &= D[H(\mathbf X(t),t)] = \frac{\pt H}{\pt t} + D^i x \frac{\pt H}{\pt x^i} + D_i p \frac{\pt H}{\pt p_i} + D_{jk} o \frac{\pt H}{\pt o_{jk}} + \frac{1}{2} Q^{jk} x \frac{\pt^2 H}{\pt x^j \pt x^k} \\
      &\quad + \frac{1}{2} Q_{jk} p \frac{\pt^2 H}{\pt p_j \pt p_k} + \frac{1}{2} Q_{ijkl} o \frac{\pt^2 H}{\pt o_{ij} \pt o_{kl}} + Q^j_k(x,p) \frac{\pt^2 H}{\pt x^j \pt p_k} + Q^j_{kl}(x,o) \frac{\pt^2 H}{\pt x^j \pt o_{kl}} + Q_{jkl}(p,o) \frac{\pt^2 H}{\pt p_j \pt o_{kl}} \\
      &= \frac{\pt H}{\pt t} + D^i x \frac{\pt H}{\pt x^i} + D_i p \frac{\pt H}{\pt p_i} + D_{jk} o \frac{\pt H}{\pt o_{jk}} + \frac{1}{2} Q^{jk} x \frac{\pt^2 H}{\pt x^j \pt x^k} + C_{ij} \frac{\pt H}{\pt o_{ij}} = \frac{\pt H}{\pt t}.
    \end{split}
  \end{equation*}
  The result follows.
\end{proof}

In particular, when $H$ is time-independent, we have
\begin{equation}\label{stoch-cons-energy}
  \mathbf D_t H = 0,
\end{equation}
which is also a consequence of \eqref{2-Hamiltonian-vf-prop}. Equivalently, $H$ is harmonic with respect to the horizontal integral process $\mathbf X$. In this case, we can say that $H$ is stochastically conserved, or is a stochastic conserved quantity. In particular, the expectation $\E[H(\mathbf X)]$ is a constant.

\subsection{Two inspirational examples}\label{exmp-1}

  Let $M$ be a Riemannian manifold with Riemannian metric $g$. Assume for simplicity that $M$ is compact. Let $\nabla$ be the Levi-Civita connection on $TM$ with Christoffel symbols $(\Gamma^k_{ij})$. In this section, we will consider two types of processes on $M$, to provide some intuition of our stochastic Hamiltonian formalism.

\subsubsection{Diffusion processes on Riemannian manifolds}\label{sec-6-3-1}

Consider a second-order Hamiltonian $H$ on $\mathcal T^{S*} M$ with the following coordinate expression:
  \begin{equation}\label{exmp-Hamiltonian-0}
    H(x,p,o) = b^i(x) p_i - \frac{1}{2} g^{ij}(x) \Gamma_{ij}^k(x) p_k + \frac{1}{2} g^{ij}(x) o_{ij} + F(x),
  \end{equation}
  where $b$ is a given smooth vector field on $M$ and $F$ a smooth function on $M$. One can easily verify that the expression at RHS of \eqref{exmp-Hamiltonian-0} is indeed invariant under changes of coordinates. 
We consider the S-H equations \eqref{stoch-Hamilton-eqns-2} subject to boundary conditions $\text{Law} (X(0)) = \mu_0$ and $(p,o)(T) = d^2 S_T$, where $\mu_0$ is a given probability distribution and $S_T$ a given smooth function on $M$.

  By the first two equations of system \eqref{stoch-Hamilton-eqns-2}, the projection diffusion $X$ satisfies the following MDEs,
  \begin{equation}\label{exmp-projection-0}\left\{
    \begin{aligned}
      (DX)^i(t) &= b^i(X(t)) - \frac{1}{2} g^{jk}(X(t)) \Gamma_{jk}^i(X(t)), \\
      (QX)^{jk}(t) &= g^{jk}(X(t)),
    \end{aligned}\right.
  \end{equation}
  subject to the initial distribution $\text{Law} (X(0)) = \mu_0$; or equivalently (according to the end of Section \ref{sec-2-4}), it can be rewritten as the following It\^o SDE in weak sense,
  \begin{equation}\label{diffusion}
    dX^i(t) = \left[ b^i(X(t)) - \frac{1}{2} g^{jk}(X(t)) \Gamma_{jk}^i(X(t)) \right] dt + \sigma_r^i(X(t)) dW^r(t), \quad
    \text{Law} (X(0)) = \mu_0,
  \end{equation}
  where $\sigma$ is the positive definite square root $(1,1)$-tensor of $g$, i.e., $\sum_{r=1}^d \sigma^i_r\sigma^j_r = g^{ij}$, $W$ denotes an $\R^d$-valued standard Brownian motion. 
  Note that equations \eqref{exmp-projection-0} are independent of coordinates $(p,o)$, so they form a closed system on the base manifold $M$ and can be solved independently. Indeed, the solution $X$ is a diffusion on $M$ with generator $A^X = (b^i - \frac{1}{2} g^{jk} \Gamma_{jk}^i) \pt_i + \frac{1}{2}g^{jk} \pt_j\pt_k = \nabla_b + \frac{1}{2} \Delta$.

  Now we consider the last two equations of \eqref{stoch-Hamilton-eqns-2}. The LHS of the third equation reads
  \begin{equation*}
    \bigg[ \frac{\pt}{\pt t} + \left( b^j - \frac{1}{2} g^{kl} \Gamma_{kl}^j \right) \vf{x^j} + \frac{1}{2} g^{jk} \frac{\pt^2}{\pt x^j\pt x^k}\bigg] p_i = \left( \frac{\pt}{\pt t} + \langle b, \nabla \rangle + \frac{1}{2} \Delta \right) p_i,
  \end{equation*}
  where $\langle \cdot,\cdot \rangle$ denotes the pairing of vectors and covectors, $\Delta$ is the Laplace-Beltrami operator and $\nabla$ the gradient, with respect to $g$. In order to find the solution of the third equation of \eqref{stoch-Hamilton-eqns-2}, we consider the following linear backward parabolic equation (where ``backward'' has a meaning different from that in Section \ref{sec-2-2})
  \begin{equation}\label{exmp-HJB-0}
    \frac{\pt S}{\pt t} + \langle b, \nabla S\rangle + \frac{1}{2} \Delta S + F = 0, \quad t\in [0,T),
  \end{equation}
  with terminal value $S(T,x) = S_T(x)$. We let
  \begin{equation}\label{eqn-24-0}
    p_i = \frac{\pt S}{\pt x^i},
  \end{equation}
  and use \eqref{exmp-HJB-0} and \eqref{non-degenerate} to derive
  \begin{equation}\label{eqn-36}
    \begin{split}
      -\frac{\pt F}{\pt x^i} &= \vf{x^i} \left( \frac{\pt S}{\pt t} + \langle b, \nabla S\rangle + \frac{1}{2} \Delta S \right) \\
      &= \left( \frac{\pt}{\pt t} + \langle b, \nabla \rangle + \frac{1}{2} \Delta \right) p_i + \left( \frac{\pt b^j}{\pt x^i} p_j - \frac{1}{2} \frac{\pt g^{kl}}{\pt x^i} \Gamma_{kl}^j p_j - \frac{1}{2} g^{kl}\frac{\pt \Gamma_{kl}^j}{\pt x^i} p_j + \frac{1}{2} \frac{\pt g^{jk}}{\pt x^i} o_{jk} \right) \\
      &= \left( \frac{\pt}{\pt t} + \langle b, \nabla \rangle + \frac{1}{2} \Delta \right) p_i + \frac{\pt}{\pt x^i}(H-F),
    \end{split}
  \end{equation}
  which agree with the third equation of \eqref{stoch-Hamilton-eqns-2}.

  Finally, we combine \eqref{eqn-24-0} with \eqref{non-degenerate} to conclude that the horizontal integral process $\mathbf X$ is
  \begin{equation*}
    \mathbf X(t) = (p,o)(t,X(t)) = \left(\frac{\pt S}{\pt x^i}, \frac{\pt^2 S}{\pt x^j\pt x^k} \right) (t,X(t)) = d^2S(t,X(t)).
  \end{equation*}

\begin{example}[Brownian motions]\label{exmp-BM}
  When $b\equiv 0$ and $F\equiv 0$, the 2nd-order Hamiltonian is $H(x,p,o) = \frac{1}{2} g^{ij}(x) (o_{ij} - \Gamma_{ij}^k(x) p_k)$, the solution process $X$ is a standard Brownian motion on $M$ with initial distribution $\mu_0$. Such 2nd-order Hamiltonian $H$ can be regarded as a ``stochastic deformation'' of the trivial classical Hamiltonian $H_0 = 0$. Indeed, $H$ is the $g$-canonical lift of $H_0$ that will be defined in forthcoming Section \ref{sec-6-6}. Therefore, we may regard Brownian motions as ``stochastization'' or ``stochastic deformation'' of trivially constant curves on the base manifold $M$.
\end{example}

We are going to describe in the next example a dynamical approach to diffusions, elaborated afterwards (Section \ref{sec-7-3}), inspired by Schr\"odinger.

\subsubsection{Reciprocal processes and diffusion bridges on Riemannian manifolds}

With the same coefficients $b, F$ and boundary data $\mu_0,S_T$ in Subsection \ref{sec-6-3-1}, we consider the S-H system \eqref{stoch-Hamilton-eqns-2} with the following second-order Hamiltonian $H$ on $\mathcal T^{S*} M$:
  \begin{equation}\label{exmp-Hamiltonian}
    H(x,p,o) = \frac{1}{2} g^{ij}(x) p_ip_j + b^i(x) p_i - \frac{1}{2} g^{ij}(x) \Gamma_{ij}^k(x) p_k + \frac{1}{2} g^{ij}(x) o_{ij} + F(x),
  \end{equation}
subject to boundary conditions $\text{Law} (X(0)) = \mu_0$ and $(p,o)(T) = d^2 S_T$. Here, $b$ and $F$ are called, respectively, vector and scalar potentials in classical mechanics. Again, it is easy to verify that the expression at RHS of \eqref{exmp-Hamiltonian} is indeed invariant under changes of coordinates. 

  The LHS of the third equation in \eqref{stoch-Hamilton-eqns-2} now reads
  \begin{equation*}
    \bigg[ \frac{\pt}{\pt t} + \left( g^{jk} p_k + b^j - \frac{1}{2} g^{kl} \Gamma_{kl}^j \right) \vf{x^j} + \frac{1}{2} g^{jk} \frac{\pt^2}{\pt x^j\pt x^k}\bigg] p_i = \left( \frac{\pt}{\pt t} + p \cdot \nabla + \langle b, \nabla \rangle + \frac{1}{2} \Delta \right) p_i,
  \end{equation*}
  In order to find the solution of the third equation of \eqref{stoch-Hamilton-eqns-2}, we first consider the positive solution of following backward parabolic equation on $M$
  \begin{equation}\label{Kolmogorov-eqn}
    \frac{\pt u}{\pt t} + \langle b, \nabla u\rangle + \frac{1}{2}\Delta u + F u = 0, \quad t\in [0,T),
  \end{equation}
  with terminal value $u(T,x) = e^{S_T(x)}$, where $\langle \cdot, \cdot\rangle$ denotes the Riemannian inner product with respect to $g$. If we let $S=\ln u$, then it is easy to verify that $S$ satisfies the following Hamilton-Jacobi-Bellman (HJB) equation
  \begin{equation}\label{exmp-HJB}
    \frac{\pt S}{\pt t} + \langle b, \nabla S\rangle + \frac{1}{2} |\nabla S|^2 + \frac{1}{2} \Delta S + F = 0, \quad t\in [0,T),
  \end{equation}
  with terminal value $S(T,x) = S_T(x)$, where $|\cdot|$ denotes the Riemannian norm with respect to $g$. Now we let
  \begin{equation}\label{eqn-24}
    p_i = \frac{\pt S}{\pt x^i} = \frac{\pt \ln u}{\pt x^i},
  \end{equation}
  and use \eqref{exmp-HJB} and \eqref{non-degenerate} to derive, in a way similar to \eqref{eqn-36},
  \begin{equation*}
    -\frac{\pt F}{\pt x^i} = \vf{x^i} \left( \frac{\pt S}{\pt t} + \langle b, \nabla S\rangle + \frac{1}{2} |\nabla S|^2 + \frac{1}{2} \Delta S \right)
    = \left( \frac{\pt}{\pt t} + p \cdot \nabla + \langle b, \nabla \rangle + \frac{1}{2} \Delta \right) p_i + \frac{\pt}{\pt x^i}(H-F),
  \end{equation*}
  which agree with the third equation of \eqref{stoch-Hamilton-eqns-2}. Therefore, the projection diffusion $X$ of the system \eqref{stoch-Hamilton-eqns-2} satisfies the following MDEs,
  \begin{equation}\label{exmp-projection}\left\{
    \begin{aligned}
      (DX)^i(t) &= g^{ij}(X(t)) \frac{\pt \ln u}{\pt x^j}(t, X(t)) + b^i(X(t)) - \frac{1}{2} g^{jk}(X(t)) \Gamma_{jk}^i(X(t)), \\
      (QX)^{jk}(t) &= g^{jk}(X(t)),
    \end{aligned}\right.
  \end{equation}
  subject to the initial distribution $\text{Law} (X(0)) = \mu_0$; or equivalently (according to the end of Section \ref{sec-2-4}), it can be rewritten as the following It\^o SDE in weak sense,
  \begin{equation}\label{diff-bridge}\left\{
    \begin{aligned}
      dX^i(t) &= \left[ g^{ij}(X(t)) \frac{\pt \ln u}{\pt x^j}(t, X(t)) + b^i(X(t)) - \frac{1}{2} g^{jk}(X(t)) \Gamma_{jk}^i(X(t)) \right] dt + \sigma_r^i(X(t)) dW^r(t), \\
      \text{Law} (&X(0)) = \mu_0,
    \end{aligned}\right.
  \end{equation}
  where $\sigma$ is the positive definite square root $(1,1)$-tensor of $g$, i.e., $\sum_{r=1}^d \sigma^i_r\sigma^j_r = g^{ij}$, $W$ denotes an $\R^d$-valued standard Brownian motion. 

  The solution process $X$ of \eqref{diff-bridge} is called a \emph{Bernstein process} \cite{Ber32,CWZ00} (or the \emph{reciprocal process} derived from the $M$-valued diffusion in \eqref{diffusion} \cite{Jam75}).
  The time marginal distribution $\mu_t$ of $X$ satisfies a Born-type formula
  $\mu_t(dx) = u(t,x) v(t,x) dx$ (see, e.g., \cite[Corollary 3.3.1]{Zam86} or \cite[Equations (2.9), (4.6) and (4.8)]{CZ91}), where $v$ satisfies the adjoint equation of \eqref{Kolmogorov-eqn}. 
  The terminal law of $X$ can be determined in the following way: we first solve \eqref{Kolmogorov-eqn} to get $u(0,x)$, and then find out the initial value for $v$ via $\mu_0(dx) = u(0,x) v(0,x) dx$ and solve the equation for $v$ to get $v(T,x)$, finally the terminal law of $X$ is given by $\mu_T(dx) = u(T,x)v(T,x)dx$. In particular, when $\mu_0 = \delta_{q_1}$ and $\mu_T = \delta_{q_2}$ for $q_1, q_2 \in M$, the solution $X$ of \eqref{diff-bridge} is the Markovian bridge of the diffusion $Y$ conditioning on ending point $q_2$ \cite{Cet16}.

  Again we combine \eqref{eqn-24} with \eqref{non-degenerate} to conclude that the horizontal integral process $\mathbf X$ is
  \begin{equation}\label{exmp-int-proc}
    \mathbf X(t) = (p,o)(t,X(t)) = \left(\frac{\pt S}{\pt x^i}, \frac{\pt^2 S}{\pt x^j\pt x^k} \right) (t,X(t)) = d^2S(t,X(t)).
  \end{equation}

\begin{remark}\label{remark-2}

  (i). 
  The derivation of the reciprocal process \eqref{diff-bridge} from the diffusion \eqref{diffusion} was the way chosen by Jamison \cite{Jam75}, inspired by Schr\"odinger's original problem \cite{Sch32}. No geometry or dynamical equations like HJB equation \eqref{exmp-HJB} was involved by him. Like here, Jamison's construction was involving only the past (nondecreasing) filtration. The dynamical content dates back to \cite{Zam86,CZ91,CZ03}, where a reciprocal process was constructed from the only data of a Hamiltonian operator as required by Schr\"odinger's original problem, and the future (nonincreasing) filtration was also used to study the time-reversed dynamics. Cf. also Example \ref{EQM} and Section \ref{sec-7-3}.

  (ii). Equations \eqref{exmp-projection} suggest that the transformation from coordinates $(x,p,o)$ to coordinates $(x,D x, Q x)$ is \emph{not} invertible. More precisely, the coordinates $(D^i x)$ are transformed from $(x,p)$ but the coordinates $(Q^{jk} x)$ are only related to $(x^i)$. Besides, these two equations have nothing to do with the coordinates $(o_{jk})$. However, if we look at the $\nabla$-canonical coordinates $(D^i_\nabla x)$ for \eqref{exmp-projection}, then
  \begin{equation*}
    (D_\nabla X)^i(t) = g^{ij}(X(t)) p_j(t, X(t)) + b^i(X(t)),
  \end{equation*}
  which indicates that the transform from $(x,p)$ to $(x,D_\nabla x)$ is invertible. These will help us establish stochastic Lagrangian mechanics and second-order Legendre transforms, in forthcoming Chapter \ref{sec-7}.

  (iii). As observed in Section \ref{sec-2-2}, every result presented here has a backward version (in the sense of backward mean derivatives with respect to the future filtration $\{\F_t\}$). Indeed, two forward-backward SDE systems for Bernstein diffusions on Euclidean space were derived in \cite{CV15}: one is under the past filtration and coincides with ours, whereas the other one is under the future filtration.
\end{remark}

There are some special cases which are of independent interests and have been considered in the literature.

\begin{example}[Brownian (free) reciprocal processes and Brownian bridges]\label{Brownian-bridge}
  Consider the case where $b\equiv 0$, $F\equiv 0$. 
  In this case, $Y$ is a Brownian motion on $M$, so we call $X$ a Brownian reciprocal process. In particular, the Brownian bridge from $q_1$ to $q_2$ of time length $T>0$ is driven by the It\^o SDE \eqref{diff-bridge} where $X(0) = q_1$, $b\equiv 0$ and
  $u$ satisfies the backward heat equation \eqref{Kolmogorov-eqn} with $F\equiv0$ and final value $u(T,x) = \delta_{q_2}(x)$.
  See also \cite[Theorem 5.4.4]{Hsu02}. 
  Thus, Brownian bridges are understood as stochastic Hamiltonian flows of the 2nd-order Hamiltonian $H(x,p,o) = \frac{1}{2} g^{ij}(x) p_ip_j - \frac{1}{2} g^{ij}(x) \Gamma_{ij}^k(x) p_k + \frac{1}{2} g^{ij}(x) o_{ij}$, compared with geodesics as Hamiltonian flows of the classical Hamiltonian $H_0(x,p) = \frac{1}{2} g^{ij}(x) p_ip_j$ (cf. \cite[Theorem 3.7.1]{AM78}). Here, the 2nd-order Hamiltonian $H$ is the $g$-canonical lift of $H_0$. We can also say that Brownian bridges are ``stochastization'' or ``stochastic deformation'' of geodesics, cf. Example \ref{exmp-BM}. Relations between geodesics and Brownian motions have attracted many studies. For example, one can find various interpolation relations between geodesics and Brownian motions in \cite{ABT15,Li16}.
\end{example}

\begin{example}[Euclidean quantum mechanics {\cite{CZ03,AYZ89,ARZ06}}]\label{EQM}
    It is insightful to consider the case $M=\R^d$ and $b\equiv 0$. The Riemannian metric under consideration is the flat Euclidean one. To catch sight of the analogy with quantum mechanics, we involve the reduced Planck constant $\hbar$ into the second-order Hamiltonian $H$ of \eqref{exmp-Hamiltonian}, so that
        \begin{equation*}
          H_\hbar (x,p,o) = \frac{1}{2}|p|^2 + \frac{\hbar}{2} \tr o + F(x).
        \end{equation*}
        The system \eqref{stoch-Hamilton-eqns} then reads as
        \begin{equation*}\left\{
          \begin{aligned}
            (DX)^i(t) &= p_i(t,X(t)), \\
            (QX)^{jk}(t) &= \hbar \delta^{jk}, \\
            D [p_i(t,X(t))] &= - \frac{\pt F}{\pt x^i}(X(t)), \\
            o_{ik}(t, x) &= \frac{\pt p_k}{\pt x^i}(t, x).
          \end{aligned}\right.
        \end{equation*}
        Note that the first three equations form a sub-system and can be solved separately, as they are independent of the coordinates $o_{ij}$'s. Equation \eqref{Kolmogorov-eqn} and its adjoint now reduce to the following $\hbar$-dependent backward and forward heat equations, respectively,
        \begin{equation*}
          \hbar \frac{\pt u}{\pt t} + \frac{\hbar^2}{2}\Delta u + F u = 0, \qquad -\hbar \frac{\pt v}{\pt t} + \frac{\hbar^2}{2}\Delta v + F v = 0,
        \end{equation*}
        which together with the Born-type formula $\mu_t(dx) = u(t,x) v(t,x) dx$ display the strong analogy to quantum mechanics \cite{Zam86}.
        The function $S=\hbar \ln u$ solves the following $\hbar$-dependent HJB equation
        \begin{equation*}
          \frac{\pt S}{\pt t} + \frac{1}{2} |\nabla S|^2 + \frac{\hbar}{2} \Delta S + F = 0.
        \end{equation*}
        The first three equations then can be solved by letting $p = \nabla S$. The first and third equations imply a Newton-type equation
        \begin{equation}\label{EQM-eqn-motion}
          DDX(t) = -\nabla F(X(t)).
        \end{equation}
        This is indeed the equation of motion of the Euclidean version of quantum mechanics, which was the original motivation of Schr\"odinger in his well-known problem to be discussed in Section \ref{sec-7-3}. See \cite[p.~158]{CZ03} and \cite[Eq.~(4.17)]{Zam15} for more. Note that \cite{CZ03,Zam15} used $V=-F$ to denote the physical scalar potential and used the relation $S= -\hbar \ln u$ and $p = -\nabla S$ to formulate the HJB equation from backward heat equation in the case of nondecreasing (past) filtration.

        There are two special cases of which more will be studied later.

        (i). When $d=1$ and $F(x) = \frac{1}{2} x^2$, i.e., $H(x,p,o) = \frac{1}{2}(p^2 + x^2)+ \frac{1}{2}o$, we call its projective integral process $X$ the (forward) \emph{stochastic harmonic oscillator}. It is a stochastization of the classical harmonic oscillator with Hamiltonian $H_0(x,p) = \frac{1}{2}(p^2 + x^2)$ \cite[Example 5.2.3]{AM78}. Likewise, here $H$ is the canonical lift $H_0$, see Section \ref{sec-6-6}.

        (ii). When $d=1$ and $F(x) = -\frac{1}{2} x^2$, i.e., $H(x,p,o) = \frac{1}{2}(p^2 - x^2)+ \frac{1}{2}o$, we call it the (forward) \emph{Euclidean harmonic oscillator}.
\end{example}

\subsection{The mixed-order contact structure on $\mathcal T^{S*} M\times\R$}

In the later sections we will investigate time-dependent systems. The proper space for consideration is now $\mathcal T^{S*} M\times\R$. Recall in \eqref{germ} that $\mathcal T^{S*} M\times\R = J^2\hat\pi$, where the latter is the second-order jet bundle of $(M\times \R, \hat\pi, M)$.

In classical differential geometry, the first-order jet bundle $J^1\hat\pi = T^* M\times\R$ can be equipped with an exact contact structure in several ways \cite[Section 5.1]{AM78}. Among others, the canonical symplectic form $\omega_0$ on $T^* M$ corresponds to a contact structure on $J^1\hat\pi$ via $\tilde \omega_0 = \hat\pi^* \omega_0$, which is indeed exact as $\tilde \omega_0 = -d\tilde \theta_0$ for $\tilde \theta_0 = dt + \hat\pi^*\theta_0$. Another commonly used contact structure is the Poincar\'e-Cartan form $\omega^0_{H_0} = \tilde \omega_0 + dH_0\wedge dt$ for a given function $H_0\in C^\infty(J^1\hat\pi)$. It is also exact as $\omega^0_{H_0} = - d\theta^0_{H_0}$ where $\theta^0_{H_0} = \hat\pi^*\theta_0 - H_0dt$. The advantage of the Poincar\'e-Cartan form, compared with the contact form $\omega_0$, is that it can be related to the (time-dependent) Hamiltonian vector field $V_{H_0}$ on $T^* M$ of ${H_0}$. More precisely, the vector field $\tilde V_{H_0} = \vf t + V_{H_0}$, treated as a vector field on $J^1\hat\pi$ and called the characteristic vector field of $\omega^0_{H_0}$, is the unique vector field satisfying $\tilde V_{H_0} \lrcorner\, \omega^0_{H_0} = 0$ and $\tilde V_{H_0}\lrcorner\, dt=1$.

Now we proceed in a similar way for the second-order jet bundle $J^2\hat\pi$. Define
$$\tilde \omega = \hat\pi^{S*} \omega \quad\text{and}\quad \tilde \theta = dt + \hat\pi^{S*}\theta.$$
Then $\tilde \omega = -d\tilde \theta$. We call the pair $(J^2\hat\pi, \tilde \omega)$ a second-order contact manifold and the pair $(J^2\hat\pi, \tilde \theta)$ a mixed-order exact contact manifold. In local coordinates, $\tilde \omega$ has the same expression as $\omega$ in \eqref{2-symp-form}, but we stress that it is now a second-order form on $\mathcal T^{S*} M\times\R$. The form $\tilde \theta$ has the local expression
\begin{equation*}
  \tilde \theta = dt + p_i d^2 x^i + \textstyle{\frac{1}{2}} o_{jk} dx^j\cdot dx^k.
\end{equation*}
This makes clear that $\tilde \theta$ is a mixed-order form on $\mathcal T^{S*} M\times\R$.

A time-dependent second-order Hamiltonian $H$ is a smooth function on $J^2\hat\pi\cong \mathcal T^{S*} M\times\R$.
The second-order Hamiltonian vector field $A_H$ of $H$ is now a time-dependent second-order vector field on $\mathcal T^{S*} M$, its horizontal integral process share the same equations as \eqref{stoch-Hamilton-eqns} or \eqref{stoch-Hamilton-eqns-2}, only with $H$ explicitly depending on time. Define a mixed-order vector field $\tilde A_H$ on $\mathcal T^{S*} M\times \R$ by
\begin{equation*}
  \tilde A_H := A_H + \vf t,
\end{equation*}
where $A_H$ is a second-order Hamiltonian vector field of the form \eqref{Hamiltonian-vf}. We call $\tilde A_H$ the extended second-order Hamiltonian vector field of $H$.

We define the second-order counterpart of Poincar\'e-Cartan form by
$$\omega_H := \tilde \omega + d^\circ H\wedge dt = d^2 x^i \wedge d^2 p_i + \textstyle{\frac{1}{2}} dx^j\cdot dx^k \wedge d^2 o_{jk} + d^2 H\wedge dt,$$
and call it the mixed-order Poincar\'e-Cartan form on $\mathcal T^{S*} M\times\R$. It is exact in the sense that $\omega_H = - d^\circ \theta_H$, where $\theta_H = \hat\pi^{S*}\theta - Hdt = p_i d^2 x^i + \textstyle{\frac{1}{2}} o_{jk} dx^j\cdot dx^k - Hdt$.

The following lemma gives the relations between $\omega_H$ and $\tilde A_H$.

\begin{lemma}\label{charact-2-vf}
  The class of extended second-order Hamiltonian vector fields $\tilde A_H$ is the unique class of mixed-order vector fields on $\mathcal T^{S*} M\times\R$ satisfying
  \begin{equation*}
    \tilde A_H \lrcorner\, \omega_H = 0 \quad\text{and}\quad \tilde A_H\lrcorner\, dt=1.
  \end{equation*}
\end{lemma}
\begin{proof}
  Firstly, we show that $\tilde A_H$ satisfies the two equalities. The second equality is trivial. For the first one, we pick a mixed-order vector field $B$ on $\mathcal T^{S*} M\times\R$; then,
  \begin{equation*}
    \begin{split}
      \omega_H(\tilde A_H, B) & = \tilde \omega(\tilde A_H, B) + d^\circ H(\tilde A_H) dt(B) - dt(\tilde A_H) d^\circ H(B) \\
      & = \omega(A_H, \hat\pi^{S}_*(B)) + \left[ d^\circ H(A_H) + d^\circ H(\textstyle{\vf t}) \right] dt(B) - d^\circ H(B) \\
      & = d^2 H(\hat\pi^{S}_*(B)) + \textstyle{\frac{\pt H}{\pt t}} dt(B) - d^\circ H(B) \\
      & = 0.
    \end{split}
  \end{equation*}
  To prove the uniqueness, it suffices to show that any mixed-order vector field $A$ on $\mathcal T^{S*} M\times\R$ satisfying $A \lrcorner\, \omega_H = 0$ is a multiplier of $\tilde A_H$. Suppose that $A$ has the local expression
  \begin{equation*}
    \begin{split}
      A =&\ A^0 \vf t + A^i \vf{x^i} + A_i \vf{p_i} + A^{jk} \frac{\pt^2}{\pt x^j \pt x^k} + A^2_{jk} \vf{o_{jk}} \\
      &\ + A^{11}_{jk} \frac{\pt^2}{\pt p_j \pt p_k} + A_{ijkl} \frac{\pt^2}{\pt o_{ij} \pt o_{kl}} + A^j_k \frac{\pt^2}{\pt x^j \pt p_k} + A^j_{kl} \frac{\pt^2}{\pt x^j \pt o_{kl}} + A_{jkl} \frac{\pt^2}{\pt p_j \pt o_{kl}}.
    \end{split}
  \end{equation*}
  Then, it follows that
  \begin{equation*}
    \begin{split}
      0 = A \lrcorner\, \omega_H =&\ A^i d^2p_i - A_i d^2 x^i + A^{jk} d^2o_{jk} - \textstyle{\frac{1}{2}} A^2_{jk} dx^j\cdot dx^k + \text{terms}\left( A^{11}_{jk}, A_{ijkl}, A^j_k, A^j_{kl}, A_{jkl} \right) \\
      &\ - A^0 \left( \frac{\pt H}{\pt x^i} d^2 x^i + \frac{\pt H}{\pt p_i} d^2 p_i + \frac{\pt H}{\pt o_{jk}} d^2 o_{jk} + \frac{1}{2} \frac{\pt^2 H}{\pt x^j \pt x^k} dx^j\cdot dx^k + \cdots \right) \\
      &\ + \left( A^i \frac{\pt H}{\pt x^i} + A_i \frac{\pt H}{\pt p_i} + + A^{jk} \frac{\pt^2H}{\pt x^j \pt x^k} + A^2_{jk} \frac{\pt H}{\pt o_{jk}} + \cdots \right) dt.
    \end{split}
  \end{equation*}
  The vanishing of each coefficient gives
  \begin{equation*}
    A^i = A^0 \frac{\pt H}{\pt p_i}, \quad A_i = -A^0 \frac{\pt H}{\pt x^i}, \quad A^{jk} = A^0 \frac{\pt H}{\pt o_{jk}}, \quad A^2_{jk} = -A^0 \left( \frac{\pt^2 H}{\pt x^j \pt x^k} + \cdots \right), \quad \cdots.
  \end{equation*}
  Therefore, $A = A^0 \tilde A_H$.
\end{proof}

\subsection{Canonical transformations and Hamilton-Jacobi-Bellman equations}\label{canonical-trans}

Let us study the second-order analogs of canonical transformations and their generating functions. To do so, we need to find a change of coordinates from $(x^i, p_i, o_{jk},t)$ to $(y^i, P_i, O_{jk}, s)$ that preserves the form of stochastic Hamilton's equations \eqref{stoch-Hamilton-eqns} (with time-dependent 2nd-order Hamiltonian). 
More precisely, we have the following definition of canonical transformations between mixed-order contact structures, which is adapted from those between classical contact structures in \cite{ACI83}.

\begin{definition}\label{can-trans-def}
  Let $(\mathcal T^{S*} M\times\R, \tilde \omega)$ and $(\mathcal T^{S*} N\times\R, \tilde \eta)$ be two second-order contact manifolds corresponding to second-order tautological forms $\theta$ and $\vartheta$. 
  A bundle isomorphism $\mathbf F: (\mathcal T^{S*} M\times\R, \hat\pi_{2,1}, T^* M\times\R)\to (\mathcal T^{S*} N\times\R, \hat\rho_{2,1}, T^* N\times\R)$ is called a canonical transformation if its projection $\mathbb F$ is a bundle isomorphism from $(T^* M\times\R, \hat\pi^1_{0,1}, \R)$ to $(T^* N\times\R, \hat\rho^1_{0,1}, \R)$ projecting to $F^0:\R\to\R$,
  and there is a function $H_{\mathbf F} \in C^\infty(\mathcal T^{S*} M\times\R)$ such that
  \begin{equation}\label{2-canonical-transf}
    \mathbf F^{R*}\tilde \eta = \omega_{H_{\mathbf F}},
  \end{equation}
  where $\omega_{H_{\mathbf F}} = \tilde \omega + d^\circ H_{\mathbf F} \wedge dF^0$.
\end{definition}

The map $\mathbf F$ in the definition is also a bundle isomorphism from $(\mathcal T^{S*} M\times\R, \hat\pi^2_{0,1}, \R)$ to $(\mathcal T^{S*} N\times\R, \hat\rho^2_{0,1}, \R)$ projecting to $F^0$. 
Hence, we may assume $\mathbf F(\alpha_q, t) = (\bar{\mathbf F}(\alpha_q, t), F^0(t))$ for all $(\alpha_q, t) \in \mathcal T^{S*} M\times\R$, where $\bar{\mathbf F}$ is a smooth map from $\mathcal T^{S*} M\times\R$ to $\mathcal T^{S*} N$.
For each $t\in\R$, we define a map $\bar{\mathbf F}_t: \mathcal T^{S*} M \to \mathcal T^{S*} N$ by $\bar{\mathbf F}_t(\alpha_q) = \bar{\mathbf F}(\alpha_q,t)$. We also introduce an injection $\jmath_t: \mathcal T^{S*} M \to \mathcal T^{S*} M\times\R$ by $\jmath_t(\alpha_q) = (\alpha_q,t)$. Then, we have $\bar{\mathbf F}_t = \hat\rho_{1,1} \circ \mathbf F \circ \jmath_t$.


\begin{lemma}
  The map $\bar{\mathbf F}_t$ is second-order symplectic for each $t\in\R$ if and only if there is a mixed-order form $\alpha$ on $\mathcal T^{S*} M\times\R$ such that
  \begin{equation*}
    \mathbf F^{R*}\tilde \eta = \tilde \omega + \alpha \wedge dt.
  \end{equation*}
  In particular, condition \eqref{2-canonical-transf} implies that each $\bar{\mathbf F}_t$ is a second-order symplectomorphism.
\end{lemma}
\begin{proof}
  The sufficiency follows from
  \begin{equation*}
    \begin{split}
      (\bar{\mathbf F}_t)^{S*} \eta &= (\jmath_t)^{R*} \circ \mathbf F^{R*} \circ (\hat\rho_{1,1})^{S*} \eta = (\jmath_t)^{R*} \circ \mathbf F^{R*} \tilde \eta \\
      &= (\jmath_t)^{R*} \tilde \omega + (\jmath_t)^{R*} \alpha \wedge (\jmath_t)^{R*} dt = \omega + (\jmath_t)^{R*} \alpha \wedge 0 = \omega.
    \end{split}
  \end{equation*}
  For the necessity, we observe that
  \begin{equation*}
    (\jmath_t)^{R*}(\mathbf F^{R*}\tilde \eta - \tilde \omega) = (\bar{\mathbf F}_t)^{S*} \eta - \omega = 0.
  \end{equation*}
  So we can write $\mathbf F^{R*}\tilde \eta - \tilde \omega = \alpha \wedge dt + \gamma$, where $\gamma$ is a mixed-order form which does not involve $dt$. This leads to $\gamma = (\hat\pi_{1,1})^{R*} \circ (\jmath_t)^{R*}\gamma = (\hat\pi_{1,1})^{R*} \circ (\jmath_t)^{R*}(\mathbf F^{R*}\tilde \eta - \tilde \omega - \alpha \wedge dt) = 0$. The result follows.
\end{proof}

The following lemma gives some equivalent statements to the condition \eqref{2-canonical-transf}.

\begin{lemma}\label{equiv-canonical-transf}
  Condition \eqref{2-canonical-transf} is equivalent to the following: \\
  (i) $\mathbf F^{R*}\tilde \vartheta - \tilde \theta + H_{\mathbf F} dF^0$ is mixed-order closed; \\
  (ii) for all $K\in C^\infty(\mathcal T^{S*} N\times\R)$, $\mathbf F^{R*} \eta_K = \omega_H$; \\
  (iii) for all $K\in C^\infty(\mathcal T^{S*} N\times\R)$, $\mathbf F^{R}_* \tilde A_H = \tilde A_K$; \\
  where $H = (K\circ \mathbf F + H_{\mathbf F})\dot F^0$.
\end{lemma}
\begin{proof}
  The equivalence between \eqref{2-canonical-transf} and (i) is clear. For \eqref{2-canonical-transf}$\Rightarrow$(ii), since $\mathbf F$ projects to $F^0$,
  \begin{equation*}
    \begin{split}
      \mathbf F^{R*} \eta_K &= \mathbf F^{R*} \tilde \eta + d^\circ (K\circ \mathbf F)\wedge d(t\circ \mathbf F) = \tilde \omega + d^\circ H_{\mathbf F} \wedge dF^0 + d^\circ(K\circ \mathbf F)\wedge dF^0 \\
      &= \tilde \omega + d^\circ H \wedge dt = \omega_H.
    \end{split}
  \end{equation*}
  The converse (ii)$\Rightarrow$\eqref{2-canonical-transf} is straightforward by letting $K\equiv 0$. To show (ii)$\Rightarrow$(iii), by applying Lemma \ref{charact-2-vf}, it suffices to prove that
  \begin{equation*}
    \mathbf F^{R}_* \tilde A_H \lrcorner\, \eta_K = 0 \quad\text{and}\quad \mathbf F^{R}_* \tilde A_H\lrcorner\, dt=1,
  \end{equation*}
  while
  \begin{equation*}
    \mathbf F^{R}_* \tilde A_H \lrcorner\, \eta_K = (\mathbf F^{R*})^{-1} (\tilde A_H \lrcorner\, \mathbf F^{R*}\eta_K ) = (\mathbf F^{R*})^{-1} (\tilde A_H \lrcorner\, \omega_H ) = 0,
  \end{equation*}
  and
  \begin{equation*}
    \mathbf F^{R}_* \tilde A_H\lrcorner\, dt = (\mathbf F^{R*})^{-1} (\tilde A_H \lrcorner\, \mathbf F^{R*}dt ) = (\mathbf F^{R*})^{-1} (\dot F^0 \tilde A_H \lrcorner\, dt ) = (\mathbf F^{R*})^{-1} (\dot F^0) = 1.
  \end{equation*}
  (iii)$\Rightarrow$(ii) is similar.
\end{proof}

\begin{definition}
  Let $\mathbf F: \mathcal T^{S*} M\times\R \to \mathcal T^{S*} N\times\R$ be canonical. If we can locally write
  \begin{equation}\label{generating-func}
    \mathbf F^{R*}\tilde \vartheta - \tilde \theta + H_{\mathbf F} d F^0 =  -d^{\circ} G
  \end{equation}
  for $G\in C^\infty(M\times\R)$, then we call $G$ a generating function for the canonical transformation $\mathbf F$.
\end{definition}

We use $(x, p, o, t)$ for local coordinates on $\mathcal T^{S*} M\times\R$ and $(y, P, O,s)$ for those on $\mathcal T^{S*} N\times\R$. Recall that $\mathbf F(\alpha_q, t) = (\bar{\mathbf F}(\alpha_q, t), F^0(t))$. Then, using \eqref{local-rep-pull-mixed}, the relation \eqref{generating-func} reads in coordinates as
\begin{equation*}
  \begin{split}
    &\left[ \dot F^0 + (P_i\circ \mathbf F) \frac{\pt\bar{\mathbf F}^i}{\pt t} \right] dt + (P_i\circ \mathbf F) \frac{\pt\bar{\mathbf F}^i}{\pt x^j} d^2x^j + \frac{1}{2} \left[ (P_i\circ \mathbf F) \frac{\pt^2\bar{\mathbf F}^i}{\pt x^k\pt x^l} + (O_{ij}\circ \mathbf F)\frac{\pt\bar{\mathbf F}^i}{\pt x^k}\frac{d\bar{\mathbf F}^j}{dx^l} \right] dx^k \cdot dx^l \\
    &- \left( dt + p_i d^2 x^i + \frac{1}{2} o_{jk} dx^j\cdot dx^k \right) + H_{\mathbf F} dF^0 + \frac{\pt G}{\pt t} dt + \frac{\pt G}{\pt x^i} d^2 x^i + \frac{1}{2} \frac{\pt^2 G}{\pt x^j \pt x^k} dx^j\cdot dx^k = 0.
  \end{split}
\end{equation*}
Balancing the coefficient of $dt$, we get

\begin{equation*}
    \frac{\pt G}{\pt t} + H_{\mathbf F} + (P_i\circ \mathbf F) \frac{\pt\bar{\mathbf F}^i}{\pt t} + \dot F^0 - 1 = 0.
\end{equation*}


By Lemma \ref{equiv-canonical-transf}, the new Hamiltonian function $K$ after transformation $\mathbf F$ is related to the old Hamiltonian $H$ by $(H - K\circ \mathbf F)\dot F^0 = H_{\mathbf F}$. Let us further assume that we can choose coordinates in which $(y^i)$ and $(x^i)$ are independent, so that the independent variables in \eqref{generating-func} are $(x, y, t)$. Then, relation \eqref{generating-func} means
\begin{equation}\label{formal-transf}
  \left( P_i d^2 y^i + \textstyle{\frac{1}{2}} O_{jk} dy^j\cdot dy^k + dF^0 \right) - \left( p_i d^2 x^i + \textstyle{\frac{1}{2}} o_{jk} dx^j\cdot dx^k + dt \right) + (H dt - K d F^0 ) = - d^\circ G,
\end{equation}
which implies that the generating function of the canonical transformation $G(x, y, t)$ satisfies
\begin{equation}\label{formal-relation}
  \left\{\begin{aligned}
    &p_i = \frac{\pt G}{\pt x^i}, \quad o_{jk} \frac{\pt x^k}{\pt y^l} = \frac{\pt^2 G}{\pt x^j \pt x^k} \frac{\pt x^k}{\pt y^l} + \frac{\pt^2 G}{\pt x^j \pt y^l}, \quad P_i = - \frac{\pt G}{\pt y^i}, \quad O_{jk} = - \frac{\pt^2 G}{\pt y^j \pt y^k} - \frac{\pt^2 G}{\pt y^j \pt x^l} \frac{\pt x^l}{\pt y^k}, \\
    &(K - 1) \dot F^0 - H + 1 = \frac{\pt G}{\pt t}.
  \end{aligned}\right.
\end{equation}
The expressions for $(o_{jk})$ and $(O_{jk})$ are due to the mixed differential term in $d^\circ G$, and correspond to the relation \eqref{non-degenerate}.

\begin{remark}
  Unlike the canonical transformations of classical Hamiltonian systems which have four types of generating functions related via classical Legendre transform (see \cite[Section 9.1]{GPS02}), here we can only have the type using $(x, y, t)$ as independent variables but not others. This can be attributed to the ill-behaveness of the 2nd-order analog of Legendre transform, as indicated in Remark \ref{remark-2}.(iii). However, if the configuration space $M$ is a Riemannian manifold, stochastic Hamiltonian mechanics can be simplified to share the same phase space $T^* M$ as classical Hamiltonian mechanics, so that we can also have four types of generating functions. See Subsection \ref{sec-7-4-2} for details and examples of canonical transformations.
\end{remark}

The Hamilton-Jacobi-Bellman (HJB) equation can be introduced as a special case of a time-dependent canonical transformation \eqref{formal-relation}. In the case where $F^0 = \id_\R$ and the new Hamiltonian $K$ vanishes formally, we denote by $S$ the corresponding generating function $G$. It follows from \eqref{formal-relation} that $S$ solves the Hamilton-Jacobi-Bellman equation,
\begin{equation}\label{HJB}
  \frac{\pt S}{\pt t} + H\left( x^i, \frac{\pt S}{\pt x^i}, \frac{\pt^2 S}{\pt x^j \pt x^k}, t \right) = 0.
\end{equation}
We will refer to equation \eqref{HJB} as the HJB equation \emph{associated with} second-order Hamiltonian $H$, and a solution $S$ of \eqref{HJB} as a \emph{second-order Hamilton's principal function} of $H$.

More generally, we have

\begin{theorem}\label{HJB-int-proc}
   Let $A_H$ be a second-order Hamiltonian vector field on $(\mathcal T^{S*} M, \omega)$ and let $S\in C^\infty(M\times\R)$. Then, the following statements are equivalent: \\
   (i) for every $M$-valued diffusion $X$ satisfying
   \begin{equation*}
     (DX(t), QX(t)) = d^2 (\tau_M^*)_{d^2 S(t,X(t))} A_H,
   \end{equation*}
   the $\mathcal T^{S*} M$-valued process $d^2S\circ X$ is a horizontal integral process of $A_H$; \\
   (ii) $S$ satisfies the Hamilton-Jacobi-Bellman equation
   \begin{equation}\label{HJB-3}
     \frac{\pt S}{\pt t} + H(d^2 S, t) = f(t),
   \end{equation}
   for some function $f$ depending only on $t$.
\end{theorem}

\begin{proof}
  Let $\mathbf X = d^2S\circ X$ and set $x^i = x^i\circ d^2S$, $p_i = p_i\circ d^2S$, $o_{jk} = o_{jk}\circ d^2S$. Then
  \begin{equation}\label{eqn-9}
    p_i(t,x) = \frac{\pt S}{\pt x^i}(t,x), \quad o_{jk}(t,x) = \frac{\pt^2 S}{\pt x^j \pt x^k}(t,x).
  \end{equation}
  These imply that the last equation of the system \eqref{stoch-Hamilton-eqns-2} holds. Since
  $$d^2 (\tau_M^*)_{\mathbf X(t)} A_H = \frac{\pt H}{\pt p_i}(\mathbf X(t)) \vf{x^i} + \frac{\pt H}{\pt o_{jk}}(\mathbf X(t)) \frac{\pt^2}{\pt x^j \pt x^k},$$
  the first two equations in \eqref{stoch-Hamilton-eqns} or \eqref{stoch-Hamilton-eqns-2} hold. Hence, to turn the process $\mathbf X = d^2S\circ X$ into a horizontal integral process of $A_H$, it is sufficient and necessary to make sure that the third equation in \eqref{stoch-Hamilton-eqns-2} holds. Plugging the first equation of \eqref{eqn-9} into the third equation, it reads as
  \begin{equation*}
    \bigg( \frac{\pt}{\pt t} + \frac{\pt H}{\pt p_j} \vf{x^j} + \frac{\pt H}{\pt o_{jk}} \frac{\pt^2}{\pt x^j\pt x^k}\bigg) \frac{\pt S}{\pt x^i} = - \frac{\pt H}{\pt x^i}.
  \end{equation*}
  A straightforward reinterpretation yields
  \begin{equation*}
    \vf{x^i} \left[ \frac{\pt S}{\pt t} + H\left( x^j, \frac{\pt S}{\pt x^j}, \frac{\pt^2 S}{\pt x^j \pt x^k}, t \right) \right] = 0.
  \end{equation*}
  The result follows.
\end{proof}

\begin{remark}\label{remark-3}
  If $S$ solves the HJB equation \eqref{HJB-3}, then $\tilde S = S - \tilde f$ solve \eqref{HJB} with $\tilde f$ a primitive function of $f$. As a matter of fact, one can always integrate the time-dependent function $f$ into the 2nd-order Hamiltonian function $H$ such that the HJB equation \eqref{HJB-3} has the same form as \eqref{HJB}. More precisely, if we let $\tilde H = H - f$, then Theorem \ref{HJB-int-proc} also holds with $\tilde H$ and zero function in place of $H$ and $f$, respectively. A similar argument holds for S-H equations \eqref{stoch-Hamilton-eqns}. Indeed, adding a function $f$ depending only on time to a 2nd-order Hamiltonian does not change its S-H equations.
\end{remark}

\begin{example}
  The function $S=\ln u$ considered in Section \ref{exmp-1} satisfies the Hamilton-Jacobi-Bellman equation \eqref{exmp-HJB}, which is exactly $\frac{\pt S}{\pt t} + H(d^2 S) = 0$ with the second-order Hamiltonian $H$ given in \eqref{exmp-Hamiltonian}. Hence, this theorem yields that the process $d^2S\circ X$ is a horizontal integral process of $A_H$, which coincides with \eqref{exmp-int-proc}. The Euclidean case for such argument has been discovered in \cite[p.~180]{CZ03} or \cite[Eq.~(4.20)]{Zam15}.
\end{example}

By \eqref{HJB} and \eqref{eqn-9}, the total mean derivative of a 2nd-order Hamilton's principal function $S$ 
is given by
\begin{equation}\label{eqn-25}
  \mathbf D_t S = \frac{\pt S}{\pt t} + D^i x \frac{\pt S}{\pt x^i} + \frac{1}{2} Q^{jk} x \frac{\pt^2 S}{\pt x^j \pt x^k} = p_i D^i x + \frac{1}{2} o_{jk} Q^{jk} x - H(x,p,o,t).
\end{equation}
where $(p(t,x),o(t,x)) = d^2 S(t,x)$ as in \eqref{eqn-9}.

%

\subsection{Second-order Hamiltonian functions from classical ones}\label{sec-6-6}

In the presence of a linear connection $\nabla$ on $M$, we are able to
reduce (or produce) second-order Hamiltonian functions to (from) classical ones.

Let be given a second-order Hamiltonian function $H: \mathcal T^{S*} M\times\R \to \R$. We make use of the fiber-linear bundle injection $\hat\iota^*_\nabla: T^* M \to \mathcal T^{S*} M$ in \eqref{iota-conn} to define a classical Hamiltonian by
\begin{equation}\label{reduction}
  H_0 = H\circ (\hat\iota^*_\nabla \times \id_{\R}): T^* M\times\R \to \R.
\end{equation}
In canonical coordinates, it maps as $H_0(x, p, t) = H(x, p, (\Gamma_{jk}^i(x) p_i), t)$. If we introduce a family of auxiliary variables by
\begin{equation}\label{o-hat}
  \hat o_{jk} = \hat o_{jk}(x,p) := \Gamma_{jk}^i(x) p_i.
\end{equation}
Then, we can write
\begin{equation*}
  H_0(x,p,t) = H(x,p,\hat o(x,p),t).
\end{equation*}
We say $H$ \emph{reduces} to $H_0$ under the connection $\nabla$, or $H_0$ is the $\nabla$-\emph{reduction} of $H$.

Clearly, the way to lift from a classical Hamiltonian $H_0: T^* M\times\R \to \R$ to a second-order Hamiltonian function that reduces to $H_0$ under $\nabla$ is \emph{not} unique. But there is a canonical reduction when we are provided with a symmetric $(2,0)$-tensor field $g$ (not necessarily Riemannian), given by
\begin{equation}\label{lift-Ham}
  \overline H^g_0(x,p,o,t) := H_0(x,p,t) + \textstyle{\frac{1}{2}} g^{jk}(x) \left(o_{jk} - \Gamma^i_{jk}(x)p_i \right) = H_0(x,p,t) + \textstyle{\frac{1}{2}} g^{jk}(x) o_{jk}^\nabla.
\end{equation}
Then, $H_0$ is the $\nabla$-reduction of $\overline H^g_0$, and
\begin{equation}\label{canonical-ext}
  \textstyle{\frac{1}{2}} o_{jk} g^{jk} - \overline H^g_0(x,p,o,t) = \textstyle{\frac{1}{2}} \hat o_{jk} g^{jk} - \overline H^g_0(x,p,\hat o,t).
\end{equation}
We call $\overline H^g_0$ the $(g,\nabla)$-\emph{canonical lift} of $H_0$. If $g$ is a Riemannian metric and $\nabla$ is the associated Levi-Civita connection, then we simply call $\overline H^g_0$ the $g$-\emph{canonical lift} of $H_0$. If there is a classical Hamiltonian $H_0$ such that the second-order Hamiltonian $H$ is the $(g,\nabla)$- (or $g$-) canonical lift of $H_0$, we say $H$ is $(g,\nabla)$- (or $g$-) canonical.

As an example, the second-order Hamiltonian $H$ of \eqref{exmp-Hamiltonian} is $g$-canonical and reduces to $H_0(x,p) = \frac{1}{2} g^{ij}(x) p_ip_j + b^i(x)p_i + F(x)$. 

Furthermore, for the canonical transformation $\mathbf F: \mathcal T^{S*} M \to \mathcal T^{S*} N$ in Definition \ref{can-trans-def}, we can reduce its associated function $H_{\mathbf F} \in C^\infty(\mathcal T^{S*} M\times\R)$ to a classical function $H^0_{\mathbf F} \in C^\infty(T^* M\times\R)$ via \eqref{reduction}. As a consequence of \eqref{2-canonical-transf}, the projection of $\mathbf F$, i.e., the map $\mathbb F: T^* M\times\R\to T^* N\times\R$ satisfies $\mathbb F^*\tilde \eta_0 = \omega^0_{H^0_{\mathbf F}}$ where $\omega^0_{H^0_{\mathbf F}} = \tilde \omega_0 + d H^0_{\mathbf F} \wedge dF^0$. It follows that $\mathbb F$ is a classical canonical transformation \cite[Definition 5.2.6]{AM78}.

We will go back to this issue in Section \ref{sec-7-4} where the second-order Legendre transform will be developed. In particular, we will show there that for the canonical 2nd-order Hamiltonian in \eqref{lift-Ham}, the corresponding 2nd-order Hamilton's equations \eqref{stoch-Hamilton-eqns-2} can be rewritten on the cotangent bundle $T^* M$ in a global fashion, see Theorem \ref{canonical-reduction}.

%

\section{Stochastic Lagrangian mechanics}\label{sec-7}

In this chapter, we specify a Riemannian metric $g$ for the manifold $M$, and a $g$-compatible linear connection $\nabla$. Note that such $g$ and $\nabla$ always exist but are not unique in general.

We will denote by $|\cdot|$ and $\langle\cdot,\cdot\rangle$ the Riemannian norm and inner product, respectively. Also, denote by $\check g$ the inverse metric tensor of $g$, and $(\Gamma_{jk}^i)$ the Christoffel symbols of $\nabla$. We observe that $\check g$ is a $(2,0)$-tensor field. Denote by $R$ the Riemann curvature tensor and $\Ric$ the Ricci $(1,1)$-tensor.

\subsection{Mean covariant derivatives}

\begin{definition}[Vector fields and 1-forms along diffusions]
  Let $X$ be diffusion on $M$. By a vector field along $X$, we mean a $TM$-valued process $V$, such that $\tau_M (V(t)) = X(t)$ for all $t$. Similarly, by a 1-form along $X$, we mean a $T^* M$-valued process $\eta$, such that $\tau^*_M (\eta(t)) = X(t)$ for all $t$.
\end{definition}

Clearly, for a time-dependent vector field $V$ on $M$, the restriction of $V$ on $X$, i.e., $\{V_{(t,X(t))}\}$, is a vector field along $X$. In this case, we call $\{V_{(t,X(t))}\}$ a vector field restricted on $X$. In this way, vector fields restricted on $X$ are just $TM$-valued horizontal diffusions projecting to $X$. Similarly for 1-forms.

\begin{definition}[Parallelisms along diffusions]
  Let $X \in I_{t_0}(M)$. A vector field $V$ along $X$ is said to be parallel along $X$ if the following Stratonovich SDE in local coordinates holds,
  \begin{equation}\label{parallel-vf}
    d V^i(t) + \Gamma_{jk}^i(X(t)) V^j(t) \circ dX^k(t) = 0.
  \end{equation}
  A 1-form $\eta$ along $X$ is said to be parallel along $X$ if
  \begin{equation*}
    d \eta_j(t) - \Gamma_{jk}^i(X(t)) \eta_i(t) \circ dX^k(t) = 0.
  \end{equation*}
\end{definition}

\begin{definition}[Stochastic parallel displacements]
  Given a diffusion $X\in I_{t_0} (M)$ and a (random) vector $v \in T_{X(t_0)} M$, the stochastic parallel displacement of $v$ along $X$ is the extension of $v$ to a parallel vector field $V$ along $X$, that is, $V$ satisfies the SDE \eqref{parallel-vf} with initial condition $V(t_0) = v$. We denote $\Gamma(X)_{t_0}^t v := V(t)$ and $\Gamma(X)_t^{t_0} V(t) := v$. The stochastic parallel displacement of a (random) covector $\eta\in T^*_{X(t_0)} M$ along $X$ is defined in a similar fashion.
\end{definition}

\begin{definition}[Damped parallel displacements]
  Let $X\in I_{t_0} (M)$. Given a (random) vector $v \in T_{X(t_0)} M$ and covector $\eta_0 \in T^*_{X(t_0)} M$, the damped parallel displacement of $v$ along $X$ is the extension of $v$ to a vector field $V$ along $X$ that satisfies the SDE
  \begin{equation}\label{geodesic-vf}
    d V^i(t) + \Gamma_{jk}^i(X(t)) V^j(t) \circ dX^k(t) + \frac{1}{2} R^i_{kjl}(X(t)) V^j(t) (QX)^{kl}(t) dt = 0, \quad V(t_0) = v.
  \end{equation}
  The damped parallel displacement of $\eta_0$ along $X$ is the extension of $\eta$ to a vector field $\eta$ along $X$ that satisfies the SDE
  \begin{equation}\label{geodesic-form}
    d \eta_j(t) - \Gamma_{jk}^i(X(t)) \eta_i(t) \circ dX^k(t) - \frac{1}{2} R^i_{kjl}(X(t)) \eta_i(t) (QX)^{kl}(t) dt = 0, \quad \eta(t_0) = \eta_0.
  \end{equation}
  We denote $\overline\Gamma(X)_{t_0}^t v := V(t)$, $\overline\Gamma(X)_{t_0}^t \eta_0 := \eta(t)$, and $\overline\Gamma(X)_t^{t_0} V(t) := v$, $\overline\Gamma(X)_t^{t_0} \eta(t) := \eta_0$.
\end{definition}

If $V$ and $\eta$ are restrictions on $X$, that is, $V(t) = V_{(t,X(t))}$ and $\eta(t) = \eta_{(t,X(t))}$, then equations \eqref{geodesic-vf} and \eqref{geodesic-form} can be rewritten, respectively, as
\begin{equation*}
  \frac{\pt V}{\pt t} dt + \nabla^{}_{\circ dX} V + \frac{1}{2} R(V, \circ dX) \circ dX = 0, \qquad \frac{\pt \eta}{\pt t} dt + \nabla^{}_{\circ dX} \eta - \frac{1}{2} R(\eta, \circ dX) \circ dX = 0,
\end{equation*}
where we mean by $R(\eta, V) W$ the 1-form $[R(\eta^\sharp, V) W]^\flat$. The Stratonovich stochastic differentials can be transformed into It\^o ones. For example, \eqref{geodesic-form} is equivalent to
\begin{equation}\label{parallel-form-Ito}
  d \eta_j(t) = \Gamma_{jk}^i(X(t)) \eta_i(t) dX^k(t) + \frac{1}{2} (QX)^{kl}(t)\left( \frac{\pt \Gamma_{jk}^i}{\pt x^l} + \Gamma_{jk}^m \Gamma_{ml}^i \right)(X(t)) \eta_i(t) dt + \frac{1}{2} R^i_{kjl}(X(t)) \eta_i(t) (QX)^{kl}(t) dt.
\end{equation}

\begin{remark}
  The notion of stochastic parallel displacements was introduced by It\^o \cite{Ito75} and Dynkin \cite{Dyn68}. The notion of damped parallel displacement is due to Malliavin \cite{Mal97}. It was originally introduced by Dohrn and Guerra \cite{DG79}, where they call it geodesic correction to the stochastic parallel displacement.
\end{remark}

\begin{lemma}\label{parallel-vf-form}
  Let $X \in I_{t_0}(M)$.  \\
  (i). Let $\eta$ be a 1-form on $M$ parallel along $X$. If $V$ is a vector field on $M$ which is also parallel along $X$, then $\eta(V)(t) = \eta(V)(t_0)$ for all $t\ge t_0$; if $v \in T_{X(t_0)} M$, then $\eta(\Gamma(X)_{t_0}^t v)(t) = \eta(v)(t_0)$ for all $t\ge t_0$. \\
  (ii). Let $\eta$ be a 1-form on along $X$ satisfying the SDE \eqref{geodesic-form}. If $V$ is a vector field along $X$ satisfying the SDE \eqref{geodesic-vf}, then $\eta(V)(t) = \eta(V)(t_0)$ for all $t\ge t_0$; if $v \in T_{X(t_0)} M$, then $\eta(\overline\Gamma(X)_{t_0}^t v)(t) = \eta(v)(t_0)$ for all $t\ge t_0$.
\end{lemma}
\begin{proof}
  We only prove Assertion (ii), as (i) is similar. Since Stratonovich stochastic differentials obey Leibniz's rule, we have
  \begin{equation*}
    \begin{split}
      d[\eta(V)] &= \eta_i \circ dV^i + V^j \circ d\eta_j \\
      &= -\eta_i \Gamma_{jk}^i V^j \circ dX^k - \frac{1}{2} \eta_i R^i_{kjl} V^j (QX)^{kl} dt + V^j \Gamma_{jk}^i \eta_i \circ dX^k + \frac{1}{2} V^j R^i_{kjl} \eta_i (QX)^{kl} dt \\
      &= 0.
    \end{split}
  \end{equation*}
  This proves the first statement of (ii). The second statement of (ii) follows by letting $V(t) := \overline\Gamma(X)_{t_0}^t v$.
\end{proof}

\begin{definition}[Mean covariant derivatives along diffusions]
  Given a diffusion $X$ on $M$. Let $V$ and $\eta$ be time-dependent vector field and 1-form along $X$, respectively. The (forward) mean covariant derivative of $V$ with respect to $X$ is a time-dependent vector field $\frac{\D V}{dt}$ along $X$, defined by
  \begin{equation}\label{mean-cov-d}
    \frac{\D V}{dt} (t) = \lim_{\e\to0^+} \E \left[ \frac{\Gamma(X)_{t+\e}^t V(t+\e) - V(t) }{\e} \Bigg| \Pred_t \right].
  \end{equation}
  The damped mean covariant derivative of $V$ with respect to $X$ is a time-dependent vector field $\frac{\overline\D V}{dt}$ along $X$ with $\overline\Gamma$ instead of $\Gamma$ in \eqref{mean-cov-d}. Similarly, we can define $\frac{\D \eta}{dt}$ and $\frac{\overline\D \eta}{dt}$.
\end{definition}

\begin{lemma}\label{stoch-mean-derv-prop}
  (i). Let $V$ and $\eta$ be vector field and 1-form along $X$. If $\eta$ is parallel along $X$, then
  \begin{equation}\label{eqn-20}
    \textstyle{\E \left[ \eta \left( \frac{\D V}{dt} \right) \right] = \E \left( D[\eta(V)] \right).}
  \end{equation}
  If $\eta$ satisfies the SDE \eqref{geodesic-form}, then \eqref{eqn-20} holds true with $\frac{\overline\D}{dt}$ instead of $\frac{\D}{dt}$.

  (ii). Let $V$ be a vector field restricted on $X$. Then
  $$\frac{\overline\D V}{dt} = \frac{\D V}{dt} + \frac{1}{2} (QX)^{ij} R(V,\pt_i)\pt_j = \frac{\pt V}{\pt t} + \nabla^{}_{D_\nabla X} V + \frac{1}{2} (QX)^{ij} \left( \nabla^2_{\pt_i,\pt_j} V + R(V,\pt_i)\pt_j \right).$$

  (iii). Let $\eta$ be a 1-form restricted on $X$. Then
  $$\frac{\overline\D \eta}{dt} = \frac{\D \eta}{dt} - \frac{1}{2} (QX)^{ij} 
  R(\eta,\pt_j)\pt_i = \frac{\pt \eta}{\pt t} + \nabla^{}_{D_\nabla X} \eta + \frac{1}{2} (QX)^{ij} \left( \nabla^2_{\pt_i,\pt_j} \eta - 
  R(\eta,\pt_j)\pt_i \right).$$

  (iv). Let $V$ and $\eta$ be a vector field and a 1-form restricted on $X$. Then
  \begin{equation*}
    \D_t[\eta(V)] = \eta \left( \frac{\D V}{dt} \right) + \frac{\D \eta}{dt} (V) + (QX)^{ij} (\nabla_{\pt_i} \eta) (\nabla_{\pt_j} V) = \eta \left( \frac{\overline\D V}{dt} \right) + \frac{\overline\D \eta}{dt} (V) + (QX)^{ij} (\nabla_{\pt_i} \eta) (\nabla_{\pt_j} V).
  \end{equation*}
\end{lemma}

\begin{proof}
  (i). By Lemma \ref{parallel-vf-form}.(i), we have
  \begin{equation*}
    \begin{split}
      \E \left[ \eta \left( \frac{\D V}{dt} \right) (t) \right] &= \lim_{\e\to0} \E \left[ \frac{ \eta(t) (\Gamma(X)_{t+\e}^t V(t+\e) ) - \eta(t)(V(t)) }{\e} \right] \\
      &= \lim_{\e\to0} \E \left[ \frac{ \eta( V ) (t+\e) - \eta( V ) (t) }{\e} \right] \\
      &= \E \left( D[\eta(V)(t)] \right).
    \end{split}
  \end{equation*}
  This proves the first statement of (i). The second statement of (i) follows by a similar argument with $\frac{\overline\D}{dt}$ in place of $\frac{\D}{dt}$ and $\overline\Gamma$ in place of $\Gamma$.

  (ii). It suffices to derive the expression for $\frac{\overline\D V}{dt}$. Suppose that $\eta$ is a 1-form satisfying the SDE \eqref{geodesic-form} and the diffusion $X$ satisfies $QX(t) = (\sigma\circ\sigma^*)(t,X(t))$. Then, we apply It\^o's formula to $\eta(V)(X(t))$ and make use of \eqref{relation-two-drvtv} and \eqref{parallel-form-Ito}. We get
  \begin{equation*}
    \begin{split}
      d[\eta(V)] &= d(\eta_i V^i) = \eta_i \left( \frac{\pt V^i}{\pt t} dt + \frac{\pt V^i}{\pt x^j} dX^j + \frac{1}{2} \frac{\pt^2 V^i}{\pt x^j\pt x^k} d[X^j, X^k] \right) + V^j d\eta_j + d[\eta_j, V^j] \\
      &= \eta_i \left( \frac{\pt V^i}{\pt t} + \frac{\pt V^i}{\pt x^j} (DX)^j + \frac{1}{2} \frac{\pt^2 V^i}{\pt x^j\pt x^k} (QX)^{jk} \right) dt + \eta_i \frac{\pt V^i}{\pt x^j} \sigma_r^j dB^r \\
      &\quad + V^j \left[ \Gamma_{jk}^i (DX)^k + \frac{1}{2} (QX)^{kl}\left( \frac{\pt \Gamma_{jk}^i}{\pt x^l} + \Gamma_{jk}^m \Gamma_{ml}^i \right) + \frac{1}{2} R^i_{kjl} (QX)^{kl} \right] \eta_i dt + V^j \Gamma_{jk}^i \eta_i \sigma_r^k dB^r \\
      &\quad + \Gamma_{jk}^i \eta_i \frac{\pt V^j}{\pt x^l} (QX)^{kl} dt \\
      &= \eta_i \left[ \frac{\pt V^i}{\pt t} + \left( \frac{\pt V^i}{\pt x^k} + V^j \Gamma_{jk}^i \right) (D_\nabla X)^k \right] dt \\
      &\quad + \frac{1}{2} \eta_i (QX)^{kl} \left[ -\frac{\pt V^i}{\pt x^j} \Gamma^j_{kl} + \frac{\pt^2 V^i}{\pt x^k\pt x^l} + V^j \left( -\Gamma_{jm}^i \Gamma^m_{kl} + \frac{\pt \Gamma_{jk}^i}{\pt x^l} + \Gamma_{jk}^m \Gamma_{ml}^i \right) + 2 \Gamma_{jk}^i \frac{\pt V^j}{\pt x^l} \right] dt \\
      &\quad + \frac{1}{2} \eta_i R^i_{kjl} (QX)^{kl} V^j dt + \eta_i \left( \frac{\pt V^i}{\pt x^k} + V^j \Gamma_{jk}^i \right) \sigma_r^k dB^r \\
      &= \eta \left( \frac{\pt V}{\pt t} + \nabla^{}_{D_\nabla X} V + \frac{1}{2} (QX)^{ij} \left( \nabla^2_{\pt_i,\pt_j} V + R(V,\pt_i)\pt_j \right) \right) dt + \eta \left( \nabla_{\sigma_r} V \right) dB^r.
    \end{split}
  \end{equation*}
  Hence, the result (i) implies
  \begin{equation*}
    \E \left[ \eta \left( \frac{\overline\D V}{dt} \right) \right] = \E \left( D[\eta(V)(t)] \right) = \E \left[ \eta \left( \frac{\pt V}{\pt t} + \nabla^{}_{D_\nabla X} V + \frac{1}{2} (QX)^{ij} \left( \nabla^2_{\pt_i,\pt_j} V + R(V,\pt_i)\pt_j \right) \right) \right].
  \end{equation*}
  The arbitrariness of $\eta$ yields (ii).

  (iii). Similar to (ii).

  (iv). We only prove the first equality as the second is similar. By \eqref{local-rep-total-mean-2},
  \begin{equation*}
    \begin{split}
      \D_t [\eta(V)] &= \left( \frac{\pt}{\pt t} + (D_\nabla X)^i \pt_i + \frac{1}{2} (QX)^{ij} \nabla^2_{\pt_i,\pt_j} \right) [\eta(V)] \\
      &= \left( \frac{\pt \eta}{\pt t} \right)(V) + \eta \left( \frac{\pt V}{\pt t} \right) + \left( \nabla^{}_{D_\nabla X} \eta \right)(V) + \eta \left( \nabla^{}_{D_\nabla X} V \right) \\
      &\quad + \frac{1}{2} (QX)^{ij} \left[ \left( \nabla^2_{\pt_i,\pt_j} \eta \right) (V) + \eta \left( \nabla^2_{\pt_i,\pt_j} V \right) + \left( \nabla_{\pt_i} \eta \right) \left( \nabla_{\pt_j} V \right) + \left( \nabla_{\pt_j} \eta \right) \left( \nabla_{\pt_i} V \right) \right] \\
      &= \eta \left( \frac{\D V}{dt} \right) + \frac{\D \eta}{dt} (V) + (QX)^{ij} (\nabla_{\pt_i} \eta) (\nabla_{\pt_j} V).
    \end{split}
  \end{equation*}
  The result follows.
\end{proof}

If $QX(t) = \check g(X(t))$, then
\begin{equation*}
  \frac{\overline\D V}{dt} = \frac{\pt V}{\pt t} + \nabla^{}_{D_\nabla X} V + \frac{1}{2} \Delta V + \frac{1}{2} \Ric(V),
\end{equation*}
and similarly,
\begin{equation}\label{damped-mean-cov}
  \frac{\overline\D \eta}{dt} = \frac{\pt \eta}{\pt t} + \nabla^{}_{D_\nabla X} \eta + \frac{1}{2} \Delta \eta - \frac{1}{2} \Ric(\eta) = \frac{\pt \eta}{\pt t} + \nabla^{}_{D_\nabla X} \eta + \frac{1}{2} \Delta_{\mathrm{LD}} \eta,
\end{equation}
where $\Delta$ is the connection Laplacian, and $\Delta_{\mathrm{LD}} = -( dd^*+d^* d)$ is the Laplace-de Rham operator on forms. The relation $\Delta_{\mathrm{LD}} = \Delta - \Ric$ is due to the Weitzenb\"ock identity \cite[Theorem 9.4.1]{Pet16}. We remark here that the operator $\Delta + \Ric$ acting on vector fields is also called Laplace-de Rham operator in \cite{DG79}.

In the context of fluid dynamics, the operator $\frac{\pt}{\pt t}+ \nabla_v$, with $v$ a vector field, is often referred to as material derivative or hydrodynamic derivative. So the mean covariant derivative $\frac{\D}{dt}$ and its damped variant $\frac{\overline\D}{dt}$ can be regarded as stochastic deformations of material derivative.

\subsection{A stochastic stationary-action principle}\label{sec-7-2}

In this section, we will establish a type of stochastic stationary-action principle: the stochastic Hamilton's principle. Another version for systems with conserved energy, the stochastic Maupertuis's principle, can be found in Appendix \ref{app-3}.

In contrast to second-order Hamiltonians, not all real-valued functions on $\mathcal T^S M$ can be used as second-order Lagrangians in stochastic Lagrangian mechanics. This has been hinted in Section \ref{exmp-1}, as we have mentioned in Remark \ref{remark-2}. For this reason, we will produce a class of second-order Lagrangians from classical Lagrangians, via the fiber-linear bundle projection $\varrho_\nabla$ in \eqref{varrho-2} and the $\nabla$-canonical coordinates $(D_\nabla^i x)$ in \eqref{partial-coord}.

\begin{definition}\label{adm-2-Lag}
  By an admissible second-order Lagrangian, we mean a function $L:\R\times\mathcal T^S M\to\R$ such that there exists a classical Lagrangian $L_0: \R\times TM\to\R$ satisfying $L = L_0 \circ (\id_\R\times\varrho_\nabla)$. We call $L$ the $\nabla$-lift of $L_0$.
\end{definition}

In local coordinates, the $\nabla$-lift $L$ of $L_0$ is expressed as
\begin{equation}\label{lift-Lag}
  L(t, x, Dx, Qx) = L_0 \circ \varrho_\nabla(t, x, Dx, Qx) = L_0(t, x, D_\nabla x).
\end{equation}
Let $T>0$. Our stochastic variational problem consists in finding the extrema (maxima or minima) of the stochastic action functional
\begin{equation}\label{action}
  \mathcal S [X;0,T] := \E \int_0^T L\left(t, X(t), D X(t), Q X(t) \right) dt = \E \int_0^T L_0\left(t, X(t), D_\nabla X(t) \right) dt
\end{equation}
over a suitable domain of diffusions $X$ on $M$, where $L$ is an admissible second-order Lagrangian lifted from $L_0$.

In order to formulate a well-posed stochastic variational problem in an economical way, we assume that the manifold $M$ is compact and the metric $g$ is geodesically complete (which will be used to characterize the variations of diffusions in Lemma \ref{variation-aspt}), and that the connection $\nabla$ is the associated Levi-Civita connection. The geodesic completeness can be ensured, for example, if $M$ is connected (see, e.g., \cite[p. 346]{Lee13}). Whenever the metric $g$ is given, the associated Levi-Civita connection is uniquely determined, due to the fundamental theorem of Riemannian geometry \cite[Theorem IV.2.2]{KN63}. We will refer to such a geodesically complete Riemannian metric as a \emph{reference metric tensor}. 

For a fixed point $q \in M$ and a probability distribution $\mu\in \Pred(M)$ on $M$, we define an admissible class of diffusions by
\begin{equation}\label{diff-space}
  \A_g([0,T];q, \mu) = \left\{ X\in I_{(0,q)}^{(T,\mu)}(M): QX(t) = \check g(X(t)), \forall t\in [0,T], \text{a.s.} \right\},
\end{equation}
where $I_{(0,q)}^{(T,\mu)}(M)$ denotes the set all $M$-valued diffusion processes starting from $q$ at $t=0$ and with final distribution $\mu$, i.e., $\P\circ (X(T))^{-1} = \mu$.
The action functional $\mathcal S$ is now defined on the set $\A_g([0,T];q, \mu)$, that is, $\mathcal S: \A_g([0,T];q, \mu) \to \R$.

Note that the admissible class $\A_g$ is similar to the Wiener space, so that a candidate for its ``tangent space'' is Cameron--Martin space. Denote by $\mathcal H([0,T];q)$ the Hilbert space of absolutely continuous curves $v:[0,T]\to T_q M$ such that $\int_0^T |\dot v(t)|^2 dt < \infty$. Let $\mathcal H_0([0,T];q)$ be the subspace consisting of all $v\in\mathcal H([0,T];q)$ satisfying $v(0) = v(T) = 0$.

\begin{definition}
  Let $X\in \A_g([0,T];q, \mu)$. For a curve $v\in \mathcal H_0([0,T];q)$, the vector field along $X$ given by $V(t):=\Gamma(X)_0^t v(t)$ is called a tangent vector to $\A_g([0,T];q, \mu)$ at $X$. The tangent space to $\A_g ([0,T];q, \mu)$ at $X$ is the set of all such tangent vectors, that is,
  \begin{equation*}
    T_X \A_g([0,T];q, \mu) := \left\{ \Gamma(X)_0^\cdot v(\cdot) : v\in \mathcal H_0([0,T];q) \right\}.
  \end{equation*}
\end{definition}

\begin{definition}\label{variation}
  By a variation (or deformation) of a diffusion $X\in \A_g([0,T];q, \mu)$ along $v\in \mathcal H_0([0,T];q)$, we mean a one-parameter family of diffusions $\{X^v_\e\}_{\e\in(-\varepsilon,\varepsilon)}$, where for each $t\in [0,T]$, $X^v_\e(t)$ satisfies the following ODE
  \begin{equation}\label{variation-diff}
    \frac{\pt}{\pt\e}X^v_\e(t) = \Gamma(X^v_\e)_0^t v(t), \quad X^v_0(t) = X(t).
  \end{equation}
  The diffusion $X\in \A_g([0,T];q, \mu)$ is called a stationary (or critical) point of $\mathcal S$, if the first variation $\delta \mathcal S$ vanishes at $X$, i.e.,
  \begin{equation}\label{stationary}
    \frac{d}{d\e}\bigg|_{\e=0} \mathcal S[X^v_\e;0,T] = 0, \quad \text{for all } v\in \mathcal H_0([0,T];q).
  \end{equation}
\end{definition}

\begin{remark}
  (i). The variations of diffusions on manifolds, via differential equation \eqref{variation-diff}, is standard in stochastic analysis on path spaces of Riemannian manifolds. See for example \cite[Eq.~(2.3)]{Dri92} and \cite[Theorem 4.1]{Hsu95}, where it is shown that Wiener measure is quasi-invariant under such variations. This kind of variations has some equivalent constructions. For instance, the previous two references also provided an approach by lifting to the frame bundle and projecting to the Euclidean space (a stochastic analog of Cartan's development), while Malliavin and Fang \cite{FM93} provided an alternative perspective via Bismut connection.

  (ii). The stochastic variational problem \eqref{action}--\eqref{stationary} in the Euclidean context has also been familiar in stochastic optimal transport/control. See Section \ref{sec-7-3} and Subsection \ref{sec-7-4-4} for connections to those areas.

  (iii). Unlike the infinitesimal variation used in Definition \ref{prog-vf} for studying symmetries of SDEs, the infinitesimal variation here in \eqref{variation-diff} needs to be a parallel vector field.
\end{remark}

The following lemma is the key for establishing stochastic Hamilton's principle. The first statement shows that the variation $X^v_\e$ is well defined on the path space $\A_g([0,T];q, \mu)$. The second one describes the infinitesimal changes of $D_\nabla X^v_\e$ with respect to the variation parameter $\e$. The proof of the latter is based on a geodesic approximation technique, which is originally due to It\^o \cite{Ito62}.

\begin{lemma}\label{variation-aspt}
  Given $X\in \A_g([0,T];q, \mu)$ and $v\in \mathcal H_0([0,T];q)$. We have \\
  (i) for each $\e\in(-\varepsilon,\varepsilon)$, $X^v_\e\in \A_g([0,T];q, \mu)$; and \\
  (ii) for all $t\in[0,T]$,
  \begin{equation}\label{variation-deriv}
    \frac{D}{d\e}\bigg|_{\e=0} D_\nabla X^v_\e(t) = \Gamma(X)_0^t \dot v(t) + \frac{1}{2}(QX)^{ij}(t) R\left(\Gamma(X)_0^t v(t),\pt_i\right) \pt_j,
  \end{equation}
  where $\dot v(t) = \frac{d}{dt}v(t) \in T_{v(t)} T_{q} M \cong T_{q} M$, $\frac{D}{d\e}$ is the (classical) covariant derivative with respect to the parameter $\e$.
\end{lemma}

\begin{proof}
  (i). Let $\xi$ and $\xi_\e$ be the anti-development (\cite[Definition 2.3.1]{Hsu02}) of $X$ and $X^v_\e$, respectively, with fixed initial frame $r(0) \in O_{q}M$. Equivalently, for example, $\xi$ is an $\R^d$-valued diffusion related to $X$ by the following SDEs \cite[Section 2.3]{Hsu02}
  \begin{equation*}\left\{
    \begin{aligned}
      dX^i(t) &= r_j^i(t) \circ d\xi^j(t), \\
      dr_j^i(t) &= -\Gamma_{kl}^i(X(t)) r_j^l(t)r_m^k(t) \circ d\xi^m(t).
    \end{aligned}\right.
  \end{equation*}
  Applying the fact that $\sum_{k=1}^d r_k^i r_k^j= g^{ij}$ (e.g., \cite[Proposition 1.5]{KN63}) and the condition $QX(t) = \check g(X(t))$, we have
  \begin{equation}\label{eqn-13}
    r_k^i(t) r_l^j(t) \delta^{kl} = g^{ij}(X(t)) = (QX)^{ij}(t) = r_k^i(t) r_l^j(t) (Q\xi)^{kl}(t),
  \end{equation}
  and consequently, $Q\xi\equiv \I_d$. Meanwhile, it follows from \cite[Section 3.5]{FM93} (or \cite[Theorem 5.1]{Dri92}, \cite[Section 3]{Hsu95}) that
  \begin{equation*}
    d\xi_\e(t) = \exp\left(\e \int_0^t \Omega\left( \left(r(0)^{-1} v\right)(s), \circ d\xi(s) \right)\right) \circ d\xi(t) + \e d\left(r(0)^{-1} v\right)(t),
  \end{equation*}
  where $\Omega$ is the curvature form on the orthogonal frame bundle $OM$, taking values in $\mathfrak{so}(d)$, and the frame $r(0)$ is viewed as an isomorphism from $\R^d$ to $T_{q} M$. It follows that $Q\xi_\e = Q\xi \equiv \I_d$. For the reason similar to \eqref{eqn-13}, we have $QX^v_\e(t) = \check g(X^v_\e(t))$. The result follows. See \cite[Theorem 8.3]{Dri92} for a more elaborate proof.

  (ii). Fix $n,m \in \N_+$. Let $0=t_0<t_1<\cdots<t_n=T$ be a division of the time interval $[0,T]$, and let $-\varepsilon=\e_{m-} <\cdots <\e_{-1} <0= \e_0<\e_1<\cdots<\e_m=\varepsilon$ be one of the variation parameter interval $(-\varepsilon,\varepsilon)$. Denote $\Delta t_i := t_i - t_{i-1}$. Consider the polygonal curve $x^n = \{x^n(t)\}_{t\in[0,T]}$, which is an approximation of $X$ made of minimizing geodesic segments joining $X(t_{i-1})$ with $X(t_i)$ for all $1\le i\le n$. This is attainable by the geodesic completeness. We will construct an approximation scheme for the variational processes $X^v_\e$'s.

  For $\e\in [\e_0,\e_1]$, we construct the approximation $x^n_\e$ of $X^v_\e$ as follows. We extend each $X(t_i)$, $0\le i\le n$, to a geodesic
  $$\gamma_0^{(i)}(\e) = \exp_{X(t_i)} \left( \e \Gamma(x^n)_0^{t_i} v(t_i) \right), \quad \e\in [\e_0,\e_1].$$
  Let $x^n_\e = \{x^n_\e(t)\}_{t\in[0,T]}$ be the polygonal curve consisting of minimizing geodesic segments joining $\gamma_0^{(i-1)}(\e)$ with $\gamma_0^{(i)}(\e)$ for all $1\le i\le n$.

  Then, we construct $x^n_\e$ for $\e\in [\e_j,\e_{j+1}]$, $1\le j\le m-1$, by induction. Suppose $x^n_\e$, $\e\in [\e_{j-1},\e_j]$, has been defined. Then, in particular, we have a curve $x^n_{\e_j}$. Extend each $x^n_{\e_j}(t_i)$, $0\le i\le n$, to a geodesic by
  $$\gamma_j^{(i)}(\e) = \exp_{x^n_{\e_j}(t_i)} \left( (\e -\e_j) \Gamma(x^n_{\e_j})_0^{t_i} v(t_i) \right), \quad \e\in [\e_j,\e_{j+1}].$$
  Let $x^n_\e$ be the polygonal curve consisting of minimizing geodesic segments joining $\gamma_j^{(i-1)}(\e)$ with $\gamma_j^{(i)}(\e)$ for all $1\le i\le n$. In a similar way, we can define $x^n_\e$ for $\e\in [\e_j,\e_{j+1}]$, $-m\le j\le -1$.

  Now we have a family of polygonal curves $\{x^n_\e: \e \in (-\varepsilon,\varepsilon)\}$, which satisfies $x^n_0 = x^n$ and
  \begin{equation*}
    \frac{\pt^{\sign(\e)}}{\pt\e}\bigg|_{\e=\e_j} x^n_\e(t_i) = \Gamma(x^n_{\e_j})_0^{t_i} v(t_i).
  \end{equation*}
  As for each $\e\in (-\varepsilon,\varepsilon)$ and $1\le i\le n$, $\{x^n_\e(t)\}_{t\in[t_{i-1},t_i]}$ is a geodesic, the vector field $$J(t):=\frac{\pt}{\pt\e}\bigg|_{\e=0} x^n_\e(t), \quad t\in[t_{i-1},t_i]$$
  is a Jacobi field along $\{x^n(t)\}_{t\in[t_{i-1},t_i]}$. This leads to the following Jacobi equation
  \begin{equation}\label{eqn-14}
    \frac{D^2}{dt^2} J(t) + R\left( J(t), \dot x^n(t) \right) \dot x^n(t) = 0, \quad t\in[t_{i-1},t_i],
  \end{equation}
  with boundary values
  \begin{equation}\label{eqn-15}
    J(t_{i-1}) = \Gamma(x^n)_0^{t_{i-1}} v(t_{i-1}), \quad  J(t_i) = \Gamma(x^n)_0^{t_i} v(t_i).
  \end{equation}
  Since the connection is torsion-free, we can exchange the covariant derivative and standard derivative to have
  \begin{equation}\label{eqn-16}
    \frac{D}{dt} J(t_{i-1}) = \frac{D}{dt} \frac{\pt}{\pt\e} x^n_\e(t) \bigg|_{\e=0,t=t_{i-1}} = \frac{D}{d\e} \frac{\pt}{\pt t} x^n_\e(t) \bigg|_{\e=0,t=t_{i-1}} = \frac{D}{d\e}\bigg|_{\e=0} \dot x^n_\e(t_{i-1}),
  \end{equation}
  On the other hand, Taylor's theorem yields
  \begin{equation}\label{eqn-17}
    \Gamma(x^n)_{t_i}^{t_{i-1}} J(t_i) = J(t_{i-1}) + \frac{D}{dt} J(t_{i-1}) \Delta t_i + \frac{1}{2}\frac{D^2}{dt^2} J(t_{i-1}) (\Delta t_i)^2 + o\left( (\Delta t_i)^2 \right).
  \end{equation}
  Combining \eqref{eqn-14}--\eqref{eqn-17}, we have
  \begin{equation*}
    \frac{D}{d\e}\bigg|_{\e=0} \dot x^n_\e(t_{i-1}) = \Gamma(x^n)_0^{t_{i-1}} \frac{v(t_i) - v(t_{i-1})}{\Delta t_i} + \frac{1}{2} R\left( \Gamma(x^n)_0^{t_{i-1}} v(t_{i-1}), \dot x^n(t_{i-1}) \right) \dot x^n(t_{i-1}) \Delta t_i + o\left( \Delta t_i \right).
  \end{equation*}
  A standard limit theorem yields the result (ii).
\end{proof}

\begin{remark}
  (i). The constraint $QX(t) = \check g(X(t))$ in \eqref{diff-space} looks strong. A possibly better viewpoint is to force all diffusions under consideration to have the same nondegenerate diffusion tensor $a$, i.e., $QX(t) = a(X(t))$. Then, the inverse of $a$ defines a Riemannian metric $g$, cf. \cite[Section V.4]{IW89}. As can be seen from the first part of the above proof, the constraint of fixing the diffusion tensor is a natural one in the literature of variational calculus on the path space. An intuitive reason for this constraint is to assure that the induced measures are equivalent, which is necessary for equation \eqref{variation-diff} to be well-posed, cf. \cite{Dri92}. The assumption of Levi-Civita connection $\nabla$ may be relaxed to that the connection $\nabla$ is $g$-compatible and torsion skew symmetric \cite[Definition 8.1]{Dri92}, in which case the second assertion of this lemma needs to add the effect of torsion.

  (ii). One may expect from the limits of \eqref{eqn-14} and \eqref{eqn-15} that there is a ``stochastic'' Jacobi equation with two boundary values describing the difference between a diffusion and an ``infinitesimally close'' diffusion, cf. \cite{AT98}.
\end{remark}

For a smooth function $f$ on $TM$, we denote by $d_{\dot x} f$ the differential of $f$ with respect to the coordinates $(\dot x^i)$. Since $T_{(x,\dot x)} T_x M \cong T_x M$, $d_{\dot x} f$ is treated as a 1-form on $T_x M$ and
\begin{equation}\label{diff-dot-x}
  d_{\dot x} f = \frac{\pt f}{\pt \dot x^i} d x^i.
\end{equation}
We call $d_{\dot x} f$ the \emph{vertical differential} of $f$. Regarding the differential with respect to the coordinates $(x^i)$, we introduce the \emph{horizontal differential} which depends on the connection $\nabla$, by
\begin{equation}\label{diff-x}
  d_x f = \left( \frac{\pt f}{\pt x^i} - \Gamma_{ij}^k \dot x^j \frac{\pt f}{\pt \dot x^k} \right) d x^i.
\end{equation}
It is easy to check that both definitions \eqref{diff-dot-x} and \eqref{diff-x} are invariant under change of coordinates. In fact, by the classical theory \cite[Section 3.5 and Example 4.6.7]{Sau89}, we know that the connection $\nabla$ can uniquely determine a $TTM$-valued 1-form on $TM$ horizontal over $M$, which is given in local coordinates by
\begin{equation*}
  \Gamma = dx^i \otimes \left( \frac{\pt}{\pt x^i} - \Gamma_{ij}^k \dot x^j \frac{\pt}{\pt \dot x^k} \right).
\end{equation*}
Hence, the horizontal differential is $d_x f = \Gamma(df)$, where $df$ is the total differential of $f$. Given a vector field $V$ on $M$, $f\circ V: q\mapsto f(V_q)$ is a smooth function on $V$. Then, it is easy to check that
\begin{equation}\label{eqn-22}
  d (f\circ V) = d_x f \circ V + (d_{\dot x} f \circ V)(\nabla_{\pt_i} V) dx^i.
\end{equation}

The following integration-by-parts formula will be used. Its proof is straightforward from definitions of stochastic integrals and mean derivatives, cf. \cite[Lemma 4.4]{CZ91}.
\begin{lemma}\label{integration-by-parts}
  Let $X = \{X(t)\}_{t\in[0,T]}$ be a real-valued continuous semimartingale such that $DX$ exists, let $f$ be a real-valued continuous process on $[0,T]$, of finite variation. Then
  \begin{equation*}
    \E \int_0^T X(t) \dot f(t) dt = E\left[ f(T)X(T)-f(0)X(0) \right] - \E \int_0^T f(t) DX(t) dt.
  \end{equation*}
\end{lemma}


Now we are in position to present the stochastic version of Hamilton's principle.

\begin{theorem}[Stochastic Hamilton's principle]\label{stoch-Hamilton-prin}
  Let $L_0$ be a regular Lagrangian on $\R\times TM$. A diffusion $X\in \A_g([0,T];q, \mu)$ is a stationary point of $\mathcal S$, if and only if $X$ satisfies the following stochastic Euler-Lagrange (S-EL) equation
  \begin{equation}\label{stoch-EL}
    \frac{\overline\D}{dt} \big( d_{\dot x} L_0\left(t, X(t), D_\nabla X(t) \right) \big) = d_x L_0\left(t, X(t), D_\nabla X(t) \right),
  \end{equation}
  where $\frac{\overline\D}{dt}$ is the damped mean covariant derivative with respect to $X$.
\end{theorem}

We remark that since $QX(t) = \check g(X(t))$, the operator $\frac{\overline\D}{dt}$ in \eqref{stoch-EL} is just the one of \eqref{damped-mean-cov}. The unknown in \eqref{stoch-EL} is the process $X$, so the conditions $X(0) = q$ and $\P\circ (X(T))^{-1} = \mu$, indicated in the assumption $X\in \A_g([0,T];q, \mu)$, can be regarded as boundary conditions of \eqref{stoch-EL}.

\begin{proof}
  Denote $V(t) = \Gamma(X)_0^t v(t)$. It follows from \eqref{variation-deriv} and \eqref{eqn-22} that
  \begin{equation}\label{eqn-18}
    \begin{split}
      \frac{d}{d\e}\bigg|_{\e=0} \mathcal S[X^v_\e;0,T] &= \E \int_0^T \frac{d}{d\e}\bigg|_{\e=0} L_0\left(t, X^v_\e(t), D_\nabla X^v_\e(t) \right) dt \\
      &= \E \int_0^T \left[ d_x L_0 \left( \frac{\pt}{\pt\e}\bigg|_{\e=0}X^v_\e(t) \right) + d_{\dot x} L_0 \left( \frac{D}{d\e}\bigg|_{\e=0} D_\nabla X^v_\e(t) \right) \right] dt \\
      &= \E \int_0^T \left[ d_x L_0 \left( V(t)\right) + d_{\dot x} L_0 \left( \Gamma(X)_0^t \dot v(t) \right) + \frac{1}{2} (QX)^{ij}(t) d_{\dot x} L_0 \left( R(V(t),\pt_i) \pt_j \right) \right] dt.
    \end{split}
  \end{equation}
  By Lemmas \ref{parallel-vf-form}.(ii) and \ref{integration-by-parts} and the fact that $v(0) = v(T) = 0$, we have
  \begin{equation}\label{eqn-19}
    \begin{split}
      \E \int_0^T d_{\dot x} L_0 \left( \Gamma(X)_0^t \dot v(t) \right) dt &= \E \int_0^T \Gamma(X)_t^0 (d_{\dot x} L_0) \left( \dot v(t) \right) dt \\
      &= -\E \int_0^T \lim_{\e\to0} \E \left[ \left(\frac{\Gamma(X)_{t+\e}^0 (d_{\dot x} L_0) - \Gamma(X)_t^0 (d_{\dot x} L_0)}{\e}\right) \left( v(t) \right) \Bigg| \Pred_t \right] dt \\
      &= -\E \int_0^T \lim_{\e\to0} \E \left[ \left( \frac{\Gamma(X)_{t+\e}^t (d_{\dot x} L_0) - d_{\dot x} L_0}{\e} \right) \left( \Gamma(X)_0^t v(t) \right) \Bigg| \Pred_t \right] dt \\
      &= -\E \int_0^T \lim_{\e\to0} \E \left[ \frac{\Gamma(X)_{t+\e}^t (d_{\dot x} L_0) - d_{\dot x} L_0}{\e}\Bigg| \Pred_t \right] \left( \Gamma(X)_0^t v(t) \right) dt \\
      &= -\E \int_0^T \frac{\D}{dt} (d_{\dot x} L_0) \left( V(t) \right) dt.
    \end{split}
  \end{equation}
  Thus, by Lemma \ref{stoch-mean-derv-prop}.(iii),
  \begin{equation*}
    \begin{split}
      \frac{d}{d\e}\bigg|_{\e=0} \mathcal S[X^v_\e;0,T]
      &= \E \int_0^T \left[ d_x L_0 \left( V(t)\right) - \frac{\D}{dt} (d_{\dot x} L_0) \left( V(t) \right) + \frac{1}{2} (QX)^{ij}(t) 
      R(d_{\dot x} L_0,\pt_j) \pt_i \left( V(t) \right) \right] dt \\
      &= \E \int_0^T \left( d_x L_0 - \frac{\overline\D}{dt} (d_{\dot x} L_0)\right) \left( V(t) \right) dt.
    \end{split}
  \end{equation*}
  The arbitrariness of $v$ yields the desired result.
\end{proof}

\begin{remark}
  (i). For a special class of Lagrangians in the Euclidean context, the stochastic Euler-Lagrange equation \eqref{stoch-EL} has been established in \cite[Subsection 5.1]{CZ91} where they called it stochastic Newton equation, see also \cite{Zam15}. For general Lagrangians on Riemmannian manifolds, equation \eqref{stoch-EL} is new (to the authors' best knowledge). See Section \ref{sec-7-3} for discussions of a special case.

  (ii). The second author and his collaborator formulated a weak stochastic Euler-Lagrange equation in \cite{LZ16}. They mean by ``weak'' that their stochastic Euler-Lagrange equation holds in the sense of stochastic integrals. The main differences between their formulation and ours is that we get rid of the stochastic integral (martingale) part in our equation since we use mean derivatives instead of stochastic differentials.
\end{remark}

\subsection{An inspirational example: Schr\"odinger's problem}\label{sec-7-3}

The inspirational example of stochastic Hamiltonian mechanics presented in Section \ref{exmp-1} also provides an example of our stochastic Lagrangian mechanics. Consider the following Lagrangian defined on $\R\times TM$:
\begin{equation}\label{Lagrangian}
  L_0(t,x,\dot x) = \frac{1}{2} |\dot x-b(t,x)|^2 - F(t,x),
\end{equation}
where $b$ is a given time-dependent vector field on $M$. It actually relates to the 2nd-order Hamiltonian $H$ in \eqref{exmp-Hamiltonian} via the 2nd-order Legendre transform, which will be considered in Section \ref{sec-7-4}. For such Lagrangian, we can directly figure out the relation between stochastic Euler-Lagrange
equation \eqref{stoch-EL} and Hamilton-Jacobi-Bellman equation. We denote by $I_0^T(M)$ the set all $M$-valued diffusion processes over time interval $[0,T]$.

\begin{theorem}[S-EL \& HJB]\label{SEL-HJB}
  Let $L_0$ be as in \eqref{Lagrangian}.
  If $X\in I_0^T(M)$ satisfies
  \begin{equation}\label{SEL-HJB-cond}
    D_\nabla X(t) = \nabla S(t,X(t)) + b(t,X(t))
  \end{equation}
  for a function $S:\R\times M\to\R$, then $X$ is a solution of the stochastic Euler-Lagrange equation \eqref{stoch-EL} if and only if $S$ solves the following Hamilton-Jacobi-Bellman equation
  \begin{equation}\label{HJB-2}
    \frac{\pt S}{\pt t} + \langle b, \nabla S\rangle + \frac{1}{2} |\nabla S|^2 + \frac{1}{2} \Delta S + F = f,
  \end{equation}
  with $f$ a function depending only on $t$.
\end{theorem}

\begin{proof}
  For a function $g$ on $\R\times M$, we will denote by $dg$ the exterior differential of $g$ on $M$, i.e., with respect to coordinates $(x^i)$. Condition \eqref{SEL-HJB-cond} can be rewritten in local coordinates as
  \begin{equation}\label{x-S}
    \dot x = \nabla S+ b.
  \end{equation}
  Then, it is clear that
  \begin{equation}\label{eqn-10}
    d_{\dot x} L_0 = \frac{\pt L_0}{\pt \dot x^i} d x^i = g_{ij}(\dot x^j - b^j) d x^i = d S.
  \end{equation}
  Since $\nabla g = 0$, we use Leibniz's rule to derive
  \begin{equation}\label{eqn-11}
    \begin{split}
      d_x L_0(\pt_k) &= \frac{1}{2} d[g(\dot x-b,\dot x-b)](\pt_k) - d F(\pt_k) = -g\left( \nabla_{\pt_k} b, \dot x-b \right) - d F(\pt_k) = -dS\left( \nabla_{\pt_k} b \right) - d F(\pt_k).
    \end{split}
  \end{equation}

  Now we take the differential with respect to $x$ to the HJB equation \eqref{HJB-2}. Obviously,
  $$d\frac{\pt S}{\pt t}= \frac{\pt}{\pt t}dS = \frac{\pt}{\pt t} d_{\dot x} L_0.$$
  For the second term,
  \begin{equation*}
    \begin{split}
      d( \langle b, \nabla S\rangle )(\pt_k) &= d[dS(b)](\pt_k) = \left(\nabla_{\pt_k}dS\right)(b) + dS\left(\nabla_{\pt_k}b\right) \\
      &= \nabla^2_{\pt_k,b} S + dS\left(\nabla_{\pt_k}b\right) = \left(\nabla_b dS\right)(\pt_k) + dS\left(\nabla_{\pt_k}b\right).
    \end{split}
  \end{equation*}
  For the third term, we use again $\nabla g = 0$. Then, we have
  \begin{equation*}
    \begin{split}
      \frac{1}{2} d\left( |\nabla S|^2 \right)(\pt_k) &= \frac{1}{2}d[dS\otimes dS(\check g)](\pt_k) =\left( (\nabla_{\pt_k} dS) \otimes dS\right)(\check g) = \left(\nabla_{\pt_k} dS\right)(\nabla S) = \left( \nabla_{\nabla S} dS\right)(\pt_k).
    \end{split}
  \end{equation*}
  For the fourth term, in the same way we have
  \begin{equation*}
    \begin{split}
      &\quad\ d(\Delta S)(\pt_k) = d\left(g^{ij} \nabla^2_{\pt_i,\pt_j} S\right)(\pt_k) = d\left( \nabla^2 S(\check g) \right)(\pt_k) = \left(\nabla_{\pt_k} \nabla^2 S\right)(\check g) = g^{ij} \nabla^3_{\pt_k,\pt_i, \pt_j} S \\
      &= g^{ij} \left[ \left( \nabla^3_{\pt_k,\pt_i, \pt_j} S - \nabla^3_{\pt_i,\pt_k, \pt_j} S \right) + \left( \nabla^3_{\pt_i,\pt_k, \pt_j} S - \nabla^3_{\pt_i,\pt_j, \pt_k} S \right) + \nabla^3_{\pt_i,\pt_j, \pt_k} S \right] \\
      &= g^{ij} \left[ \left( \nabla^2_{\pt_k,\pt_i} dS - \nabla^2_{\pt_i,\pt_k} dS \right)(\pt_j) + 0 + \nabla^2_{\pt_i,\pt_j} dS (\pt_k) \right] = g^{ij} \left[ R(\pt_k,\pt_i) d S (\pt_j) + \nabla^2_{\pt_i,\pt_j} dS (\pt_k) \right] \\
      &= g^{ij} \left[ - R(d S, \pt_j)\pt_i(\pt_k) + \nabla^2_{\pt_i,\pt_j} dS (\pt_k) \right] 
      = [\Delta dS - \Ric(dS)](\pt_k) = \Delta_{\mathrm{LD}} (dS) (\pt_k).
    \end{split}
  \end{equation*}
  Combining these together and applying \eqref{SEL-HJB-cond}--\eqref{eqn-11} as well as \eqref{damped-mean-cov}, we obtain
  \begin{equation*}
    \begin{split}
      &\ d\left( \frac{\pt S}{\pt t} + \langle b, \nabla S\rangle + \frac{1}{2} |\nabla S|^2 + \frac{1}{2} \Delta S + F \right) (\pt_k) = \left(\frac{\pt}{\pt t} + \nabla_{b+\nabla S} + \frac{1}{2}\Delta_{\mathrm{LD}} \right) (dS)(\pt_k) + dS\left(\nabla_{\pt_k}b\right) + dF(\pt_k) \\
      =&\ \frac{\overline\D}{dt} (dS)(\pt_k) + dS\left(\nabla_{\pt_k}b\right) + dF(\pt_k) = \left[\frac{\overline\D}{dt} (d_{\dot x} L_0) - d_x L_0\right] (\pt_k).
    \end{split}
  \end{equation*}
  The result follows.
\end{proof}

\begin{remark}
  Equation \eqref{eqn-10} gives the relation between Lagrangians and 2nd-order Hamilton's principal functions. It is valid for more general Lagrangians, see Remark \ref{remark-4}.(i).
\end{remark}


Theorem \ref{SEL-HJB} strongly suggests some relations between stochastic Lagrangian (and also Hamiltonian) mechanics and Schr\"odinger's problem in the reintepretation of optimal transport. In the setting of the latter (see, e.g., \cite{CWZ00,Leo14,Leo14a}), there is a given reversible positive measure $\mathbf R$ on the path space $\C_0^T = C([0,T], M)$, called \emph{reference measure}, as well as two probability distributions $\mu_0,\mu_T\in \Pred(M)$. Schr\"odinger's problem aims to minimize the following relative entropy
  \begin{equation}\label{entropy}
    H(\P|\mathbf R) =
    \begin{cases}
      \int_{\C_0^T} \log\left( \frac{d\P}{d\mathbf R} \right) d\P, & \P \ll \mathbf R, \\
      +\infty, & \text{otherwise},
    \end{cases}
  \end{equation}
over all probability measures $\P$ on $\C_0^T$ such that $\mu_0,\mu_T$ are the initial and final time marginal distributions of $\P$, i.e., $\P_0 = \mu_0$ and $\P_T = \mu_T$, where $\P_t := \P\circ (X(t))^{-1}$ is the time marginal distribution of $\P$ and $X(t) : \C_0^T\to M, X(t,\omega) = \omega(t)$ is the coordinate mapping. Denote, respectively, by $X_\mathbf R$ and $X_\mathbf P$, the coordinate process $X$ under the measure $\mathbf R$ and $\mathbf P$.  Then, Girsanov theorem implies that \cite[Theorem 1]{Leo12} a necessary condition for the \emph{finite entropy condition} $H(\P|\mathbf R) < \infty$ is $QX_\P = QX_\mathbf R$, $\P$-a.s.. Furthermore, if $\mathbf R$ is a diffusion measure, i.e., $X_\mathbf R$ is a diffusion process, then a similar application of  Girsanov theorem yields that a necessary condition for $H(\P|\mathbf R) < \infty$ is that $\P$ is also a diffusion measure and there exists a time-dependent vector field $w$ such that
$$\left( DX_\P(t), QX_\P(t) \right) = \left( DX_\mathbf R(t) + w(t, X(t)), QX_\mathbf R(t) \right),\quad \forall t\in[0,T], \text{ a.s.}.$$
The solution $\P$ of Schr\"odinger's problem, i.e., minimizing \eqref{entropy}, is related to the reference measure $\mathbf R$ by a time-symmetric version of Doob's $h$-transform \cite[Section 3]{Leo14}. Its coordinate process $X_\P$ is sometimes called a \emph{Schr\"odinger bridge} or Schr\"odinger process. When the reference measure $\mathbf R$ is Markovian, i.e., the law of a Markov process, the solution process $X_\P$ is also called a reciprocal \cite{Ber32,Jam75} or Bernstein process \cite{CWZ00,CV15}.

If the manifold $M$ is endowed with a Riemannian metric $g$, and the reference coordinate process $X_\mathbf R$ has generator
  \begin{equation*}
    A^{X_\mathbf R} = \langle b, \nabla \rangle + \ts{\frac{1}{2}} \Delta + F,
  \end{equation*}
for some time-dependent vector field $b$ on $M$, then the density $\mu(t,x) = \frac{d\P^*_t}{d\Vol}(x)$ of the minimizer $\P^*$ of \eqref{entropy} solves the following Kolmogorov forward equation
  \begin{equation}\label{Kol-forw}\left\{
    \begin{aligned}
      &\frac{\pt}{\pt t} \mu(t,x) + \divg \left[\mu (\nabla S + b) \right] - \frac{1}{2} \Delta \mu(t,x) = 0, \quad (t,x)\in(0,1]\times M, \\
      &\mu(0,x) = \mu_0(x), \quad x\in M.
    \end{aligned}\right.
  \end{equation}
where $S$ solves the HJB equation \eqref{HJB-2} with $f \equiv 0$, or \eqref{exmp-HJB}.

Moreover, an analog of Benamou-Brenier formula was derived (see \cite{Leo14}). Consider the problem of minimizing the average action
\begin{equation}\label{Benamou-Brenier}
  \int_0^T \int_{M} \left( \frac{1}{2}|v(t,x) - b(t,x)|^2 - F(t,x) \right) \rho(t,dx) dt
\end{equation}
among all pairs $(\rho,v)$, where is $\rho=(\rho(t))_{t\in[0,T]}$ is a measurable path in $\Pred(M)$, $v=(v(t))_{t\in[0,T]}$ is a measurable time-dependent vector field and the following constraints are satisfied (in the weak sense of PDEs):
\begin{equation}\label{FP-eqn}\left\{
  \begin{aligned}
    &\frac{\pt}{\pt t} \rho + \divg \left(\rho v \right) - \frac{1}{2} \Delta \rho = 0, \\
    &\rho(0) = \mu_0, \ \rho(T) = \mu_T,
  \end{aligned}\right.
\end{equation}
The relation between $\rho$ in \eqref{Benamou-Brenier} and $\P$ in \eqref{entropy} is just that $\rho$ is the time marginal of $\P$, namely,
\begin{equation}\label{relation-Sch}
  \rho(t) = \P_t = \P\circ (X(t))^{-1}.
\end{equation}
The minimizer of \eqref{Benamou-Brenier} is the pair $(\mu,\nabla S+b)$ where $\mu$ solves \eqref{Kol-forw} and $S$ solves \eqref{exmp-HJB}.

These results are summarized in the following equivalent relations:
\begin{equation}\label{BB-sum}
  \begin{split}
    &\ \inf\left\{ H(\P|\mathbf R): \P\in \Pred(\C_0^T), \P_0 = \mu_0, \P_T = \mu_T \right\} - H\left( \mu_0|\mathbf R_0 \right) \\
    =&\ \inf\left\{ \int_0^T \int_{M} \left( \frac{1}{2}|v(t,x) - b(t,x)|^2 - F(t,x) \right) \rho(t,dx) dt : (\rho,v) \text{ satisfies \eqref{FP-eqn}} \right\} \\
    =&\ \int_0^T \int_{M} \left( \frac{1}{2}|\nabla S(t,x)|^2 - F(t,x) \right) \mu(t,dx) dt.
  \end{split}
\end{equation}

Now if the coordinate process $X_\mathbf R$ under the reference measure $\mathbf R$ is a nondegenerate $M$-valued diffusion in $I_0^T(M)$ which is diffusion-homogeneous, then assigning such a reference measure $\mathbf R$ amounts to assigning a pair $(b_\mathbf R,g_\mathbf R)\in \Gamma(TM \otimes \Sym^2(T^* M))$, where $g_\mathbf R$ is a positive-definite symmetric $(0,2)$-tensor, i.e., a Riemannian metric tensor. More precisely, we let $A^{X_\mathbf R} = (\mathfrak{b}, a) + F$ be the generator of $X_\mathbf R$. Since $X_\mathbf R$ is nondegenerate and diffusion-homogeneous, $a$ is a time-independent nondegenerate symmetric $(2,0)$-tensor field. Let $g_\mathbf R = \hat a$ be the inverse of $a$, so that $g_\mathbf R$ is a Riemannian metric tensor. We then equip the Riemannian manifold $(M, g_\mathbf R)$ with the associated Levi-Civita connection $\nabla$. The isomorphism \eqref{dcpst-tang} implies that
\begin{equation*}
  A^{X_\mathbf R} = b_\mathbf R^i \pt_i + \ts{\frac{1}{2}} g_\mathbf R^{ij} \nabla^2_{\pt_i,\pt_j} + F = \langle b_\mathbf R, \nabla \rangle + \ts{\frac{1}{2}} \Delta + F,
\end{equation*}
where $b_\mathbf R$ is the time-dependent vector field given by $b_\mathbf R^i = ( \mathfrak b^i + \ts{\frac{1}{2}} g_\mathbf R^{jk} \Gamma^i_{jk} )$, and $\nabla$ and $\Delta$ are the gradient and Laplace-Beltrami operator with respect to $g_\mathbf R$, respectively.

We set that $\P$ is a diffusion measure and $QX_\P = QX_\mathbf R = \check g_\mathbf R$, $\P$-a.s., which is a necessary condition for $H(\P|\mathbf R) < \infty$. Then, by \eqref{2-tangent-equiv}, the generator of $X_\P$ is given by
\begin{equation*}
  (DX_\P(t), QX_\P(t) ) = (D_\nabla X_\P)^i(t) \pt_i|_{X(t)} + \ts{\frac{1}{2}} \Delta|_{X(t)}.
\end{equation*}
From \eqref{FP-eqn} and \eqref{relation-Sch}, one can see that $v(t,X(t)) = D_\nabla X_\P(t)$ and the action \eqref{Benamou-Brenier} equals to
\begin{equation}\label{action-2}
  \E_\P \int_0^T \left( \frac{1}{2}|D_\nabla X(t) - b_\mathbf R(t,X(t))|^2 - F(t,X(t)) \right) dt.
\end{equation}
So the minimizing problem turns into minimizing the action \eqref{action-2} over all diffusion measures $\P\in\Pred(\C_0^T)$ with $\P_0 = \mu_0$, $\P_T = \mu_T$ and $QX_\P = \check g_\mathbf R$, $\P$-a.s.. If $\mu_0 = \delta_{q}$ and $\mu_T = \mu$, this brings us back to our stochastic variational problem, that is, to minimize the action functional $\mathcal S$ in \eqref{action} over $\A_{g_\mathbf R}([0,T];q, \mu)$, with Lagrangian $L_0(t, x,\dot x) = \frac{1}{2} |\dot x-b_\mathbf R(t,x)|^2 - F(t,x)$. Note that in this case, since $\P_0 = \mu_0$ is Dirac, the relative entropy in \eqref{entropy} and $H(\mu_0|\mathbf R_0)$ are always infinite, while their difference $H(\P|\mathbf R) - H(\mu_0|\mathbf R_0)$ can be finite as in \eqref{BB-sum}. Moreover, by Theorem \ref{stoch-Hamilton-prin} and \ref{SEL-HJB}, a necessary condition for $X_\P$ to be the minimizer of $\mathcal S$ is that $X_\P$ satisfies \eqref{SEL-HJB-cond} and \eqref{HJB-2}, which coincides with \eqref{Kol-forw}.

\begin{remark}
  (i). Compared to the Lagrangian \eqref{Lagrangian} used here for addressing Schr\"odinger's problem, there is another type of Lagrangians used in the Euclidean version of quantum mechanics in \cite[Eq.~(5.4)]{CZ91}. The latter has an additional term of divergence of $b$, which helps to express part of the action functional as a Stratonovich integral. The stochastic Euler-Lagrange equation \eqref{stoch-EL} applied to their Lagrangians recovers the equations of motion in \cite[Theorem 5.3]{CZ91}.

  (ii). In the seminal paper \cite{Ott01}, F.~Otto provided a geometric perspective for numerous PDEs by introducing a Riemannian structure in the Wasserstein space. It is known as Otto's calculus. A similar idea can ascend to V.I.~Arnold, who established a geometric framework for hydrodynamics by studying the Riemannian nature of the infinite-dimensional group of diffeomorphisms \cite{AK21}. The recent paper \cite{GLR20} formulated Schr\"odinger's problem via Otto calculus, where the equation of motion is given by an infinite-dimensional Newton equation, cf. \cite{KMM21,Von12} on related matters. All these works can be called a ``geometrization'' of (stochastic) dynamics. In contrast, the present framework can be called a ``stochastization'' of geometric mechanics. The difference and relations between our framework and theirs are similar to those between two ways of producing HJ equations for quantum mechanics mentioned in the introduction. More precisely, while (second-order) HJB equations play a key role in our framework, various HJ equations with density-dependent potential terms were derived by them (see \cite[Corollary 23]{GLR20} and \cite[Proposition 2.4]{KMM21}).
\end{remark}

\subsection{Second-order Legendre transform}\label{sec-7-4}

\subsubsection{From $\mathcal T^{S*} M$ to $\mathcal T^S M$ and back}\label{sec-7-4-1}

Let us fix a linear connection $\nabla$ on $M$. Here, for simplicity, we consider time-independent Hamiltonians and Lagrangians.

We first produce second-order Lagrangians from second-order Hamiltonians. To this end, we first reduce the second-order Hamiltonian to a classical one. Given a time-independent second-order Hamiltonian $H: \mathcal T^{S*} M\to \R$, its $\nabla$-reduction is the classical Hamiltonian $H_0 = H\circ \hat\iota^*_\nabla: T^* M\to \R$, as in \eqref{reduction}. If $H_0$ is hyperregular (see \cite[Section 3.6]{AM78}), then its fiber derivative $\textbf{F}H_0: T^* M\to TM$, which is given in canonical coordinates by $\dot x^i = \frac{\pt H_0}{\pt p_i}$, is a diffeomorphism and defines the classical Legendre transform \cite[Section 3.6]{AM78}:
\begin{equation}\label{classical-Legendre}
  L_0(x, \dot x) = p_i \dot x^i - H_0(x, p) = p_i \dot x^i - H\left(x, p, \hat o \right),
\end{equation}
where $(\hat o_{jk})$ is a family of auxiliary variables introduced in \eqref{o-hat}. Then we lift $L_0$ to an admissible second-order Lagrangian $L: \mathcal T^S M\to\R$ as in Definition \ref{adm-2-Lag}, that is, $L = L_0 \circ \varrho_\nabla$. Combining \eqref{classical-Legendre} with \eqref{lift-Lag}, the relation between $L$ and $H$ is
\begin{equation}\label{2-Legendre}
  L(x, Dx, Qx) = p_i D_\nabla x^i - H\left(x, p, \hat o \right) = p_i D^ix + \textstyle{\frac{1}{2}} \hat o_{jk} Q^{jk} x - H(x, p, \hat o).
\end{equation}
We call \eqref{2-Legendre} the \emph{second-order Legendre transform}. In particular, if we restrict the admissible 2nd-order Lagrangian $L$ to the subbundle of $\mathcal T^S M$ with coordinate constraint $Q^{jk} x = g^{jk}(x)$ for some symmetric $(2,0)$-tensor field $g$ (which is just the condition in \eqref{diff-space}), and let $H$ be $(g,\nabla)$-canonical, then by \eqref{canonical-ext}, we have
\begin{equation}\label{2-Legendre-special}
  L(x, Dx, Qx) = p_i D^ix + \textstyle{\frac{1}{2}} o_{jk} Q^{jk} x - H(x, p, o).
\end{equation}
Consequently, we can find the relation between 2nd-order Hamilton's principal functions and action functionals. 
By \eqref{eqn-25} and \eqref{2-Legendre-special},
\begin{equation*}
  \mathbf D_tS = L(t, x, Dx, Qx) = L_0(t, x, D_\nabla x).
\end{equation*}
One concludes, from Dynkin's formula, that for an $M$-valued diffusion $X\in \A_g([0,T];q, \mu)$,
\begin{equation*}
  \E S(T, X(T)) - S(0, q) = \E \int_0^T L_0\left(t, X(t), D_\nabla X(t) \right) dt = \mathcal S[X;0,T],
\end{equation*}
and
\begin{equation*}
  S(t,x) - S(0, q) = \E_{(t,x)} [S(t, X(t)) - S(0, X(0))] = \E_{(t,x)} \int_0^t L_0\left(s, X(s), D_\nabla X(s) \right) ds,
\end{equation*}
where $\E_{(t,x)}$ is the conditional expectation $\E(\cdot|X(t)=x)$. These mean that the action functional is the expectation of 2nd-order Hamilton's principal function (up to an undetermined constant), while the 2nd-order Hamilton's principal function is the conditional expectation version of action functional.

Conversely, let us be given an admissible 2nd-order Lagrangian $L: \mathcal T^S M\to \R$ which is the $\nabla$-lift of a classical Lagrangian $L_0: T M\to \R$. If $L_0$ is hyperregular, then its fiber derivative
\begin{equation}\label{Legendre-tsf}
  \textbf{F}L_0: T M\to T^*M, \quad (x,\dot x)\mapsto (x,d_{\dot x}L_0),
\end{equation}
which is written in coordinates as $p_i = \frac{\pt L_0}{\pt \dot x^i}$, is a diffeomorphism and defines the classical inverse Legendre transform:
\begin{equation}\label{Legendre-tsf-2}
  H_0(x, p) = p_i \dot x^i - L_0(x, \dot x).
\end{equation}
We replace coordinates $(\dot x^i)$ by $(D_\nabla^i x)$, due to \eqref{partial-coord}. Now, given a symmetric $(2,0)$-tensor field $g$, we lift $H_0$ to the $(g,\nabla)$-canonical $\overline H^g_0$ in \eqref{lift-Ham}. The relation between $\overline H^g_0$ and $L$ is
\begin{equation}\label{2-Legendre-inverse}
  \begin{split}
    \overline H^g_0(x,p,o) &= p_i D_\nabla^i x - L_0(x, D_\nabla x) + \textstyle{\frac{1}{2}} g^{jk}(x) \left(o_{jk} - \Gamma^i_{jk}(x)p_i \right) \\
    &= p_i D^i x + \textstyle{\frac{1}{2}} o_{jk} Q^{jk}x - L(x, Dx, Qx) + \textstyle{\frac{1}{2}} \left( g^{jk}(x) - Q^{jk}x \right) o^\nabla_{jk},
  \end{split}
\end{equation}
where $(o^\nabla_{jk})$ is the tensorial conjugate diffusivities defined in \eqref{tensorial-diffusivities}. We call \eqref{2-Legendre-inverse} the $(g,\nabla)$-\emph{canonical inverse 2nd-order Legendre transform}. When $g$ is Riemannian and $\nabla$ is the associated Levi-Civita connection, we call \eqref{2-Legendre-inverse} the $g$-canonical inverse 2nd-order Legendre transform. In particular, when restricting $L$ onto the subbundle of $\mathcal T^S M$ with coordinate constraint $Q^{jk} x = g^{jk}(x)$, we have
\begin{equation}\label{2-Legendre-inverse-2}
  \overline H^g_0(x,p,o) = p_i D^i x + \textstyle{\frac{1}{2}} o_{jk} Q^{jk}x - L(x, Dx, Qx).
\end{equation}

Following the procedure in classical mechanics \cite[Definition 3.5.11]{AM78}, for a given classical Lagrangian $L_0:TM\to \R$, we define a function $A_0:TM\to \R$ by $A_0(v_x) = \mathbf FL_0(v_x)\cdot v_x$, and the \emph{classical energy} $E_0:TM\to \R$ by $E_0 = A_0-L_0$. Notice that in local coordinates, $A_0 = \dot x^i \frac{\pt L_0}{\pt \dot x^i}$ and $E_0 = \dot x^i \frac{\pt L_0}{\pt \dot x^i} - L_0$. 

\begin{example}
It is easy to check that the $\nabla$-lift of the classical Lagrangian $L_0$ in \eqref{Lagrangian} is the second-order Legendre transform of the second-order Hamiltonian $H$ in \eqref{exmp-Hamiltonian}. And conversely, the latter is the $g$-canonical inverse 2nd-order Legendre transform of the former. The classical energy associated with this Lagrangian is given by
\begin{equation}\label{classical-energy}
  E_0(t,x,\dot x) = \frac{1}{2} |\dot x-b(t,x)|^2 + \langle \dot x-b(t,x), b(t,x) \rangle + F(t,x).
\end{equation}
Each term at RHS corresponds to a kinetic energy, a vector potential energy and a scalar potential energy, respectively.
\end{example}

\subsubsection{Stochastic Hamiltonian mechanics on Riemannian manifolds}\label{sec-7-4-2}

Given a reference metric tensor $g$, i.e., a geodesically complete Riemannian metric as in Section \ref{sec-7-2}, let $\nabla$ be the associated Levi-Civita connection. If a 2nd-order Hamiltonian $H$ is the $g$-canonical lift of a classical Hamiltonian $H_0$, namely, $H = \overline H_0^g$ as in \eqref{lift-Ham}, then the stochastic Hamilton's equations \eqref{stoch-Hamilton-eqns-2} can reduce to a simpler Hamilton-type system on $T^*M$, which is exactly equivalent to the stochastic Euler-Lagrange equation \eqref{stoch-EL} via the classical Legendre transform \eqref{Legendre-tsf} and \eqref{Legendre-tsf-2}.

Similarly to \eqref{diff-dot-x} and \eqref{diff-x}, we introduce, for a smooth function $f$ on $T^* M$, the vertical gradient $\nabla_p f$ and horizontal differential $d_x f$ which are given in local coordinates $(x,p)$ by
\begin{equation*}
  \nabla_p f = \frac{\pt f}{\pt p_i} \vf{x^i}, \quad d_x f = \left( \frac{\pt f}{\pt x^i} + \Gamma_{ij}^k p_k \frac{\pt f}{\pt p_j} \right) dx^i.
\end{equation*}
Both are invariant under change of coordinates. Still by the classical theory, the connection $\nabla$ can uniquely determine a $TT^*M$-valued 1-form on $T^*M$ horizontal over $M$, given by
\begin{equation*}
  \Gamma^* = dx^i \otimes\left( \frac{\pt}{\pt x^i} + \Gamma_{ij}^k p_k \frac{\pt}{\pt p_j} \right).
\end{equation*}
Hence, we have $d_x f = \Gamma^*(df)$. Given a 1-form $\eta$ on $M$, $f\circ \eta: q\mapsto f(\eta_q)$ is a smooth function on $M$. Then, it is easy to verify that
\begin{equation}\label{eqn-21}
  d(f\circ \eta) = d_x f \circ \eta + \nabla_{(\nabla_p f\circ \eta)} \eta.
\end{equation}

\begin{theorem}\label{canonical-reduction}
  Given a smooth function $H_0 : T^*M\times\R \to\R$.

  (i). Let $H = \overline H_0^g: \mathcal T^{S*} M\times\R \to\R$ be the $g$-canonical lift of $H$. Let $\mathbf X$ be the horizontal integral process of stochastic Hamilton's equations \eqref{stoch-Hamilton-eqns-2} corresponding to $H$ and $X=\tau_M^{S*}(\mathbf X)$. Define a $T^*M$-valued horizontal diffusion by $\mathbb X:= \hat\varrho^*(\mathbf X)$. Then $\mathbb X(t) = p(t,X(t))$ solves the following system on $T^*M$,
  \begin{equation}\label{canonical-stoch-Hamilton-eqns}\left\{
    \begin{aligned}
      D_\nabla X(t) &= \nabla_p H_0 (\mathbb X(t),t), \\
      \frac{\overline\D}{dt} p(t,X(t)) &= -d_x H_0 (\mathbb X(t),t),
    \end{aligned}\right.
  \end{equation}
  subject to $QX(t) = \check g(X(t))$, where $\frac{\overline\D}{dt}$ is the damped mean covariant derivative with respect to $X$. In this case, we refer to the system \eqref{canonical-stoch-Hamilton-eqns} as the $g$-canonical reduction of \eqref{stoch-Hamilton-eqns-2}, or global stochastic Hamilton's equations.

  (ii). If $H_0$ is hyperregular, then the global stochastic Hamilton's equations \eqref{canonical-stoch-Hamilton-eqns} are equivalent to the  stochastic Euler-Lagrange equation \eqref{stoch-EL} via the classical Legendre transform $p=d_{\dot x} L_0$ and $H_0(x,p,t) = p\cdot \dot x - L_0(t,x,\dot x)$.

  (iii). Let $S\in C^\infty(M\times\R)$. Then the following statements are equivalent: \\
  (a) for every $M$-valued diffusion $X$ satisfying
  \begin{equation}\label{cond-1}
    D_\nabla X(t) = \nabla_p H_0 (dS(t, X(t) ),t), \quad QX(t) = \check g(X(t)),
  \end{equation}
  the $T^* M$-valued process $dS\circ X$ solves the global stochastic Hamilton's equations \eqref{canonical-stoch-Hamilton-eqns}; \\
  (b) $S$ satisfies the following Hamilton-Jacobi-Bellman equation
  \begin{equation}\label{HJB-4}
    \frac{\pt S}{\pt t} + H_0(d S, t) + \frac{1}{2}\Delta S = f(t),
  \end{equation}
  for some function $f$ depending only on $t$.
\end{theorem}

\begin{proof}
  (i). Since $H = \overline H_0^g = H_0 + \textstyle{\frac{1}{2}} g^{jk} (o_{jk} - \Gamma^i_{jk} p_i )$, $(QX)^{jk} = 2\frac{\pt H}{\pt o_{jk}}$ if and only if $QX(t) = \check g(X(t))$. Since,
  \begin{equation*}
    \frac{\pt H}{\pt p_i} = \frac{\pt H_0}{\pt p_i} - \frac{1}{2} g^{jk} \Gamma^i_{jk} = dx^i( \nabla_p H_0 ) - \frac{1}{2} (QX)^{jk} \Gamma^i_{jk},
  \end{equation*}
  we have $(DX)^i = \frac{\pt H}{\pt p_i}$ if and only if $D_\nabla X= \nabla_p H_0$, due to \eqref{relation-two-drvtv}. This proves the first equation of \eqref{canonical-stoch-Hamilton-eqns}. Furthermore,
  \begin{equation*}
    \frac{\pt H}{\pt x^i} = \frac{\pt H_0}{\pt x^i} + \frac{1}{2} \pt_i g^{jk} \left( o_{jk} - \Gamma^l_{jk} p_l \right) - \frac{1}{2} g^{jk} \pt_i \Gamma^l_{jk} p_l = \frac{\pt H_0}{\pt x^i} - g^{jm} \Gamma_{im}^k \left( o_{jk} - \Gamma^l_{jk} p_l \right) - \frac{1}{2} g^{jk} \pt_i \Gamma^l_{jk} p_l.
  \end{equation*}
  On the other hand, by applying Lemma \ref{stoch-mean-derv-prop} (ii) and (iv), and the equation $D_\nabla X= \nabla_p H_0$, we have
  \begin{equation}\label{eqn-23}
    \begin{split}
      (D (p\circ\textbf{X}))_i &= \mathbf D_t p_i = \mathbf D_t [p(\pt_i)] = \frac{\overline\D p}{dt} (\pt_i) + p \left( \frac{\overline\D \pt_i}{dt} \right) + (QX)^{jk} (\nabla_{\pt_j} p) (\nabla_{\pt_k} \pt_i) \\
      &= \frac{\overline\D p}{dt} (\pt_i) + p \left( \nabla^{}_{D_\nabla X}\pt_i + \frac{1}{2} g^{jk} \nabla^2_{\pt_j,\pt_k} \pt_i + \frac{1}{2} g^{jk} R(\pt_i,\pt_j)\pt_k \right) + g^{jk} (\nabla_{\pt_j} p) (\nabla_{\pt_k} \pt_i) \\
      &= \frac{\overline\D p}{dt} (\pt_i) + p_l \left( \frac{\pt H_0}{\pt p_j} \Gamma_{ij}^l + \frac{1}{2}g^{jk} \pt_i \Gamma_{jk}^l \right) + g^{jk} \Gamma_{ik}^m \left( \pt_j p_m - \Gamma_{jm}^l p_l \right).
    \end{split}
  \end{equation}
  Hence,
  \begin{equation*}
    \begin{split}
      (D (p\circ\textbf{X}))_i + \frac{\pt H}{\pt x^i} = \frac{\overline\D p}{dt} (\pt_i) + d_x H_0(\pt_i) + g^{jm} \Gamma_{im}^k \left( \pt_j p_k - o_{jk} \right).
    \end{split}
  \end{equation*}
  The second equation of \eqref{canonical-stoch-Hamilton-eqns} follows from \eqref{non-degenerate}.

  (ii). The equivalence between \eqref{canonical-stoch-Hamilton-eqns} and \eqref{stoch-EL} follows from the following calculations:
  \begin{gather*}
    \nabla_p H_0 = \nabla_p (p\cdot \dot x - L_0) = \dot x, \\
    d_x H_0 = \left( \frac{\pt H_0}{\pt x^i} + \Gamma_{ij}^k p_k \frac{\pt H_0}{\pt p_j} \right) dx^i = \left( -\frac{\pt L_0}{\pt x^i} + \Gamma_{ij}^k \frac{\pt L_0}{\pt \dot x^k} \dot x^j \right) dx^i = - d_x L_0.
  \end{gather*}

  (iii). By \eqref{damped-mean-cov}, conditions \eqref{cond-1} and \eqref{eqn-21},
  \begin{equation*}
    \begin{split}
      \frac{\overline\D}{dt} (dS) &= \left(\frac{\pt}{\pt t} + \nabla^{}_{D_\nabla X} + \frac{1}{2}\Delta_{\mathrm{LD}} \right) (dS) = d\frac{\pt S}{\pt t} + \nabla_{(\nabla_p H_0\circ dS)} dS - \frac{1}{2}( dd^*+d^* d)dS \\
      &= d\frac{\pt S}{\pt t} + d(H_0\circ dS) - d_x H_0 \circ dS - \frac{1}{2} dd^*dS = d\left( \frac{\pt S}{\pt t} + H_0\circ dS + \frac{1}{2} \Delta S \right) - d_x H_0 \circ dS.
    \end{split}
  \end{equation*}
  The result follows.
\end{proof}

\begin{remark}\label{remark-4}
  (i). Assertions (ii) and (iii) of Theorem \ref{canonical-reduction} generalize Theorem \ref{SEL-HJB}, since from the Legendre transform $p=d_{\dot x} L_0$ we observe that the S-EL equation \eqref{stoch-EL} is related to HJB equation \eqref{HJB-4} via equation \eqref{eqn-10}. However, assertion (iii) is a special case of Theorem \ref{HJB-int-proc}, since HJB equation \eqref{HJB-4} is just the one in \eqref{HJB-3} with $H = \overline H_0^g$ the $g$-canonical lift of $H_0$, due to the observation that $\overline H_0^g(d^2S, t) = H_0(d S, t) + \frac{1}{2}\Delta S$.

  (ii). The advantage of Theorem \ref{canonical-reduction} is that it formulates stochastic Hamiltonian mechanics in a global way similar to stochastic Lagrangian mechanics, while its disadvantage is that it depends on the choice of Riemannian structures. However, unlike stochastic Hamiltonian mechanics of Chapter \ref{sec-6}, neither global S-H equations \eqref{canonical-stoch-Hamilton-eqns} nor HJB equation \eqref{HJB-4} encodes any new symplectic or contact structures, as the Hamiltonian functions therein are still classical.

  (iii). By a direct calculation similar to \eqref{eqn-23}, one easily obtains the  following local version of stochastic Euler-Lagrange equation \eqref{stoch-EL}:
  \begin{equation}\label{eqn-34}
    \mathbf D_t \left( \frac{\pt L_0}{\pt \dot x^i} \right) = \frac{\pt L_0}{\pt x^i} + \frac{1}{2} g^{jk} \pt_i \Gamma^l_{jk} \frac{\pt L_0}{\pt \dot x^l} - \frac{1}{2} \pt_i g^{jk} \left( \frac{\pt^2 L_0}{\pt x^j \pt \dot x^k} - \Gamma^l_{jk} \frac{\pt L_0}{\pt \dot x^l} \right).
  \end{equation}
  This local version is related to stochastic Hamilton's equations \eqref{stoch-Hamilton-eqns} via the canonical 2nd-order Legendre transform \eqref{2-Legendre-inverse}.

  (iv). Similarly to Remark \ref{remark-3}, if we let $\tilde H = H - f$, then Theorem \ref{canonical-reduction} holds with $\tilde H$ and zero function in place of $H$ and $f$. 
  We will refer to equation \eqref{HJB-4} with $f\equiv 0$ as the HJB equation associated with Hamiltonian $H_0$, or the HJB equation associated with the Lagrangian $L_0$ related to $H_0$ via the Legendre transform (when $H_0$ is hyperregular).
\end{remark}

%

On Riemannian manifolds, canonical transformations of Section \ref{canonical-trans} can also be reduced to tangent bundles. We consider 
a bundle isomorphism $\mathbf F$ from $\mathcal T^{S*} M\times\R$ to $\mathcal T^{S*} N\times\R$, projecting to a time-change map $F^0:\R\to\R$. 
The transformation $\mathbf F$ is a map from coordinates $(x^i, p_i, o_{jk},t)$ to $(y^i, P_i, O_{jk}, s)$ satisfying $s=F^0(t)$. Both base manifolds $M$ and $N$ are  equipped with some Riemannian metrics and the corresponding Levi-Civita connections.

By the inverse 2nd-order Legendre transform \eqref{2-Legendre-inverse-2} and the integrability condition \eqref{non-degenerate}, the action functional in \eqref{action} can be rewritten as
\begin{equation*}
  \begin{split}
    \mathcal S [X;0,T] 
    &= \E \int_0^T \left[ p_i(t,X(t)) (DX)^i(t) + \frac{1}{2} \frac{\pt p_j}{\pt x^k}(t,X(t)) (QX)^{jk}(t) - \overline H^g_0(\mathbf X(t),t) \right] dt \\
    &= \E \int_0^T \left[ p_i(t,X(t)) \circ dX^i(t) - \overline H^g_0(\mathbf X(t),t) dt \right],
  \end{split}
\end{equation*}
where $\circ\,d$ denotes the Stratonovich stochastic differential and $\overline H^g_0 = H_0 + \frac{1}{2} g^{jk} (o_{jk} - \Gamma^i_{jk}p_i )$. We denote simply $x^i = x^i\circ \mathbf X$, $p_i = p_i\circ \mathbf X$ and $H= \overline H^g_0$. Then $\mathcal S = \E \int_0^T (p_i\circ dx^i(t) - H dt)$. Now we make a change of coordinates from $(x^i,p_i,t)$ to $(y^i,P_i,s)$ satisfying $s=F^0(t)$, and denote that $y^i = y^i\circ \mathbf X$ and $P_i = P_i\circ \mathbf X$. We have
\begin{equation*}
  \mathcal S = \E \int_0^T \left(P_i\circ dy^i(s) - K ds \right) = \E \int_0^T \left(P_i\circ d(y^i\circ F^0)(t) - K \dot F^0 dt \right),
\end{equation*}
where the function $K$ plays the role of the 2nd-order Hamiltonian in new coordinate system.

As in Section \ref{canonical-trans}, the general condition for a transformation to be canonical is to preserve the form of stochastic Hamilton's system \eqref{canonical-stoch-Hamilton-eqns}. This is equivalent to preserve the form of stochastic stationary-action principle \eqref{stationary}, according to Theorem \ref{canonical-reduction}.(ii). It follows from $\delta \mathcal S = 0$ that
\begin{equation*}
  \delta\,\E \int_0^T \left(p_i\circ dx^i(t) - H dt\right) = \delta\,\E \int_0^T \left(P_i\circ d(y^i\circ F^0)(t) - K \dot F^0 dt\right) =0.
\end{equation*}
Since the underlying process $X$ has zero variation at the endpoints, both equalities will be satisfied if the integrands are related by the following SDE:
\begin{equation}\label{stoch-generating-func}
  p_i\circ dx^i - H dt = P_i\circ dy^i - K \dot F^0 dt + dG,
\end{equation}
where $G$ is a function of phase space coordinates $(x,p,t)$ or $(y,P,s)$ or any mixture of them and called the generating function. Note that in contrast with the classical theory of canonical transformation and also \eqref{formal-transf}, here equation \eqref{stoch-generating-func} for canonical transformations is a stochastic differential equation, instead of equation for forms.

Consider the \emph{type one} generating function $G_1$, that is, $G = G_1(x,y,t)$ is given as a function of the old and new generalized position coordinates (cf. \cite[Section 9.1]{GPS02}). Then using It\^o's formula $dG_1 = \frac{\pt G_1}{\pt t}dt + \frac{\pt G_1}{\pt x^i} \circ dx^i + \frac{\pt G_1}{\pt y^i} \circ dy^i$, and vanishing the coefficients of every (stochastic) differentials $\circ dx$, $\circ dy$ and $dt$ in \eqref{stoch-generating-func}, we get
\begin{equation*}
  p_i = \frac{\pt G_1}{\pt x^i}, \quad P_i = - \frac{\pt G_1}{\pt y^i}, \quad K \dot F^0 - H = \frac{\pt G_1}{\pt t},
\end{equation*}
which recovers \eqref{formal-relation}. By taking $F^0 = \id_\R$ (i.e., no time-change) and requiring the new Hamiltonian $K_0$ to be identically zero, and writing $G_1$ as $S$ the last equation turns into the following HJB equation 
\begin{equation*}
    \frac{\pt S}{\pt t}(x,y,t) + H_0\left( x^i, \frac{\pt S}{\pt x^i}(x,y,t), t \right) + \frac{1}{2}\Delta_x S(x,y,t) + \frac{1}{2}\Delta_y S(x,y,t) = 0,
\end{equation*}
where $(x,y)$ are regarded as coordinates on the product manifold $M\times N$ equipped with the direct-sum Riemannian metric and its corresponding Levi-Civita connection, $\Delta_x$ and $\Delta_y$ are the Laplacian on $M$ and $N$, respectively, so that $\Delta_x + \Delta_y$ is the Laplacian on $M\times N$ under the aforementioned connection.

In contrast to the mixed-order contact approach to canonical transformations of Section \ref{canonical-trans}, since the changes of coordinates proceed on $T^* M$, one can easily formulate \emph{four types} of generating functions that are related to each other through classical Legendre transforms in the same way as in classical mechanics \cite[Section 9.1]{GPS02}. For example, the \emph{type two} generating function takes the form $G = G_2(x,P,t) - y^i P_i$, for which we have
\begin{equation}\label{canonical-trans-R}
  p_i = \frac{\pt G_2}{\pt x^i}, \quad y^i = \frac{\pt G_2}{\pt P_i}, \quad K \dot F^0 - H = \frac{\pt G_2}{\pt t}.
\end{equation}
In this case, since $(x^i)$ and $(y^i)$ are no longer independent variables, Riemannian structures on $M$ and $N$ should be related by the transformation. In view of this, we only consider point transformations, a subclass of canonical transformations. That is, we assume $G_2$ to be the form
\begin{equation*}
  G_2(x,P,t) = f^i(x,t) P_i + h(x,t)
\end{equation*}
for some diffeomorphisms $f: M\to N$'s and $h: M\to\R$. The second equation of \eqref{canonical-trans-R} implies
\begin{equation*}
  y^i = f^i(x,t).
\end{equation*}
So we equip $N$ with the (time-dependent) pushforward Riemannian metric of $g$ on $M$ by $f$, and with the Levi-Civita connection.

\begin{example}[Canonical transformations for one-dimensional Bernstein's reciprocal processes]
  Consider the scalar case of Example \ref{Brownian-bridge}, that is, the $\R$-valued Brownian reciprocal process with 2nd-order Hamiltonian $H(x,p,o) = H_0(x,p) + \frac{1}{2}o = \frac{1}{2}|p|^2+ \frac{1}{2}o$. 
  The equations of motion are $DDX=0$, $QX = 1$ (cf. \eqref{EQM-eqn-motion}). 
  In the following, we will consider two canonical transformations which transform Brownian reciprocal processes to reciprocal processes derived from diffusions with linear potentials and quadratic potentials, respectively.

  (i). Consider the time-dependent change of coordinates from $(x,p,t)$ to $(y,P,t)$ (without time-change) induced by $G_2(x,P,t)= (x + \ts{\frac{1}{2}} t^2) P -tx$. 
  By \eqref{canonical-trans-R},
  \begin{equation}\label{eqn-28}
    y = x+ \ts{\frac{1}{2}} t^2, \quad p = P-t, \quad K = H + Pt-y + \ts{\frac{1}{2}} t^2.
  \end{equation}
  For the latent 2nd-order coordinates, we have
  $$O = \frac{\pt P}{\pt y} = \frac{\pt p}{\pt x} = o.$$
  Hence, by the last equation of \eqref{canonical-trans-R}, the new 2nd-order Hamiltonian is
  $$K(y,P,O,t) = K_0(y,P,t)+ \frac{1}{2}O= \frac{1}{2}|P|^2 - y + t^2 + \frac{1}{2}O,$$
  which is still of the form \eqref{exmp-Hamiltonian}, with $b\equiv0$ and $F(t,y) = -y+t^2$. The equations of motion under new coordinates are $DDY=1$ and $QY = 1$. By Remark \ref{remark-3}, $K$ share the same equations of motion with $\tilde K(y,P,O) = \frac{1}{2}|P|^2 - y + \frac{1}{2}O$. In other words, \eqref{eqn-28} transforms Brownian reciprocal processes to reciprocal processes derived from diffusions with linear potentials. 
  This example is taken from \cite[Theorem 4.1.(1)]{LZ07}, where the authors used \eqref{eqn-28} to transform free heat equations to heat equations with linear potentials. We refer readers to \cite{LZ07} for more applications of canonical transformations of contact Hamiltonian systems to Euclidean quantum mechanics in Example \ref{EQM}.

  (ii). Consider the following change of coordinates from $(x,p,t)$ to $(y,P,s)$ (with time-change)
  \begin{equation}\label{eqn-29}
    x = y\sqrt{1-t^2}, \quad P= p\sqrt{1-t^2} + yt, \quad s=\arctanh t.
  \end{equation}
  Clearly, the map $(x,p)\mapsto (y,P)$ is induced by the type three generating function $G_3(y,p,t) = -py\sqrt{1-t^2} -\frac{y^2}{2}t$ via relations $x= -\frac{\pt G_3}{\pt p}$ and $P= -\frac{\pt G_3}{\pt y}$. The relation between the latent coordinates $o$ and $O$ is
  \begin{equation}\label{eqn-30}
    O = \frac{\pt P}{\pt y} = \frac{\pt p}{\pt x} \frac{\pt x}{\pt y} \sqrt{1-t^2} + t = (1-t^2) o +t.
  \end{equation}
  The new 2nd-order Hamiltonian $K$ satisfies $K \frac{ds}{dt} - H = \frac{\pt G_3}{\pt t}$. Hence, combining with \eqref{eqn-29} and \eqref{eqn-30}, we obtain
  $$K (y,P,O,s) = (1-t^2) \left( \frac{1}{2} |p|^2 + \frac{py t}{\sqrt{1-t^2}} - \frac{1}{2} |y|^2 + \frac{1}{2} o \right) = \frac{1}{2} |P|^2 - \frac{1}{2} |y|^2 + \frac{1}{2} O - \frac{1}{2} \tanh s.$$
  This differs with the 2nd-order Hamiltonian of Euclidean harmonic oscillators in Example \ref{EQM}.(ii), i.e., $\tilde K (y,P,O) = \frac{1}{2} |P|^2 - \frac{1}{2} |y|^2 + \frac{1}{2} O$, by a term depending only on time. So by virtue of Remark \ref{remark-3}, $K$ and $\tilde K$ share the same equations of motion $DDY=Y$, $QY = 1$. Therefore, \eqref{eqn-29} transforms free reciprocal processes to Euclidean harmonic oscillators.
\end{example}

\begin{example}[Canonical transformations for vanishing potentials]
  Let $(M,g)$ be Riemannian. Take $G_2(x,P,t) = x^i P_i - S(x,t)$ for some function $S$. Then
  \begin{equation*}
    y = x, \quad p_i = P_i - \frac{\pt S}{\pt x^i}, \quad K = H - \frac{\pt S}{\pt t}.
  \end{equation*}
  Since the transformation on base manifold $M$ is identity, it does not change the Riemannian metric, and
  $$o_{ij} = \frac{\pt p_i}{\pt x^j} = \frac{\pt P_i}{\pt y^j} - \frac{\pt^2 S}{\pt x^i\pt x^j}.$$

  (i). We consider the Hamiltonian $H_0(x, p) \equiv b^i(x) p_i - F(x)$, whose corresponding 2nd-order Hamiltonian $H = \overline H^g_0$ has a diffusion with generator $\frac{1}{2} \Delta + b\cdot\nabla - F$ for solution process (see Subsection \ref{sec-6-3-1}). Then, the new Hamiltonian is
  $$K_0(y,P,t) = b^i(y) P_i - \langle b(y), \nabla S(y,t)\rangle - F(y) - \frac{1}{2} \Delta S(y,t) - \frac{\pt S}{\pt t}(y,t).$$
  If we choose $S$ solving the backward PDE \eqref{exmp-HJB-0}, then $K_0(y,P) = b^i(y) P_i$ has a diffusion process with generator $\frac{1}{2} \Delta + \nabla_b$ for solution. In particular, such a canonical transformation can transform a diffusion process with a scalar potential into a free motion.

  (ii). Consider the Hamiltonian $H_0(x,p,t) = \frac{1}{2}g^{ij}(x) p_i p_j + g^{ij}(x)p_i\frac{\pt S}{\pt x^j}(x,t) + b^i(x) p_i - F(x)$, whose corresponding 2nd-order Hamiltonian $H = \overline H^g_0$ has a Schr\"odinger's bridge with vector potential $(b+\nabla S)$ and scalar potential $-F$ for solution process.
  Then, the new Hamiltonian is
  $$K_0(y,P,t) = \frac{1}{2}g^{ij}(y) P_iP_j + b^i(y) P_i - \langle b(y), \nabla S(y,t)\rangle - \frac{1}{2}|\nabla S(y,t)|^2 - F(y) - \frac{1}{2}\Delta S(y,t) - \frac{\pt S}{\pt t}(y,t).$$
  To transform $K_0$ into the standard form $K_0(y,P,t) = \frac{1}{2}g^{ij}(y) P_iP_j + b^i(y) P_i$ whose solution is a Schr\"odinger's bridge with vector potential $b$, we only need to assume that $S$ solves HJB equation \eqref{exmp-HJB}. In particular, such a canonical transformation transforms a Schr\"odinger's bridge with a scalar potential into a free one.
\end{example}

Regarding the classical energy introduced in the end of Subsection \ref{sec-7-4-1}, for a given classical Lagrangian $L_0:\R\times TM\to \R$, we introduce its \emph{generalized (or deformed) energy} $E :\R\times TM\to \R$ by
\begin{equation*}
  E(t,x,\dot x) = E_0(t,x,\dot x) + \textstyle{\frac{1}{2}} \Delta S(t,x),
\end{equation*}
where $S$ is the solution of the Hamilton-Jacobi-Bellman equation \eqref{HJB-4} associated with $L_0$ (with $f\equiv0$). The term $\frac{1}{2}\Delta S$ stands for the stochastic deformation.

\subsubsection{Small-noise limits}

In this part, we will see, informally, how our stochastic framework degenerates into classical mechanics as the noise goes to zero. Let $\e>0$ be a small parameter which we refer to as \emph{diffusivity}. The limit when $\e\to 0$ is called the \emph{small-noise limit}.

Let $\A^\e_g([0,T];q, \mu)$ be the small-noise version of the admissible class \eqref{diff-space}, that is, with the constraint $QX(t) = \e \check g(X(t))$. The $\e$-dependent stochastic variational problem is to minimize the action functional $\mathcal S [X;0,T]$ in \eqref{action} among all $X\in\A^\e_g([0,T];q, \mu)$. Then, the same procedure as Section \ref{sec-7-2} yields the following $\e$-dependent stochastic Euler-Lagrange equation,
\begin{equation}\label{stoch-EL-e}
    \frac{\overline\D^\e}{dt} \big( d_{\dot x} L_0\left(t, X_\e(t), D_\nabla X_\e(t) \right) \big) = d_x L_0\left(t, X_\e(t), D_\nabla X_\e(t) \right),
\end{equation}
which is an equivalent condition for $X_\e \in \A^\e_g([0,T];q, \mu)$ to be a stationary point of $\mathcal S$. Here $\frac{\overline\D^\e}{dt}$ is the damped mean covariant derivative with respect to $X_\e$ so that
$$\frac{\overline\D^\e}{dt} = \frac{\pt}{\pt t} + \nabla^{}_{D_\nabla X} + \frac{\e}{2} \Delta_{\mathrm{LD}}.$$
Now as $\e\to 0$, since $QX_\e \to 0$, $X_\e$ tends to some deterministic curve $\gamma=(\gamma(t))_{t\in[0,T]}$ (in a suitable probabilistic sense), and $D_\nabla X_\e(t)$ tends to $\dot \gamma(t)$. Thus, we can write informally
\begin{equation*}
  \A^\e_g([0,T];q, \mu) \to \A^0_g([0,T];q, \mu) := \left\{\gamma \text{ is adapted with paths in } C^2([0,T],M) : \gamma(0) = q, \P\circ(\gamma(T))^{-1} = \mu \right\}.
\end{equation*}
The $\e$-dependent stochastic variational problem tends to the following deterministic variational problem
\begin{equation}\label{variation-determ}
  \min_{\gamma\in\A^0_g([0,T];q, \mu)} \int_0^T L_0\left(t, \gamma(t), \dot\gamma(t) \right) dt.
\end{equation}
And the $\e$-dependent stochastic Euler-Lagrange equation \eqref{stoch-EL-e} tends to
\begin{equation}\label{stoch-EL-0}
    \frac{D}{dt} \big( d_{\dot x} L_0\left(t, \gamma(t), \dot \gamma(t) \right) \big) = d_x L_0\left(t, \gamma(t), \dot \gamma(t) \right),
\end{equation}
where, $\frac{D}{dt} = \frac{\pt}{\pt t} + \nabla_{\dot\gamma}$ is the material derivative along $\gamma$. This is the classical Euler-Lagrange equation in global form, cf. \cite[p. 153]{Vil09}.

We introduce the following $\e$-dependent version of the $g$-canonical lift \eqref{lift-Ham}:
\begin{equation*}
  H_\e(x,p,o,t) := H_0(x,p,t) + \textstyle{\frac{\e}{2}} g^{jk}(x) \left(o_{jk} - \Gamma^i_{jk}(x)p_i \right).
\end{equation*}
Let $\mathbf X_\e$ be a horizontal integral process of stochastic Hamilton's equations \eqref{stoch-Hamilton-eqns} corresponding to $H_\e$ and $X_\e=\tau_M^{S*}(\mathbf X_\e)$. Since $(Q (x\circ\textbf{X}_\e))^{jk}= 2\frac{\pt H_\e}{\pt o_{jk}} = \e\check g \to 0$ as $\e\to0$, $\mathbf X_\e$ converges to a $T^* M$-valued process. And since $\frac{\pt H_\e}{\pt p_i} \to \frac{\pt H_0}{\pt p_i}$ and $\frac{\pt H_\e}{\pt x^i} \to \frac{\pt H_0}{\pt x^i}$, the limit $T^* M$-valued process satisfies classical Hamilton's equations,
\begin{equation}\label{classical-H-E}\left\{
  \begin{aligned}
    \dot x^i(t) &= \ts{\frac{\pt H_0}{\pt p_i}} (x(t), p(t), t), \\
    \dot p_i(t) &= - \ts{\frac{\pt H_0}{\pt x^i}} (x(t), p(t), t).
  \end{aligned}\right.
\end{equation}
Let $\mathbb X_\e:= \hat\varrho^*(\mathbf X_\e)$. Then, $\mathbb X_\e(t) = p(t,X_\e(t))$ solves the system of global stochastic Hamilton's equations \eqref{canonical-stoch-Hamilton-eqns}, with $\mathbb X_\e$, $X_\e$ and $\frac{\overline\D^\e}{dt}$ in place of $\mathbb X$, $X$ and $\frac{\overline\D}{dt}$, respectively, subject to $QX_\e(t) = \e\check g(X_\e(t))$. As $\e$ goes to 0, this system tend to the following deterministic system,
\begin{equation}\label{canonical-stoch-Hamilton-eqns-0}\left\{
    \begin{aligned}
      \dot x(t) &= \nabla_p H_0 (x(t), p(t), t), \\
      \frac{D}{dt} p(t) &= -d_x H_0 (x(t), p(t),t),
    \end{aligned}\right.
  \end{equation}
This is indeed the global form of \eqref{classical-H-E} which is equivalent to the global Euler-Lagrange equation \eqref{stoch-EL-0} via the classical Legendre transform.

The corresponding $\e$-dependent Hamilton-Jacobi-Bellman equation is now
  \begin{equation*}
    \frac{\pt S}{\pt t} + H_\e(d^2 S, t) = \frac{\pt S}{\pt t} + H_0(d S, t) + \frac{\e}{2}\Delta S = f(t),
  \end{equation*}
which, as $\e\to0$, goes to the classical Hamilton-Jacobi equation
  \begin{equation*}
    \frac{\pt S}{\pt t} + H_0(d S, t) = f(t).
  \end{equation*}
The latter corresponds to \eqref{stoch-EL-0}--\eqref{canonical-stoch-Hamilton-eqns-0} via classical Hamilton-Jacobi theory (e.g., \cite[Chapter 5]{AM78}).

We list here some previous works that have independent interests in the above small-noise limits, in some special cases. The time-asymptotic large deviation for Brownian bridges of Example \ref{Brownian-bridge} was studied in \cite{Hsu90}. The second author of the present paper and his collaborator proved in \cite{PYZ16} a large deviation result for one-dimensional Bernstein bridges which are solution processes of Euclidean quantum mechanics in Example \ref{EQM}. The paper \cite{Leo12a} proved that the $\Gamma$-limit of Schr\"odinger's problem in Section \ref{sec-7-3} with small variance is the Monge-Kantorovich problem. The latter is the optimal transport problem associated with the classical variational problem \eqref{variation-determ} \cite[Chapter 7]{Vil09}. See \cite[Section 2.3]{Mik21} for more on small-noise limits of stochastic optimal transport.

\begin{remark}
  There are various terminologies in other areas related to the small-noise limit. In thermodynamics \cite{HZ23}, $\e$ stands for the Boltzmann constant which relates to the diffusion coefficient via Einstein relation, as consistent with Schr\"odinger's original statistical problem \cite{Sch32}; when applied to quantum mechanics as in Example \ref{EQM}, the small-noise limit is called the \emph{semiclassical limit} and the parameter $\e$ stands for the reduced Planck constant $\hbar$; when/if applied to hydrodynamics (cf. \cite{ACC14,CCR23}), it is often called the \emph{vanishing viscosity limit} and $\e$ stands for the kinematic viscosity $\nu$. The latter may be expected to solve Kolmogorov's conjecture that the ``stochastization'' of dynamical systems is related to hydrodynamic PDEs as viscosity vanishes \cite{AK21}. In physics, diffusivity, Planck constant and viscosity are indeed related to each other \cite{TB21}.
\end{remark}

\subsubsection{Relations to stochastic optimal control}\label{sec-7-4-4}

Following the way of converting problems of classical calculus of variations into optimal control problems (see \cite{FS06}), we can regard the stochastic variational problem of Section \ref{sec-7-2} as a stochastic optimal control problem.

Assume that $(M,g)$ is compact (for simplicity). Consider a stochastic control model in which the state evolves according to an $M$-valued diffusion $X$ governed by a system of MDEs on the time interval $[t,T]$, of the form
\begin{equation}\label{opt-cont}\left\{
  \begin{aligned}
    &D_\nabla X(s) = U(s), \\
    &Q X(s) = g(X(s)),
  \end{aligned} \right.
\end{equation}
or equivalently, by an It\^o SDE of the form
\begin{equation*}
  dX^i(s) = \left( U^i(s) - \frac{1}{2} g^{jk}(X(s)) \Gamma_{jk}^i(X(s)) \right) ds + \sigma_r^i(X(s)) dW^r(s), \quad s\in [t,T],
\end{equation*}
where 
$\sigma$ is the positive definite square root $(1,1)$-tensor of $g$, i.e., $\sum_{r=1}^d \sigma^i_r\sigma^j_r = g^{ij}$, $W$ is an $\R^d$-valued standard Brownian motion and, most importantly, $U$ is a $TM$-valued process called the control process. There are no control constraints for $U$ as it is admissible in the sense of \cite[Definition 2.1]{FS06}. As endpoint condition, we require that $X(t) = x$. 

The control problem on a finite time interval $s\in[t,T]$ is to choose $U$ to minimize
\begin{equation}\label{opt-cont-cost}
  J(t,x;X,U) := \E_{(t,x)} \left[ \int_t^T L_0\left(s, X(s), U(s) \right) ds - S_T(X(T)) \right], 
\end{equation}
among all pairs $(X,U)$ satisfying the system \eqref{opt-cont} and the endpoint condition, where $S_T$ is a given smooth function on $M$. The real-valued smooth function $L_0$ on $\R\times TM$ is called running cost function and $J$ the payoff functional. The problem is called a \emph{stochastic Bolza problem}. In the case $S_T\equiv0$, this stochastic control problem is of the same form as our stochastic variational problem of Section \ref{sec-7-2}. For this reason, we call the latter stochastic control problem to be in \emph{Lagrange form}. By an  argument similar to Theorem \ref{stoch-Hamilton-prin}, one can derive the same equation as \eqref{stoch-EL}, but with boundary conditions $X(t) = x$ and $d_{\dot x}L_0(T,X(T),D_\nabla X(T)) = dS_T(X(T))$.

The starting point of \emph{dynamic programming} is to regard the infimum of $J$
being minimized as a function $S(t, x)$ of the initial data:
\begin{equation*}
  S(t,x) = -\inf_{(X,U)} J (t,x;X,U).
\end{equation*}
Then, Bellman's principle of dynamic programming \cite[Section III.7]{FS06} states that for $t\le t+\e\le T$,
\begin{equation*}
   0= \inf_{X\in I_0^T(M)} \E_{(t,x)} \left[ \int_t^{t+\e} L_0\left(s, X(s), D_\nabla X(s) \right) ds - S(t+\e,X(t+\e)) + S(t,x) \right].
\end{equation*}
Divide the equation by $\e$, let $\e\to 0^+$, and then use Dynkin's formula. We get the dynamic programming equation
\begin{equation}\label{dynamic-prog}
  0 = \inf \left[ L_0(t,x,D_\nabla x) - (\D_t S)(t,x, Dx, Qx) \right],
\end{equation}
subjected to terminal data $S(T,x) = S_T(x)$. By \eqref{local-rep-total-mean} and \eqref{opt-cont},
\begin{equation*}
  \D_t S = \pt_t S + D^i x \pt_i S + \ts{\frac{1}{2}} Q^{ij} x \pt_i \pt_j S = \pt_t S + \left( D^i_\nabla x - \ts{\frac{1}{2}} \Gamma^i_{jk} g^{jk} \right) \pt_i S + \ts{\frac{1}{2}} g^{ij} \pt_i \pt_j S.
\end{equation*}
We let
\begin{equation*}
  H(x,p,o,t) = \sup \left[\left( D^i_\nabla x - \ts{\frac{1}{2}} \Gamma^i_{jk}(x) g^{jk}(x) \right) p_i + \ts{\frac{1}{2}} g^{ij}(x) o_{ij} - L_0(t,x,D_\nabla x) \right]
\end{equation*}
where the supremum can be ignored if $L_0$ is convex, so that $H = \overline H^g_0$ is exactly the canonical inverse 2nd-order Legendre transform in \eqref{2-Legendre-inverse}. Then, the dynamic programming equation \eqref{dynamic-prog} can be written as the HJB equation \eqref{HJB},  cf. \cite[Section IV.3]{FS06}.

There is also a stochastic version of Pontryagin's maximum principle \cite[Theorem 3.3.2]{Yon99}. The crucial objects in stochastic Pontryagin's principle are first- and second-order adjoint processes, $p$ and $o$, respectively. Corresponding to the stochastic control problem \eqref{opt-cont}--\eqref{opt-cont-cost}, its adjoint processes $p$ and $o$ satisfy the following backward SDEs \cite[Section 3.3.2]{Yon99} (where ``backward'' is again in a different sense from ours in Chapter \ref{sec-2}),
\begin{equation}\label{BSDE-p}\left\{
  \begin{aligned}
    dp_i(t) &= \left[ \frac{1}{2} \left( \pt_i g^{kl} \Gamma_{kl}^j + g^{kl} \pt_i \Gamma_{kl}^j \right) (X(t)) p_j(t) - \sum_{r=1}^d \pt_i \sigma_r^j (X(t)) z_{jr}(t) + \frac{\pt L_0}{\pt x^i} \left(t, X(t), U(t) \right) \right] dt \\
    &\quad + z_{ir}(t) dW^r(t), \\
    p_i(T) &= \pt_i S_T(X(T)),
  \end{aligned}\right.
\end{equation}
and
\begin{equation}\label{BSDE-o}\left\{
  \begin{aligned}
    do_{ij}(t) &= - \Bigg[ \frac{\pt^2 \overline H^g_0}{\pt x^i \pt x^j} (X(t),p(t),o(t),t) \\
    &\qquad\ + o_{ik}(t) \left( o_{jl}(t) \frac{\pt^2 H_0}{\pt p_k \pt p_l} (X(t),p(t),t) 
    + 2\frac{\pt^2 H_0}{\pt x^j \pt p_k} (X(t),p(t),t) \right) \\
    &\qquad\ - o_{ik}(t) \left( \pt_j g^{lm} \Gamma_{lm}^k + g^{lm} \pt_j \Gamma_{lm}^k \right) (X(t)) \\
    &\qquad\ + \sum_{r=1}^d \left( \pt_j \sigma_r^k (X(t)) Z_{ikr}(t) + \pt_j \sigma_r^l (X(t)) Z_{ilr}(t) \right)
    \Bigg] dt + Z_{ijr}(t) dW^r(t), \\
    o_{ij}(T) &= \pt_i \pt_j S_T(X(T)),
  \end{aligned}\right.
\end{equation}
which are called first- and second-order adjoint equation, respectively. The unknowns in \eqref{BSDE-p} and \eqref{BSDE-o} are the pairs $(p,z)$ and $(o,Z)$, respectively. Suppose that $p_i(t) = p_i(t,X(t))$ and $o_{ij}(t) = o_{ij}(t,X(t))$ for time-dependent 2nd-order form $(p,o)$ that satisfies 2nd-order Maxwell relations \eqref{non-degenerate}. Then
\begin{equation*}
  z_{ir} = \frac{\pt p_i}{\pt x^j} \sigma^j_r, \quad Z_{ijr} = \frac{\pt o_{ij}}{\pt x^k} \sigma^k_r.
\end{equation*}
Plugging them into \eqref{BSDE-p} and \eqref{BSDE-o}, we get
\begin{align}
  D_i p &= \frac{1}{2} \left( \pt_i g^{kl} \Gamma_{kl}^j + g^{kl} \pt_i \Gamma_{kl}^j \right) p_j - \frac{1}{2} \pt_i g^{jk} \frac{\pt p_j}{\pt x^k} + \frac{\pt L_0}{\pt x^i} = -\frac{\pt \overline H^g_0}{\pt x^i}, \label{eqn-33} \\
  D_{ij} o &= - \left( \frac{\pt^2 \overline H^g_0}{\pt x^i \pt x^j} + \frac{\pt p_k}{\pt x^i} \frac{\pt p_l}{\pt x^j} \frac{\pt^2 \overline H^g_0}{\pt p_k \pt p_l} 
  + 2 \frac{\pt p_k}{\pt x^i} \frac{\pt^2 \overline H^g_0}{\pt x^j \pt p_k} + 2 \frac{\pt o_{kl}}{\pt x^i} \frac{\pt^2 \overline H^g_0}{\pt x^j \pt o_{kl}} 
    \right). \notag
\end{align}
These coincide with the corresponding equations in the S-H system \eqref{stoch-Hamilton-eqns} for 2nd-order Hamiltonian $\overline H^g_0$. The first equality of \eqref{eqn-33} also recovers \eqref{eqn-34}.

\subsection{Stochastic variational symmetries}

\begin{definition}
  Given an action functional $\mathcal S$ as in \eqref{action}, a bundle automorphism $F$ on $(\R\times M, \pi, \R)$ projecting to $F^0$ is called a variational symmetry of $\mathcal S$ if, whenever $[t_1,t_2]$ is a subinterval of $[0,T]$, we have $\mathcal S[F\cdot X,F^0(t_1),F^0(t_2)] = \mathcal S[X,t_1,t_2]$. A $\pi$-projectable vector field $V$ on $\R\times M$ is called an infinitesimal variational symmetry of $\mathcal S$, if its flow consists of variational symmetries of $\mathcal S$.
\end{definition}

\begin{lemma}
  The $\pi$-projectable vector field $V$ of the form \eqref{general-vf} is an infinitesimal variational symmetry of $\mathcal S$ if and only if
  \begin{equation*}
    \left[(j^\nabla V)(L_0) + L_0\dot V^0\right] (j^\nabla_t X), \quad t\in[0,T]
  \end{equation*}
  is a martingale, for all $X\in I_0^T(M)$.
\end{lemma}

\begin{proof}
  As in the proof of Theorem \ref{prog-proj-vf}, we let $\psi = \{(\psi^0_\e, \bar \psi_\e)\}_{\e\in\R}$ be the flow generated by $V$, and denote $\tilde X_\e = \psi_\e \cdot X$. Then, by a change of variable $s=\psi^0_\e(t)$,
  \begin{equation*}
    \begin{split}
      \mathcal S[\tilde X_\e, \psi^0_\e(t_1), \psi^0_\e(t_2)] &= \E \int_{\psi^0_\e(t_1)}^{\psi^0_\e(t_2)} L_0\left(s, \tilde X_\e(s), D_\nabla \tilde X_\e(s) \right) ds \\
      &= \E \int_{t_1}^{t_2} L_0\left(\psi^0_\e(t), \bar \psi_\e(t,X(t)), D_\nabla \tilde X_\e(\psi^0_\e(t)) \right) \frac{d\psi^0_\e}{dt}(t) dt.
    \end{split}
  \end{equation*}
  Since for all $[t_1,t_2]\subset [0,T]$ and each $\e$, $\mathcal S[\tilde X_\e, \psi^0_\e(t_1), \psi^0_\e(t_2)] = S[X,t_1,t_2]$, we have that the difference
  \begin{equation*}
    L_0\left(\psi^0_\e(t), \bar \psi_\e(t,X(t)), D_\nabla \tilde X_\e(\psi^0_\e(t)) \right) \frac{d\psi^0_\e}{dt}(t) - L_0\left(t, X(t), D_\nabla X(t) \right).
  \end{equation*}
  is a martingale (depending on $\e$).   Taking derivatives with respect to $\e$ and evaluating at $\e=0$ for the above equality, and recalling that $j^\nabla V = \frac{d}{d\e} \big|_{\e=0}j^\nabla\psi_\e$, we can obtain the desired result.
\end{proof}

\begin{definition}
  Given a smooth function $\Phi:\R\times M\to\R$. A $\pi$-projectable vector field $V$ on $\R\times M$ is called an infinitesimal $\Phi$-divergence symmetry of $\mathcal S$, if
  \begin{equation*}
    \left[(j^\nabla V)(L_0) + L_0\dot V^0\right] (j^\nabla_t X) = \D_t \Phi (j^\nabla_t X),
  \end{equation*}
  for all $X\in I_0^T(M)$ and $t\in[0,T]$.
\end{definition}

Recall that for the $\pi$-projectable vector field $V$ of the form \eqref{general-vf}, we denote $\bar V = V^i \vf{x^i}$, as in Corollary \ref{prog-tang}.

\begin{proposition}\label{div-symm}
  A vector field $V$ of the form \eqref{general-vf} is an infinitesimal $\Phi$-divergence symmetry of $\mathcal S$ if and only if
  \begin{equation*}
    V^0 \pt_t L_0 + d_x L_0(\bar V) + d_{\dot x} L_0 \left( \frac{\overline\D \bar V}{dt} \right) - \dot V^0 E_0 = \D_t \Phi.
  \end{equation*}
\end{proposition}

\begin{proof}
  It follows from Corollary \ref{prog-tang} and \eqref{diff-dot-x}, \eqref{diff-x} that
  \begin{equation*}
    \begin{split}
      \D_t \Phi &= V^0 \pt_t L_0 + V^i \pt_i L_0 + \left[ \left( \pt_t + \dot x^j \pt_j \right) V^i + \ts{\frac{1}{2}} \left( \Delta \bar V + \Ric(\bar V) \right)^i - \dot V^0 \dot x^i \right] \pt_{\dot x^i} L_0 +\dot V^0 L_0 \\
      &= V^0 \pt_t L_0 + d_x L_0(\bar V) + d_{\dot x} L_0 \left( \left( \pt_t + \nabla_{\dot x} + \ts{\frac{1}{2}} \Delta_{\text{LD}} \right) \bar V \right) - \dot V^0 \left( \dot x^i \pt_{\dot x^i} L_0 - L_0 \right) \\
      &= V^0 \pt_t L_0 + d_x L_0(\bar V) + d_{\dot x} L_0 \left( \frac{\overline\D \bar V}{dt} \right) - \dot V^0 E_0.
    \end{split}
  \end{equation*}
  This concludes the proof.
\end{proof}

\begin{corollary}
  Let $L_0:\R\times TM\to\R$ be a hyperregular Lagrangian.
  Let $V$ be a vector field of the form \eqref{general-vf}. Given a smooth function $\Phi:\R\times M\to\R$, define the $\Phi$-extension of $V$ by
  \begin{equation}\label{Phi-ext}
    V_\Phi = V + \Phi\vf{u},
  \end{equation}
  which is a vector field on $\R\times M\times\R$. Suppose that $V$ satisfies
  \begin{equation*}
    \frac{1}{2} \dot V^0 \Delta S = g^{ij} \nabla^2_{\pt_i, \nabla_{\pt_j} \bar V} S,
  \end{equation*}
  for $S$ the solution of the Hamilton-Jacobi-Bellman equation \eqref{HJB-4} associated with $L_0$ (for $f\equiv0$). Then, $V$ is an infinitesimal $\Phi$-divergence symmetry of $\mathcal S$ if and only if $V_\Phi$ is an infinitesimal symmetry of equation \eqref{HJB-4}.
\end{corollary}

\begin{proof}
  By the classical jet bundle theory, we know that $V$ is an infinitesimal symmetry of Hamilton-Jacobi-Bellman equation \eqref{HJB-4} if and only if \cite[Theorem 2.31]{Olv98}
  \begin{equation}\label{symm-HJB}
    j^{1,2} V \left( u_t + 
    H_0(x, (u_i), t) + \ts{\frac{1}{2}} g^{ij}(x) u_{ij} - \ts{\frac{1}{2}} g^{ij}(x) \Gamma_{ij}^k(x) u_k 
    \right) = 0,
  \end{equation}
  where
  \begin{equation*}
    j^{1,2} V = V^0 \vf t + V^i \vf{x^i} + \Phi \vf{u} + V_t \frac{\pt}{\pt u_t} + V_i \frac{\pt}{\pt u_i} + V_{ij} \frac{\pt}{\pt u_{ij}},
  \end{equation*}
  with coefficients given by \cite[Theorem 2.36 or Example 2.38]{Olv98}
  \begin{equation*}
    V_t = \frac{\pt \Phi}{\pt t} - \dot V^0 u_t - \frac{\pt V^i}{\pt t} u_i, \quad V_i = \frac{\pt \Phi}{\pt x^i} - \frac{\pt V^j}{\pt x^i} u_j, \quad V_{ij} = \frac{\pt^2 \Phi}{\pt x^i\pt x^j} - \frac{\pt^2 V^k}{\pt x^i\pt x^j} u_k - \frac{\pt V^k}{\pt x^i} u_{jk} - \frac{\pt V^k}{\pt x^j} u_{ik}.
  \end{equation*}
  Moreover, the jet coordinates $(u_t, u_i, u_{ij})$ satisfy
  $$(u_t, u_i, u_{ij}) = (\pt_t S, \pt_i S, \pt_{ij} S) = (-E_0- \ts{\frac{1}{2}}\Delta S, \pt_{\dot x^i} L_0, \pt_{ij} S),$$
  where we recall $dS = d_{\dot x}L_0$ from equation \eqref{eqn-10} and Remark \ref{remark-4}, and also that $\pt_t S = - H_0(dS,t) - \frac{1}{2} \Delta S = - E_0- \frac{1}{2} \Delta S$.
  Plugging these into \eqref{symm-HJB} and using the fact that $\pt_t H_0 = - \pt_t L_0$ and $\pt_{x^i} H_0 = - \pt_{x^i} L_0$ due to classical Legendre transform, we have
  \begin{equation}\label{criterion-HJB-symm}
    \begin{split}
      0=&\ V^0 \pt_t H_0 
      + V^i \left( \pt_{x^i} H_0 
      + \ts{\frac{1}{2}} \pt_i g^{jk} u_{jk} - \ts{\frac{1}{2}} \pt_i g^{jk} \Gamma_{jk}^l u_l - \ts{\frac{1}{2}} g^{jk} \pt_i \Gamma_{jk}^l u_l 
      \right) + \left( \pt_t \Phi - \dot V^0 u_t - \pt_t V^i u_i \right) \\
      &\ + \left( \pt_i \Phi - \pt_i V^l u_l \right) \left( \pt_{p_i} H_0 
      - \ts{\frac{1}{2}} g^{jk} \Gamma_{jk}^i\right) + \ts{\frac{1}{2}} g^{ij} \left( \pt_i \pt_j \Phi - \pt_i \pt_j V^k u_k - \pt_i V^k u_{jk} - \pt_j V^k u_{ik} \right) \\
      =&\ V^0 \pt_t H_0 
      + V^i \pt_{x^i} H_0 
      - \left( \pt_t + \pt_{p_i} H_0 
      \pt_j \right) V^i u_i - \ts{\frac{1}{2}} g^{ij} \left( \pt_i \pt_j V^k - \Gamma_{ij}^l \pt_l V^k + 2 \Gamma_{il}^k \pt_j V^l + \pt_l \Gamma_{ij}^k V^l \right) u_k \\
      &\ - \dot V^0 u_t - g^{ij} \left( \pt_j V^k + \Gamma_{jm}^k V^m \right) \left( u_{ik} - \Gamma_{ik}^l u_l \right) + \left[ \pt_t \Phi + \left( 
      \pt_{p_i} H_0 - \ts{\frac{1}{2}} g^{jk} \Gamma_{jk}^i\right) \pt_i \Phi + \ts{\frac{1}{2}} g^{ij} \pt_i \pt_j \Phi \right] \\
      =&\ - V^0 \pt_t L_0 - V^i \pt_{x^i} L_0 - \left( \pt_t + \dot x^j \pt_j \right) V^i \pt_{\dot x^i} L_0 - \ts{\frac{1}{2}} \left[ \Delta \bar V + \Ric(\bar V) \right]^k \pt_{\dot x^k} L_0 \\
      &\ + \dot V^0 \left(E_0 + \ts{\frac{1}{2}} \Delta S\right) - g^{ij} \nabla^2_{\pt_i, \nabla_{\pt_j} \bar V} S + \left( \pt_t \Phi + \dot x^i \pt_i \Phi + \ts{\frac{1}{2}} \Delta \Phi \right) \\
      =&\ - \left[ V^0 \pt_t L_0 + d_x L_0(\bar V) + d_{\dot x} L_0 \left( \frac{\overline\D \bar V}{dt} \right) - \dot V^0 E_0 \right] + \left( \ts{\frac{1}{2}} \dot V^0 \Delta S - g^{ij} \nabla^2_{\pt_i, \nabla_{\pt_j} \bar V} S \right) + \D_t \Phi,
    \end{split}
  \end{equation}
  where, in the last equality, we used the fact that $(QX)^{ij}(t) = g^{ij}(X(t))$ to derive $\D_t \Phi$. The result then follows from Proposition \ref{div-symm}.
\end{proof}

\begin{theorem}[Stochastic Noether's theorem]
  Let $L_0:\R\times TM\to\R$ be a hyperregular Lagrangian. 
  Suppose that the vector field $V_\Phi$ in \eqref{Phi-ext} is an infinitesimal symmetry of the Hamilton-Jacobi-Bellman equation \eqref{HJB-4} associated with $L_0$ (with $f\equiv0$). Then the following stochastic conservation law holds for the stochastic Euler-Lagrange equation \eqref{stoch-EL},
  \begin{equation}\label{cond-noether}
    \D_t \left[ V^i \pt_{\dot x^i} L_0 - V^0 E - \Phi \right] = 0.
  \end{equation}
\end{theorem}

\begin{proof}
  Recall that $dS = d_{\dot x}L_0$ and $\pt_t S = - E_0- \frac{1}{2} \Delta S = -E$. By applying Lemma \ref{stoch-mean-derv-prop}.(iv) and \eqref{stoch-EL}, as well as the fact that $(QX)^{ij}(t) = g^{ij}(X(t))$, we have
  \begin{equation*}
    \begin{split}
      \D_t \left[ d_{\dot x} L_0(\bar V) \right] &= d_{\dot x}L_0 \left( \frac{\overline\D \bar V}{dt} \right) + \frac{\overline\D (d_{\dot x}L_0)}{dt} (\bar V) + (QX)^{ij} (\nabla_{\pt_i} (d_{\dot x} L_0) ) (\nabla_{\pt_j} \bar V) \\
      &= d_{\dot x}L_0 \left( \frac{\overline\D \bar V}{dt} \right) + d_x L_0 (\bar V) + g^{ij} \nabla^2_{\pt_i,\nabla_{\pt_j} \bar V} S.
    \end{split}
  \end{equation*}
  Then, we use HJB equation \eqref{HJB-4} (with $f\equiv0$) and the classical Legendre transform $H_0 = d_{\dot x}L_0\cdot \dot x - L_0$ to derive
  \begin{equation*}
    \begin{split}
      \D_t E &= - \D_t \pt_t S = - \pt_t \left( \pt_t + \nabla_{\dot x} + \ts{\frac{1}{2}} \Delta \right) S = - \pt_t \left[ dS \cdot \dot x + \left( \pt_t + \ts{\frac{1}{2}} \Delta \right) S \right] \\
      & = - \pt_t \left( d_{\dot x}L_0\cdot \dot x - H_0 \right) = - \pt_t L_0.
    \end{split}
  \end{equation*}
  Combining these with the S-EL equation \eqref{stoch-EL} and the criterion \eqref{criterion-HJB-symm} for symmetries of the HJB equation \eqref{HJB-2}, we have
  \begin{equation*}
    \begin{split}
      \D_t \left[ V^i \pt_{\dot x^i} L_0 - V^0 E - \Phi \right] &= \D_t \left[ d_{\dot x} L_0(\bar V) \right] - \dot V^0 E - V^0 \D_t E - \D_t \Phi \\
      &= d_{\dot x}L_0 \left( \frac{\overline\D \bar V}{dt} \right) + d_x L_0 (\bar V) + g^{ij} \nabla^2_{\pt_i,\nabla_{\pt_j} \bar V} S - \dot V^0 \left(E_0 + \ts{\frac{1}{2}} \Delta S\right) + V^0 \pt_t L_0 - \D_t \Phi \\
      &= 0.
    \end{split}
  \end{equation*}
  The result follows.
\end{proof}

\begin{remark}
  (i). In stochastic Hamiltonian formalism, \eqref{cond-noether} reads as $\D_t \left[ V^i p_i - V^0 H - \Phi \right] = 0$.

  (ii). The stochastic conservation law \eqref{stoch-cons-energy} of a time-independent $g$-canonical 2nd-order Hamiltonian $H = \overline H_0^g$ can be regarded as a special case of the above stochastic Noether's theorem. Indeed, consider the infinitesimal unit time translation $V = \frac{\pt}{\pt t}$, i.e., $V^0 = 1$, $\bar V = 0$, $\Phi =0$. Then, the criterion \eqref{criterion-HJB-symm} reduces to $0 = \pt_t L_0 = -\pt_t H_0$, which means that $H = \overline H_0^g$ is time-independent. The resulting stochastic conservation law is $\D_t E = \D_t H =0$.
\end{remark}

Applying stochastic Noether's theorem to Schr\"odinger's problem of Section \ref{sec-7-3}, we have the following corollary. Its Euclidean case with zero vector potential (i.e., $b\equiv0$) has already been formulated in \cite{TZ97}.

\begin{corollary}[Stochastic Noether's theorem for Schr\"odinger's problem]
  Let $L_0$ be the Lagrangian given in \eqref{Lagrangian}. Suppose that the vector field $V_\Phi$ in \eqref{Phi-ext} is an infinitesimal symmetry of Hamilton-Jacobi-Bellman equation \eqref{HJB-2} with $f\equiv0$. Then the following stochastic conservation law holds for the coordinate process of the solution of Schr\"odinger's problem in \eqref{Benamou-Brenier},
  \begin{equation*}
    \D_t \left[ g_{ij} \left( D_\nabla^j x - b^j \right) V^i - V^0 \left( E_0 + \ts{\frac{1}{2}} \Delta S \right) - \Phi \right] = 0,
  \end{equation*}
  where $E_0$ is the classical energy given in \eqref{classical-energy} and $S$ is the solution of \eqref{HJB-2}.
\end{corollary}

\begin{appendices}

\section{Mixed-order tangent and cotangent bundles}\label{app-1}

\subsection{Mixed-order tangent and cotangent maps}\label{sec-A-1}

Clearly, the mixed-order tangent bundle $T \R \times \mathcal T^S M$ is a subbundle of the totally second-order tangent bundle $\mathcal T^S (\R\times M)$, and contains the tangent bundle $T (\R\times M) \cong T \R \times T M$ as a subbundle. Similar properties hold for the mixed-order cotangent bundle. 

It is easy to verify that the mixed-order tangent bundle can be characterized as follows:
\begin{equation*}
  T \R \times \mathcal T^S M = \{A\in \mathcal T^S (\R\times M): \pi^S_*(A) \in T \R \}.
\end{equation*}
We also define the stochastic analog of the vertical bundle as
\begin{equation*}
  V^S\pi = \{ A\in T \R \times \mathcal T^S M: \pi^S_*(A) = 0 \}.
\end{equation*}
Then, it is easy to see that $V^S\pi \cong \R\times \mathcal T^S M$.

Given a smooth map $F: \R\times M \to \R\times N$, we can define its second-order pushforward $F_*^S$ as in Definitions \ref{push-pull-map-point} and \ref{push-pull-map}, so that $F_*^S$ is a bundle homomorphism from $\tau^S_{\R\times M}$ to $\tau^S_{\R\times N}$. In general, $F_*^S$ neither maps the mixed-order tangent bundle to the mixed-order tangent bundle, nor maps the vertical bundle to the vertical bundle. But if $F$ is projectable, then it does.

\begin{lemma}\label{bd-morph-mixed-vertical}
  Let $M$ and $N$ be two smooth manifolds and $M$ be connected. Let $F: \R\times M\to \R\times N$ be a smooth map. Then the following statements are equivalent: \\
  (i) $F$ is a bundle homomorphism from $(\R\times M, \pi, \R)$ to $(\R\times N, \rho, \R)$; \\
  (ii) $F^S_* (T \R \times \mathcal T^S M ) \subset T \R \times \mathcal T^S N$; \\
  (iii) $F^S_* (V^S \pi ) \subset V^S \rho$.
\end{lemma}

\begin{proof}
  We first prove that (i) implies both (ii) and (iii). Suppose that $F$ is a bundle homomorphism projecting to $F^0$. Then, $\rho\circ F = F^0\circ \pi$ and hence, for any $A\in \mathcal T^S (\R\times M)$,
  \begin{equation*}
    \rho^S_*(F^S_* (A)) = (F^0)^S_*\pi^S_*(A).
  \end{equation*}
  If $A\in T \R \times \mathcal T^S M$, then $\pi^S_*(A) \in T \R$ and thus $\rho^S_*(F^S_* (A)) \in (F^0)^S_*(T \R) = (F^0)_*(T \R) \subset T \R$. This implies $F^S_* (A)\in T \R \times \mathcal T^S N$. If $A\in V^S \pi$, then $\pi^S_*(A) = 0$, it follows $\rho^S_*(F^S_* (A))=0$ and therefore $F^S_* (A)\in V^S \rho$.

  Next we prove either (ii) or (iii) implies (i). Choose local coordinates $(t,x^i)$ around $(t_0,q) \in \R\times M$ and $(s,y^j)$ around $F(t_0,q)$. Suppose $F$ has a local expression $F=(F^0, \bar F^j)$. Let $A\in T \R \times \mathcal T^S M|_{(t_0,q)}$ having the following local expression:
  \begin{equation}\label{eqn-3}
    A = A^0 \vf t\bigg|_{t_0} + A^i \vf{x^i}\bigg|_q + A^{jk} \frac{\pt^2}{\pt x^j \pt x^k}\bigg|_q.
  \end{equation}
  Then, Lemma \ref{push-pull-map-prop} yields
  \begin{equation*}
    \begin{split}
      F^S_* A =&\ (A F^0) \frac{\partial}{\partial s}\bigg|_{F^0(t_0,q)} + (A \bar F^i) \frac{\partial}{\partial y^i}\bigg|_{\bar F(t_0,q)} + \Gamma_A(F^0, F^0) \frac{\partial^2}{\partial s^2}\bigg|_{F^0(t_0,q)} \\
      &\ + \Gamma_A(\bar F^i, \bar F^j) \frac{\partial^2}{\partial y^i\partial y^j}\bigg|_{\bar F(t_0,q)} + 2 \Gamma_A(F^0, \bar F^i) \frac{\partial^2}{\partial s\partial y^i}\bigg|_{F(t_0,q)}.
    \end{split}
  \end{equation*}
  If (ii) holds, then $F^S_*(A) \in T \R \times \mathcal T^S N|_{F(t_0,q)}$. It then follows
  \begin{equation}\label{eqn-4}
    \Gamma_A(F^0, F^0) = A^{jk} \frac{\pt F^0}{\pt x^j} \frac{\pt F^0}{\pt x^k} = 0 \quad\text{and}\quad \Gamma_A(F^0, \bar F^i) = A^{jk} \frac{\pt F^0}{\pt x^j} \frac{\pt \bar F^i}{\pt x^k} = 0.
  \end{equation}
  Since $A$ is arbitrary, we know that $\frac{\pt F^0}{\pt x^i} = 0$ for all $i$. Then, by the connectness of $M$, $F^0$ is independent of $q\in M$. This implies that $F$ is a bundle homomorphism. Now assume that $A\in V^S M|_{(t_0,q)}$ has a local expression in \eqref{eqn-3} with $A^0 = 0$. If (iii) holds, then $F^S_*(A) \in V^S N|_{F(t_0,q)}$. This amounts to \eqref{eqn-4} together with
  \begin{equation*}
    A F^0 = A^i \frac{\pt F^0}{\pt x^i} + A^{jk} \frac{\pt^2 F^0}{\pt x^j \pt x^k} = 0.
  \end{equation*}
  Again, the arbitrariness of $A$ yields that $\frac{\pt F^0}{\pt x^i} = 0$ for all $i$. Thus, $F$ is a bundle homomorphism.
\end{proof}

It is easy to deduce from the proof that if $F = (F^0, \bar F)$ is a bundle homomorphism from $\pi$ to $\rho$, then $F^S_*|_{T \R \times \mathcal T^S M}$ is a bundle homomorphism from $\tau_R\times\tau^S_{M}$ to $\tau_\R\times \tau^S_{N}$.

When $F: \R\times M \to \R\times N$ is a diffeomorphism, we can also consider the second-order pullback map $F^{S*}$ which is a bundle homomorphism from $\tau^{S*}_{\R\times M}$ to $\tau^{S*}_{\R\times N}$. But when we restrict $F^{S*}$ to the mixed-order cotangent bundle $T^* \R \times \mathcal T^{S*} M$, there are difficulties. We can check that even if $F$ is a bundle homomorphism, $F^{S*}$ does not necessarily map $T^* \R \times \mathcal T^{S*} M$ into $T^* \R \times \mathcal T^{S*} M$. The reason is basically that the restrictions of second-order pullbacks to the cotangent bundle do not coincide with usual pullbacks. To overcome this, we consider the dual map of $F^S_*|_{T \R \times \mathcal T^S M}$. This motivates the following definition, which contrasts with Definitions \ref{push-pull-map-point} and \ref{push-pull-map}.

\begin{definition}[Mixed-order pushforward and pullback]
  Let $F$ be a bundle homomorphism from $(\R\times M, \pi, \R)$ to $(\R\times N, \rho, \R)$. The mixed-order tangent map of $F$ at $(t,q)\in\R\times M$ is the linear map $d^\circ F_{(t,q)}: T \R\times \mathcal T^S M|_{(t,q)} \to T \R\times\mathcal T^S N|_{F(t,q)}$ defined by
  \begin{equation*}
    d^\circ F_{(t,q)} = d^2 F_{(t,q)}|_{T_t \R\times \mathcal T^S_q M}.
  \end{equation*}
  The mixed-order cotangent map of $F$ at $(t,q)\in\R\times M$ is the linear map $d^\circ F^*_{(t,q)}: T^* \R\times\mathcal T^{S*} N|_{F(t,q)} \to T^* \R\times\mathcal T^{S*} M|_{(t,q)}$ dual to $d^\circ F_{(t,q)}$, that is,
  \begin{equation*}
    d^\circ F^*_{(t,q)} (\alpha) (A) = \alpha (d^\circ F_{(t,q)}( A)), \quad\text{for } A\in T_t \R\times \mathcal T^S_q M, \alpha\in T^* \R\times\mathcal T^{S*} N|_{F(t,q)}.
  \end{equation*}
  The mixed-order pushforward by $F$ is the bundle homomorphism $F^R_*: (T \R \times \mathcal T^S M, \tau_\R\times \tau^S_M, \R\times M) \to (T \R \times \mathcal T^S N, \tau_\R\times \tau^S_N, \R\times N)$ defined by
  $$F^R_*|_{T_t \R \times \mathcal T_q^S M} = d^\circ F_{(t,q)}.$$
  Given a mixed-order form $\alpha$ on $\R\times N$, the mixed-order pullback of $\alpha$ by $F$ is the mixed-order form $F^{R*}\alpha$ on $\R\times M$ defined by
  \begin{equation*}
    (F^{R*}\alpha)_{(t,q)} = d^\circ F^*_{(t,q)} \left( \alpha_{F(t,q)} \right), \quad (t,q)\in \R\times M.
  \end{equation*}
  If, moreover, $F$ is a bundle isomorphism, then the mixed-order pullback by $F$ is the bundle isomorphism $F^{R*}: (T \R \times \mathcal T^{R*} N, \tau_\R\times \tau^{S*}_N, \R\times N) \to (T \R \times \mathcal T^{S*} M, \tau_\R\times \tau^{S*}_M, \R\times M)$ defined by
  $$F^{R*}|_{\mathcal T_s \R \times \mathcal T_{q'}^{S*} N} = d^\circ F^*_{F^{-1}(s,q')}.$$
  Given a mixed-order vector field $A$ on $\R\times M$, the mixed-order pushforward of $A$ by $F$ is the mixed-order vector field $F^R_*A$ on $\R\times N$ defined by
  \begin{equation*}
    (F^R_*A)_{(s,q')} = d^\circ F_{F^{-1}(s,q')} \left( A_{F^{-1}(s,q')} \right), \quad (s,q')\in\R\times N.
  \end{equation*}
\end{definition}

Clearly, the mixed-order pushforward $F^R_*$ is nothing but $F^S_*|_{T \R \times \mathcal T^S M}$. Write $F = (F^0, \bar F)$. Then, in local coordinates, $F^R_*$ acts on $A$ of \eqref{eqn-3} as follows:
\begin{equation}\label{local-rep-push-mixed}
  \begin{split}
    F^R_* A =&\ A^0 \frac{dF^0}{dt}(t_0) \frac{\partial}{\partial s}\bigg|_{F^0(t_0)} + (A \bar F^i)(t_0,q) \frac{\partial}{\partial y^i}\bigg|_{\bar F(t_0,q)} + A^{kl} \frac{\pt \bar F^i}{\pt x^k} \frac{\pt \bar F^j}{\pt x^l} (t_0,q) \frac{\partial^2}{\partial y^i\partial y^j}\bigg|_{\bar F(t_0,q)} \\
    =&\ A^0 \frac{dF^0}{dt}(t_0) \frac{\partial}{\partial s}\bigg|_{F^0(t_0)} + \left[ A^0 \frac{\pt\bar F^i}{\pt t}(t_0,q) + A^j \frac{\pt\bar F^i}{\pt x^j}(t_0,q) + A^{jk} \frac{\pt^2 \bar F^i}{\pt x^j \pt x^k} (t_0,q)\right] \frac{\partial}{\partial y^i}\bigg|_{\bar F(t_0,q)} \\
    &\ + A^{kl} \frac{\pt \bar F^i}{\pt x^k} \frac{\pt \bar F^j}{\pt x^l} (t_0,q) \frac{\partial^2}{\partial y^i\partial y^j}\bigg|_{\bar F(t_0,q)}.
  \end{split}
\end{equation}
And $F^{R*}$ acts on the mixed-order cotangent vector $\alpha = \alpha_0 ds|_{F^0(t_0)} + \alpha_i d^2 y^i|_{\bar F(t_0,q)} + \alpha_{ij}dy^i\cdot dy^j|_{\bar F(t_0,q)} \in T \R \times \mathcal T^{S*} N|_{F(t_0,q)}$ by
\begin{equation}\label{local-rep-pull-mixed}
  \begin{split}
    F^{R*}\alpha =&\ \alpha_0 \frac{dF^0}{dt}(t_0) dt|_{t_0} + \alpha_i d^\circ F^i|_{(t_0,q)} + \alpha_{ij}\frac{\pt\bar F^i}{\pt x^k}\frac{d\bar F^i}{dx^l} (t_0,q) dx^k \cdot dx^l|_q \\
    =&\ \left[ \alpha_0 \frac{dF^0}{dt}(t_0) + \alpha_i \frac{\pt\bar F^i}{\pt t}(t_0,q) \right] dt|_{t_0} + \alpha_i \frac{\pt\bar F^i}{\pt x^j}(t_0,q) d^2x^j|_q \\
    &\ + \left[ \frac{\alpha_i}{2} \frac{\pt^2\bar F^i}{\pt x^k\pt x^l} (t_0,q) + \alpha_{ij}\frac{\pt\bar F^i}{\pt x^k}\frac{d\bar F^j}{dx^l} (t_0,q) \right] dx^k \cdot dx^l|_q.
  \end{split}
\end{equation}
By virtue of these local expressions, one easily deduce that
\begin{equation*}
  F^R_*|_{T_t \R \times \mathcal T^S_q M} = F(q)_*|_{T_t \R} \times \bar F(t)^S_*|_{\mathcal T^S_q M}, \quad F^{R*}|_{\mathcal T^*_s \R \times \mathcal T^{S*}_{q'} N} = F(q)^*|_{T^*_s \R} \times \bar F(t)^{S*}|_{\mathcal T^{S*}_{q'} N}.
\end{equation*}
And in turn, these verify the linearity of $F^R_*$ and $F^{R*}$. The following property is easy to check.

\begin{lemma}\label{push-mixed}
  Let $F$ be a bundle isomorphism from $(\R\times M, \pi, \R)$ to $(\R\times N, \rho, \R)$ and $A$ be a mixed-order vector field. Let $f$ be a smooth functions on $\R\times N$. Then $((F^R_*A)f)\circ F = A(f\circ F)$.
\end{lemma}

\subsection{Pushforwards of generators}

A smooth map $F: M\to N$ can be associated naturally with a bundle homomorphism $\id_\R\times F: (\R\times M, \pi, \R) \to (\R\times N, \rho, \R)$ that projects to the identity on $\R$. In this case, the pushforward of a diffusion $X$ by $\id_\R\times F$ is just $(\id_\R\times F)\cdot X = F(X)$. The stochastic prolongations of the bundle homomorphism $\id_\R\times F$ is then
\begin{equation*}
  j (\id_\R\times F) (j_{(t,q)} X) = j_{(t,F(q))} (F(X)).
\end{equation*}

\begin{corollary}\label{push-generator}
  Let $F: M\to N$ be a diffeomorphism. If a diffusion $X$ on $M$ has a generator $A=(A_t)$, then the process $F(X)$ is a diffusion on $N$, with generator $F^S_*A = (F^S_*A_t)$.
\end{corollary}
\begin{proof}
  Assume $X\in I_{t_0}(M)$. For every $f\in C^\infty(N)$, $f\circ F\in C^\infty(M)$, by the assumption, we have
  \begin{equation*}
    \begin{split}
      &\ f\circ F(X(t)) - f\circ F(X(t_0)) - \int_{t_0}^t A_s (f\circ F)(X(s)) ds \\
      =&\ f(F(X(t))) - f(F(X(t_0))) - \int_{t_0}^t\left((F^S_*A_s)f\right)(F(X(s))) ds
    \end{split}
  \end{equation*}
  is a real-valued continuous $\{\Pred_t\}$-martingale. This proves that $F(X)\in I_{t_0}(N)$ has generator $F^S_*A$.
\end{proof}

This corollary together with the identification between $\R\times \mathcal T^S M$ and $\R\times \mathcal T^E M$ in \eqref{diff-1} and \eqref{diff-2}, give rise to the relation between prolongations and pushforwards as follows:
\begin{equation*}
  \begin{split}
    j (\id_\R\times F) (t, A_q) &= j (\id_\R\times F) (j_{(t,q)} X^A) = j_{(t,F(q))} (F\circ X^A) = \left( t, (F^S_*A_t)_{F(q)} \right) \\
    &= \left( t, d^2 F_q(A_{(t,q)})\right ) = \left( t, d^2 F_q(A_q) \right) = (\id_\R\times F^S_*) (t, A_q),
  \end{split}
\end{equation*}
so that $j (\id_\R\times F) = \id_\R\times F^S_*$.

The following corollary is an extension of Corollary \ref{push-generator} and a straightforward consequence of Lemma \ref{push-mixed-coord}. Here, we will present another proof, using notions of Appendix \ref{sec-A-1}.

\begin{corollary}\label{push-ext-generator}
  Let $F$ be a bundle isomorphism from $(\R\times M, \pi, \R)$ to $(\R\times N, \rho, \R)$ projecting to $F^0$. If $X$ is a diffusion on $M$ with respect to $\{\Pred_t\}$ and has a extended generator $\vf t + A$ where $A$ is a time-dependent second-order vector field, then the pushforward $F\cdot X$ is a diffusion on $N$ with respect to $\{\F_{(F^0)^{-1}(s)}\}$, with extended generator
  $$\frac{d(F^0)^{-1}}{ds} F^R_*\left( \vf t + A \right).$$
\end{corollary}

\begin{proof}
  Assume that $X\in I_{t_0}(M)$ and $F = (F^0, \bar F)$. For every $f\in C^\infty(\R\times N)$, Lemma \ref{push-mixed} yields that the process
  \begin{equation*}
    \begin{split}
      &\ f\circ F(t, X(t)) - f\circ F(t_0, X(t_0)) - \int_{t_0}^t \left( \vf t + A \right) (f\circ F) (u, X(u)) du \\
      =&\ f\left(F^0(t), \bar F(t, X(t))\right) - f\left(F^0(t_0), \bar F(t_0, X(t_0))\right) - \int_{t_0}^t F^R_*\left( \vf t + A \right)f \left(F^0(u), \bar F(u, X(u))\right) du
    \end{split}
  \end{equation*}
  is a continuous $\{\Pred_t\}$-martingale. Denote $s_0 = F^0(t_0)$. By substituting $t = (F^0)^{-1}(s)$ which can be done because $F^0$ is an isomorphism, and using the change of variable $u = (F^0)^{-1}(v)$, and recalling that $F\cdot X(s) = \bar F\left( (F^0)^{-1}(s), X((F^0)^{-1}(s)) \right)$, the process
  \begin{equation*}
    \begin{split}
      &\ f(s,F\cdot X(s)) - f(s_0,F\cdot X(s_0)) - \int_{(F^0)^{-1}(s_0)}^{(F^0)^{-1}(s)} F^R_*\left( \vf t + A \right)f \left(F^0(u), \bar F(u, X(u))\right) du \\
      =&\ f(s,F\cdot X(s)) - f(s_0,F\cdot X(s_0)) - \int_{s_0}^s \frac{d(F^0)^{-1}}{ds}(v) F^R_*\left( \vf t + A \right)f (v, F\cdot X(v)) dv
    \end{split}
  \end{equation*}
  is a continuous $\{\F_{(F^0)^{-1}(s)}\}$-martingale. The result follows.
\end{proof}

\begin{remark}\label{remark-1}
  (i) As a consequence, the generator of the pushforward $F\cdot X$ is given in local coordinates by
  \begin{equation*}
    \frac{d(F^0)^{-1}}{ds}\left[ \left( \vf t + A \right) \bar F^i \circ F^{-1}\right] \frac{\partial}{\partial y^i} + \frac{d(F^0)^{-1}}{ds} \left[ \left( A^{kl} \frac{\pt \bar F^i}{\pt x^k} \frac{\pt \bar F^j}{\pt x^l} \right) \circ F^{-1} \right] \frac{\partial^2}{\partial y^i\partial y^j}.
  \end{equation*}
  This coincides with Lemma \ref{push-mixed-coord}.

  (ii) This corollary together with Lemma \ref{bd-morph-mixed-vertical} indicates that the bundle homomorphisms from $\R\times M$ to $\R\times N$ are the only (deterministic) smooth maps between them that map diffusions to diffusions. Indeed, if a smooth map $F$ from $\R\times M$ to $\R\times N$ pushes forward a diffusion to another diffusion, then a similar argument as in Corollary \ref{push-ext-generator} implies that $F^S_*$ would map the extended generator of the former diffusion to that of the latter, whereas Lemma \ref{bd-morph-mixed-vertical} says such $F^S_*$ must be the second-order pushforward of some bundle homomorphism.

  (iii) In particular, if $F$ is a smooth map from $M$ to $N$ and $X$ is a diffusion on $M$ with generator $A$, then $F(X)$ is a diffusion on $N$ with respect to the same filtration, with generator $F^S_*( A )$.
\end{remark}

\subsection{Pushforwards and pullbacks by diffusions}

\begin{definition}[Pushforwards and pullbacks by diffusions]\label{push-pull-diff}
  Let $X$ be an $M$-valued diffusion process. Let $(\R\times U, (t,x^i))$ be a coordinate chart on $\R\times M$. The pushforward map $X_*$ from $T_t \R$ to $T_t \R \times \mathcal T^S_{X(t)} M$ is defined in the local coordinate by
  \begin{equation}\label{pushforward}
    X_*\left(\tau \frac{d}{dt} \bigg|_{t_0} \right) = \tau \left( \vf t\bigg|_{t_0} + (DX)^i(t_0) \vf{x^i}\bigg|_{X(t_0)} + \frac{1}{2} (QX)^{jk} (t_0) \frac{\pt^2}{\pt x^j \pt x^k} \bigg|_{X(t_0)} \right).
  \end{equation}
  The pullback map $X^*$ from $\mathcal T^*_t \R \times \mathcal T^{S*}_{X(t)} M$ to $\mathcal T^*_t \R$ is defined by
  \begin{equation}\label{pullback}
    X^*\left(\alpha_0 dt|_{t_0} + \alpha_i d^2 x^i|_{X(t_0)} + \ts{\frac{1}{2}} \alpha_{jk}dx^j\cdot dx^k |_{X(t_0)} \right) = \left( \alpha_0 + \alpha_i (DX)^i(t_0) + \ts{\frac{1}{2}} \alpha_{jk} (QX)^{jk} (t_0) \right) dt|_{t_0}.
  \end{equation}
\end{definition}

\begin{remark}
  Recall that in classical differential geometry, the pushforward by a smooth curve $\gamma = (\gamma(t))_{t\in[-1,1]}$ on $M$ is a map $\gamma_*: T\R \to T M$ given by $\gamma_*(\frac{d}{dt}|_{t_0}) = \dot\gamma^i(t_0)\vf{x^i}|_{\gamma(t_0)}$. While if we look at the graph of $\gamma$ as a section of the trivial bundle $(\R\times M, \pi, \R)$, denoted by $\bar\gamma$, then the pushforward map by $\bar\gamma$ is $\bar\gamma_*(\frac{d}{dt}|_{t_0}) = \frac{d}{dt}|_{t_0} + \dot\gamma^i(t_0)\vf{x^i}|_{\gamma(t_0)}$. For this reason, it would be more appropriate to call $X_*$ and $X^*$ in Definition \ref{push-pull-diff} the pushforward and pullback by graph of $X$, or by random section corresponding to $X$, instead of by $X$ itself. But we avoid that for simplicity.
\end{remark}

One can see from the definition that the pushforward $X_*$ maps the time vector $\frac{d}{dt}|_{t_0}$ to the value of the extended generator of $X$ at $(t_0,X(t_0))$. There is an informal way to look at the pullback map $X^*$: one first replace all $x$'s by $X$'s in the brackets at LHS of \eqref{pullback} and obtain
\begin{equation*}
  \alpha_0 dt + \alpha_i d X^i + \ts{\frac{1}{2}} \alpha_{jk}dX^j\cdot dX^k;
\end{equation*}
then substituting $d X^i$ and $dX^j\cdot dX^k$, and following It\^o's calculus,
\begin{equation*}
  d X^i = (DX)^i dt + \text{martingale part}, \qquad dX^j\cdot dX^k = (QX)^{jk} dt,
\end{equation*}
and getting rid of the martingale part, we get the RHS of \eqref{pullback}.

The following corollary is straightforward. We will see that pushforward and pullback maps by diffusions are also closely related to the concept of ``total derivatives''.

\begin{corollary}\label{push-pull-prop}
  (i). Let $X$ be an $M$-valued diffusion process. For all $\tau \frac{d}{dt}|_{t_0}\in \mathcal T_{t_0}\R$ and $\alpha\in \mathcal T_{t_0}^* \R \times \mathcal T_{X(t_0)}^{S*} M$,
  \begin{equation}\label{dual-pp}
    \left\langle X^*\left( \alpha \right), \tau \textstyle{\frac{d}{dt}|_{t_0}} \right\rangle = \left\langle\alpha, X_* ( \tau \textstyle{\frac{d}{dt}|_{t_0}} ) \right\rangle.
  \end{equation}
  (ii). If $X\in I_{(t_0,q)}(M)$, $f$ is a smooth function on $\R\times M$ and $g$ a smooth function on $M$, then
  \begin{align*}
    \left\langle X^*(d^\circ f), \textstyle{\frac{d}{dt}} \right\rangle \big|_{t_0} &= X_*( \textstyle{\frac{d}{dt}} )(f)\big|_{(t_0,q)} = (\D_t f)(j_{(t_0,q)}X) = \langle \ts{\frac{\pt}{\pt t}} + A^X, d^\circ f \rangle (t_0,q), \\
    \left\langle X^*(dg\cdot dg), \textstyle{\frac{d}{dt}} \right\rangle \big|_{t_0} &= \left\langle dg\cdot dg, X_*( \textstyle{\frac{d}{dt}} )\right\rangle \big|_{(t_0,q)} = (\Q_t g)(j_{(t_0,q)}X).
  \end{align*}
  (iii). Let $X, Y$ be $M$-valued diffusion processes satisfying $X(t) = Y(t)$ a.s.. Then, $j_t X = j_t Y$ a.s. if and only if $X_*( \frac{d}{dt}|_t) = Y_*( \frac{d}{dt}|_t)$ a.s.. In particular, if $X, Y \in I_{(t,q)}(M)$, then $j_{(t,q)} X = j_{(t,q)} Y$ if and only if $X_*( \frac{d}{dt}|_t) = Y_*( \frac{d}{dt}|_t)$. \\
  (iv). Let $F$ be a bundle homomorphism from $(\R\times M, \pi, \R)$ to $(\R\times N, \rho, \R)$ projecting to $F^0$, and $X$ be an $M$-valued diffusion process. Then $F^R_* \circ X_* = (F\cdot X)_*\circ (F^0)_*$. \\
  (v). Let $F$ be a smooth function from $M$ to $M$, and $X$ be an $M$-valued diffusion process. Then $(\id_{T\R}\times F^S_*) \circ X_* = (F\circ X)_*$.
\end{corollary}
\begin{proof}
  Assertions (i), (ii) and (iii) are easy to deduce from the definitions. We prove (iv) using local expressions. Assume that $F = (F^0, \bar F)$ and denote $\tilde X = F\cdot X$. Recall that $\tilde X(F^0(t)) = \bar F(t,X(t))$. Then
  \begin{equation*}
    \begin{split}
      F^R_* \circ X_* \left( \frac{d}{dt} \bigg|_t \right) =&\ \frac{dF^0}{dt}(t) \frac{\partial}{\partial s}\bigg|_{F^0(t)} + \bigg[ \frac{\pt\bar F^i}{\pt t}(t,X(t)) + (DX)^j(t) \frac{\pt\bar F^i}{\pt x^j}(t,X(t)) \\
      &\ +\frac{1}{2}(QX)^{jk}(t) \frac{\pt^2 \bar F^i}{\pt x^j \pt x^k} (t,X(t))\bigg] \frac{\partial}{\partial y^i}\bigg|_{\bar F(t,X(t))} + \frac{1}{2}(QX)^{kl}(t) \frac{\pt \bar F^i}{\pt x^k} \frac{\pt \bar F^j}{\pt x^l} (t,X(t)) \frac{\partial^2}{\partial y^i\partial y^j}\bigg|_{\bar F(t,X(t))} \\
      =&\ \frac{dF^0}{dt}(t) \Bigg[ \frac{\partial}{\partial s}\bigg|_{F^0(t)} + (D\tilde X)^i(F^0(t)) \frac{\partial}{\partial y^i}\bigg|_{\tilde X(F^0(t))} + \frac{1}{2} (Q\tilde X)^{ij}(F^0(t)) \frac{\partial^2}{\partial y^i\partial y^j}\bigg|_{\tilde X(F^0(t))} \Bigg] \\
      =&\ \frac{dF^0}{dt}(t) (F\cdot X)_* \left( \frac{\partial}{\partial s}\bigg|_{F^0(t)} \right) \\
      =&\ (F\cdot X)_* \circ (F^0)_* \left( \frac{d}{dt} \bigg|_t \right).
    \end{split}
  \end{equation*}
  The result follows.
\end{proof}

\subsection{Lie derivatives}

\begin{definition}[Lie derivatives]\label{Lie-def}
  Let $V$ be a vector field on $M$ and $\psi = \{\psi_\e\}_{\e\in\R}$ be its flow. Let $A$ be a second-order vector field and $\alpha$ be a second-order form on $M$. The Lie derivative of $A$ with respect to $V$ is a second-order vector field on $M$, denoted by $\L_V A$, and defined by
  \begin{equation*}
    (\L_V A)_q = \frac{d}{d\e}\bigg|_{\e=0} (\psi_{-\e})^S_* (A_{\psi_\e(q)}) = \lim_{\e\to0} \frac{(\psi_{-\e})^S_* (A_{\psi_\e(q)}) - A_q}{\e}.
  \end{equation*}
  The Lie derivative of $\alpha$ with respect to $V$ is a second-order form on $M$, denoted by $\L_V \alpha$, and defined by
  \begin{equation*}
    (\L_V \alpha)_q = \frac{d}{d\e}\bigg|_{\e=0} (\psi_{\e})^{S*} (\alpha_{\psi_\e(q)}) = \lim_{\e\to0} \frac{(\psi_{\e})^{S*} (\alpha_{\psi_\e(q)}) - \alpha_q}{\e}.
  \end{equation*}
\end{definition}

For sufficient small $\e\ne0$, $\psi_\e$ is defined in a neighborhood of $q\in M$ and $\psi_{-\e}$ is the inverse of $\psi_\e$. So the difference quotients in the above definitions of Lie derivatives make sense. It is easy to verify that the derivatives exist for each $q\in M$, and $\L_V A$ is a smooth second-order vector field, $\L_V \alpha$ is a smooth second-order covector field. Likewise, the restrictions of $\L_V$ to $\mathcal T_q M$ and $\mathcal T^{*}_{F(q)} N$ coincide with the classical Lie derivatives. In the following, we will seek properties of $\L$. Some of them can be found in \cite[Section 6.(d)]{Mey81a}.

\begin{lemma}\label{Lie-derivative-prop}
  Let $V$ be a vector field and $f$ be a smooth function. Let $A$ and $\alpha$ be a second-order vector field and second-order form, respectively. Then \\
  (i) $\L_V A = [V, A]$, where the RHS denotes the commutator of $V$ and $A$ as linear operators; \\
  (ii) $\L_V (f A) = (Vf) A + f \L_V A$; \\
  (iii) $\langle \L_V \alpha, A \rangle = V (\langle \alpha, A \rangle) - \langle \alpha, \L_V A \rangle$; \\
  (iv) $\L_V (f \alpha) = (Vf) \alpha + f \L_V \alpha$; \\
  (v) $\L_V(d^2 f) = d^2(Vf)$.
\end{lemma}

\begin{remark}
  Note that the commutator $[V, A]$ is a second-order vector field. Indeed, if $V$ and $A$ have coordinate expressions $V = V^i \vf{x^i}$ and $A = A^i \frac{\partial}{\partial x^i} + A^{ij} \frac{\partial^2}{\partial x^i\partial x^j}$, then the following local expression for $[V, A]$ is easy to verify,
  \begin{equation*}
    \begin{split}
      [V, A] =&\ \left( V^j \frac{\pt A^i}{\pt x^j} - A^j \frac{\pt V^i}{\pt x^j} - A^{jk} \frac{\pt^2 V^i}{\pt x^j \pt x^k} \right) \vf{x^i} + V^i \frac{\pt A^{jk}}{\pt x^i} \frac{\pt^2}{\pt x^j \pt x^k} - A^{jk} \left( \frac{\pt V^i}{\pt x^j} \frac{\pt^2}{\pt x^i \pt x^k} + \frac{\pt V^i}{\pt x^k} \frac{\pt^2}{\pt x^i \pt x^j} \right).
    \end{split}
  \end{equation*}
\end{remark}

\begin{proof}
  (i) For a function $f \in C^\infty(M)$,
  \begin{equation*}
    \begin{split}
      (\L_V A)_q f &= \lim_{\e\to0} \frac{(\psi_{-\e})^S_* (A_{\psi_\e(q)})f - A_q f}{\e} = \lim_{\e\to0} \frac{(A_{\psi_\e(q)})(f\circ \psi_{-\e}) - A_q f}{\e} \\
      &= \lim_{\e\to0} \frac{(A_{\psi_\e(q)})(f\circ \psi_{-\e} - f)}{\e} + \lim_{\e\to0} \frac{(A_{\psi_\e(q)})f - A_q f}{\e}.
    \end{split}
  \end{equation*}
  Then, a similar argument to the derivation of classical Lie derivatives yields
  \begin{equation*}
    (\L_V A)_q f = -A_q(Vf) + V_q (Af) = [V, A]_q f.
  \end{equation*}

  (ii) $\L_V (f A)g = [V, fA] g = V(fAg) - fA Vg = Vf Ag + f VAg - fA Vg = Vf Ag + f (\L_V A) g$.

  (iii) For a second-order vector field $A$,
  \begin{equation*}
    \begin{split}
      \langle \L_V \alpha, A \rangle &= \lim_{\e\to0} \frac{\langle (\psi_{\e})^{S*} (\alpha_{\psi_\e(q)}), A \rangle - \langle \alpha_q, A \rangle }{\e} = \lim_{\e\to0} \frac{\langle \alpha_{\psi_\e(q)}, (\psi_{\e})^S_* A \rangle - \langle \alpha_q, A \rangle }{\e} \\
      &= \lim_{t\to0} \frac{\langle \alpha_{\psi_\e(q)} - \alpha_q, (\psi_{\e})^S_* A \rangle}{\e} + \lim_{\e\to0} \frac{\langle \alpha_q, (\psi_{\e})^S_* A - A \rangle }{t} \\
      &= \lim_{\e\to0} \frac{\langle \alpha_{\psi_\e(q)} - \alpha_q, A \rangle}{\e} - \lim_{\e\to0} \frac{\langle \alpha_q, (\psi_{-\e})^S_* A - A \rangle }{\e} \\
      &= V (\langle \alpha, A \rangle) - \langle \alpha, \L_V A \rangle.
    \end{split}
  \end{equation*}

  (iv) Use (iii) to derive
  \begin{equation*}
    \begin{split}
      \langle \L_V (f \alpha), A \rangle &= V (f \langle \alpha, A \rangle) - f \langle \alpha, \L_V A \rangle = (Vf) \langle \alpha, A \rangle + f V (\langle \alpha, A \rangle) - f \langle \alpha, \L_V A \rangle \\
      &= (Vf) \langle \alpha, A \rangle + f \langle \L_V \alpha, A \rangle.
    \end{split}
  \end{equation*}

  (v) Again using (iii) we have $\langle \L_V (d^2 f), A \rangle = V (\langle d^2 f, A \rangle) - \langle d^2 f, \L_V A \rangle = V A f - [V, A] f = AVf = \langle d^2(Vf), A \rangle$.
\end{proof}

\begin{corollary}
  (i) $\L_V (df\cdot dg) = d(Vf)\cdot dg + df\cdot d(Vg)$. \\
  (ii) $\L_V (\omega\cdot\eta) = \L_V \omega\cdot\eta + \omega\cdot\L_V\eta$. \\
  (iii) $\L_V$ commutes with the symmetric product operator $\bullet$.
\end{corollary}

\begin{proof}
  For the first assertion,
  \begin{equation*}
    \begin{split}
      &\ \langle \L_V (df\cdot dg), A \rangle = V (\langle df\cdot dg, A \rangle) - \langle df\cdot dg, \L_V A \rangle = V(\Gamma_A(f,g)) - \Gamma_{[V,A]}(f,g) \\
      =&\ V( A(fg) - fAg - gAf ) - ([V,A](fg) - f[V,A]g - g[V,A]f) \\
      =&\ VA(fg) - VfAg - fVAg - VgAf - gVAf \\
      &\ - ( VA(fg)- AV(fg) - fVAg + fAVg - gVAf + gAVf ) \\
      =&\ AV(fg) - VfAg - VgAf - fAVg - gAVf \\
      =&\ [A(Vfg) - VfAg - gAVf] - [A(fVg) - VgAf - fAVg] \\
      =&\ \langle d(Vf)\cdot dg, A \rangle + \langle df\cdot d(Vg), A \rangle.
    \end{split}
  \end{equation*}
  We use the local expressions to prove the second assertion. Assume, locally, that $\omega = \omega_i dx^i$ and $\eta = \eta_i dx^i$. Then, by \eqref{product-general}, Lemma \ref{Lie-derivative-prop}.(ii) and Corollary \ref{Lie-derivative-prop}.(iv),
  \begin{equation*}
    \begin{split}
      \L_V(\omega\cdot\eta) &= \L_V(\omega_i \eta_j dx^i\cdot dx^j) = V(\omega_i \eta_j) dx^i\cdot dx^j + \omega_i \eta_j \L_V(dx^i\cdot dx^j) \\
      &= ( \eta_j V\omega_i + \omega_i V\eta_j ) dx^i\cdot dx^j + \omega_i \eta_j (dV^i\cdot dx^j + dx^i\cdot dV^j) \\
      &= ( V\omega_i dx^i + \omega_i dV^i ) \cdot (\eta_j dx^j) + (\omega_i dx^i) \cdot ( V\eta_j dx^j + \eta_j dV^j ) \\
      &= \L_V \omega\cdot\eta + \omega\cdot\L_V\eta.
    \end{split}
  \end{equation*}
  The last assertion is a consequence of the second one. Indeed,
  \begin{equation*}
    \L_V (\bullet(\omega\otimes\eta)) = \L_V (\omega\cdot\eta) = \L_V \omega\cdot\eta + \omega\cdot\L_V\eta = \bullet( \L_V \omega\otimes\eta + \omega\otimes\L_V\eta ) = \bullet (\L_V (\omega\otimes\eta)).
  \end{equation*}
\end{proof}

Given a vector field $V$ on $\R\times M$, the Lie derivative $\L_V$ can also be defined for second-order vector fields and second-order forms on $\R\times M$, as in Definition \ref{Lie-def}, without any changes. But when restricting to the mixed-order vector fields and mixed-order forms, it is necessary that the flow in Definition \ref{Lie-def} consists of bundle homomorphisms on $(\R\times M, \pi, \R)$, so that its mixed-order pushforwards and pullbacks are well defined. This feeding back to the vector field $V$ amounts to $V$ is $\pi$-\emph{projectable}. In this case, we just replace the second-order pushforwards and pullbacks in Definition \ref{Lie-def} by mixed-order pushforwards and pullbacks, to define the Lie derivative $\L_V$ for mixed-order vector fields and mixed-order forms on $\R\times M$.

Now let $V$ be a $\pi$-projectable vector field on $\R\times M$. Then, Lemma \ref{Lie-derivative-prop}.(i)--(iv) still holds for smooth functions $f$ on $\R\times M$, mixed-order vector fields $A$ and mixed-order forms $\alpha$ on $\R\times M$. The assertion (v) will hold with the mixed differential instead of the second-order differential, that is, $\L_V(d^\circ f) = d^\circ(Vf)$. Moreover, if $V$ and $A$ have coordinate expressions $V = V^0 \vf{t} + V^i \vf{x^i}$ and $A = A^0 \vf{t} + A^i \frac{\partial}{\partial x^i} + A^{ij} \frac{\partial^2}{\partial x^i\partial x^j}$ where $V^0$ only depends on time, then the Lie derivative $\L_V A$ has the following expression:
  \begin{equation*}
    \begin{split}
      \L_V A = [V, A] =&\ \left( V^0 \frac{\pt A^0}{\pt t}+ V^j \frac{\pt A^0}{\pt x^j} - A^0 \frac{\pt V^0}{\pt t} \right) \vf{t} \\
      &\ + \left( V^0 \frac{\pt A^i}{\pt t} + V^j \frac{\pt A^i}{\pt x^j} - A^0 \frac{\pt V^i}{\pt t} - A^j \frac{\pt V^i}{\pt x^j} - A^{jk} \frac{\pt^2 V^i}{\pt x^j \pt x^k} \right) \vf{x^i} \\
      &\ + \left( V^0 \frac{\pt A^{jk}}{\pt t} + V^i \frac{\pt A^{jk}}{\pt x^i} \right) \frac{\pt^2}{\pt x^j \pt x^k} - A^{jk} \left( \frac{\pt V^i}{\pt x^j} \frac{\pt^2}{\pt x^i \pt x^k} + \frac{\pt V^i}{\pt x^k} \frac{\pt^2}{\pt x^i \pt x^j} \right).
    \end{split}
  \end{equation*}

\section{The mixed-order contact structure on $\R\times \mathcal T^S M$}\label{app-2}

\subsection{Mixed-order total derivatives and mixed-order contact forms}

We denote by $\pi_{1,0}^*(T \R \times \mathcal T^S M)$ the pullback bundle (see \cite[Definition 1.4.5]{Sau89}) of $\tau_\R \times \tau^S_M$ by $\pi_{1,0}$. It is a fiber bundle over $\R\times \mathcal T^S M$. 

\begin{definition}[Mixed-order holonomic lift]
  Let $t\in\R$, $q\in M$, $X\in I_{(t,q)}(M)$ and $\tau \frac{d}{dt}|_t\in T_t \R$. The mixed-order holonomic lift of $\tau \vf t|_t$ by $X$ is defined to be
  \begin{equation*}
    \left(X_*(\tau \textstyle{\frac{d}{dt}}|_t), j_{(t,q)}X \right) \in \pi_{1,0}^*(T \R \times \mathcal T^S M).
  \end{equation*}
  The set of all mixed-order holonomic lifts is denoted by $H^R\pi_{1,0}$, that is,
  \begin{equation*}
    H^R \pi_{1,0}:= \left\{ \left(X_*(\tau \textstyle{\frac{d}{dt}}|_t), j_{(t,q)}X \right) \in \pi_{1,0}^*(T \R \times \mathcal T^S M) : j_{(t,q)}X \in \R\times \mathcal T^S M, \tau \textstyle{\frac{d}{dt}}|_t \in T_t \R \right\}.
  \end{equation*}
\end{definition}

Since $X_*$ depends only upon the mean derivatives of $X$ at $t$, the holonomic lift of a tangent vector is completely determined by $j_{(t,q)}X$ and does not depend on the choice of the representative diffusion $X$. In particular, the set $H^R \pi_{1,0}$ is well defined and is clearly a subbundle of $\pi_{1,0}^*(T \R \times \mathcal T^S M)$.

\begin{lemma}\label{dcpt-1}
  The fiber bundle $(\pi_{1,0}^*(T \R \times \mathcal T^S M), \pi_{1,0}^*(\tau_\R \times \tau^S_M), \R\times \mathcal T^S M)$ can be written as the Whitney sum of two subbundles
  \begin{equation*}
    \pi_{1,0}^* (V^S\pi) \times_{\R\times \mathcal T^S M} H^R\pi_{1,0}.
  \end{equation*}
\end{lemma}
\begin{proof}
  Suppose that $( A, j_{(t,q)}X) \in \pi_{1,0}^*(T \R \times \mathcal T^S M)$. Then $A \in T \R \times \mathcal T^S M$, and
  \begin{equation*}
    \left( X_*(\pi^R_*(A)), j_{(t,q)}X \right) \in H^R\pi_{1,0}.
  \end{equation*}
  It follows easily from the definition of pushforward \eqref{pushforward} that $\pi^R_*(A - X_*(\pi^R_*(A))) = 0$. Hence, $A - X_*(\pi^R_*(A))\in V^S\pi$ and
  \begin{equation*}
    \left( A - X_*(\pi^R_*(A)), j_{(t,q)}X \right) \in \pi_{1,0}^* (V^S\pi).
  \end{equation*}
  The result follows.
\end{proof}

The decomposition of $( A, j_{(t,q)}X) \in \pi_{1,0}^*(T \R \times \mathcal T^S M)$ may then be found by letting
\begin{equation*}
  \begin{split}
    A =&\ A^0 \vf t\bigg|_t + A^i \vf{x^i}\bigg|_q + A^{jk} \frac{\pt^2}{\pt x^j \pt x^k}\bigg|_q \\
    =&\ \left(A^i- A^0 D^i x(j_{(t,q)}X) \right) \vf{x^i}\bigg|_q + \left(A^{jk} - A^0 Q^{jk} x(j_{(t,q)}X) \right) \frac{\pt^2}{\pt x^j \pt x^k}\bigg|_q \\
    &\  + A^0 \left( \vf t\bigg|_t + D^i x(j_{(t,q)}X) \vf{x^i}\bigg|_q + \frac{1}{2} Q^{jk} x(j_{(t,q)}X) \frac{\pt^2}{\pt x^j \pt x^k}\bigg|_q \right).
  \end{split}
\end{equation*}

\begin{definition}
  A section of the bundle $(H^R\pi_{1,0}, \pi_{1,0}^*(\tau_\R \times \tau^S_M)|_{H^R\pi_{1,0}}, \R\times \mathcal T^S M)$ is called a mixed-order total derivative. The specific section
  \begin{equation*}
    \vf t + D^i x \vf{x^i} + \frac{1}{2} Q^{jk} x \frac{\pt^2}{\pt x^j \pt x^k}
  \end{equation*}
  is called the coordinate mixed-order total derivative, and is denoted by $\mathbf D_t$.
\end{definition}

The coordinate mixed-order total derivative is just the total mean derivative  in Definition \ref{total-mean-d}. The dual construction is the mixed-order contact cotangent vector, which may be described as being in the kernel of $X^*$.
\begin{definition}
  An element $(\alpha, j_{(t,q)}X) \in \pi_{1,0}^*(T^* \R \times \mathcal T^{S*} M)$ is called a mixed-order contact cotangent vector if $X^*(\alpha) = 0$. The set of all mixed-order contact cotangent vectors is denoted by $C^{R*}\pi_{1,0}$, that is,
  \begin{equation*}
    C^{R*}\pi_{1,0}:= \left\{ (\alpha, j_{(t,q)}X) \in \pi_{1,0}^*(T^* \R \times \mathcal T^{S*} M) : j_{(t,q)}X \in \R\times \mathcal T^S M, X^*(\alpha) = 0 \right\}.
  \end{equation*}
\end{definition}

It is straightforward to check that the vanishing of $X^*$ does not depend on the particular choice of the representative diffusion $X$. The dual relation between $X^*$ and $X_*$ in \eqref{dual-pp} implies that the mixed-order contact and holonomic elements annihilate each other.

To express a mixed-order contact cotangent vector $(\alpha, j_{(t,q)}X)$ in coordinates, let us consider
\begin{equation}\label{contact-cot}
  \alpha = \alpha_0 dt|_t + \alpha_i d^2 x^i|_q + \alpha_{jk} dx^j\cdot dx^k|_q.
\end{equation}
Using the definition \eqref{pullback} we get
\begin{equation*}
  0 = X^*(\alpha) = \left( \alpha_0 + \alpha_i (DX)^i + \alpha_{jk} (QX)^{jk} \right) dt|_t.
\end{equation*}
There are two basic nontrivial solutions of the above equation, say,
\begin{equation*}\left\{
  \begin{aligned}
    &\alpha_0 = - \alpha_i (DX)^i, \\
    &\alpha_{jk} = 0,
  \end{aligned}\right. \qquad\text{and}\qquad \left\{
  \begin{aligned}
    &\alpha_0 = - \alpha_{jk} (QX)^{jk}, \\
    &\alpha_i = 0.
  \end{aligned}\right.
\end{equation*}
Plugging these solutions in \eqref{contact-cot}, we get two basic types of mixed-order contact cotangent vectors
\begin{equation*}
  (d^2x^i - D^i x dt)|_{j_{(t,q)}X} \qquad\text{and}\qquad (dx^j\cdot dx^k - Q^{jk} x dt)|_{j_{(t,q)}X}.
\end{equation*}
Thus, every mixed-order contact cotangent vector in ${(C^{R*}\pi_{1,0})}_{j_{(t,q)}X}$ is a linear combination of these basic mixed-order contact cotangent vectors.

\begin{lemma}
  The fiber bundle $(\pi_{1,0}^*(T^* \R \times \mathcal T^{S*} M), \pi_{1,0}^*(\tau^*_\R \times \tau^{S*}_M), \R\times \mathcal T^S M)$ can be written as the Whitney sum of two subbundles
  \begin{equation*}
    \pi_1^* (T^* \R) \times_{\R\times \mathcal T^S M} C^{R*}\pi_{1,0}.
  \end{equation*}
\end{lemma}
\begin{proof}
  Suppose that $(\alpha, j_{(t,q)}X) \in \pi_{1,0}^*(T^* \R \times \mathcal T^{S*} M)$. Then, $\alpha \in T^* \R \times \mathcal T^{S*} M$, and the definition of pullback yields
  \begin{equation*}
    \left( X^*(\alpha), j_{(t,q)}X \right) \in \pi_1^* (T^* \R).
  \end{equation*}
  Since $X^*(\alpha - X^*(\alpha)) = 0$, it follows that
  \begin{equation*}
    \left( \alpha - X^*(\alpha), j_{(t,q)}X \right) \in C^{R*}\pi_{1,0}.
  \end{equation*}
  This ends the proof.
\end{proof}

The decomposition of $( \alpha, j_{(t,q)}X) \in \pi_{1,0}^*(T \R \times \mathcal T^S M)$ may then be found by letting
\begin{equation*}
  \begin{split}
    \alpha =&\ \alpha_0 dt|_t + \alpha_i d^2 x^i|_q + \alpha_{jk}dx^j\cdot dx^k|_q \\
    =&\ \left( \alpha_0 + \alpha_i D^i x(j_{(t,q)}X) + \alpha_{jk} Q^{jk} x(j_{(t,q)}X) \right) dt|_t \\
    &\  + \alpha_i \left( d^2x^i - D^i x(j_{(t,q)}X) dt \right)\big|_{(t,q)} + \alpha_{jk} \left( dx^j\cdot dx^k - Q^{jk} x(j_{(t,q)}X) dt \right)\Big|_{(t,q)}.
  \end{split}
\end{equation*}

\begin{definition}
  A section of the bundle $(C^{R*}\pi_{1,0}, \pi_{1,0}^*(\tau^*_\R \times \tau^{S*}_M)|_{C^{R*}\pi_{1,0}}, \R\times \mathcal T^S M)$ is called a mixed-order contact form. The following specific sections
  \begin{equation*}
    d^2x^i - D^i x dt, \quad dx^j\cdot dx^k - Q^{jk} x dt, \qquad 1\le i,j,k \le d,
  \end{equation*}
  are called basic mixed-order contact forms.
\end{definition}

It follows from the construction that the set of basic mixed-order contact forms defines a local frame of the bundle $\pi_{1,0}^*(\tau^*_\R \times \tau^{S*}_M)|_{C^{R*}\pi_{1,0}}$.

\begin{remark}
  In contrast, we recall the classical contact forms on the first-order jet bundle $J^1 \pi = \R\times TM$. Using the coordinates $(t,x^i,\dot x^i)$, the classical basic contact forms are $dx^i - \dot x^i dt$, $1\le i \le d$. See \cite[Section 4.3]{Sau89} and \cite[Theorem 4.23]{Olv95}, also cf. \cite[p.~9]{Gei08}, for a one-dimensional example.
\end{remark}

\begin{corollary}
  Let $(\R\times U, (t,x^i))$ be a coordinate chart on $\R\times M$. Let $\mathbf X$ be a $\mathcal T^S M$-valued diffusion process. In local coordinates, the pushforward map $\mathbf X_*$ from $T \R$ to $T \R \times \mathcal T^S \mathcal T^S M$ is given by
  \begin{equation*}
    \begin{split}
      \mathbf X_*\left(\tau \frac{d}{dt}\bigg|_t \right) =&\ \tau \bigg( \vf t + D^i (x\circ \mathbf X) \vf{x^i} + D^i (Dx\circ \mathbf X) \vf{D^ix} + D^{jk} (Qx\circ \mathbf X) \vf{Q^{jk}x} \\
      &\ + \frac{1}{2} Q^{jk} (x\circ \mathbf X) \frac{\pt^2}{\pt x^j \pt x^k} + \frac{1}{2} Q^{jk} (Dx\circ \mathbf X) \frac{\pt^2}{\pt D^j x \pt D^k x} + \frac{1}{2} Q^{jklm} (Qx\circ \mathbf X) \frac{\pt^2}{\pt Q^{jk}x \pt Q^{lm}x} \\
      &\ + \frac{1}{2} Q^{jk} (x\circ \mathbf X, Dx\circ \mathbf X) \frac{\pt^2}{\pt x^j \pt D^kx} + \frac{1}{2} Q^{jkl} (x\circ \mathbf X, Qx\circ \mathbf X) \frac{\pt^2}{\pt x^j \pt Q^{kl}x} \\
      &\ + \frac{1}{2} Q^{jkl} (Dx\circ \mathbf X, Qx\circ \mathbf X) \frac{\pt^2}{\pt D^j x \pt Q^{kl}x}\bigg)\bigg|_{(t,\mathbf X(t))}.
    \end{split}
  \end{equation*}
  The pullback map $\mathbf X^*$ from $T^* \R \times \mathcal T^{S*} \mathcal T^S M$ to $T^* \R$ is given by
  \begin{equation*}
    \begin{split}
      &\ \mathbf X^*\Big(\alpha_0 dt + \alpha_i d^2 x^i + \alpha^1_i d^2 D^ix + \alpha^2_{jk} d^2 Q^{jk}x + \alpha_{jk}dx^j\cdot dx^k + \alpha^1_{jk}dD^jx\cdot dD^kx + \alpha^2_{jklm}dQ^{jk}x\cdot dQ^{lm}x \\
      & \qquad + \alpha^{01}_{jk}dx^j\cdot dD^k x + \alpha^{02}_{jkl}dx^j\cdot dQ^{kl}x + \alpha^{12}_{jkl}dD^jx\cdot dQ^{kl}x \Big) \Big|_{(t,\mathbf X(t))} \\
      =&\ \Big( \alpha_0 + \alpha_i D^i (x\circ \mathbf X) + \alpha^1_i D^i (Dx\circ \mathbf X) + \alpha^2_{jk} D^{jk} (Qx\circ \mathbf X) \\
      &\ + \alpha_{jk} Q^{jk} (x\circ \mathbf X) + \alpha^1_{jk} Q^{jk} (Dx\circ \mathbf X) + \alpha^2_{jklm} Q^{jklm} (Qx\circ \mathbf X) \\
      &\ + \alpha^{01}_{jk} Q^{jk} (x\circ \mathbf X, Dx\circ \mathbf X) + \alpha^{02}_{jkl} Q^{jkl} (x\circ \mathbf X, Qx\circ \mathbf X) + \alpha^{12}_{jkl} Q^{jkl} (Dx\circ \mathbf X, Qx\circ \mathbf X) \Big) dt|_t.
    \end{split}
  \end{equation*}
\end{corollary}

\begin{corollary}
   Let $\alpha$ be a section of $(T^* \R \times \mathcal T^{S*} \mathcal T^S M, \tau^*_\R \times \tau^{S*}_{T^S M}, \R \times \mathcal T^S M)$. Then $\alpha$ is a mixed-order contact form if and only if for every $t\in\R$ and every $X\in \cup_{q\in M}I_{(t,q)}(M)$,
  \begin{equation*}
    (jX)^*(\alpha|_{j_{(t,q)}X}) = 0.
  \end{equation*}
\end{corollary}

\begin{proof}
  We first let $\alpha = \alpha_0 dt + \alpha_i d^2 x^i + \alpha_{jk}dx^j\cdot dx^k$ be a mixed-order contact form and let $X\in I_{(t,q)}(M)$. Then
  \begin{equation}\label{eqn-2}
    (jX)^*(\alpha|_{j_{(t,q)}X}) = \left( \alpha_0 + \alpha_i D^i x + \alpha_{jk} Q^{jk} x \right)(j_{(t,q)}X) dt|_t = X^*(\alpha|_{j_{(t,q)}X}) = 0.
  \end{equation}
  To prove the converse, we suppose
  \begin{equation*}
    \begin{split}
      \alpha =&\ \alpha_0 dt + \alpha_i d^2 x^i + \alpha^1_i d^2 D^ix + \alpha^2_{jk} d^2 Q^{jk}x + \alpha_{jk}dx^j\cdot dx^k + \alpha^1_{jk}dD^jx\cdot dD^kx + \alpha^2_{jklm}dQ^{jk}x\cdot dQ^{lm}x \\
      &\ + \alpha^{01}_{jk}dx^j\cdot dD^k x + \alpha^{02}_{jkl}dx^j\cdot dQ^{kl}x + \alpha^{12}_{jkl}dD^jx\cdot dQ^{kl}x
    \end{split}
  \end{equation*}
  Fix a particular index $i_0$ with $1\le i_0 \le d$. Let $Y\in I_{(t,q)}(M)$ such that $j_{(t,q)}X = j_{(t,q)}Y$, $D^i D Y = D^i D X + \delta_{i_0}^i$ and
  \begin{equation*}
    \begin{split}
      &\ \left( D^{jk}Q Y, Q^{jk} DY, Q^{jklm}Q Y, Q^{jk}(Y,DY), Q^{jkl}(Y,QY), Q^{jkl}(DY,QY) \right) \\
    =&\ \left( D^{jk}Q X, Q^{jk} DX, Q^{jklm}Q X, Q^{jk}(X,DX), Q^{jkl}(X,QX), Q^{jkl}(DX,QX) \right).
    \end{split}
  \end{equation*}
  Then,
  \begin{equation*}
    0 = (jY)^*(\alpha|_{j_{(t,q)}Y}) = (jX)^*(\alpha|_{j_{(t,q)}X}) + \alpha^1_i \delta_{i_0}^i = \alpha^1_{i_0}.
  \end{equation*}
  It follows from the arbitrariness of $i_0$ that $\alpha^1_i = 0$ for all $1\le i\le d$. Similarly, all $\alpha_{jk}^1$, $\alpha_{jk}^2$ and $\alpha_{jklm}^2$ vanish. Consequently, $\alpha = \alpha_0 dt + \alpha_i d^2 x^i + \alpha_{jk}dx^j\cdot dx^k$. As in \eqref{eqn-2}, we have $(jX)^*(\alpha|_{j_{(t,q)}X}) = X^*(\alpha|_{j_{(t,q)}X}) = 0$. Hence, $\alpha$ is a mixed-order contact form.
\end{proof}

\begin{corollary}
  Let $\mathbf X$ be a $\mathcal T^S M$-valued diffusion process. Then $\mathbf X = j X$, with $X$ an $M$-valued diffusion process, if and only if $\mathbf X^*(\alpha) = 0$ for every mixed-order contact form $\alpha$ on $\R\times \mathcal T^S M$.
\end{corollary}

\begin{proof}
  We first suppose $\mathbf X = j X$ with $X$ an $M$-valued diffusion process. Then, for a mixed-order contact form $\alpha$,
  \begin{equation*}
    \mathbf X^*(\alpha) = (jX)^*(\alpha) = X^*(\alpha) = 0.
  \end{equation*}
  To prove the converse, it suffices to show, in local coordinates, that
  \begin{equation*}
    D^i x(\mathbf X) = D^i (x\circ \mathbf X), \qquad Q^{jk} x(\mathbf X) = Q^{jk} (x\circ \mathbf X).
  \end{equation*}
  This can be done as soon as we let $\alpha$ be a basic mixed-order contact form. For example, let $\alpha = d^2x^i - D^i x dt$, then
  \begin{equation*}
    0 = \mathbf X^*(\alpha) = \left( D^i (x\circ \mathbf X) - D^i x \circ\mathbf X \right) dt,
  \end{equation*}
  which leads to $D^i x(\mathbf X) = D^i (x\circ \mathbf X)$.
\end{proof}

\subsection{The mixed-order Cartan distribution and its symmetries}

The model bundle $\R\times \mathcal T^S M$ is a trivial bundle over $\R$ in its own right, and so we may consider its mixed-order tangent bundle $(T \R \times \mathcal T^S \mathcal T^S M, \tau_\R \times \tau^S_{\mathcal T^S M}, \R\times \mathcal T^S M)$.

\begin{definition}
  The bundle endomorphisms $(v,\id_E)$ of $\pi_{1,0}^*(\tau_\R \times \tau^S_M)$ is defined by
  \begin{equation*}
    v(A^h + A^v) = A^v,
  \end{equation*}
  where $A^h \in H^R\pi_{1,0}$ and $A^v \in \pi_{1,0}^* (V^S\pi)$.
\end{definition}

\begin{definition}[Mixed-order Cartan distribution]
  The mixed-order Cartan distribution is the kernel of the vector bundle homomorphism over $\id_{\R \times \mathcal T^S M}$
  \begin{equation*}
    v\circ (\pi_{1,0*}, \tau_\R \times \tau^S_{\mathcal T^S M}): T \R \times \mathcal T^S \mathcal T^S M \to \pi_{1,0}^*(\tau_\R \times \tau^S_M)
  \end{equation*}
  and is denoted by $C^R\pi_{1,0}$.
\end{definition}

Note that $C^R\pi_{1,0}$ is a subbundle of $\tau_\R \times \tau^S_{\mathcal T^S M}$. It follows from the above two definitions that
\begin{equation*}
  C^R\pi_{1,0} = (\pi_{1,0*}, \tau_\R \times \tau^S_{\mathcal T^S M})^{-1} H^R\pi_{1,0}.
\end{equation*}
Hence, for each $X \in I_{(t,q)}(M)$,
\begin{equation*}
  C^R\pi_{1,0}|_{j_{(t,q)}X} = (jX)_*(T_t \R) \oplus V^S \pi_{1,0}|_{j_{(t,q)}X}.
\end{equation*}
Similarly to the proof of Lemma \ref{dcpt-1}, we can decompose an element $\mathbf A \in C^R\pi_{1,0}|_{j_{(t,q)}X}$ as
\begin{equation}\label{Cartan-dcpt}
  \mathbf A = (jX)_*((\pi_1)^R_*(\mathbf A)) + \left[ \mathbf A - (jX)_*((\pi_1)^R_*(\mathbf A)) \right],
\end{equation}
where $(jX)_*((\pi_1)^R_*(\mathbf A)) \in (jX)_*(T_t \R)|_{j_{(t,q)}X}$ and $\mathbf A - (jX)_*((\pi_1)^R_*(\mathbf A)) \in V^S_{j_{(t,q)}X} \pi_{1,0}$.

From the duality relations it also follows that $(\tau^*_\R \times \tau^{S*}_{\mathcal T^S M})|_{C^{R*}\pi_{1,0}}$ is the annihilator of $(\tau_\R \times \tau^S_{\mathcal T^S M})|_{C^R\pi_{1,0}}$, or in other words, the basic mixed-order contact forms are local defining forms for the mixed-order contact distribution $C^R\pi_{1,0}$.
A typical element $\mathbf A\in C^R\pi_{1,0}|_{j_{(t,q)}X}$ may be written in coordinates as
\begin{equation}\label{local-rep-cartan-dist}
  \begin{split}
    \mathbf A =&\ \mathbf A^0 \left( \vf t\bigg|_{j_{(t,q)}X} + D^i x(j_{(t,q)}X) \vf{x^i}\bigg|_{j_{(t,q)}X} + \frac{1}{2} Q^{jk} x(j_{(t,q)}X) \frac{\pt^2}{\pt x^j \pt x^k}\bigg|_{j_{(t,q)}X} \right) \\
    &\ + \mathbf A^i_1 \vf{D^ix}\bigg|_{j_{(t,q)}X} + \mathbf A^{jk}_{2} \vf{Q^{jk}x}\bigg|_{j_{(t,q)}X} + \mathbf A^{jk}_{11} \frac{\pt^2}{\pt D^j x \pt D^k x}\bigg|_{j_{(t,q)}X} + \mathbf A^{jklm}_{22} \frac{\pt^2}{\pt Q^{jk}x \pt Q^{lm}x}\bigg|_{j_{(t,q)}X} \\
    &\ + \mathbf A^{jk}_{01} \frac{\pt^2}{\pt x^j \pt D^k x}\bigg|_{j_{(t,q)}X} + \mathbf A^{jkl}_{02} \frac{\pt^2}{\pt x^j \pt Q^{kl} x}\bigg|_{j_{(t,q)}X} + \mathbf A^{jkl}_{12} \frac{\pt^2}{\pt D^j x \pt Q^{kl} x}\bigg|_{j_{(t,q)}X}.
  \end{split}
\end{equation}
From this it is easy to deduce $(\pi_{1,0})_*^R \mathbf A \in H^R\pi_{1,0}$.

\begin{definition}
  A symmetry of the mixed-order Cartan distribution on $\R \times \mathcal T^S M$ is a bundle automorphism $\mathbf F$ of $\R \times \mathcal T^S M$ which satisfies $\mathbf F^R_*(C^R\pi_{1,0}) = C^R\pi_{1,0}$.
\end{definition}

It follows by duality that symmetries of the mixed-order Cartan distribution are those bundle automorphisms which satisfy $\mathbf F^{R*}(C^{R*}\pi_{1,0}) = C^{R*}\pi_{1,0}$. For this reason, $\mathbf F$ is also called a mixed-order contact transformation. Similarly, $\mathbf F$ may be characterized by the fact that whenever $\alpha$ is a mixed-order contact form then so is $\mathbf F^{R*}(\alpha)$.

\begin{proposition}\label{bd-endo}
  Let $\mathbf F$ be a bundle homomorphism from $(\R\times \mathcal T^S M, \pi_1, \R)$ to $(\R\times \mathcal T^S N, \rho_1, \R)$ that projects to a diffeomorphism $F^0:\R\to\R$. Then $\mathbf F^R_*(C^R \pi_{1,0}) \subset C^R \rho_{1,0}$ if and only if $\mathbf F = j F$ where $F$ is a bundle homomorphism from $(\R\times M, \pi, \R)$ to $(\R\times N, \rho, \R)$ that projects to $F^0$.
\end{proposition}

\begin{proof}
  First, we prove the sufficiency. Let $\mathbf A\in C^R\pi_{1,0}|_{j_{(t,q)}X}$. According to \eqref{Cartan-dcpt}, we decompose $\mathbf A$ by $\mathbf A = \mathbf A_1 + \mathbf A_2$ with $\mathbf A_1 = (jX)_*((\pi_1)^R_*(\mathbf A)) \in (jX)_*(T_t \R)$ and $\mathbf A_2 \in V^S_{j_{(t,q)}X} \pi_{1,0}$. Then, since by Corollaries \ref{push-diff-bd} and \ref{push-pull-prop}.(iv), $(jF)^R_* \circ (jX)_* = (jF\cdot jX)_* \circ (F^0)_* = (j\tilde X)_* \circ (F^0)_*$ where $\tilde X = F\cdot X$ is the pushforward of $X$ by $F$, we have
  \begin{equation*}
    \mathbf F^R_*(\mathbf A_1) = (jF)^R_*(\mathbf A_1) = (jF)^R_* (jX)_* (\pi_1)^R_* \mathbf A = (j\tilde X)_* (F^0)_* (\pi_1)^R_* \mathbf A \in (j\tilde X)_*(\mathcal T_{F^0(t)}\R).
  \end{equation*}
  Besides, since $jF: \pi_{1,0} \to \rho_{1,0}$ is a bundle homomorphism projecting to $F$ by Corollary \ref{jf-bd-morph}.(ii), we have $\rho_{1,0} \circ jF = F \circ \pi_{1,0}$. Then,
  \begin{equation*}
    (\rho_{1,0})^S_* (\mathbf F^R_*(\mathbf A_2)) = (\rho_{1,0})^S_* ((jF)^R_*(\mathbf A_2)) = F^S_* (\pi_{1,0})^S_* (\mathbf A_2) = 0,
  \end{equation*}
  which yields $\mathbf F^R_*(\mathbf A_2) \in V^S \rho_{1,0}$. This proves $\mathbf F^R_*(C^R\pi_{1,0}) \subset C^R\rho_{1,0}$.

  For the necessity, we first prove that $\mathbf F$ is bundle homomorphism from $\pi_{1,0}$ to $\rho_{1,0}$ by showing $\mathbf F^S_*(V^S \pi_{1,0}) \subset V^S \rho_{1,0}$, by virtue of Lemma \ref{bd-morph-mixed-vertical}. Let $\mathbf A \in V^S \pi_{1,0}$. Set $\mathbf F^R_* \mathbf A = \mathbf A_1 + \mathbf A_2$, where $\mathbf A_1\in (j Y)_*(\mathcal T_{F^0(t)} \R)$ and $\mathbf A_2\in V^S \rho_{1,0}$ for some diffusion $Y$. Since $\mathbf F$ projects to $F^0$,
  \begin{equation*}
    (\rho_1)^S_*(\mathbf F^S_* \mathbf A) = (F^0)^S_* (\pi_1)^S_* \mathbf A = (F^0)^S_* \pi^S_* (\pi_{1,0})^S_* \mathbf A = 0,
  \end{equation*}
  while $(\rho_1)^S_* \mathbf A_2 = \rho^S_* (\rho_{1,0})^S_* \mathbf A_2 = 0$. Thus, $(\rho_1)^S_* \mathbf A_1 = 0$. Since $\mathbf A_1\in (j Y)_*(\mathcal T_{F^0(t)} \R)$, we set $\mathbf A_1 = (jY)_* (\tau \vf s|_{F^0(t)} )$. Then $(\rho_1)^S_* \mathbf A_1 = \tau \vf s|_{F^0(t)} = 0$. Hence, $\tau = 0$ and so $\mathbf A_1 = 0$. This leads to $\mathbf F^R_*(V^S \pi_{1,0}) \subset V^S \rho_{1,0}$ and so that $\mathbf F$ is bundle homomorphism from $\pi_{1,0}$ to $\rho_{1,0}$. Denote the projection of $\mathbf F$ onto a map from $\R\times M$ to $\R\times N$ by $F$. It follows that \begin{equation*}
    \rho\circ F\circ \pi_{1,0} = \rho\circ\rho_{1,0} \circ \mathbf F = \rho_1\circ \mathbf F = F^0 \circ \pi_1 = F^0 \circ\pi\circ\pi_{1,0}.
  \end{equation*}
  Since $\pi_{1,0}$ is surjective, we obtain $\rho\circ F = F^0 \circ\pi$, so that $F$ is a bundle homomorphism from $\pi$ to $\rho$ projecting to $F^0$. We shall write $F = (F^0, \bar F)$ and $\mathbf F = (F^0, \bar{\mathbf F})$. 

  Next, we will show $\mathbf F =jF$. Fix a $j_{(t,q)}X \in \R\times \mathcal T^S M$. Let $\mathbf F(j_{(t,q)}X) = j_{(s,q')}Y$. Then, $s = F^0(t)$ and $(s,q') = F(t,q)$.
  For an element $\mathbf A\in C^R\pi_{1,0}|_{j_{(t,q)}X}$ with local expression in \eqref{local-rep-cartan-dist}, we have from \eqref{local-rep-push-mixed} that
  \begin{equation*}
    \begin{split}
      \mathbf F^R_* \mathbf A =&\ \mathbf A^0 \frac{dF^0}{dt}(t) \frac{\partial}{\partial s}\bigg|_{j_{(s,q')}Y} + (\mathbf A \bar F^i)(j_{(t,q)}X) \frac{\partial}{\partial y^i}\bigg|_{j_{(s,q')}Y} + \frac{\mathbf A^0}{2} Q^{jk} x(j_{(t,q)}X) \frac{\pt \bar F^i}{\pt x^k} \frac{\pt \bar F^j}{\pt x^l} (t,q) \frac{\partial^2}{\partial y^i\partial y^j}\bigg|_{j_{(s,q')}Y} \\
      &\ + \text{terms}\left( \frac{\partial}{\partial D^iy}\bigg|_{j_{(s,q')}Y}, \frac{\partial}{\partial Q^{ij}y}\bigg|_{j_{(s,q')}Y}, \cdots \right).
    \end{split}
  \end{equation*}
  Since $\bar F$ only depends on the variables on $\R\times M$, we have
  \begin{equation*}
    \begin{split}
      (\mathbf A \bar F^i)(j_{(t,q)}X) &= \left( (\pi_{1,0})_*^R \mathbf A \right) \bar F^i(j_{(t,q)}X) \\
      &= \mathbf A^0 \left[ \frac{\pt \bar F^i}{\pt t}(t,q) + D^j x(j_{(t,q)}X) \frac{\pt \bar F^i}{\pt x^j}(t,q) + \frac{1}{2} Q^{jk} x(j_{(t,q)}X) \frac{\pt^2 \bar F^i}{\pt x^j \pt x^k}(t,q) \right].
    \end{split}
  \end{equation*}
  Then, the local expressions for $jF$ in \eqref{local-rep-jF-1} and \eqref{local-rep-jF-2} yield
  \begin{equation*}
    \mathbf F^R_* \mathbf A = \mathbf A^0 \frac{dF^0}{dt}(t) \left[ \frac{\partial}{\partial s}\bigg|_{j_{(s,q')}Y} + D^i y \circ jF(j_{(t,q)}X) \frac{\partial}{\partial y^i}\bigg|_{j_{(s,q')}Y} + \frac{1}{2} Q^{ij} y \circ jF(j_{(t,q)}X) \frac{\partial^2}{\partial y^i\partial y^j}\bigg|_{j_{(s,q')}Y} \right].
  \end{equation*}
  Since $\mathbf F^R_* \mathbf A\in C^R\pi_{1,0}|_{j_{(s,q')}Y}$ by the assumption, it follows that $jF(j_{(t,q)}X) = j_{(s,q')}Y = \mathbf F(j_{(t,q)}X)$. This proves that $\mathbf F = jF$.
\end{proof}

\begin{corollary}\label{bd-endo-2}
  Let $\mathbf F$ be a bundle automorphism on $(\R\times \mathcal T^S M, \pi_1, \R)$ projecting to a diffeomorphism $F^0:\R\to\R$. Then $\mathbf F$ is a symmetry of $C^R\pi_{1,0}$ if and only if $\mathbf F = j F$ where $F$ is a bundle automorphism on $(\R\times M, \pi, \R)$ that projects to $F^0$.
\end{corollary}

\begin{proof}
  If $\mathbf F$ is a symmetry, then $\mathbf F^R_*(C^R\pi_{1,0}) \subset C^R\pi_{1,0}$ and $(\mathbf F^{-1})^R_*(C^R\pi_{1,0}) \subset C^R\pi_{1,0}$. By Proposition \ref{bd-endo}, $\mathbf F = j F$ and $\mathbf F^{-1} = j G$ for some bundle endomorphisms $F$ and $G$ on $(\R\times M, \pi_1, \R)$ that projects to $F^0$ and $(F^0)^{-1}$, respectively. Then, Corollary \ref{jf-bd-morph}.(iii) implies that $j(F\circ G) = jF \circ jG = F\circ F^{-1} = \id_{\R\times \mathcal T^S M}$ and hence $F\circ G = \id_{\R\times M}$. For the same reason, $G\circ F = \id_{\R\times M}$. Thus, $F$ is a bundle automorphism on $\pi$. Conversely, if $\mathbf F = j F$ and $F$ is a bundle automorphism, then $\mathbf F \circ j F^{-1} = j F^{-1}\circ F = \id_{\R\times \mathcal T^S M}$, which yields $\mathbf F^{-1} = j F^{-1}$ and hence $\mathbf F$ is a bundle automorphism on $\pi_1$.
\end{proof}

\subsection{Infinitesimal symmetries}

\begin{definition}
  An infinitesimal symmetry of the mixed-order Cartan distribution is a $\pi_1$-projectable vector field $\mathbf V$ on $\R \times \mathcal T^S M$ with the property that, whenever the mixed-order vector field $\mathbf A$ belongs to $C^R\pi_{1,0}$, then so does the mixed-order vector field $\L_{\mathbf V}\mathbf A$.
\end{definition}

Like in the classical case, an infinitesimal symmetry of the mixed-order Cartan distribution may also be called an infinitesimal mixed-order contact transformation. By duality, $\mathbf V$ is such an infinitesimal symmetry precisely when $\L_{\mathbf V} \alpha$ is a contact form for every mixed-order contact form $\alpha$.

The following lemma is a consequence of the definition of Lie derivatives.
\begin{lemma}\label{Cartan-symm}
  Let $\mathbf V$ be a $\pi_1$-projectable vector field on $\R \times \mathcal T^S M$ with flow $\Psi=\{\Psi_\e\}_{\e\in\R}$. Then, $\mathbf V$ is an infinitesimal symmetry of the mixed-order Cartan distribution if and only if for each $\e$, the diffeomorphism $\Psi_\e$ is a symmetry of the mixed-order Cartan distribution.
\end{lemma}

The following result is the infinitesimal version of Corollary \ref{bd-endo-2}. It can be deduced directly from Lemma \ref{Cartan-symm} and Corollary \ref{bd-endo-2}. But here we give a computational proof based on the Lie derivative of mixed-order contact forms.

\begin{theorem}
  Let $\mathbf V$ be a $\pi_1$-projectable vector field on $\R \times \mathcal T^S M$. Then, $\mathbf V$ is an infinitesimal symmetry of the mixed-order Cartan distribution if and only if $\mathbf V$ is the prolongation of a $\pi$-projectable vector field $V$ on $\R \times M$.
\end{theorem}

\begin{proof}
  Let the vector field $\mathbf V$ having the following local expression:
  \begin{equation*}
    \mathbf V = \mathbf V^0 \vf t + \mathbf V^i \vf{x^i} + \mathbf V_1^i \vf{D^i x} + \mathbf V_2^i \vf{Q^{jk} x},
  \end{equation*}
  where $\mathbf V^0$ only depends on time due to the projectability of $\mathbf V$. We then derive the Lie derivative $\L_{\mathbf V}$ of the basic mixed-order contact forms $d^2x^i - D^i x dt$ and $dx^j\cdot dx^k - Q^{jk} x dt$ as follows:
  \begin{equation*}
    \begin{split}
      &\ \L_{\mathbf V}(d^2x^i - D^i x dt) \\
      =&\ d^\circ \mathbf V^i - \mathbf V_1^i dt - D^i x d\mathbf V^0 \\
      =&\ \frac{\pt \mathbf V^i}{\pt t} dt + \frac{\pt \mathbf V^i}{\pt x^j} d^2 x^j + \frac{1}{2} \frac{\pt^2 \mathbf V^i}{\pt x^j \pt x^k} dx^j\cdot dx^k + \text{terms}\left( \frac{\pt \mathbf V^i}{\pt D^j x}, \frac{\pt \mathbf V^i}{\pt Q^{jk} x}, \cdots \right) \\
      &\ - \mathbf V_1^i dt - D^i x \frac{d\mathbf V^0}{dt} dt \\
      =&\ \frac{\pt \mathbf V^i}{\pt x^j} (d^2 x^j - D^j x dt) + \frac{1}{2} \frac{\pt^2 \mathbf V^i}{\pt x^j \pt x^k} (dx^j\cdot dx^k - Q^{jk} x dt) \\
      &\ + \left( \frac{\pt \mathbf V^i}{\pt t} + \frac{\pt \mathbf V^i}{\pt x^j} D^j x + \frac{1}{2} \frac{\pt^2 \mathbf V^i}{\pt x^j \pt x^k} Q^{jk} x - \mathbf V_1^i - D^i x \frac{d\mathbf V^0}{dt} \right) dt + \text{terms}\left( \frac{\pt \mathbf V^i}{\pt D^j x}, \frac{\pt \mathbf V^i}{\pt Q^{jk} x}, \cdots \right),
    \end{split}
  \end{equation*}
  and
  \begin{equation*}
    \begin{split}
      &\ \L_{\mathbf V}(dx^j\cdot dx^k - Q^{jk} x dt) \\
      =&\ d \mathbf V^j \cdot d x^k + d x^j \cdot d \mathbf V^k - \mathbf V_2^{jk} dt - Q^{jk} x d\mathbf V^0 = \frac{\pt \mathbf V^j}{\pt x^i} d x^i \cdot d x^k + \frac{\pt \mathbf V^k}{\pt x^i} d x^j \cdot d x^i - \mathbf V_2^{jk} dt - Q^{jk} x d\mathbf V^0 \\
      =&\ \frac{\pt \mathbf V^j}{\pt x^i} (d x^i \cdot d x^k - Q^{ik} x dt) + \frac{\pt \mathbf V^k}{\pt x^i} (d x^i \cdot d x^j - Q^{ij} x dt) + \left( \frac{\pt \mathbf V^j}{\pt x^i} Q^{ik} x + \frac{\pt \mathbf V^k}{\pt x^i} Q^{ij} x - \mathbf V_2^{jk} - Q^{jk} x \frac{d\mathbf V^0}{dt} \right) dt.
    \end{split}
  \end{equation*}
  Thus, the mixed-order forms $\L_{\mathbf V}(d^2x^i - D^i x dt)$ and $\L_{\mathbf V}(dx^j\cdot dx^k - Q^{jk} x dt)$ are mixed-order contact forms if and only if
  \begin{gather}
    \text{terms } \frac{\pt \mathbf V^i}{\pt D^j x}, \frac{\pt \mathbf V^i}{\pt Q^{jk} x}, \text{ etc, vanish and} \label{eqn-5} \\
    \frac{\pt \mathbf V^i}{\pt t} + \frac{\pt \mathbf V^i}{\pt x^j} D^j x + \frac{1}{2} \frac{\pt^2 \mathbf V^i}{\pt x^j \pt x^k} Q^{jk} x - \mathbf V_1^i - D^i x \frac{d\mathbf V^0}{dt} = 0, \label{eqn-6} \\
    \frac{\pt \mathbf V^j}{\pt x^i} Q^{ik} x + \frac{\pt \mathbf V^k}{\pt x^i} Q^{ij} x - \mathbf V_2^{jk} - Q^{jk} x \frac{d\mathbf V^0}{dt} = 0. \label{eqn-7}
  \end{gather}
  Now \eqref{eqn-5} means that $\mathbf V^i$'s only depend on the variables on $\R\times M$, so that the vector field $\mathbf V$ is also $\pi_{1,0}$-projectable. The two equations \eqref{eqn-6} and \eqref{eqn-7} are just restatements of the prolongation formulae in Theorem \ref{prog-proj-vf}.
\end{proof}

\section{Stochastic Maupertuis's principle}\label{app-3}

Based on Definition \ref{variation}, if we further consider the variation caused by time-change, as in classical mechanics (cf. \cite[Definition 3.8.4]{AM78} or the so called $\Delta$-variation in \cite[Section 8.6]{GPS02}), then we need to impose the constraint of constant
energy. So the path space $\A_g([0,T];q, \mu)$ in \eqref{diff-space} is modified to
\begin{equation*}
  \begin{split}
    \A_g([0,T];q, \mu;e) := \Big\{ (X, \tau): & \tau\in C^2([0,T],\R), \tau' > 0, X\in I_{(\tau(0),q)}^{(\tau(T),\mu)}(M), \\
    & QX(t) = \check g(X(t)), \forall t\in [\tau(0),\tau(T)], \text{a.s.}, \\
    & \E E_0(t, X(t), D_\nabla X(t)) = e, \forall t\in [\tau(0),\tau(T)] \Big\},
  \end{split}
\end{equation*}
where $e\in\R$ is a regular value of $E_0$.

\begin{definition}\label{variation-2}
  Given $v\in \mathcal H([0,T];q)$ and $\varsigma\in\C^1([0,T],\R)$, by a variation of the pair $(X,\tau)\in \A_g([0,T];q, \mu;e)$ along $(v,\varsigma)$, we mean a family of pairs $\{(X_\e^{v,\varsigma},\tau^{\varsigma}_\e)\}_{\e\in(-\varepsilon,\varepsilon)}$ where $\tau^{\varsigma}_0 = \tau$, $\frac{\pt}{\pt t}\tau^{\varsigma}_\e >0$, such that for each $\e$, $\frac{\pt}{\pt\e}\tau^{\varsigma}_\e|_{\e=0} =\varsigma$, $X_\e^{v,\varsigma}\in I_{(\tau^{\varsigma}_\e(0),q)}^{(\tau^{\varsigma}_\e(T),\mu)}(M)$, and for each $t\in[\tau^{\varsigma}_\e(0),\tau^{\varsigma}_\e(T)]$, $\E E_0(t, X_\e^{v,\varsigma}(t), D_\nabla X_\e^{v,\varsigma}(t)) = e$, $X^{v,\varsigma}_\e(t)$ satisfies the ODE
  \begin{equation}\label{variation-diff-2}
    \frac{\pt}{\pt\e}X^{v,\varsigma}_\e(t) = \Gamma(X^{v,\varsigma}_\e)_{\tau^{\varsigma}_\e(0)}^t v(t), \quad X^{v,\varsigma}_0(t) = X(t).
  \end{equation}
  Define a functional $\mathcal I: \A_g([0,T];q, \mu;e) \to \R$ by
  \begin{equation*}
    \mathcal I[X,\tau] := 
    \E \int_{\tau(0)}^{\tau(T)} A_0\left(t, X(t), D_\nabla X(t) \right) dt.
  \end{equation*}
  The pair $(X,\tau)\in \A_g([0,T];q, \mu;e)$ is called a stationary point of $\mathcal I$, if
  \begin{equation*}
    \frac{d}{d\e}\bigg|_{\e=0} \mathcal I[X_\e^{v,\varsigma},\tau^{\varsigma}_\e] = 0, \quad \text{for all } v\in \mathcal H([0,T];q) \text{ and } \varsigma\in\C^1([0,T],\R).
  \end{equation*}
\end{definition}

As in Lemma \ref{variation-aspt}, it is easy to deduce from \eqref{variation-diff-2} that $QX^{v,\varsigma}_\e(t) = \check g(X^{v,\varsigma}_\e(t))$ for each $t\in[\tau^{\varsigma}_\e(0),\tau^{\varsigma}_\e(T)]$ so that $X^{v,\varsigma}_\e\in \A_g([0,T];q, \mu;e)$. Moreover, formula \eqref{variation-deriv} still holds for all $t\in[\tau(0),\tau(T)]$, with $X^{v,\varsigma}_\e$ in place of $X^v_\e$.

\begin{lemma}\label{variation-formula}
  Keep the notations in Definition \ref{variation-2}. Then, in normal coordinates $(x^i)$ we have
  \begin{equation*}
    \frac{\pt}{\pt\e}\bigg|_{\e=0} \E\left[ (X^{v,\varsigma}_\e)^i(\tau^{\varsigma}_\e(s)) \big| \Pred_{\tau(s)} \right] = \left( \Gamma(X)_{\tau(0)}^{\tau(s)} v(\tau(s)) \right)^i + \varsigma(s) (D_\nabla X)^i (\tau(s)).
  \end{equation*}
\end{lemma}

\begin{proof}
  Without loss of generality, we assume $\tau^{\varsigma}_\e(s)\ge \tau(s)$. It follows from \eqref{variation-diff-2} and Definition \ref{def-vector-deriv} that
  \begin{equation*}
    \begin{split}
      \text{LHS} &= \lim_{\e\to 0} \E\left[ \frac{(X^{v,\varsigma}_\e)^i(\tau^{\varsigma}_\e(s)) - X^i(\tau(s))}{\e} \bigg| \Pred_{\tau(s)} \right] \\
      &= \lim_{\e\to 0} \E\left[ \frac{(X^{v,\varsigma}_\e)^i(\tau^{\varsigma}_\e(s)) - X^i(\tau^{\varsigma}_\e(s))}{\e} \bigg| \Pred_{\tau(s)} \right] + \lim_{\e\to 0} \E\left[ \frac{X^i(\tau^{\varsigma}_\e(s)) - X^i(\tau(s))}{\tau^{\varsigma}_\e(s)- \tau(s)} \bigg| \Pred_{\tau(s)} \right] \varsigma(s) \\
      &= \text{RHS}.
    \end{split}
  \end{equation*}
  Done.
\end{proof}

\begin{theorem}[Stochastic  Maupertuis's principle]
  Let $L_0$ be a regular Lagrangian on $\R\times TM$. Let $X\in I_{(0,q)}^{(T,\mu)}(M)$ such that $(X,\id_{[0,T]})\in \A_g([0,T];q,\mu;e)$. Then, the pair $(X,\id_{[0,T]})$ is a stationary point of $\mathcal I$ if and only if $X$ satisfy the stochastic Euler-Lagrange equation \eqref{stoch-EL}.
\end{theorem}

\begin{proof}
  Since all diffusions in $\A_g([0,T];q, \mu;e)$ have the same average energy $e$, we have
  \begin{equation*}
    \mathcal I[X,\tau] := \E \int_{\tau(0)}^{\tau(T)} [ L_0\left(t, X(t), D_\nabla X(t) \right) + e ] dt.
  \end{equation*}
  Denote $V(t) = \Gamma(X)_0^t v(t)$. As in \eqref{eqn-18},
  \begin{equation*}
    \begin{split}
      \frac{d}{d\e}\bigg|_{\e=0} I[X_\e^{v,\varsigma},\tau^{\varsigma}_\e] &= \E \int_0^T \frac{d}{d\e}\bigg|_{\e=0} L_0\left(t, X^{v,\varsigma}_\e(t), D_\nabla X^{v,\varsigma}_\e(t) \right) dt + \varsigma(t)\E[ L_0\left(t, X(t), D_\nabla X(t) \right) + e ]\big|_0^T \\
      &= \E \int_0^T \left[ d_x L_0 \left( V(t)\right) + d_{\dot x} L_0 \left( \Gamma(X)_0^t \dot v(t) \right) + \frac{1}{2} (QX)^{ij}(t) d_{\dot x} L_0 \left( R(V(t),\pt_i) \pt_j \right) \right] dt \\
      &\quad + \varsigma(t)\E[ L_0\left(t, X(t), D_\nabla X(t) \right) + e ]\big|_0^T.
    \end{split}
  \end{equation*}
  We apply \eqref{eqn-19} and notice that in the present situation we do not have $v(0) = v(T) = 0$ in general. Hence,
  \begin{equation*}
    \begin{split}
      \E \int_0^T d_{\dot x} L_0 \left( \Gamma(X)_0^t \dot v(t) \right) dt &= \E \int_0^T \Gamma(X)_t^0 (d_{\dot x} L_0) \left( \dot v(t) \right) dt \\
      &= \E[ d_{\dot x} L_0 \left( V(t) \right) ] \big|_0^T -\E \int_0^T \frac{\D}{dt} (d_{\dot x} L_0) \left( V(t) \right) dt.
    \end{split}
  \end{equation*}
  One the other hand, since for all $\e$, $X_\e^{v,\varsigma}(\tau^{\varsigma}_\e(0))=q$ and $\P\circ(X_\e^{v,\varsigma}(\tau^{\varsigma}_\e(T)))^{-1}=\mu$. It follows from Lemma \ref{variation-formula} that
  \begin{equation*}
    V(s) + \varsigma(s) D_\nabla X (s) = 0, \quad \text{for } s=0 \text{ or } s=T.
  \end{equation*}
  Therefore,
  \begin{equation*}
    \begin{split}
      \frac{d}{d\e}\bigg|_{\e=0} I[X_\e^{v,\varsigma},\tau^{\varsigma}_\e] &= \E \int_0^T \left( d_x L_0 - \frac{\overline\D}{dt} (d_{\dot x} L_0)\right) \left( V(t) \right) dt \\
      &\quad + \varsigma(t)\E\left[ L_0\left(t, X(t), D_\nabla X(t) \right) - (d_{\dot x} L_0) \left( D_\nabla X(t) \right) + e \right]\big|_0^T.
    \end{split}
  \end{equation*}
  By the definition of the energy $E_0$, we know that
  $$\E\left[ L_0\left(t, X(t), D_\nabla X(t) \right) - (d_{\dot x} L_0) \left( D_\nabla X(t) \right) \right] = -\E E_0\left(t, X(t), D_\nabla X(t) \right) = -e.$$
  The result follows.
\end{proof}

\end{appendices}

\section*{Data Availability}

Our manuscript has no associated data.

\paragraph{Acknowledgements.}
We would like to thank Prof.~Ana Bela Cruzeiro and Prof.~Marc Arnaudon for their careful reading and helpful discussions, which helped us a lot especially in improving Section \ref{sec-6-2} and \ref{sec-7-2}. We also would like to thank Prof.~Maosong Xiang for his helpful suggestions and kind experience-sharing.
This paper is supported by FCT, Portugal, project PTDC/MAT-STA/28812/2017, ``Schr\"odinger's problem and optimal transport: a multidisciplinary perspective (Schr\"oMoka)''.

\footnotesize{
\bibliographystyle{MyStyle-plainnat}
\bibliography{StochasticJets-ref}
}

\end{document}